\documentclass[11pt]{article}

\usepackage[a4paper,left=20mm,top=20mm,right=20mm,bottom=25mm]{geometry}

\usepackage{graphicx}
\usepackage{mathptmx} 
\usepackage[mathscr]{euscript}
\usepackage{color}
\usepackage[noend]{algorithmic}
\usepackage{algorithm}

\usepackage{mathtools}
\usepackage{amsfonts}
\usepackage{amssymb,amsthm,amsmath}

\usepackage{natbib}
\usepackage{authblk}
\usepackage{hyperref}

\usepackage[normalem]{ulem}
\let\cite\citep

\makeatletter
\def\moverlay{\mathpalette\mov@rlay}
\def\mov@rlay#1#2{\leavevmode\vtop{%
   \baselineskip\z@skip \lineskiplimit-\maxdimen
   \ialign{\hfil$\m@th#1##$\hfil\cr#2\crcr}}}
\newcommand{\charfusion}[3][\mathord]{
    #1{\ifx#1\mathop\vphantom{#2}\fi
        \mathpalette\mov@rlay{#2\cr#3}
      }
    \ifx#1\mathop\expandafter\displaylimits\fi}
\makeatother

\newcommand{\cupdot}{\charfusion[\mathbin]{\cup}{\cdot}}
\newcommand{\bigcupdot}{\charfusion[\mathop]{\bigcup}{\cdot}}

\providecommand{\keywords}[1]{\textbf{\textit{Keywords: }} #1}

\DeclareMathOperator{\join}{\triangledown}
\DeclareMathOperator{\lca}{lca}
\def\arrowedvec{\mathaccent"017E}
\newcommand{\hc}{\emph{hc}}
\newcommand{\child}{\mathsf{child}}
\newcommand{\parent}{\mathsf{par}}
\newcommand{\rthin}{\mathsf{R}}
\newcommand{\sthin}{\mathsf{S}}
\newcommand{\sigmasthin}{\sigma_{/\mathsf{S}}}
\newcommand{\AX}[1]{\textnormal{#1}}
\newcommand{\G}{\arrowedvec{G}}  
\newcommand{\E}{\arrowedvec{E}}

\definecolor{special-color}{RGB}{200,100,100} 
\newcommand{\NEW}[1]{\begingroup\color{black}#1\endgroup}

\newtheorem{innercustomthm}{Theorem}
\newenvironment{ctheorem}[1]
  {\renewcommand\theinnercustomthm{#1}\innercustomthm}
  {\endinnercustomthm}

\newtheorem{innercustomlem}{Lemma}
\newenvironment{clemma}[1]
  {\renewcommand\theinnercustomlem{#1}\innercustomlem}
  {\endinnercustomlem}

\newtheorem{innercustomdef}{Definition}
\newenvironment{cdefinition}[1]
  {\renewcommand\theinnercustomdef{#1}\innercustomdef}
  {\endinnercustomdef}

\newtheorem{innercustomcor}{Corollary}
\newenvironment{ccorollary}[1]
  {\renewcommand\theinnercustomcor{#1}\innercustomcor}
  {\endinnercustomcor}

\newtheorem{innercustomfact}{Observation}
\newenvironment{cfact}[1]
  {\renewcommand\theinnercustomfact{#1}\innercustomfact}
  {\endinnercustomfact}

\newtheorem{theorem}{Theorem}
\newtheorem{corollary}{Corollary}
\newtheorem{proposition}{Proposition}
\newtheorem{lemma}{Lemma}
\newtheorem{definition}{Definition}
\newtheorem{fact}{Observation}
\newtheorem{example}{Example}
\newtheorem{remark}{Remark}

\begin{document}

	\title{Reciprocal Best Match Graphs}
	
	\author[1]{Manuela Gei{\ss}}
	\author[1,2,3,4,5,6,7]{Peter F.\ Stadler}
	\author[8,9]{Marc Hellmuth}

	\affil[1]{Bioinformatics Group, Department of Computer Science; and 
		Interdisciplinary Center of Bioinformatics, University of Leipzig,
		H{\"a}rtelstra{\ss}e 16-18, D-04107 Leipzig, Germany}
	\affil[2]{German Centre for Integrative Biodiversity Research (iDiv)
		Halle-Jena-Leipzig} 
	\affil[3]{Competence Center for Scalable Data Services
		and Solutions} 
	\affil[4]{Leipzig Research Center for Civilization Diseases,
		Leipzig University,
		H{\"a}rtelstra{\ss}e 16-18, D-04107 Leipzig}
	\affil[5]{Max-Planck-Institute for Mathematics in the Sciences,
		Inselstra{\ss}e 22, D-04103 Leipzig}
	\affil[6]{Inst.\ f.\ Theoretical Chemistry, University of Vienna,
		W{\"a}hringerstra{\ss}e 17, A-1090 Wien, Austria}
	\affil[7]{Santa Fe Institute, 1399 Hyde Park Rd., Santa Fe,
		NM 87501, USA}
	\affil[8]{Institute of Mathematics and Computer Science, 
		University of Greifswald, Walther-Rathenau-Stra{\ss}e 47, 
		D-17487 Greifswald, Germany}
	\affil[9]{Center for Bioinformatics, Saarland University, Building E 2.1, 
		P.O.\ Box 151150, D-66041 Saarbr{\"u}cken, Germany}

	\date{}
	\normalsize
	
	\maketitle

\begin{abstract} 
  Reciprocal best matches play an important role in numerous applications
  in computational biology, in particular as the basis of many widely used
  tools for orthology assessment. Nevertheless, very little is known about
  their mathematical structure. Here, we investigate the structure of
  reciprocal best match graphs (RBMGs). In order to abstract from the
  details of measuring distances, we define reciprocal best matches here as
  pairwise most closely related leaves in a gene tree, arguing that
  conceptually this is the notion that is pragmatically approximated by
  distance- or similarity-based heuristics. We start by showing that a
  graph $G$ is an RBMG if and only if its quotient graph w.r.t.\ a certain
  thinness relation is an RBMG. Furthermore, it is necessary and sufficient
  that all connected components of $G$ are RBMGs.  The main result of this
  contribution is a complete characterization of RBMGs with 3
  colors/species that can be checked in polynomial time. For 3 colors,
  there are three distinct classes of trees that are related to the
  structure of the phylogenetic trees explaining them. We derive an
  approach to recognize RBMGs with an arbitrary number of colors; it
  remains open however, whether a polynomial-time for RBMG recognition
  exists. In addition, we show that RBMGs that at the same time are
  cographs (co-RBMGs) can be recognized in polynomial time. Co-RBMGs are
  characterized in terms of hierarchically colored cographs, a particular
  class of vertex colored cographs that is introduced here.  The (least
  resolved) trees that explain co-RBMGs can be constructed in polynomial
  time.

\end{abstract} 

\sloppy

\section{Introduction}

An important task in computational biology is the annotation of newly
sequenced genomes, and in particular to establish correspondences between
orthologous genes. Two genes $x$ and $y$ in two different species $s$ and
$t$, respectively, are orthologs if their last common ancestor was the
speciation event that separated the lineages of $s$ and $t$
\cite{Fitch:00}. A large class of software tools for orthology assignment
are based on the pairwise reciprocal best match heuristic: the genes $x$
and $y$ are (candidate) orthologs if $y$ is a best match in terms of
sequence similarity to $x$ among all genes from species $t$, and $x$ is a
best match to $y$ among all genes from $s$. This approach, which has been
termed Symmetric Best Match \cite{Tatusov:97}, bidirectional best hits
(BBH) \cite{Overbeek:99}, reciprocal best hits (RBH) \cite{Bork:98}, or
reciprocal smallest distance (RSD) \cite{Wall:03}, provides orthology
assignments on par with elaborate phylogeny-based methods, see
\citet{Altenhoff:09,Altenhoff:16,Setubal:18a} for reviews and benchmarks.

The application of pairwise best hits methods to orthology detection relies
on the observation that given a gene $x$ in species $r$ (and disregarding
horizontal gene transfer), all its co-orthologous genes $y$ in species $s$
are by definition closest relatives of $x$. Since orthology is a symmetric
relation, orthologs are necessarily reciprocal best matches (RBMs).  In
practice, however, reciprocal best matches are approximated by sequence
similarities, making the tacit assumption that the molecular clock
hypothesis is not violated too dramatically, see \citet{Geiss:18x} for a
more detailed discussion. Modern orthology detection tools are well aware
of the shortcoming of pairwise best sequence similarity estimates and
employ additional information, in particular synteny
\cite{Lechner:14a,Jahangiri:17}, or use small subsets of pairwise matches
to identify erroneous orthology assignments \cite{Yu:11,Train:17}. In this
contribution we will not concern ourselves with the practicalities of
inferring best matches. Instead, we will focus on the best match relation
from a mathematical point of view.

Despite its practical importance, very little is known about the RBM
relation. Like orthology, RBMs are a phylogenetic concept and thus refer to
a phylogenetic tree $T$. Denote by $L$ the leaf set of $T$ and consider a
surjective map $\sigma: L\to S$ that assigns to each gene $x\in L$ the
species $\sigma(x)\in S$ within which $x$ resides. To avoid trivial cases,
we assume that there are $|S|>1$ species. In this setting we can express
the concept of RBMs by
\begin{definition}
  \label{def:rbm}
  The leaf $y$ is a \emph{best match} of the leaf $x$ in the gene tree $T$
  if and only if $\lca(x,y)\preceq \lca(x,y')$ for all leaves $y'$ with
  $\sigma(y')=\sigma(y)$. If $x$ is also a best match of $y$, we call $x$
  and $y$ \emph{reciprocal best matches}.
\end{definition}
As usual, $\lca_T(x,y)$ denotes the last common ancestor of $x$ and $y$ on
$T$, and $\preceq_T$ is the ancestor order on the vertices of $T$, where
the root of $T$ is the unique maximal element. We omit the index $T$
whenever the context is clear.  The reciprocal best match relation is
symmetric by definition. It is reflexive because every gene $x$ in species
$s$ is its own (unique) best match within $s$.

The best match relation and the reciprocal best match relation are
conveniently represented as a vertex colored digraph $(\G,\sigma)$ and
vertex colored undirected graph $(G,\sigma)$, respectively, with vertex set
$L$. Arcs and edges represent best matches and reciprocal best matches,
respectively. Since there is a 1-1 relationship between graphs with a loop
at each vertex and graphs without loops, consider $(\G,\sigma)$ and
$(G,\sigma)$ as loop-less.  The relationship between these graphs and the
trees from which they are derived is captured by
\begin{definition}
\label{def:explains}    
Let $(T,\sigma)$ be a tree with leaf set $L$, let $\G=(L,\E)$ be a digraph
and $G=(L,E)$ an undirected graph, both with vertex set $L$, and let
$\sigma:L\to S$ be a map with $|S|\ge 2$.  Then, $(T,\sigma)$
\emph{explains} $(\G,\sigma)$ if there is an arc $(x,y)\in\E$ in $\G$
precisely if $y$ is a best match of $x$ in $(T,\sigma)$ with
$\sigma(x)\ne\sigma(y)$.  Analogously, $(T,\sigma)$ \emph{explains}
$(G,\sigma)$ if there is an edge $xy\in E$ in $G$ precisely if $x$ and $y$
form a reciprocal best match in $(T,\sigma)$ with $\sigma(x)\ne\sigma(y)$.
\end{definition}
Def.~\ref{def:explains} gives rise to two classes of vertex colored graphs:
\begin{definition}
  A vertex colored digraph $(\G,\sigma)$ is a \emph{best match
    graph} (BMG) if there exists a leaf-colored tree $(T,\sigma)$ that
  explains $(\G,\sigma)$. An undirected graph $(G,\sigma)$ is a
  \emph{reciprocal best match graph} (RBMG) if there exists a
  leaf-colored tree $(T,\sigma)$ that explains $(G,\sigma)$.
\end{definition}
For BMGs we recently reported two different characterizations and
corresponding polynomial-time recognition algorithms \cite{Geiss:18x}.

Here we extend the analysis to RBMGs. Since the material is rather
extensive (many of the proofs use elementary graph theory but are very
technical) we subdivided the presentation in a main narrative text
explaining the main results and a second technical part collecting the
proofs of the main results as well as additional technical results that are
useful for later more practical applications. \NEW{In order to ensure that
  the second, technical part of the contribution is self-contained all
  definitions and results are (re)stated there. In the narrative part we
  only give those definitions that are necessary to understand the results
  presented there. We use the same numbering of statements in the narrative
  and the technical part to facilitate the cross-referencing.}

We start in Section \ref{sect:prelim} to define the concepts that we need
here. We follow the general strategy to reduce redundancy by identifying
classes of trees explaining the same graphs and equivalence classes of
graphs explained by trees with essentially the same structure. We start in
Sections \ref{sect:leastres-new} in the main text and \ref{sect:leastres}
in the technical part with the description of least resolved trees as
representatives that are sufficient to explain a given RBMG. As it turns
out, least resolved trees are not unique in general.  Complementarily, in
Sections \ref{sect:thin-new} and \ref{sect:thin} we introduce a color-aware
thinness relation $\sthin$ and show that it suffices to characterize
$\sthin$-thin RBMGs.  Combining these ideas, we demonstrate in Sections
\ref{sect:connect-new} and \ref{sect:connect} that $(G,\sigma)$ is an RBMG
if and only if each of its connected components is an RBMG and at least one
of them contains all colors, and give a simple construction for a tree
explaining $(G,\sigma)$ from trees for the connected components. In order
to characterize connected, $\sthin$-thin RBMGs, we first consider the case
of three colors (Sections \ref{sect:classes-new} and \ref{sect:classes}).
We find that there are three distinct classes of 3-RBMGs that can be
recognized in polynomial time.  One of these classes does not contain
induced paths on four vertices, so called $P_4$s, while the other two
classes do. Since $P_4$s are at the heart of cograph editing approaches to
improve orthology estimates \cite{Hellmuth:13a}, we consider these
structures in some more detail in Section \ref{sect:P4} of the technical
part and characterize three distinct types: good, bad, and ugly $P_4$s.  In
Sections \ref{sect:n-new} and \ref{sect:n} we prove that trees explaining
an $n$-RBMG can be composed from tree sets explaining the induced 3-RBMGs
for all three-color subsets. However, the computational complexity for
recognizing $n$-RBMGs is left as an open problem.  Because of their
practical relevance in orthology detection, we then characterize the
$n$-RBMGs that are so-called cographs.  As we shall see, the recognition of
cograph $n$-RBMGs and the construction of trees that explain them can be
done in polynomial time. We finish with a brief survey of potential
applications of the results presented here and some open problems.

\section{Preliminaries}
\label{sect:prelim}

Throughout this contribution, we say that two sets $P,Q$ \emph{do not
  overlap} if $P\cap Q\in \{\emptyset, P, Q\}$, and they \emph{overlap},
otherwise. We will also assume throughout that the map $\sigma:L\to S$ is
surjective.  For a subset $L'\subseteq L$ we write
$\sigma(L')=\{\sigma(x)\mid x\in L'\}$.  Moreover, we use the notation
$\sigma_{|L'}$ for the surjective map $\sigma:L'\to \sigma(L')$.  We will
frequently need to refer to the number $|S|$ of colors and often speak of
$|S|$-BMGs and $|S|$-RBMGs.

A \emph{phylogenetic tree $T=(V,E)$ (on $L$)} is a rooted tree with root
$\rho_T$, leaf set $L(T) = L\subseteq V$ and inner vertices
$V^0=V\setminus L$ such that each inner vertex of $T$ (except possibly the
root) is of degree at least three.  For $x\in V$, we denote by $T(x)$ the
subtree rooted at $x$.  For a phylogenetic tree $T$ on $L$, the
\emph{restriction $T_{|L'}$ of $T$ to $L'\subseteq L$} is the phylogenetic
tree with leaf set $L'$ that is obtained from $T$ by first taking the
minimal subtree of $T$ with leaf set $L'$ and then suppressing all vertices
of degree two with the exception of the root $\rho_{T_{|L'}}$.
	
\emph{Throughout this contribution all rooted
trees are assumed to be phylogenetic unless explicitly stated otherwise.}

A vertex $u\in V$ is an \emph{ancestor} of $v\in V$, $u\succeq_T v$, and
$v$ is a \emph{descendant} of $u$, $v\preceq_T u$, if $u$ lies on the
unique path from $v$ to the root $\rho_T$. We write $u\succ_T v$
($v\prec_T u$) for $u\succeq_T v$ ($v\preceq_T u$) and $u\neq v$.  For a
subset $\alpha\subseteq V$ we write $\alpha \preceq_T u$ to mean that
$x \preceq_T u$ for all $x\in \alpha$.  If $uv\in E$ and $u\succ_T v$, we
call $u$ the \emph{parent} of $v$, denoted by $\parent(v)$, and define the
\emph{children} of $u$ as $\child(u)\coloneqq \{v\in V\mid uv\in E\}$.  It
will be convenient to use the notation $uv\in E$ to indicate $u\succ_T v$,
i.e., $u$ is closer to the root. Moreover, we say that $e=uv$ is an
\emph{outer} edge if $v\in L$ and an \emph{inner} edge otherwise.

A tree $T'$ is displayed by $T$, denoted by $T'\le T$, if $T'$ can be
obtained from a subtree of $T$ by a series of edge contractions. For a tree
$T$ on $L$ with coloring map $\sigma: L \to S$, in symbols $(T,\sigma)$, we
say that $(T,\sigma)$ \emph{displays} or is a \emph{refinement} of
$(T',\sigma')$ if $T'\le T$ and $\sigma'(v)=\sigma(v)$ for any
$v\in L(T')\subseteq L$. The subtree $T(v)$ has leaf set
$L'\coloneqq L(T(v))$ and leaf coloring
$\sigma_{L'}: L'\to \sigma(L')$. We write $\lca_T(A)$ for the
\emph{last common ancestor} of all elements of a set $A\subseteq V$ of
vertices.  For a tree $T=(V,E)$ and some inner edge $e=uv$ of $T$, we
denote by $T_e$ the tree that is obtained from $(T,\sigma)$ by contraction
of $e$, i.e., by identifying $u$ and $v$. Analogously, $T_A$ is obtained by
contracting a sequence of edges $A=(e_1,\dots,e_k)\subseteq E$.

A \emph{triple} is a binary tree on three leaves. We write $xy|z$ if the
path from $x$ to $y$ does not intersect the path from $z$ to the root. A
set $\mathcal{R}$ of triples is consistent if there is a tree $T$ that
displays every triple in $\mathcal{R}$. Analogously, we say that a set of
trees $\mathcal{T}$ is consistent it there is a tree $T$ such that $T$
displays every tree $T'\in\mathcal{T}$. Consistency of a set of triples
$\mathcal{R}$ and more generally trees $\mathcal{T}$ can be decided in
polynomial time by explicitly constructing a supertree $T$ \cite{Aho:81}.

In the following, $G=(V,E)$ and $\G=(V,\E)$ denote simple undirected and
simple directed graphs, respectively.  Throughout, we will distinguish
directed arcs $(x,y)$ in a digraph $\G$ from edges $xy$ in an undirected
graph $G$ or tree $T$. For $x\in V$ we write
$N^+(x)\coloneqq \{y\in V\mid (x,y)\in \E\}$ for its out-neighborhood and
$N^-(x)\coloneqq \{y\in V\mid(y,x)\in \E\}$ for its in-neighborhood.  The
notation naturally extends to sets of vertices $A\subseteq V$:
$N^+(A)=\bigcup_{x\in A}N^+(x)$ and $N^-(A)=\bigcup_{x\in A}N^-(x)$.  Two
vertices $x$ and $y$ of $\G$ are in relation $\rthin$ if $N^+(x)=N^+(y)$
and $N^-(x)=N^-(y)$, see e.g.\ \cite{Hellmuth:15}. Obviously $\rthin$ is an
equivalence relation. For each $\rthin$-class $\alpha$ and every
$x\in \alpha$ holds $N^+(x)=N^+(\alpha)$ and $N^-(x)=N^-(\alpha)$.  The set
of all $\rthin$-classes of $(\G,\sigma)$ will be denoted by $\mathcal{N}$,
or $\mathcal{N}(\G)$.

For a colored di-graph $(\G,\sigma)$, we write
$N_s^+(x)\coloneqq \{y\in N^+(x)\mid\sigma(y)=s\}$ and
$N_s^-(x)\coloneqq \{y\in N^-(x)\mid\sigma(y)=s\}$.  Similarly, for an
undirected colored graph $(G,\sigma)$ with $G=(V,E)$, we write
$N(x)\coloneqq \{y\in V\mid xy\in E\}$ for the neighborhood of some vertex
$x\in V$. Moreover, we set $N_s(x)\coloneqq \{y\in N(x)\mid\sigma(y)=s\}$.

A \emph{connected component} of a graph $(G,\sigma)$ is a maximal connected
subgraph of $(G,\sigma)$. A digraph is connected whenever its underlying
undirected graph (obtained by ignoring the direction of the arcs) is
connected. For our purposes it will not be relevant to distinguish two
colored graphs $(G,\sigma)$ and $(G',\sigma')$ that are isomorphic in the
usual sense of isomorphic colored graphs, i.e., isomorphic graphs $G$ and
$G'$ that only differ by a permutation of their vertex-coloring. A vertex
coloring $\sigma$ is \emph{proper} if $xy\in E(G)$ implies
$\sigma(x)\ne\sigma(y)$. As an immediate consequence of Def.\
\ref{def:explains} we have
\begin{fact}
  \label{obs:proper}
  If $(G,\sigma)$ is an RBMG, then $\sigma$ is a proper vertex coloring.
\end{fact}
As a consequence, $(G,\sigma)$ cannot be explained by a leaf-colored tree
unless $\sigma$ is a proper vertex coloring. We may therefore assume
throughout this contribution that $(G,\sigma)$ is a properly vertex colored
graph. Moreover, for $W\subseteq V(G)$ we denote with G[W] the
\emph{induced} subgraph of $G$ and put
$(G,\sigma)[W]\coloneqq (G[W],\sigma_{|W})$.

We write $\langle x_1\dots x_k\rangle \in \mathscr{P}_k$ to denote that the
vertices $x_1,\dots, x_k$ form an induced path
$P=\langle x_1\dots x_k \rangle$ on $k$ vertices and with edges
$x_ix_{i+1}$, $1\leq i \leq k-1$. Analogously,
$\langle x_1\dots x_k \rangle \in \mathscr{C}_k$ denotes the fact that the
vertices $x_1,\dots, x_k$ induce a cycle $C=\langle x_1\dots x_k \rangle$
on $k$ vertices with edges $x_ix_{i+1}$, $1\leq i \leq k-1$, and $x_kx_1$.
An induced cycle on six vertices is called \emph{hexagon}.  We will write
that $\langle x_1\dots x_k\rangle\in\mathscr{P}_k$, resp.  $\mathscr{C}_k$
is of the form $(\sigma(x_1), \dots, \sigma(x_k))$ to indicate the vertex
colors along induced paths, resp., cycles.

\emph{Cographs} form a class of undirected graphs that play an important in
the context of this contribution. They are defined recursively
\cite{Corneil:81}:
\begin{definition}
  An undirected graph $G$ is a cograph if
 \begin{itemize}
  \item[(1)] $G=K_1$
  \item[(2)] $G= H \join H'$, where $H$ and $H'$ are cographs and
    $\join$ denotes the join,
  \item[(3)] $G= H \cupdot H'$, where $H$ and $H'$ are cograph and $\cupdot$
    denotes the disjoint union.
\end{itemize}
The join of two disjoint graphs $H=(V,E)$ and $H'=(V',E')$ is defined by
$H \join H'=(V\cup V',E\cup E' \cup \{xy\mid x\in V, y\in V'\})$, whereas
their disjoint union is given by $H \cupdot H'=(V\cup V', E\cup E')$.
\label{def:cograph}
\end{definition}
A graph is a cograph if and only if does not contain an induced $P_4$
\cite{Corneil:81}.

Each cograph $G$ is associated with \emph{cotrees} $T_G$, that is,
phylogenetic trees with internal vertices labeled by $0$ or $1$, whose
leaves correspond to the vertices of $G$. In $T_G$, each subtree rooted at
an internal vertex $x$ with label $0$ corresponds to the disjoint union of
the subgraphs of $G$ induced by the leaf sets $L(T_G(y))$ of the children
$y\in\child(x)$ of $x$, and each subtree rooted at an $x$ with label $1$
corresponds to the join of the subgraphs of $G$ induced by the $L(T_G(y))$,
$y\in\child(x)$. In other words, $(T,t)$ is a cotree for $(G,\sigma)$, if
$t(\lca_T(x, y)) = 1$ if and only if $xy \in E(G)$.  For each cograph $G$
there is a unique \emph{discriminating} cotree $T_G$ with the property that
the labels $0$ and $1$ alternate along each root-leaf path in $T_G$
\cite{Corneil:81}.  For later reference, we summarize here some of the
results in \cite[Sect.\ 3]{Hellmuth:13a}.
\begin{proposition}
  Any cotree of a cograph $G$ is a refinement of the unique discriminating
  cotree of $G$.  In particular $(T_e,t_e)$ is a cotree for a cograph $G$
  if and only if $(T,t)$ is a cotree for $G$, $e=xy$ is an edge with
  $t(x)=t(y)$ that is contracted to vertex $v_e$ in $T_e$ and
  $t_e(v_e)=t(x)$ and $t_e(v)=t(v)$ for all remaining vertices.
\label{prop:cograph2}
\end{proposition}

\section{Least Resolved Trees}
\label{sect:leastres-new}

In this section we consider a notion of ``smallest'' trees explaining a
given RBMG. A we shall see, the characterization of these trees is closely
related to the one of best matches but cannot be expressed in terms of
\textit{reciprocal} best matches alone. Throughout this work, the vertex
set of BMGs and RBMGs as well as the leaf set of the trees that
explain them will be denoted by $L$ and we will write
\begin{equation*}
  L[s]\coloneqq \{x \mid x\in L, \sigma(x)=s\}
\end{equation*}
for the subset of vertices with color $s\in S$.

Given a leaf-colored tree $(T,\sigma)$, one can easily derive the
respective BMG $\G(T,\sigma)$ and RBMG $G(T,\sigma)$ that are explained by
$(T,\sigma)$.  Conversely, if $(G,\sigma)$ is an RBMG, then there is a tree
$(T,\sigma)$ that explains $(G,\sigma)$. This tree also explains the
digraph $\G(T,\sigma)$ with the property that $xy \in E(G)$ if and only if
both $(x,y)$ and $(y,x)$ are arcs in $\G(T,\sigma)$. A colored graph
$(G,\sigma)$ therefore is an RBMG if and only if it is the symmetric part
of some BMG.
  
It is important to note that distinct trees $(T',\sigma)$ and
$(T'',\sigma)$ may explain the same RBMG, i.e.,
$G(T',\sigma)=G(T'',\sigma)$, albeit the leaf-set $L$ and the leaf-coloring
$\sigma$ of course must be the same. In general the BMGs $\G(T',\sigma)$
and $\G(T'',\sigma)$ can also be different, even if
$G(T',\sigma)=G(T'',\sigma)$.  For an example, consider the RBMG
$(G,\sigma)$ and the two distinct trees $(T_e,\sigma)$ and $(T_f,\sigma)$
in Fig.~\ref{fig:lr_nonunique}. We have
$(G,\sigma) = G(T_e,\sigma)=G(T_f,\sigma)$. However, $\G(T_e,\sigma)$
contains the arc $(a,b')$ which is not contained in
$\G(T_f,\sigma)$. Hence, $\G(T_e,\sigma) \neq \G(T_f,\sigma)$.

Nevertheless, some properties of BMGs will be helpful as a means to gaining
insights into the structure of RBMGs. To this end we briefly recall some
pertinent results by \citet{Geiss:18x}.

\begin{clemma}{\ref{lem:bmg-basics}}
  Let $(\G,\sigma)$ be a BMG with vertex set $L$.  Then, $x\rthin y$
  implies $\sigma(x)=\sigma(y)$. In particular, $(\G,\sigma)$ has no arcs
  between vertices within the same $\rthin$-class.  Moreover,
  $N^+(x)\ne\emptyset$, while the in-neighborhood $N^-(x)$ may be empty for
  all $x\in L$.
\end{clemma}

\begin{figure}
  \begin{center}
    \includegraphics[width=1\textwidth]{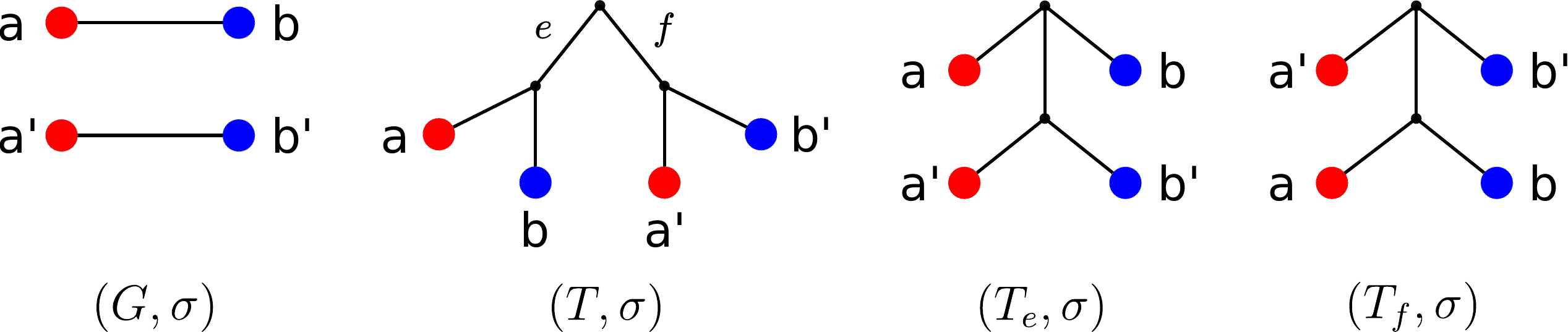}
  \end{center}
  \caption{The reciprocal best match graph $(G,\sigma)$ on two colors (red
    and blue) is explained by $(T,\sigma)$ which contains the redundant
    edges $e$ and $f$. Contraction of one of these edges gives
    $(T_e,\sigma)$ and $(T_f,\sigma)$, respectively, which are both
    least resolved \NEW{but distinct from each other}, i.e., there exists no unique least resolved
    tree w.r.t.\ $(G,\sigma)$. In particular, none of the trees
    $(T_{ef},\sigma)$ and $(T_{fe},\sigma)$ explains $(G,\sigma)$.}
  \label{fig:lr_nonunique}
\end{figure}

Although there are in general many different trees that explain the same
BMG or RBMG, it is shown by \citet{Geiss:18x} that every BMG $(\G,\sigma)$
is explained by a uniquely defined ``smallest'' tree, its so called
\emph{least resolved} tree.  The notion of least resolved trees are also of
interest for RBMGs even though we shall see below that they are not unique
in the reciprocal setting.
\begin{cdefinition}{\ref{def:redundant}}
  Let $(G,\sigma)$ be an RBMG that is explained by a tree $(T,\sigma)$. An
  inner edge $e$ is called \emph{redundant} if $(T_e,\sigma)$ also explains
  $(G,\sigma)$, otherwise $e$ is called \emph{relevant}.
\end{cdefinition}

Lemma \ref{lem:lr} in the technical part provides a characterization of
redundant edges.  It is interesting to note that this characterization of
redundancy (w.r.t.\ an RBMG) of edges in $(T,\sigma)$ requires
information on (directed) best matches and apparently cannot be expressed
entirely in terms of the reciprocal best match relation.

\begin{cdefinition}{\ref{def:series-edge-contract}}
  Let $(G,\sigma)$ be an RBMG explained by $(T,\sigma)$. Then $(T,\sigma)$
  is \emph{least resolved w.r.t.\ $(G,\sigma)$} if $(T_A,\sigma)$ does not
  explain $(G,\sigma)$ for any non-empty series of edges $A$ of
  $(T,\sigma)$.
\end{cdefinition}

Given two distinct redundant edges $e\neq f$ of $(T,\sigma)$, the edge $f$
is not necessarily redundant in $(T_e,\sigma)$, i.e., the tree
$(T_{ef},\sigma)$ obtained by sequential contraction of $e$ and $f$ does
not necessarily explain $(G,\sigma)$. This in particular implies that the
contraction of all redundant edges of $(T,\sigma)$ does not necessarily
result in a least resolved tree for the same RBMG.  Moreover, there may be
more than one least resolved tree that explains a given $n$-RBMG
$(G,\sigma)$. Fig.\ \ref{fig:lr_nonunique} gives an example of least
resolved trees that are not unique. The following theorem summarizes some
key properties of least resolved trees.
\begin{ctheorem}{\ref{thm:lr}}
  Let $(G,\sigma)$ be an RBMG explained by $(T,\sigma)$. Then there exists
  a (not necessarily unique) least resolved tree
  $(T_{e_1\dots e_k},\sigma)$ explaining $(G,\sigma)$ obtained from
  $(T,\sigma)$ by a series of edge contractions $e_1 e_2 \dots e_k$ such
  that the edge \NEW{$e_1$ is redundant in $(T,\sigma)$ and} $e_{i+1}$ is
  redundant in $(T_{e_1\dots e_i},\sigma)$ for $i\in \{1,\dots,k-1\}$. In
  particular, $(T,\sigma)$ displays $(T_{e_1\dots e_k},\sigma)$.
\end{ctheorem}

We will return to least resolved trees in Section \ref{ssec:n-RBMG}, where
the concept will be needed as a means to construct a tree explaining an
$n$-RBMG from sets of least resolved trees that explain the induced 3-RBMGs
for all subsets on three colors.

\section{$\sthin$-Thinness}
\label{sect:thin-new}

The $\rthin$ relation introduced in the previous sections is the natural
generalization of thinness in undirected graphs \cite{McKenzie:71}. As an
immediate consequence of Lemma~\ref{lem:bmg-basics}, all vertices within an
$\rthin$-class of a BMG have the same color. However, a corresponding
result does not hold for RBMGs. Fig.\ \ref{fig:sameNeighbors} shows a
counterexample, where $N(a)=N(b)$ holds for vertices with different colors
$\sigma(a)\ne\sigma(b)$. Since color plays a key role in our context, we
introduce a color-preserving thinness relation:
\begin{cdefinition}{\ref{def:sthin}} 
  Let $(G,\sigma)$ be an undirected colored graph. Then two vertices $a$
  and $b$ are in relation $\sthin$, in symbols $a\sthin b$, if
  $N(a) = N(b)$ and $\sigma(a) = \sigma(b)$.\\
  An undirected colored graph $(G,\sigma)$ is $\sthin$-thin if no two
  distinct vertices are in relation $\sthin$. We denote the $\sthin$-class
  that contains the vertex $x$ by $[x]$.
\end{cdefinition}

As a consequence of Lemma \ref{lem:bmg-basics} and the fact that every
RBMG $(G,\sigma)$ is the symmetric part of some BMG $\G(T,\sigma)$, we
obtain
\begin{clemma}{\ref{lem:SRthin}}
  Let $(G,\sigma)$ be an RBMG, $(T,\sigma)$ a tree explaining $(G,\sigma)$,
  and $\G(T,\sigma)$ the corresponding BMG. Then $x\rthin y$ in
  $\G(T,\sigma)$ implies that $x\sthin y$ in $(G,\sigma)$.
\end{clemma}
The converse of Lemma \ref{lem:SRthin} is not true, however. In Fig.\
\ref{fig:sameNeighbors}, for instance, changing the color of the leaf $b_3$
from blue to red in the tree $(T,\sigma)$ implies $N(a_2)=N(b_3)$ in the
RBMG $(G,\sigma)$ and the set $\{a_2,b_3\}$ forms an $\sthin$-class. On
the other hand, we have $N^+(a_2)\neq N^+(b_3)$ in the corresponding BMG
$\G(T,\sigma)$, thus $a_2$ and $b_3$ do not belong to the same
$\rthin$-class of $\G(T,\sigma)$.

\begin{figure}
  \begin{center}
    \includegraphics[width=0.7\textwidth]{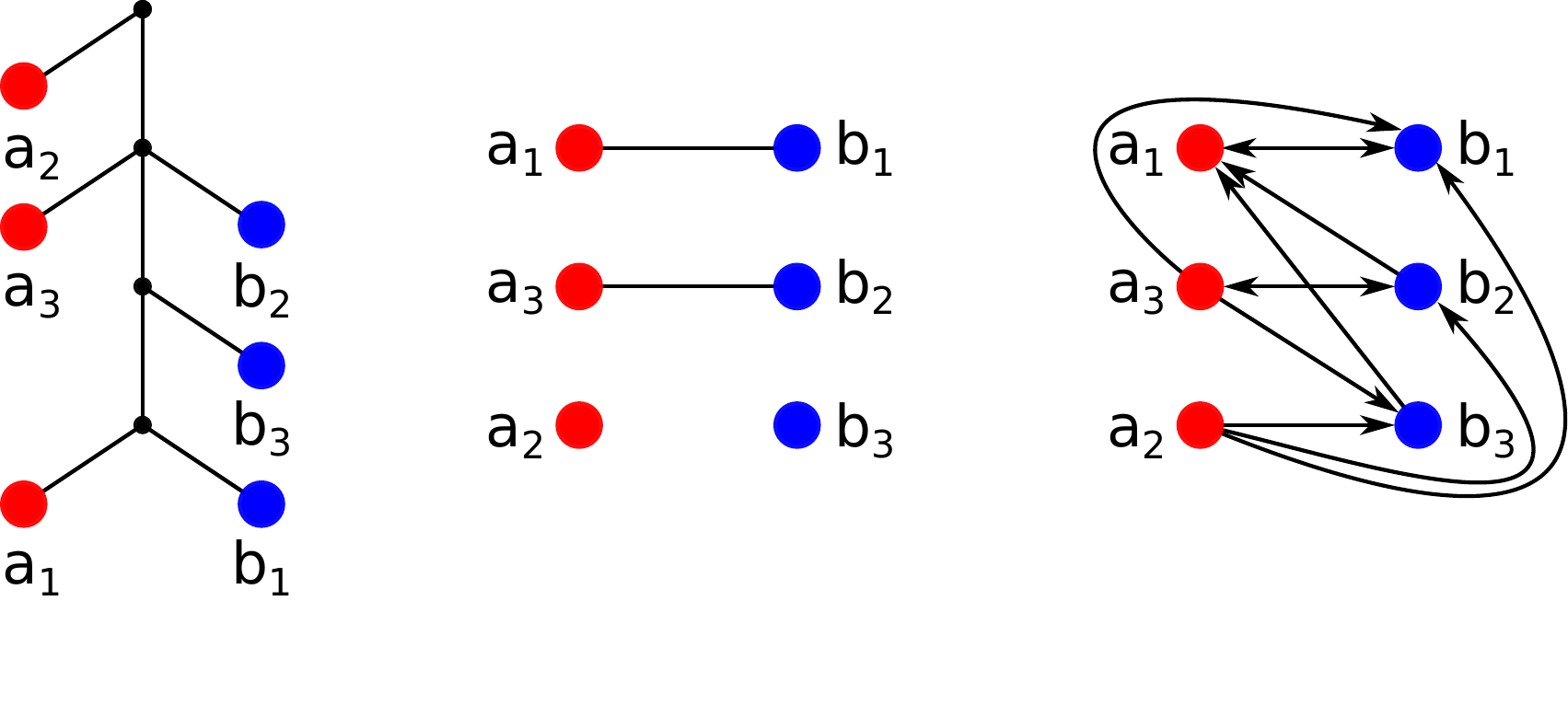}
  \end{center}
  \caption{The leaf-colored tree $(T,\sigma)$ on the left explains the RBMG
    $G(T,\sigma)$ (middle) and the BMG $\G(T,\sigma)$ (right). The colored
    graph $\G(T,\sigma)$ is $\rthin$-thin.  Thus, all leaves within an
    $\rthin$-class are trivially of the same color.  However, in the RBMG
    we have $N(a_2)=N(b_3)=\emptyset$ but $a_2$ and $b_3$ are of different
    color. Note, by definition $a_2$ and $b_3$ are not within the same
    $\sthin$-class.}
  \label{fig:sameNeighbors}
\end{figure}

For an undirected colored graph $(G,\sigma)$, we denote by $G/\sthin$ the
graph whose vertex set are exactly the $\sthin$-classes of $(G,\sigma)$,
and two distinct classes $[x]$ and $[y]$ are connected by an edge in
$G/\sthin$ if there is an $x'\in [x]$ and $y'\in [y]$ with $x'y'\in
E(G)$. Moreover, since the vertices within each $\sthin$-class have the
same color, the map $\sigmasthin \colon V(G/\sthin) \to S$ with
$\sigmasthin([x]) = \sigma(x)$ is well-defined.

\begin{clemma}{\ref{lem:sthin}}
  $(G/\sthin, \sigmasthin)$ is $\sthin$-thin for every undirected colored
  graph $(G,\sigma)$. Moreover, $xy\in E(G)$ if and only if
  $[x][y] \in E(G/\sthin)$. Thus, $G$ is connected if and only if
  $G/\sthin$ is connected.
\end{clemma}

The map $\gamma_{\sthin} \colon V(G) \to V(G/\sthin): x\mapsto [x]$
collapses all elements of an $\sthin$-thin class in $(G,\sigma)$ to a
single node in $(G_{/\sthin},\sigmasthin)$. Hence, the
$\gamma_{\sthin}$-image of a connected component of $(G,\sigma)$ is a
connected component in $(G_{/\sthin},\sigmasthin)$. Conversely, the
pre-image of a connected component of $(G_{/\sthin},\sigmasthin)$ that
contains an edge is a single connected component of
$(G,\sigma)$. Furthermore, $(G_{/\sthin},\sigmasthin)$ contains at most one
isolated vertex of each color $r\in S$. If it exists, then its pre-image is
the set of all isolated vertices of color $r$ in $(G,\sigma)$; otherwise
$(G,\sigma)$ has no isolated vertex of color $r$. Surprisingly, it suffices
to characterize the $\sthin$-thin RBMGs:
\begin{clemma}{\ref{lem:Sthin-tree}}
  $(G,\sigma)$ is an RBMG if and only if $(G/\sthin, \sigmasthin)$ is an
  RBMG.  Moreover, every RBMG $(G,\sigma)$ is explained by a tree
  $(\widehat{T},\sigma)$ in which any two vertices $x,x'\in [x]$ of each
  $\sthin$-classes $[x]$ of $(G,\sigma)$ have the same parent.
\end{clemma}

Lemma~\ref{lem:Sthin-tree} is illustrated in Fig.\ \ref{fig:SthinTree},
where $a_2$ and $a_4$ belong to the same $\sthin$-thin class
$[a_2]$. However, in the tree representation on the l.h.s., $a_2$ and $a_4$
are attached to different parents. Substituting the edge $\parent(a_4)a_4$
by $\parent(a_2)a_4$ and suppressing the vertex $\parent(a_4)$, which now
has degree $2$, yields a tree $(\hat T,\sigma)$ with
$\parent(a_2)=\parent(a_4)$ that still explains $(G,\sigma)$. Next, we can
remove the edges $\parent(a_2)a_2$ and $\parent(a_2)a_4$ as well as the
leaves $a_2$ and $a_4$ from $(\hat T,\sigma)$ and add the edge
$\parent(a_2)[a_2]$. Finally, we replace any vertex $y\neq a_2,a_4$ by
$[y]$ and set $\sigma(x)=\sigma_{/\sthin}([x])$ for all $x\in V(\hat
T)$. The resulting tree explains the $\sthin$-thin RBMG
$(G_{/\sthin},\sigma_{/\sthin})$.

\begin{figure}
  \begin{center}
    \includegraphics[width=0.9\textwidth]{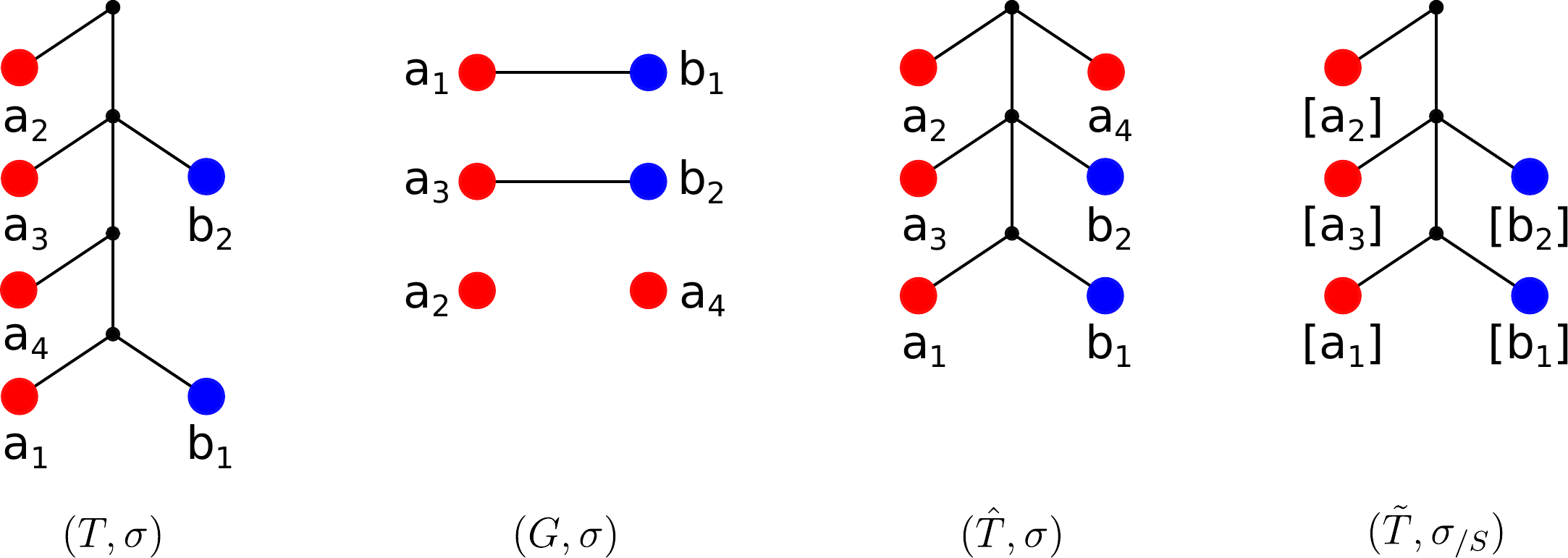}
  \end{center}
  \caption{The leaf-colored tree $(T,\sigma)$ on the left explains the
    RBMG $(G,\sigma)$. \NEW{Here,} $a_2, a_4\in [a_2]$ but they do not have
    the same parent in $T$. The tree $(\hat T,\sigma)$ is obtained from
    $(T,\sigma)$ by re-attaching the leaf \NEW{$a_4$ to $\parent(a_2)$}
    and suppressing the 2-degree vertex \NEW{$\parent(a_4)$.} The resulting
    tree still explains $(G,\sigma)$ and $a_2$ and $a_4$ are now
    siblings. Retaining only one representative of each $\sthin$-class
    finally gives the tree $(\tilde T, \sigma_{/S})$ on the right that
    explains the $\sthin$\NEW{-thin} graph $(G_{/S},\sigma_{/S})$.}
  \label{fig:SthinTree}
\end{figure}

\section{Connected Components} 
\label{sect:connect-new}

This section aims at simplifying the problem of finding a characterization
for RBMGs by showing that an undirected colored graph is an RBMG if and
only if each of its connected components is an RBMG (cf.\ Theorem
\ref{thm:connected}) and at least one of them contains all colors. This, in
turn, reduces the problem to connected graphs.  This is a non-trivial
observation since BMGs are not hereditary. Hence, we cannot expect RBMGs to
be hereditary, either.  They do satisfy a somewhat weaker property,
however:
\begin{clemma}{\ref{lem:fact:1}}
  Let $(G,\sigma)$ be an RBMG with vertex set $L$ explained by $(T,\sigma)$
  and let $(T_{|L'},\sigma_{|L'})$ be the restriction of $(T,\sigma)$ to
  $L'\subseteq L$. Then the induced subgraph
  $(G,\sigma)[L']\coloneqq (G[L'],\sigma_{|L'})$ of $(G,\sigma)$ is a (not
  necessarily induced) subgraph of $G(T_{|L'},\sigma_{|L'})$.
\end{clemma}

Lemma~\ref{lem:fact:1} provides a starting point to show that the connected
components of an RBMG are again RBMGs that can be explained by
corresponding restrictions of a leaf-colored tree:
\begin{ctheorem}{\ref{thm:conn_comp}}
  Let $(G^*,\sigma^*)$ with vertex set $L^*$ be a connected component of
  some RBMG $(G,\sigma)$ and let $(T,\sigma)$ be a leaf-colored tree
  explaining $(G,\sigma)$. Then, $(G^*,\sigma^*)$ is again an RBMG and is
  explained by the restriction $(T_{|L^*},\sigma_{|L^*})$ of $(T,\sigma)$
  to $L^*$.
\end{ctheorem}
It is worth noting that there is no similar result for BMGs.

Every connected component of an $n$-RBMG is therefore a $k$-RBMG possibly
with a strictly smaller number $k$ of colors.  Our aim in the remainder of
this section is to show that the disjoint union of RBMGs is again an RBMG
provided that one of these RBMGs contains all colors.

Let $(G,\sigma)$ be an undirected, vertex colored graph with vertex set $L$
and $|\sigma(L)|=n$. We denote the connected components of $(G,\sigma)$ by
$(G_i^n,\sigma_i^n)$, $1\le i \le k$, with vertex sets $L_i^n$ if
$\sigma(L_i^n)=\sigma(L)$ and $(G_j^{<n},\sigma_j^{<n})$, $1\le j \le l$,
with vertex sets $L_j^{<n}$ if $\sigma(L_j^{<n})\subsetneq\sigma(L)$. That
is, the upper index distinguishes components with all colors present from
those that contain only a proper subset. Suppose that each
$(G_i^n,\sigma_i^n)$ and $(G_j^{<n},\sigma_j^{<n})$ is an RBMG. Then there
are trees $(T_i^n,\sigma_i^n)$ and $(T_j^{<n},\sigma_j^{<n})$ explaining
$(G_i^n,\sigma_i^n)$ and $(G_j^{<n},\sigma_j^{<n})$, respectively. The
roots of these trees are $u_i$ and $v_j$, respectively. We construct a tree
$(T_G^*,\sigma)$ with leaf set $L$ in two steps:
\begin{itemize}
\item[(1)] Let $(T',\sigma^n)$ be the tree obtained by attaching the trees
  $(T_i^n,\sigma_i^n)$ with their roots $u_i$ to a common root $\rho'$.
\item[(2)] First, construct a path $P=v_1v_2\dots v_{l-1}v_l\rho'$, where
  $\rho'$ is omitted whenever $T'$ is empty.  Now attach the trees
  $(T_j^{<n},\sigma_j^{<n})$, $1\le j\le l$, to $P$ by identifying the root
  of each $T_j^{<n}$ with the vertex $v_j$ in $P$.  Finally, if
  $(T',\sigma^n)$ exists, attach it to $P$ by identifying the root of $T'$
  with the vertex $\rho'$ in $P$. The coloring of $L$ is the one given for
  $(G,\sigma)$.
\end{itemize}
This construction is illustrated in Fig.~\ref{fig:tree_constr} \NEW{for
  $n\ge 2$. For $n=1$, the resulting tree is simply the star tree on $L$. }
  
\begin{figure}
  \begin{tabular}{lcr}
    \begin{minipage}{0.5\textwidth}
      \begin{center}
        \includegraphics[width=\textwidth]{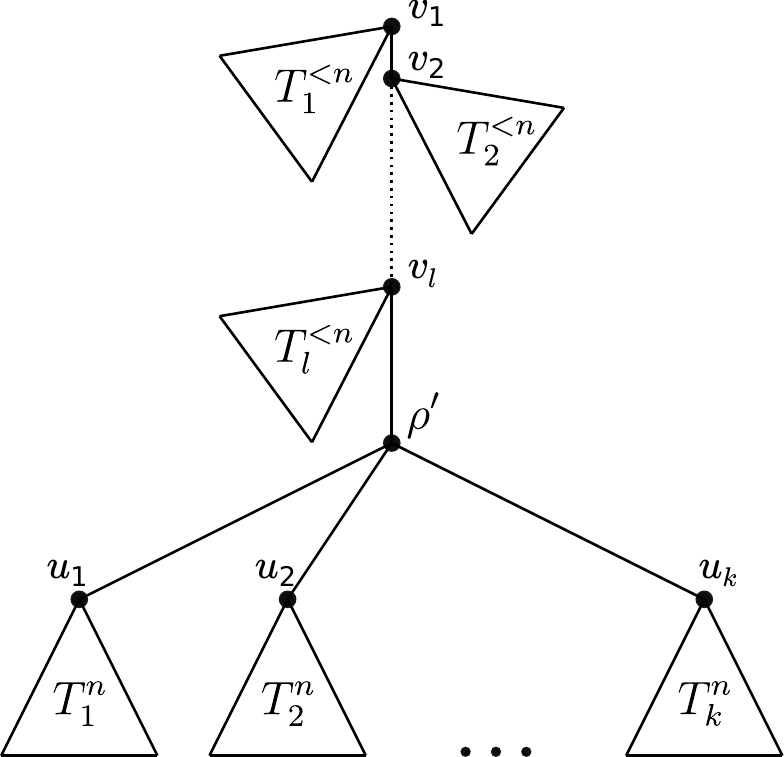}
      \end{center}
    \end{minipage}
    & & 
    \begin{minipage}{0.4\textwidth}
      \caption{Shown is a tree $(T^*_G,\sigma)$ that explains the graph
        $(G,\sigma)=\bigcup_{1\le i\le k}G((T_i^n,\sigma_i^n))\cup
        \bigcup_{1\le j\le l}G((T_j^{<n},\sigma_j^{<n}))$ such that each of
        the subtrees $(T_i^n,\sigma_i^n)$ and $(T_j^{<n},\sigma_j^{<n})$
        induces one connected component of $(G,\sigma)$. The subtree
        $(T_i^n,\sigma_i^n)$ explains the $n$-colored connected component
        $(G_i^{n},\sigma_i^{n})$ of $(G,\sigma)$. Each connected component
        $(G_j^{<n},\sigma_j^{<n})$ that does not contain all colors of $S$,
        is explained by a subtree $(T_j^{<n},\sigma_j^{<n})$. Any $n$-RBMG
        $(G,\sigma)$ can be explained by a tree of such form (cf.\ Lemma
        \ref{lem:tree}). See Fig.\ \ref{fig:tree_ex} for an explicit
        example of such a tree $(T^*_G,\sigma)$.}
    \label{fig:tree_constr}
    \end{minipage}
  \end{tabular}
\end{figure}

We then proceed by demonstrating that it suffices to consider trees of the
form $(T^*_G,\sigma)$. The key result, Lemma~\ref{lem:tree} in the
technical part, shows that an undirected vertex colored graph $(G,\sigma)$
whose connected components are RBMGs can be explained by $(T^*_G,\sigma)$
provided $(G,\sigma)$ contains a connected component in which all colors
are represented. It is not hard to check that these conditions are also
necessary.

\begin{ctheorem}{\ref{thm:connected}}
  An undirected leaf-colored graph $(G,\sigma)$ is an RBMG if and only if
  each of its connected components is an RBMG and at least one connected
  component contains all colors.
\end{ctheorem}

\begin{figure}
  \begin{center}
    \includegraphics[width=\textwidth]{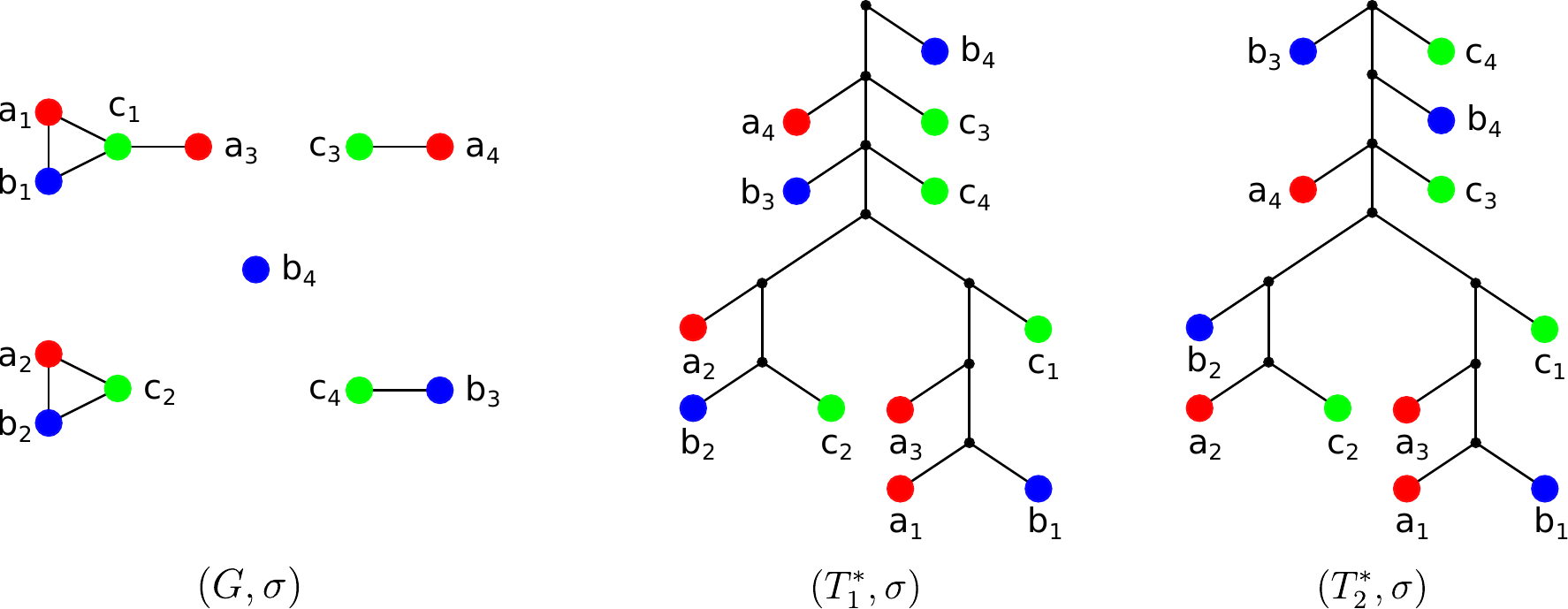}      
  \end{center}
  \caption{The trees $(T_1^*,\sigma)$ and $(T_2^*,\sigma)$ both explain the
    3-RBMG $(G,\sigma)$ with \NEW{five} connected components and are of the
    form $(T^*_G,\sigma)$.}
  \label{fig:tree_ex}
\end{figure}

The existence of an connected component using all colors is crucial for the
statement above. Consider, for instance, an edge-less graph on two
vertices, where both vertices have different color. Each of the two
connected components is clearly an RBMG, however, one easily checks that
their disjoint union is not.
\begin{ccorollary}{\ref{cor:tree2}}
  Every RBMG can be explained by a tree of the form $(T^*_G,\sigma)$.
\end{ccorollary}
By Theorem \ref{thm:connected}, it suffices to consider each connected
component of an RBMG separately. In the following section, hence, we will
consider the characterization of connected RBMGs.

\section{Three Classes of Connected 3-RBMGs}
\label{sect:classes-new}

Reciprocal best match graphs on two colors convey very little structural
information. Their connected components are either single vertices or
complete bipartite graphs \cite[Cor.\ 6]{Geiss:18x}, which reduce to a
$K_2$ with two distinctly colored vertices under
$\sthin$-thinness. Connected 3-RBMGs, in contrast, can be quite complex.
As we shall see, they fall into three distinct classes which correspond to
trees with different shapes. In the next section we will make use of
3-RBMGs to characterize general $n$-RBMGs.

Our starting point are three types of leaf-colored trees on three colors:
\begin{cdefinition}{\ref{def:3-col-tree}}
  Let $(T,\sigma)$ be a 3-colored tree with color set $S=\{r,s,t\}$.  The
  tree $(T,\sigma)$ is of
  \begin{description}
  \item[\textbf{Type \AX{\bf (I)}},] if there exists $v\in \child(\rho_T)$ such
    that $|\sigma(L(T(v)))|=2$ and
    $\child(\rho_T)\setminus \{v\}\subsetneq L$.
  \item[\textbf{Type \AX{\bf (II)}},] if there exists
    $v_1, v_2\in \child(\rho_T)$ such that
    $|\sigma(L(T(v_1)))|=|\sigma(L(T(v_2)))| = 2$,
    $\sigma(L(T(v_1)))\neq \sigma(L(T(v_2)))$ and
    $\child(\rho_T)\setminus \{v_1,v_2\}\subsetneq L$,
  \item[\textbf{Type \AX{\bf (III)}},] if there exists
    $v_1, v_2, v_3\in \child(\rho_T)$ such that
    $\sigma(L(T(v_1)))=\{r,s\}$, $\sigma(L(T(v_2)))=\{r,t\}$,
    $\sigma(L(T(v_3)))=\{s,t\}$, and
    $\child(\rho_T)\setminus \{v_1,v_2, v_3\} \subsetneq L$.
  \end{description}
\end{cdefinition}
Fig.\ \ref{fig:categories} illustrates these three tree types.

\begin{figure}
  \includegraphics[width=\textwidth]{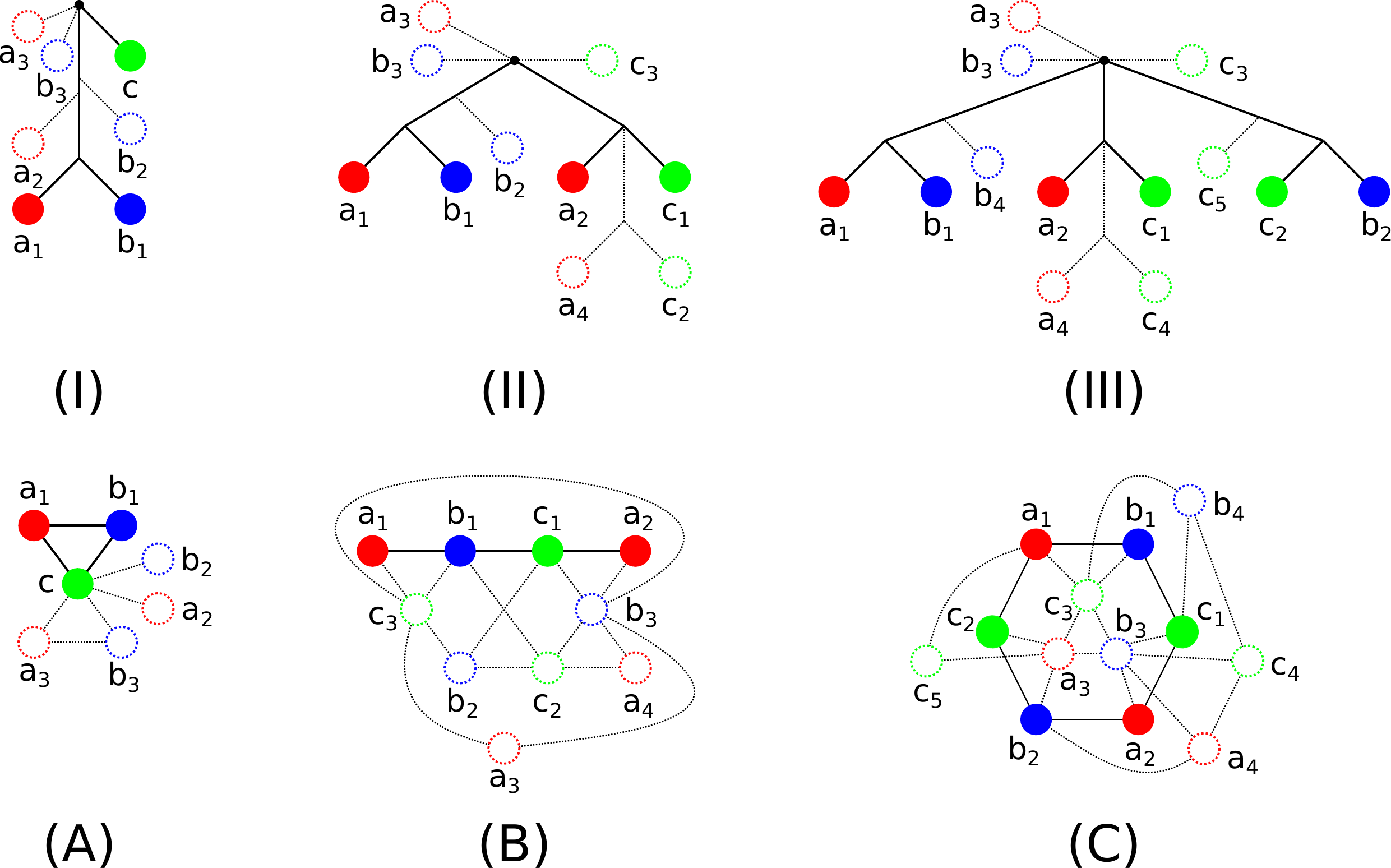}
  \centering
  \caption[]{The three categories of three-colored connected 3-RBMGs are
    shown on the bottom: (A) Contains a $K_3$ on three colors but no
    induced \NEW{$C_n$, $n\ge 5$} or $P_4$, (B) contains an induced $P_4$,
    whose endpoints have the same color, but no induced \NEW{$C_n$ with
      $n\ge 5$}, (C) contains a $C_6$ of the form $(r,s,t,r,s,t)$. The
    corresponding tree Types (I), (II), and (III) are shown on top. Solid
    lines represent edges and vertices that must necessarily be contained
    in the graph, dashed elements may be missing.}
  \label{fig:categories}
\end{figure}

Correspondingly, we distinguish three classes of 3-colored graphs:
\begin{cdefinition}{\ref{def:typesABC}}
  An undirected, connected graph $(G,\sigma)$ on three colors is of
  \begin{description}
  \item[\textbf{Type \AX{\bf (A)}}] if $(G,\sigma)$ contains a $K_3$ on
    three colors but no induced $P_4$, and thus also no induced \NEW{$C_n$,
      $n\ge 5$.}
  \item[\textbf{Type \AX{\bf (B)}}] if $(G,\sigma)$ contains an induced
    $P_4$ on three colors whose endpoints have the same color, but no
    no induced $C_n$ for $n\ge 5$.
  \item[\textbf{Type \AX{\bf (C)}}] if $(G,\sigma)$ contains an induced
    $C_6$ along which the three colors appear twice in the same
    permutation, i.e., $(r,s,t,r,s,t)$.
  \end{description}
\end{cdefinition}
The main result of this section is that each tree class explains the RBMGs
belonging to one of the classes of colored graphs. More precisely:
\begin{ctheorem}{\ref{thm:3c-types}}
  Let $(G,\sigma)$ be an $\sthin$-thin connected 3-RBMG.  Then $(G,\sigma)$
  is either of Type \AX{(A)}, \AX{(B)}, or \AX{(C)}.  An RBMG of Type
  \AX{(A)}, \AX{(B)}, and \AX{(C)}, resp., can be explained by a tree of
  Type \AX{(I)}, \AX{(II)}, and \AX{(III)}, respectively.
\end{ctheorem}

An undirected, colored graph $(G,\sigma)$ contains an induced $K_3$, $P_4$,
or $C_6$, respectively, if and only if $(G/\sthin,\sigma/\sthin)$ contains
an induced $K_3$, $P_4$, or $C_6$, resp., on the same colors (cf.\ Lemma
\ref{lem:sthin}). An immediate consequence of this fact is
\begin{ctheorem}{\ref{thm:thinABC}}
  A connected (not necessarily $\sthin$-thin) 3-RBMG $(G,\sigma)$ is either
  of Type \AX{(A)}, \AX{(B)}, or \AX{(C)}.
\end{ctheorem}

As a further consequence of Theorem \ref{thm:3c-types} and the well-known
properties of cographs \cite{Corneil:81} we obtain
\begin{cfact}{\ref{fact:cograph}}
  Let $(G,\sigma)$ be a connected, $\sthin$-thin 3-RBMG.  Then it is of
  Type \AX{(A)} if and only if it is a cograph.
\end{cfact}
This observation is of practical interest because orthology relations are
necessarily cographs \cite{Hellmuth:13a}. Hence, 3-RBMGs cannot perfectly
reflect orthology unless they are of Type \AX{(A)}.

So-called \emph{hub-vertices} play a key role for the characterization of
Type \AX{(A)} 3-RBMGs:
\begin{cdefinition}{\ref{def:hub}}
  Let $G=(V,E)$ be an undirected graph. A vertex $x\in V(G)$ such that
  $N(x)=V\setminus \{x\}$ is a \emph{hub-vertex}.
\end{cdefinition}

\begin{clemma}{\ref{lem:charA}} 
  A properly vertex colored, connected, $\sthin$-thin graph $(G,\sigma)$ on
  three colors with vertex set $L$ is a 3-RBMG of Type \AX{(A)} if and
  only if $G \notin \mathscr{P}_3$ and it satisfies the following
  conditions:
  \begin{enumerate}
  \item[\AX{(A1)}] $G$ contains a hub-vertex $x$, i.e.,
    $N(x)=V(G)\setminus \{x\}$
  \item[\AX{(A2)}] $|N(y)|<3$ for every $y\in V(G)\setminus \{x\}$.
  \end{enumerate}
\end{clemma}

We proceed by characterizing Type \AX{(B)} and \AX{(C)} 3-RBMGs. To this
end, we first introduce \emph{B-like} and \emph{C-like} graphs
$(G,\sigma)$:
\begin{cdefinition}{\ref{def:Ltsr}}
  Let $(G,\sigma)$ be an undirected, connected, properly colored,
  $\sthin$-thin graph with vertex set $L$ and color set
  $\sigma(L)=\{r,s,t\}$, and assume that $(G,\sigma)$ contains the induced
  path $P\coloneqq \langle \hat x_1 \hat y \hat z \hat x_2 \rangle$ with
  $\sigma(\hat x_1)=\sigma(\hat x_2)=r$, $\sigma(\hat y)=s$, and
  $\sigma(\hat z)=t$. Then $(G,\sigma)$ is \emph{B-like w.r.t.\ $P$} if (i)
  $N_{r}(\hat y)\cap N_{r}(\hat z)=\emptyset$, and (ii) $G$ does not
  contain an induced cycle \NEW{$C_n$, $n\ge 5$.}
\end{cdefinition}
An example is given in Fig.\ \ref{fig:NOindep-P-choice}.

\begin{figure}
  \begin{center}
    \includegraphics[width=0.8\textwidth]{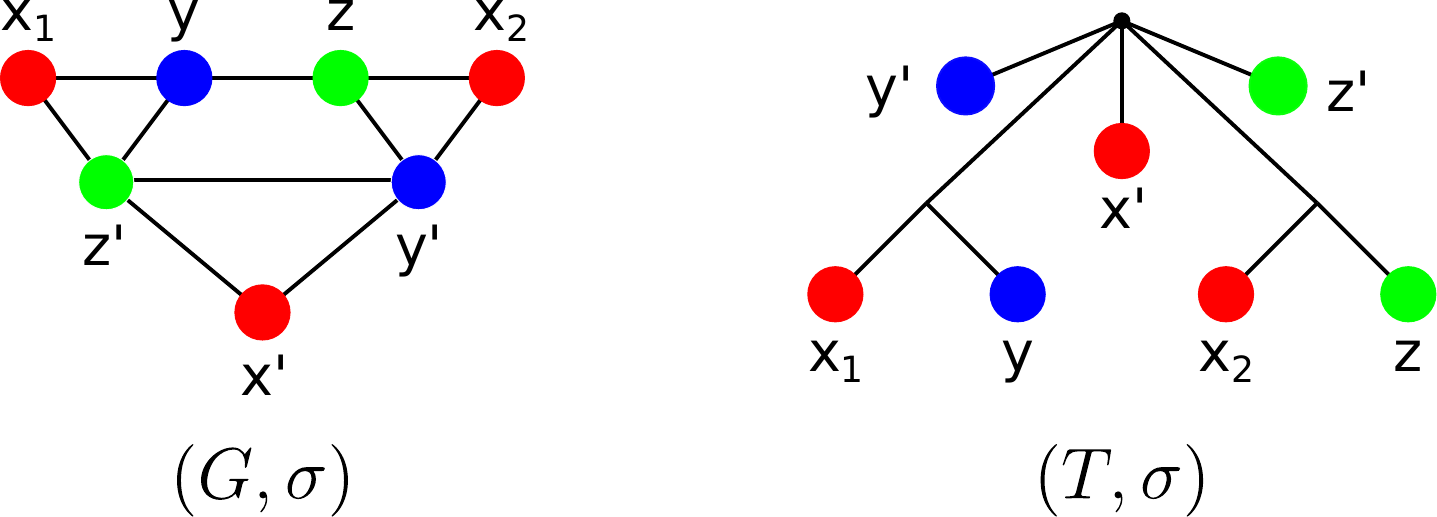}      
  \end{center}
  \caption{The graph $(G,\sigma)$ is a 3-RBMG since it is explained by
    $(T,\sigma)$. Moreover, $(G,\sigma)$ does not contain an induced
    \NEW{$C_n$, $n\ge 5$}
    but induced $P_4$s, thus it is of Type \AX{(B)}.  It is easy to see
    that $(G,\sigma)$ is B-like w.r.t.\ $\langle x_1 y z x_2\rangle$.
    However, $(G,\sigma)$ is not B-like w.r.t.\
    $\langle x_1 z' y' x_2\rangle$, since $ x' \in N_r(y')\cap N_r(z')$.  }
  \label{fig:NOindep-P-choice}
\end{figure}

For a $3$-colored, $\sthin$-thin graph $(G,\sigma)$ that is B-like w.r.t.\
the induced path
$P\coloneqq \langle \hat x_1 \hat y \hat z \hat x_2 \rangle$ we define the
following subsets of vertices:
\begin{align*} 	
  L_{t,s}^{P} \coloneqq & \{y \mid  \langle xy\hat z\rangle \in
                          \mathscr{P}_3 \text{ for any } x\in N_{r}(y)\}\\
  L_{t,r}^{P} \coloneqq & \{x\mid N_{r}(y)=\{x\} \text{ and }
                          \langle xy\hat z\rangle \in \mathscr{P}_3\}
                          \cup \\
                        & \{x\mid x\in L[r],\, N_{s}(x)=\emptyset,\,
                          L[s]\setminus L_{t,s}^P\neq\emptyset\}\\
  L_{s,t}^{P} \coloneqq & \{z \mid \langle xz\hat y\rangle \in \mathscr{P}_3
                          \text{ for any } x\in N_{r}(z)\} \\ 
  L_{s,r}^{P} \coloneqq & \{x\mid N_{r}(z)=\{x\} \text{ and }
                            xz\hat y \in \mathscr{P}_3\} \cup \\
                        & \{x\mid x\in L[r], N_{t}(x)=\emptyset,
                          L[t]\setminus L_{s,t}^P\neq\emptyset\}\\
\end{align*}
The first subscripts $t$ and $s$ refer to the color of the vertices
$\hat z$ and $\hat y$, respectively, that ``anchor'' the $P_3$s within the
defining path $P$. The second index identifies the color of the vertices in
the respective set, since by definition we have
$L_{t,s}^{P}\subseteq L[s]$, $L_{t,r}^{P}\subseteq L[r]$,
$L_{s,t}^{P}\subseteq L[t]$ and $L_{s,r}^{P}\subseteq L[r]$. Furthermore,
we set
\begin{align*}
  L_{t}^{P} \coloneqq & L_{t,s}^{P} \cup L_{t,r}^{P} \\
  L_{s}^{P} \coloneqq & L_{s,t}^{P} \cup L_{s,r}^{P} \\
  L_*^P     \coloneqq & L \setminus (L_t^P \cup L_s^P). 
\end{align*}
By definition, $L_{s,r}^P=L_s^P\cap L[r]$, $L_{t,r}^P=L_t^P\cap L[r]$,
$L_{s,t}^P=L_s^P\cap L[t]$, and $L_{t,s}^P=L_t^P\cap L[s]$. For simplicity
we will often write $L_*^P[i] \coloneqq L_*^P\cap L[i]$ for  $i\in \{s,t\}$.

The construction of Type \AX{(B)} 3-RBMGs can be extended to a similar
one of Type \AX{(C)} 3-RBMGs.

\begin{cdefinition}{\ref{def:C-like}}
  Let $(G,\sigma)$ be an undirected, connected, properly colored,
  $\sthin$-thin graph.  Moreover, assume that $(G,\sigma)$ contains the
  hexagon
  $H\coloneqq \langle \hat x_1 \hat y_1 \hat z_1 \hat x_2 \hat y_2 \hat z_2
  \rangle$ such that $\sigma(\hat x_1)=\sigma(\hat x_2) =r$,
  $\sigma(\hat y_1)=\sigma(\hat y_2)=s$, and
  $\sigma(\hat z_1)=\sigma(\hat z_2) =t$.  Then, $(G,\sigma)$ is
  \emph{C-like w.r.t.\ $H$} if there is a vertex
  $v\in\{\hat x_1 ,\hat y_1, \hat z_1, \hat x_2 ,\hat y_2, \hat z_2\}$ such
  that $|N_c(v)|>1$ for some color $c\neq \sigma(v)$.  Suppose that
  $(G,\sigma)$ is C-like w.r.t.\
  $H =\langle \hat x_1 \hat y_1 \hat z_1 \hat x_2 \hat y_2 \hat z_2
  \rangle$ and assume w.l.o.g. that $v=\hat x_1$ and $c=t$, i.e.,
  $|N_t(\hat x_1)|>1$. Then we define the following sets:
  \begin{align*}
    L_t^H &\coloneqq \{x \mid \langle x \hat z_2 \hat y_2 \rangle
            \in \mathscr{P}_3\}\cup
            \{y \mid \langle y \hat z_1 \hat x_2 \rangle\in \mathscr{P}_3\} \\
    L_s^H &\coloneqq \{x \mid \langle x \hat y_2 \hat z_2 \rangle
            \in \mathscr{P}_3\} \cup
            \{z \mid \langle z \hat y_1 \hat x_1 \in \rangle
            \mathscr{P}_3\} \\
    L_r^H &\coloneqq \{ y \mid \langle y \hat x_2 \hat z_1 \rangle
            \in \mathscr{P}_3\} \cup
            \{z \mid \langle z \hat x_1 \hat y_1 \rangle\in \mathscr{P}_3\}\\
    L_*^H &\coloneqq V(G)\setminus (L_r^H\cup L_s^H\cup L_t^H).
  \end{align*}		
\end{cdefinition}

The main result of this section is the following, rather technical, result,
which provides a complete characterization of 3-colored RBMGs.
\begin{ctheorem}{\ref{thm:char3cBMG}}
  Let $(G,\sigma)$ be an undirected, connected, $\sthin$-thin, and properly
  3-colored graph with color set $S=\{r,s,t\}$ and let $x\in L[r]$,
  $y\in L[s]$ and $z\in L[t]$.  Then $(G,\sigma)$ is a 3-RBMG if and only
  if one of the following is true:
  \begin{enumerate}
  \item Conditions \AX{(A1)} and \AX{(A2)} are satisfied, or
  \item Conditions \AX{(B1)} to \AX{(B3.b)} are satisfied, after possible
    permutation of the colors, where:
  \begin{description}
  \item[\AX{(B1)}] $(G,\sigma)$ is B-like w.r.t.\
    $P = \langle\hat x_1\hat y\hat z \hat x_2 \rangle$ for some
    $\hat x_1, \hat x_2\in L[r]$, $\hat y\in L[s]$, $\hat z \in L[t]$,
  \item[\AX{(B2.a)}] If $x\in L_*^P$, then $N(x)=L_*^P\setminus\{x\}$,
  \item[\AX{(B2.b)}] If $x\in L_t^P$, then $N_s(x)\subset L_t^P$ and
    $|N_s(x)|\le 1$, and $N_t(x)= L_*^P[t]$,
  \item[\AX{(B2.c)}] If $x\in L_s^P$, then
    $N_t(x)\subset L_s^P$ and $|N_t(x)|\le 1$, and $N_s(x)= L_*^P[s]$
  \item[\AX{(B3.a)}] If $y\in L_*^P$, then
    $N(y)=L_s^P\cup (L_*^P\setminus\{y\})$, 
  \item[\AX{(B3.b)}] If $y\in L_t^P$, then
    $N_r(y)\subset L_t^P$ and $|N_r(y)|\le 1$, and $N_t(y)=L[t]$,
  \end{description} 
  or
\item $(G,\sigma)$ is either a hexagon or $|L|>6$ and, up to permutation of
  colors, the following conditions are satisfied:
  \begin{description}
  \item[\AX{(C1)}] $(G,\sigma)$ is C-like w.r.t.\ the hexagon
    $H = \langle \hat x_1\hat y_1\hat z_1 \hat x_2 \hat y_2 \hat z_2
    \rangle$ for some $\hat x_i\in L[r]$, $\hat y_i\in L[s]$,
    $\hat z_i \in L[t]$ with $|N_t(\hat x_1)|>1$,
  \item[\AX{(C2.a)}] If $x\in L_*^H$, then
    $N(x)=L_r^H\cup (L_*^H\setminus\{x\})$,
  \item[\AX{(C2.b)}] If $x\in L_t^H$, then
    $N_s(x)\subset L_t^H$ and $|N_s(x)|\le 1$, and
    $N_t(x)= L_*^H[t]\cup L_r^H[t]$,
  \item[\AX{(C2.c)}] If $x\in L_s^H$, then
    $N_t(x)\subset L_s^H$ and $|N_t(x)|\le 1$, and
    $N_s(x)= L_*^H[s]\cup L_r^H[s]$
  \item[\AX{(C3.a)}] If $y\in L_*^H$, then
    $N(y)=L_s^H\cup (L_*^H\setminus\{y\})$,
  \item[\AX{(C3.b)}] If $y\in L_t^H$, then
    $N_r(y)\subset L_t^H$ and $|N_r(y)|\le 1$, and
    $N_t(y)=L_*^H[t]\cup L_s^H[t]$,
  \item[\AX{(C3.c)}] If $y\in L_r^H$, then
    $N_t(y)\subset L_r^H$ and $|N_t(y)|\le 1$, and
    $N_r(y)=L_*^H[r]\cup L_s^H[r]$.
   \end{description} 
 \end{enumerate}
\end{ctheorem}

\NEW{%
\begin{remark}
  As a consequence of Lemma \ref{lem:sthin}, every (not necessarily $\sthin$-thin) Type \AX{(B)} 3-RBMG
  $(G,\sigma)$ contains an induced path $\langle xyzx'\rangle$ with
  $\sigma(x)=\sigma(x')=r$, $\sigma(y)=s$, and $\sigma(z)=t$ for distinct
  colors $r,s,t$ such that $N_r(y)\cap N_r(z)=\emptyset$.  Similarly, every
  Type \AX{(C)} 3-RBMG $(G,\sigma)$ contains a hexagon
  $\langle xyzx'y'z'\rangle$ with $\sigma(x)=\sigma(x')=r$,
  $\sigma(y)=\sigma(y')=s$, and $\sigma(z)=\sigma(z')=t$ for distinct
  colors $r,s,t$ such that $|N_c(v)|>1$ for some $v\in\{x,x',y,y',z,z'\}$
  and $c\neq\sigma(v)$.
\end{remark}
}

In the technical part we describe an algorithm that determines whether a
given properly 3-colored connected graph $(G,\sigma )$ is a 3-RBMG and, in
the positive case, returns a tree $(T,\sigma )$ that explains $(G,\sigma)$
in $O(mn^2+m')$ time, where $n=|V(G/\sthin)|$, $m=|E(G/\sthin)|$ and
$m'= |E(G)|$ (cf.\ Algorithm \ref{alg:3RBMG}, Lemmas
\ref{lem:runtime:alg-3rbmg} and \ref{lem:correct:alg-3rbmg}).

In addition, we provide in Section \ref{sect:P4} in the technical part
results about the structure of induced $P_4$s in RBMGs.  In particular,
those $P_4$s can be classified as so-called \emph{good}, \emph{bad}, and
\emph{ugly} quartets and the sets $L_t^P$, $L_s^P$, and $L_*^P$ can be
determined by good quartets and are independent of the choice of the
respective good quartet.  As shown by \citet{GGL:19}, good quartets also
play an important role for the detection of false positive and false
negative orthology assignments.

\section{Characterization of $n$-RBMGs}\label{sect:n-new}
The first part of this section is dedicated to the characterization of
$n$-RBMGs by combining the sets of least resolved trees for the induced
3-RBMGs for any triplet of colors. It remains, however, an open question if
the problem of recognizing $n$-RBMGs can be solved in polynomial
time. Because of the importance of cographs in best-match-based orthology
assignment methods, we investigate \emph{co-RBMGs}, i.e., $n$-RBMGs that
are cographs, in more detail and provide a characterization of this
subclass.

\subsection{The General Case: Combination of 3-RBMGs} 
\label{ssec:n-RBMG}

It will be convenient in the following to use the simplified notation
\begin{cdefinition}{\ref{def:restrict-3}}
  $(G_{rst},\sigma_{rst})\coloneqq (G[L[r]\cup L[s]\cup L[t]],
  \sigma_{|L[r]\cup L[s]\cup L[t]})$ and
  $(T_{rst},\sigma_{rst})\coloneqq (T_{|L[r]\cup L[s]\cup L[t]},
  \sigma_{|L[r]\cup L[s]\cup L[t]})$ for any three colors $r,s,t\in S$.
\end{cdefinition}
The restriction of a BMG $\G(T,\sigma)$ to a subset $S'\subset S$ of colors
is an induced subgraph of $\G(T,\sigma)$ explained by the restriction of
$(T,\sigma)$ to the leaves with colors in $S'$, an thus again a BMG
\cite[Observation 1]{Geiss:18x}.  Since $G(T,\sigma)$ is the symmetric part
of $\G(T,\sigma)$, it inherits this property. In particular, we have
\begin{cfact}{\ref{obs:induced_sub}}
  If $(G,\sigma)$ is an $n$-RBMG, $n\ge 3$ explained by $(T,\sigma)$, then
  for any three colors $r,s,t \in S$, the restricted tree
  $(T_{rst},\sigma_{rst})$ explains $(G_{rst},\sigma_{rst})$, and
  $(G_{rst},\sigma_{rst})$ is an induced subgraph of $(G,\sigma)$.
\end{cfact}

The key idea of characterizing $n$-RBMGs is to combine the information
contained in their 3-colored induced subgraphs $(G_{rst},\sigma_{rst})$.
Observation \ref{obs:induced_sub} plays a major role in this context. It
shows that $(G_{rst},\sigma_{rst})$ is an induced subgraph of an $n$-RBMG
and it is always a 3-RBMG that is explained by $(T_{rst},\sigma_{rst})$.
Unfortunately, the converse of Observation \ref{obs:induced_sub} is in
general not true. Fig.~\ref{fig:supertree_counterex} shows a 4-colored
graph that is not a 4-RBMG while each of the four subgraphs induced by a
triplet of colors is a 3-RBMG.  We can, however, rephrase Observation
\ref{obs:induced_sub} in the following way:
\begin{cfact}{\ref{fact:induced_sub}}
  Let $(G,\sigma)$ be an $n$-RBMG for some $n\ge 3$. Then, $(T,\sigma)$
  explains $(G,\sigma)$ if and only if $(T_{rst},\sigma_{rst})$ explains
  $(G_{rst},\sigma_{rst})$ for all triplets of colors $r,s,t \in S$.
\end{cfact}

\begin{figure}
  \begin{center}
    \includegraphics[width=\textwidth]{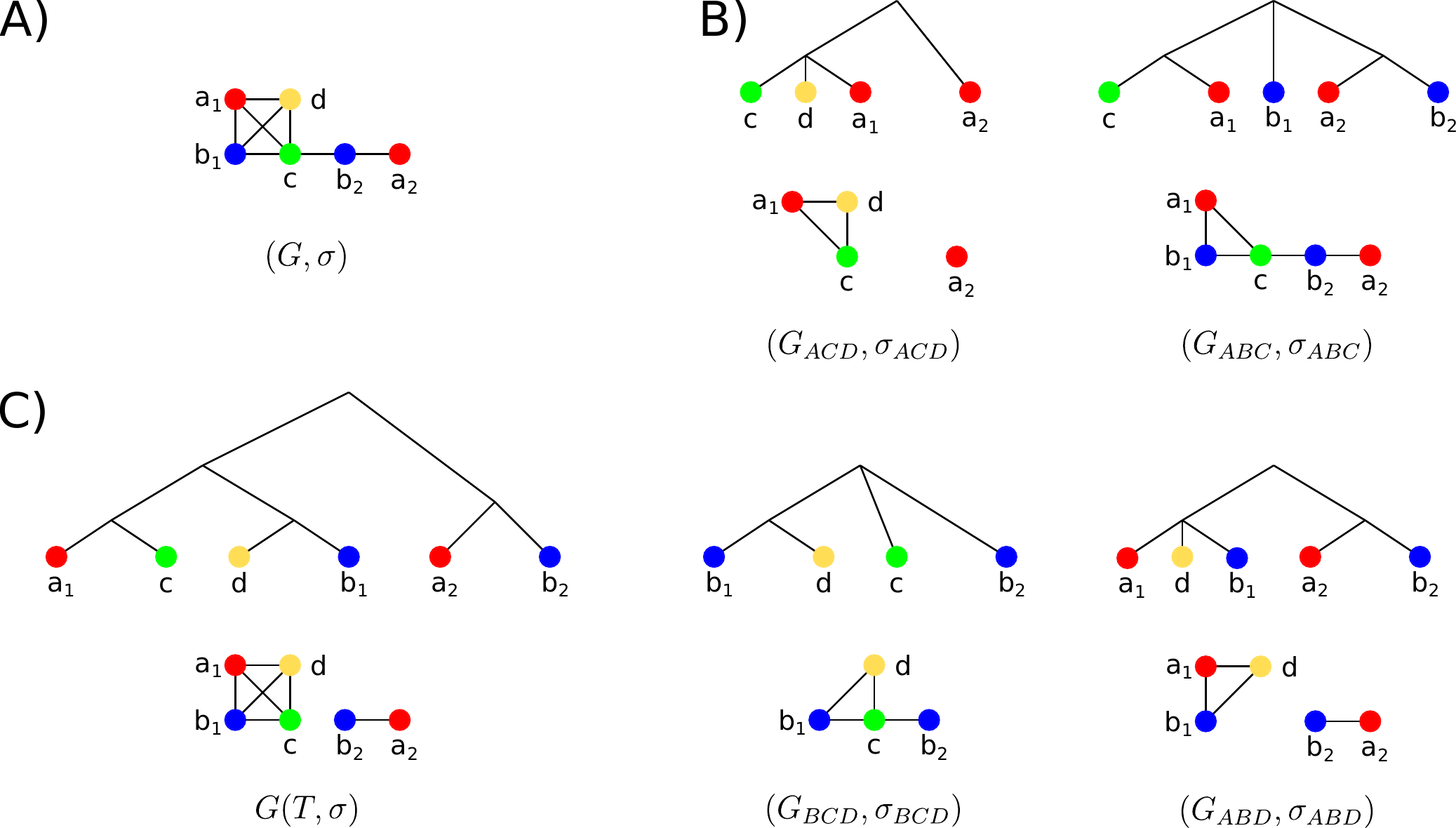}      
  \end{center}
  \caption{The 4-colored graph $(G,\sigma)$ in (A) with color set
    $S=\{A,B,C,D\}$ is not an RBMG. All four subgraphs induced by three of
    the four colors, however, are (not necessarily connected)
    3-RBMGs. These are explained by the unique least resolved trees in
    (B). Because of the uniqueness of the least resolved trees on three
    colors, the tree explaining $(G,\sigma)$ must display these four
    trees. The tree $(T,\sigma)$ in Panel (C) is the least resolved
    supertree of $\mathcal{P}\coloneqq \bigcup_{r,s,t \in S}T_{rst}$. However,
    $(T,\sigma)$ does not explain $(G,\sigma)$ since the edge $b_2c$ is not
    contained in $G(T,\sigma)$. Clearly, there exists no refinement
    $(T',\sigma)$ of $(T,\sigma)$ such that $b_2c\in E(G(T',\sigma))$ and
    therefore $(G,\sigma)$ is not an RBMG.}
  \label{fig:supertree_counterex}
\end{figure} 

\begin{cdefinition}{\ref{def:Grst-Trst}} 
  Let $(G,\sigma)$ be an $n$-RBMG. Then the \emph{tree set of $(G,\sigma)$}
  is the set
  \begin{equation*}
    \mathcal{T}(G,\sigma)\coloneqq \{(T,\sigma) \mid \NEW{(T,\sigma)}
    \text{ is least resolved and } G(T,\sigma)=(G,\sigma)\}
  \end{equation*} of all leaf-colored trees
  explaining $(G,\sigma)$. Furthermore, we write
  $\mathcal{T}_{rst}(G,\sigma)$ for the set of all least resolved trees
  explaining the induced subgraphs $(G_{rst},\sigma_{rst})$.
\end{cdefinition}
It is tempting to conjecture that the existence of a supertree for the tree
set
$\mathcal{P}\coloneqq
\{T\in\mathcal{T}_{rst}(G_{rst},\sigma_{rst}),\,r,s,t\in S\}$ is sufficient
for $(G,\sigma)$ to be an $n$-RBMG. However, this is not the case as shown
by the counterexample in Fig.\ \ref{fig:supertree_counterex}. 
  
\begin{ctheorem}{\ref{thm:n-crbmg}}
  A (not necessarily connected) undirected colored graph $(G,\sigma)$ is an
  $n$-RBMG if and only if (i) all induced subgraphs
  $(G_{rst},\sigma_{rst})$ on three colors are 3-RBMGs and (ii) there
  exists a supertree $(T,\sigma)$ of \NEW{the} tree set
  $\mathcal{P}\coloneqq \{T\in\mathcal{T}_{rst}(G,\sigma) \mid r,s,t\in
  S\}$, such that $G(T,\sigma)=(G,\sigma)$.
\end{ctheorem} 

Whether the recognition problem of $n$-RBMGs is NP-hard or not may
strongly depend on the number of least resolved trees for a given 3-colored
induced subgraph.  However, even if this number is polynomial bounded in
the input size (e.g.\ number of vertices), the number of possible (least
resolved) trees that explain a given $n$-RBMG, may grow exponentially.  In
particular, since the order of the inner nodes in the 2-colored subtrees of
Type \AX{(I)}, \AX{(II)}, and \AX{(III)} trees is in general arbitrary,
determining the number of least resolved trees seems to be far from
trivial. We therefore leave it as an open problem.

\subsection{Characterization of $n$-RBMGs that are cographs}

Probably the most important application of reciprocal best matches is
orthology detection. Since orthology relations are cographs
\cite{Hellmuth:13a}, it is of particular interest to characterize RBMGs of
this type. Since cographs are hereditary (see e.g.\ \cite{Sumner:74} where
they are called Hereditary Dacey graphs), one expects their 3-colored
restrictions to be of Type \AX{(A)}. The next theorem shows that this
intuition is essentially correct. It is based on the following observation
about cographs:
\begin{cfact}{\ref{fact:ortho-cograph}}
  Any undirected colored graph $(G,\sigma)$ is a cograph if and only if the
  corresponding $\sthin$-thin graph $(G/\sthin,\sigmasthin)$ is a cograph.
\end{cfact} 

\begin{ctheorem}{\ref{thm:cographA}}   
  Let $(G,\sigma)$ be an $n$-RBMG with $n\ge3$, and denote by
  $(G'_{rst},\sigma'_{rst})\coloneqq (G_{rst}/S,\sigma_{rst}/S)$ the
  $\sthin$-thin version of the 3-RBMG that is obtained by restricting
  $(G,\sigma)$ to the colors $r$, $s$, and $t$.  Then $(G,\sigma)$ is a
  cograph if and only if every \emph{3-colored} connected component of
  $(G'_{rst},\sigma'_{rst})$ is a 3-RBMG of Type \AX{(A)} for all triples
  of distinct colors $r,s,t$.
\end{ctheorem} 

\NEW{%
\begin{remark}
  	Theorem \ref{thm:cographA} has been stated for $\sthin$-thin induced
    3-RBMGs only. However, as a consequence of Lemma \ref{lem:sthin}, it
    extends to general RBMGs, i.e., an $n$-RBMG $(G,\sigma)$ is a cograph
    if and only if every 3-colored connected component of
    $(G_{rst},\sigma_{rst})$ is a Type \AX{(A)} 3-RBMG for all triplets of
    distinct colors $r,s,t$.
\end{remark}
}

\subsection{Hierarchically Colored Cographs}

Thm.\ \ref{thm:cographA} yields a polynomial time algorithm for recognizing
$n$-RBMGs that are cographs. It is not helpful, however, for the
reconstruction of a tree $(T,\sigma)$ that explains such a graph. Below, we
derive an alternative characterization in terms of so-called hierarchically
colored cographs (\hc-cographs).  As we shall see, the cotrees of
\hc-cographs explain a given $n$-RBMG and can be constructed in polynomial
time.
\begin{cdefinition}{\ref{def:co-RBMG}}
  A graph that is both a cograph and an RBMG is a \emph{co-RBMG}.
\end{cdefinition}

Cographs are constructed using joins and disjoint unions. We extend these
graph operations to vertex colored graphs.
\begin{cdefinition}{\ref{def:coRBMG2}}
  Let $(H_1,\sigma_{H_1})$ and $(H_2,\sigma_{H_2})$ be two vertex-disjoint
  colored graphs.  Then
  $(H_1,\sigma_{H_1}) \join (H_2,\sigma_{H_2}) \coloneqq (H_1\join
  H_2,\sigma)$ and
  $(H_1,\sigma_{H_1}) \cupdot (H_2,\sigma_{H_2}) \coloneqq (H_1\cupdot
  H_2,\sigma)$ denotes their join and union, respectively, where
  $\sigma(x) = \sigma_{H_i}(x)$ for every $x\in V(H_i)$, $i\in\{1,2\}$.
\end{cdefinition}

\begin{cdefinition}{\ref{def:hc-cograph}}
  An undirected colored graph $(G,\sigma)$ is a \emph{hierarchically
    colored cograph (\hc-cograph)} if 
 \begin{itemize}
  \item[\AX{(K1)}] $(G,\sigma)=(K_1,\sigma)$, i.e., a colored vertex, or 
  \item[\AX{(K2)}] $(G,\sigma)= (H,\sigma_H) \join (H',\sigma_{H'})$ and
    $\sigma(V(H))\cap \sigma(V(H'))=\emptyset$, or 
  \item[\AX{(K3)}] $(G,\sigma)= (H,\sigma_H) \cupdot (H',\sigma_{H'})$ and
    $\sigma(V(H))\cap \sigma(V(H')) \in \{\sigma(V(H)),\sigma(V(H'))\}$,
\end{itemize}
where both $(H,\sigma_H)$ and $(H',\sigma_{H'})$ are \hc-cographs. 
For the color-constraints (cc) in \AX{(K2)} and \AX{(K3)}, we simply write
\AX{(K2cc)} and \AX{(K3cc)}, respectively.
\end{cdefinition}
Omitting the color-constraints reduces Def.\ \ref{def:hc-cograph} to Def.\
\ref{def:cograph}. Therefore we have
\begin{cfact}{\ref{fact:hc-cographx}}
  If $(G,\sigma)$ is an \hc-cograph, then $G$ is cograph.
\end{cfact}

The recursive construction of an \hc-cograph $(G,\sigma)$ according to
Def.\ \ref{def:hc-cograph} immediately produces a binary \hc-cotree
$T^G_{\hc}$ corresponding to $(G,\sigma)$. The construction is essentially
the same as for the cotree of a cograph (cf.\ \citet[Section
3]{Corneil:81}): Each of its inner vertices is labeled by $1$ for a $\join$
operation and $0$ for a disjoint union $\cupdot$, depending on whether
\AX{(K2)} or \AX{(K3)} \NEW{is} used \NEW{in the} construction steps. We write
$t:V^0(T^G_{\hc})\to\{0,1\}$ for the labeling of the inner vertices. The
recursion terminates with a leaf of $T^G_{\hc}$ whenever a colored
single-vertex graph, i.e., \AX{(K1)} is reached. We therefore identify the
leaves of $T^G_{\hc}$ with the vertices of $(G,\sigma)$. The binary
\hc-cograph $(T^G_{\hc},t,\sigma)$ with leaf coloring $\sigma$ and labeling
$t$ at its inner vertices uniquely determines $(G,\sigma)$, i.e.,
$xy\in E(G)$ if and only if $t(\lca(x,y)) =1$.

By construction, $(T^G_{\hc},t)$ is a not necessarily discriminating cotree
for $G$. An example for different constructions of $(T^G_{\hc},t,\sigma)$
based on the particular \hc-cograph representation of $(G,\sigma)$ is given
in Fig.\ \ref{fig:hc-cograph}.

While the cograph property is hereditary, this is no longer true for
\hc-cographs, i.e., an \hc-cograph may contain induced subgraphs that are
not \hc-cographs. As an example, consider the three single vertex graphs
$(G_i,\sigma_i)$ with $V_i=\{i\}$ and colors $\sigma_1(1)=r$ and
$\sigma_2(2)=\sigma_3(3)=s \neq r$. Then
$(G,\sigma)=((G_1,\sigma_1)\join(G_2,\sigma_2))\cupdot (G_3,\sigma_3)$ is
an \hc-cograph. However, the induced subgraph
$(G,\sigma)[1,3] = (G_1,\sigma_1) \cupdot (G_3,\sigma_3)$ is not an
\hc-cograph, since $\sigma_1(V_1)\cap \sigma_3(V_3) =\emptyset$ and hence,
$(G,\sigma)[1,3]$ does not satisfy Property \AX{(K3cc)}.

Both $\join$ and $\cupdot$ are commutative and associative operations on
graphs. For a given cograph $G$, hence, alternative binary cotrees may
exist that can be transformed into each other by applying the commutative
or associative laws. This is no longer true for \hc-cographs as a
consequence of the color constraints. There are no restrictions on
commutativity, i.e., if $(G,\sigma)$ can be obtained as the join
$(H,\sigma_{H})\join(H',\sigma_{H'})$, equivalently we have
$(G,\sigma)=(H',\sigma_{H'})\join(H,\sigma_{H})$. The same holds for the
disjoint union $\cupdot$.  If $(G,\sigma)$ is obtained as
$(H,\sigma_{H})\join\big((H',\sigma_{H'})\join(H'',\sigma_{H''})\big)$,
i.e., if $(H',\sigma_{H'})\join(H'',\sigma_{H''})$ is also an \hc-cograph,
then the color sets of $H$, $H'$, and $H''$ must be disjoint by
Def.~\ref{def:hc-cograph}, and thus $(H,\sigma_H)\join(H',\sigma_{H'})$ is
also an \hc-cograph. Condition \AX{K3cc}, however, is not so well-behaved:
\begin{example}\label{exmpl:brackets}
  Consider the single vertex graphs $(G_i,\sigma_i)$ with vertex set
  $V_i=\{i\}$, $1\leq i \leq 4$ and colors $\sigma(i)=r$ if $i$ is odd and
  $\sigma(i)=s\neq r$ if $i$ is even. Consider the graph
  $G= G_1\cupdot(G_2\cupdot(G_3\join G_4))$. By construction
  $(G_3\join G_4, \sigma_{|\{3,4\}})$ is an \hc-cograph because
  $\sigma(V_3)\cap \sigma(V_4)=\emptyset$ and thus \AX{(K2cc)} is
  satisfied. Also, $(G_2\cupdot(G_3\join G_4)),\sigma_{|\{2,3,4\}})$ is an
  \hc-cograph since
  $\sigma(V_2)=\{s\}\subseteq \sigma(V_3\cup V_4)=\{r,s\}$ and thus
  \AX{(K3cc)} is satisfied. Checking \AX{(K3cc)} again, we verify that
  $(G,\sigma)$ is an \hc-cograph. By associativity of $\join$ and
  $\cupdot$, we also have
  $G'=(G_1 \cupdot G_2) \cupdot (G_3\join G_4) = G$. However,
  $(G_1 \cupdot G_2, \sigma_{\{1,2\}})$ is not an \hc-cograph because
  $\sigma(V_1)\cap \sigma(V_2)=\emptyset$ implies that
  $(G_1 \cupdot G_2, \sigma_{\{1,2\}})$ does not satisfy Property
  \AX{(K3cc)}.
\end{example}

As a consequence, we cannot simply contract edges in the \hc-cotree
$T^G_{\hc}$ with incident vertices labeled by the $\cupdot$ operation. In
other words, it is not sufficient to use discriminating trees to represent
\hc-cotrees. Moreover, not every (binary) tree with colored leaves and
internal vertices labeled with $\join$ or $\cupdot$ (which specifies a
cograph) determines an \hc-cotree, because in addition the
color-restrictions \AX{(K2cc)} and \AX{(K3cc)} must be satisfied for each
internal vertex.

\begin{clemma}{\ref{lem:hc-properties}}
  Every \hc-cograph $(G,\sigma)$ is a properly colored cograph.
\end{clemma}

Not every properly colored cograph is an \hc-cograph, however. The simplest
counterexample is $\overline{K_2}=K_1\cupdot K_1$ with two differently
colored vertices, violating \AX{(K3cc)}. The simplest connected
counterexample is the 3-colored $P_3$ since the decomposition
$P_3=(K_1\cupdot K_1)\join K_1$ is unique, and involves the non-\hc-cograph
$K_1\cupdot \NEW{K_1}$ with two distinct colors as a factor in the join.
The next result shows that co-RBMGs and \hc-cographs are actually
equivalent.
\begin{ctheorem}{\ref{thm:hc-cograph}/\ref{thm:explainThc}}
  A vertex labeled graph $(G,\sigma)$ is a co-RBMG if and only if it is an
  \hc-cograph. Moreover, every co-RBMG $(G,\sigma)$ is explained by its cotree
  $(T^G_{\hc},\sigma)$.
\end{ctheorem}

\begin{figure}
  \begin{center}
    \includegraphics[width=\textwidth]{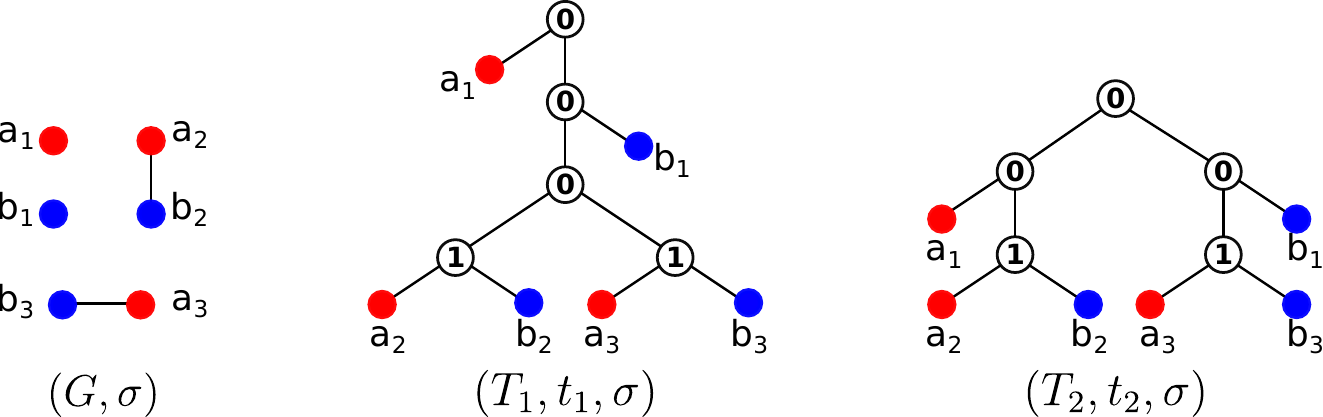}
  \end{center}
  \caption{The graph $(G,\sigma)$ is an \hc-cograph and, by Thm.\
    \ref{thm:hc-cograph}, a co-RBMG.  The trees $(T_1,t_1,\sigma)$ and
    $(T_2,t_2,\sigma)$ correspond to two possible cotrees
    $(T_{\hc}^G,t_i,\sigma)$ that explain $(G,\sigma)$. The inner labels
    ``$0$'' and ``$1$'' in the cotrees correspond to the values of the maps
    $t_i:V^0\to \{0,1\}$, $i=1,2$, such that $xy\in E(G)$ if and only if
    $t(\lca(x,y)) =1$.  Let $\mathscr{G}_{x} = ((\{x\},\emptyset),\sigma_x)$
    be the colored single vertex graph $K_1$ with $\sigma_x(x) = \sigma(x)$
    for each $x\in \{a_1,a_2,a_3,b_1,b_2,b_3\}$ as indicated in the figure.
    The tree $(T_1,t_1,\sigma)$ is constructed based on the valid
    \hc-cograph representation
    $G=\mathscr{G}_{a_1} \cupdot (\mathscr{G}_{b_1} \cupdot
    ((\mathscr{G}_{a_2}\join \mathscr{G}_{b_2})
    \cupdot(\mathscr{G}_{a_3}\join \mathscr{G}_{b_3})))$.  Here,
    $\mathscr{G}_{a_1}$ plays the role of $\mathscr{G}_{\ell}$ as in Lemma
    \ref{lem:minimal} in the technical part.  The tree $(T_2,t_2,\sigma)$
    is constructed based on the valid \hc-cograph representation
    $G=(\mathscr{G}_{a_1} \cupdot (\mathscr{G}_{a_2}\join \mathscr
    G_{b_2})) \cupdot (\mathscr{G}_{b_1} \cupdot (\mathscr{G}_{a_3}\join
    \mathscr{G}_{b_3}))$.}
  \label{fig:hc-cograph}
\end{figure}

\begin{ctheorem}{\ref{thm:coRBMG-recog}}
  Let $(G,\sigma)$ be a properly colored undirected graph.  Then it can be
  decided in polynomial time whether $(G,\sigma)$ is a co-RBMG and, in the
  positive case, a tree $(T,\sigma)$ that explains $(G,\sigma)$ can be
  constructed in polynomial time.
\end{ctheorem}

Some additional mathematical results and algorithmic considerations related
to the reconstruction of least resolved cotrees that explain a given RBMG
are discussed in Section \ref{sec:rec-hc-cogr} of the technical part.

\section{Concluding Remarks} 

Reciprocal best match graphs are the symmetric parts of best match graphs
\cite{Geiss:18x}. They have a surprisingly complicated structure that makes
it quite difficult to recognize them. Although we have succeeded here in
obtaining a complete characterization of 3-RBMGs, it remains an open
problem whether the general $n$-RBMGs can be recognized in polynomial
time. This is in striking contrast to the directed BMGs, which are
recognizable in polynomial time \cite{Geiss:18x}. The key difference
between the directed and symmetric version is that every BMG $(\G,\sigma)$
is explained by a unique least resolved tree which is displayed by every
tree $(T,\sigma)$ that explains $(\G,\sigma)$. RBMGs, in contrast, can be
explained by multiple, mutually inconsistent trees.  This ambiguity seems
to be the root cause of the complications that are encountered in the
context of RBMGs with more than $3$ colors.

An important subclass of RBMGs are the ones that have cograph structure
(co-RBMGs). These are good candidates for correct estimates of the
orthology relation. Interestingly, they are easy to recognize: by Theorem
\ref{thm:cographA} it suffices to check that all connected 3-colored
restrictions are cographs. Moreover, hierarchically colored cographs
(\hc-cographs) characterize co-RBMGs. Thm.\ \ref{thm:coRBMG-recog} shows
that co-RBMGs $(G,\sigma)$ can be recognized in polynomial time. 
Moreover, Thm.\ \ref{thm:coRBMG-recog} and \ref{thm:last} imply that 
a least resolved tree that explains 
$(G,\sigma)$ can be constructed in polynomial time.  Since every orthology
relation is equivalently represented by a cograph, every co-RBMG
$(G,\sigma)$ represents an orthology relation.  The converse, however, is
not always satisfied, as not all mathematically valid orthology relations
are \hc-cographs. The relationships of orthology relations and RBMGs,
however, will be the topic of a forthcoming contribution.

The practical motivation for considering the mathematical structure of
RBMGs is the fact that reciprocal best hit (RBH) heuristics are used
extensively as the basis of the most widely-used orthology detection
tools. A complete characterization of RBMGs is a prerequisite for the
development of algorithms for the ``RBMG-editing problem'', i.e., the task
to correct an empirically determined reciprocal best hit graph $(G,\sigma)$
to a mathematically correct RBMG. Empirical observations e.g.\ by
\citet{Hellmuth:15a} indicate that reciprocal best hit heuristics typically
yield graphs with fairly large edit distances from cographs and thus
orthology relations. We therefore suggest that orthology detection
pipelines could be improved substantially by inserting first RBMG-editing
and then the removal of good $P_4$s, followed by a variant of cograph
editing that respects the \hc-cograph structure. Some of these aspects will
be discussed in forthcoming work, some of which heavily relies on Lemmas
that we have relegated here to the technical part.

A number of interesting questions remain open for future research.  Most
importantly, $n$-RBMGs that are not cographs are not at all well
understood. While the classification of the $P_4$s in terms of the
underlying directed BMGs and the characterization of the 3-color case
provides some guidance, many problems remain. In particular, can we
recognize $n$-RBMGs in polynomial time? Is the information contained in
triples derived from 3-colored connected components sufficient, even if
this may not lead to a polynomial time recognition algorithm? Regarding the
connection of RBMGs and orthology relations, it will be interesting to ask
whether and to what extent the color information on the $P_4$s can help to
identify false positive orthology assignments. A possibly fruitful way of
attacking this issue is to ask whether there are trees that are displayed
from all $(T,\sigma)$ that explain an RBMG $(G,\sigma)$. For instance, when
is the discriminating cotree of the \hc-cotree $(T^G_{hc},t,\sigma)$
displayed by all \hc-cotrees explaining $(G,\sigma)$? And if so, are they
associated with an event-labeling $t$ at the inner vertices, and can $t$ be
leveraged to improve orthology detection?

Last but not least, we have not at all considered the question of
reconciliations of gene and species trees \cite{HHH+12,Hel-17}. While BMGs
and RBMGs do not explicitly encode information on reconciliation, it is
conceivable that the vertex coloring imposes constraints that connect to
horizontal transfer events. Conversely, can complete or partial information
on the Fitch relation \cite{Geiss:18a,HS:19}, which encodes the information on
horizontal transfer events, be integrated e.g.\ to provide additional
constraints on the trees explaining $(G,\sigma)$? The Fitch relation is
non-symmetric, corresponding to a subclass of directed co-graphs. Since
directed co-graphs \cite{CP-06} are also connected to generalization of
orthology relations that incorporate HGT events \cite{Hellmuth:17a}, it
seems worthwhile to explore whether there is a direct connection between
BMGs and directed cographs, possibly for those BMGs whose symmetric part
is an \hc-cograph.

\section*{Acknowledgements}
  Partial financial support by the German Federal Ministry of Education and
  Research (BMBF, project no.\ 031A538A, de.NBI-RBC) is gratefully
  acknowledged.

\bibliographystyle{plainnat}
\bibliography{rbmg}

\begin{appendix}

\section*{TECHNICAL PART}

\section{Least Resolved Trees}
\label{sect:leastres}

The understanding of least resolved trees, i.e., the ``smallest'' trees
that explain a given RBMGs relies crucially on the properties of BMGs. We
therefore start by recalling some pertinent results by \citet{Geiss:18x}.
\begin{lemma}
  Let $(\G,\sigma)$ be a BMG with vertex set $L$.  Then, $x\rthin y$
  implies $\sigma(x)=\sigma(y)$. In particular, $(\G,\sigma)$ has no arcs
  between vertices within the same $\rthin$-class.  Moreover,
  $N^+(x)\ne\emptyset$, while the in-neighborhood $N^-(x)$ may be empty for
  all $x\in L$.
\label{lem:bmg-basics}
\end{lemma}

For an $\rthin$-class $\alpha$ of a BMG we define its color
$\sigma(\alpha)=\sigma(x)$ for some $x\in \alpha$. This is indeed
well-defined, since, by Lemma \ref{lem:bmg-basics}, all vertices within
$\alpha$ must share the same color. Definition 9 of \citet{Geiss:18x} is a
key construction in the theory of BMGs. It introduces the \emph{root
  $\rho_{\alpha,s}$ of an $\rthin$-class $\alpha$ with color
  $\sigma(\alpha)=r$ w.r.t.\ a second, different color $s\ne r$} in a tree
$(T,\sigma)$ that explains a BMG $(\G,\sigma)$ by means of the following
equation:
\begin{equation}
  \rho_{\alpha,s} \coloneqq 
  \max_{\substack{x\in\alpha \\ y\in N^+_s(\alpha)}} \lca_{T}(x,y), 
\label{eq:root_long}
\end{equation}
where $\max$ is taken w.r.t.\ $\prec_T$.  The roots $\rho_{\alpha,s}$ are
uniquely defined by $(T,\sigma)$ because the color-restricted
out-neighborhoods $N^+_s(\alpha)$ are determined by $(T,\sigma)$
alone. Since $\lca(x,y)=\lca(x,y')$ for any two
$y,y'\in N^+_s(x), x\in \alpha$, Equ.\ \eqref{eq:root_long} simplifies to
\begin{equation}
  \rho_{\alpha,s} \coloneqq 
  \max_{x\in\alpha} \lca_{T}(x,y).
\label{eq:root}
\end{equation}
Their most important property \cite[Lemma~14]{Geiss:18x} is
\begin{equation}
  N^+_s(\alpha)=L(T(\rho_{\alpha,s}))\cap L[s]
\label{eq:bmg14}
\end{equation}
for all $s\in S\setminus\{\sigma(\alpha)\}$.

\begin{figure}[t]
  \begin{tabular}{lcr}
    \begin{minipage}{0.45\textwidth}
      \begin{center}
        \includegraphics[width=0.75\textwidth]{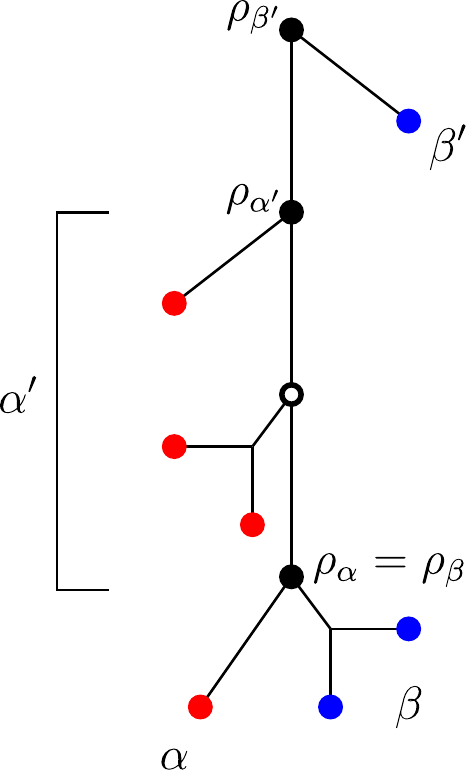}
      \end{center}
    \end{minipage}
    &    &
    \begin{minipage}{0.5\textwidth}
      \caption{Relationship between $\rthin$-classes and their roots. Shown
        is a tree with leaves from two colors (red and blue) whose leaf set
        consists of the four $\rthin$-classes $\alpha$, $\alpha'$ (red) and
        $\beta$, $\beta'$ (blue). The inner nodes of the tree corresponding
        to the roots $\rho_\alpha$, $\rho_{\alpha'}$, $\rho_\beta$ and
        $\rho_{\beta'}$ are marked in black.  \hfill\break Figure reused
        from \cite{Geiss:18x}, \copyright Springer}
      \label{fig:rthin_root}
    \end{minipage}
  \end{tabular}
\end{figure}

Least resolved for BMGs are unique \citet{Geiss:18x}. Here we consider an
analogous concept for RBMGs:
\begin{definition}
  Let $(G,\sigma)$ be an RBMG that is explained by a tree $(T,\sigma)$. An
  inner edge $e$ is called \emph{redundant} if $(T_e,\sigma)$ also explains
  $(G,\sigma)$, otherwise $e$ is called \emph{relevant}.
\label{def:redundant}
\end{definition}

The next result gives a characterization of redundant edges:
\begin{lemma}\label{lem:lr}
  Let $(G,\sigma)$ be an RBMG explained by $(T,\sigma)$. An inner edge
  $e=uv$ in $T$ is redundant if and only if $e$ satisfies the condition
  \begin{description}
  \item[\AX{(LR)}] For all colors
    $s\in \sigma(L(T(v)))\cap \sigma(\NEW{L(T(u))\setminus L(T(v))})$
    holds that if $v=\rho_{\alpha,s}$ for some $\rthin$-class
    $\alpha\in\mathcal{N}(\G(T,\sigma))$, then
    $\rho_{\beta,\sigma(\alpha)}\prec u$ for every $\rthin$-class
    $\beta\subseteq \NEW{L(T(u))\setminus L(T(v))}$
    of $\G(T,\sigma)$ with $\sigma(\beta)=s$.
  \end{description}
\end{lemma}
\begin{proof}
  The $\rthin$-classes appearing throughout this proof refer to the
  directed graph $(\G,\sigma)=\G(T,\sigma)$, and hence are completely
  determined by $(T,\sigma)$.  By definition, any redundant edge of
  $(T,\sigma)$ is an inner edge, thus we can assume that $e=uv$ is an inner
  edge of $(T,\sigma)$ throughout the whole proof.

  Suppose that Property \AX{(LR)} is satisfied. We show (with the help of
  Equ.\ \eqref{eq:bmg14}) that most neighborhoods in the BMG
  $(\G,\sigma)\coloneqq \G(T,\sigma)$ remain unchanged by the contraction
  of $e$, while those neighborhoods that change do so in such a way that
  $(T_e,\sigma)$ still explains the RBMG $(G,\sigma)$.

  We denote the inner vertex in $T_e$ obtained by contracting $e=uv$ again
  by $u$.  Recall that by convention $u\succ_T v$ in $T$. By construction,
  we have $L(T(w))=L(T_e(w))$ for all $w\ne v$ and
  $\lca_{T}(x,y)=\lca_{T_e}(x,y)$ unless $\lca_T(x,y)=v$.  Hence, a root
  $\rho_{\alpha,s}\ne v$ of $(T,\sigma)$ is also a root in $T_e$.  Equ.\
  \eqref{eq:bmg14} thus implies that $N^+_s(\alpha)$ remains unchanged upon
  contraction of $e$ whenever $\rho_{\alpha,s}\ne v$.

  Now let $\alpha$ and $s$ be such that $v=\rho_{\alpha,s}$, thus
  $N^+_s(\alpha)=L(T(v))\cap L[s]$ by Equ.\ \eqref{eq:bmg14} and in
  particular $s\in \sigma(L(T(v)))$. We distinguish two cases:\\
  (1) If $s\notin\sigma(L(T(u))\setminus L(T(v)))$, then there is no
  $\rthin$-class
  $\beta\subseteq L(T(u))\setminus L(T(v))$ of color $s$, which implies
  $L(T(u))\cap L[s]=L(T(v))\cap L[s]$. Hence, the
  set $N^+_s(\alpha)$ remains unaffected by contraction of $e$.\\
  (2) Assume $s\in\sigma(L(T(u))\setminus L(T(v)))$ and let
  $\beta\subseteq L(T(u))\setminus L(T(v))$
  be an $\rthin$-class of color $\sigma(\beta) = s$. Moreover\footnote{At
    this point, MH informed the coauthors via git commit from the delivery
    room that his daughter Lotta Merle was being born.}, let
  $\sigma(\alpha)=r\neq s$. We thus have $\rho_{\beta,r}\prec_T u$ by
  Property \AX{(LR)}. Now, $N^+_s(\alpha)=L(T(v))\cap L[s]$ and
  $\beta\subseteq L(T(u))\setminus L(T(v))$
  imply $\beta\cap N^+_s(\alpha)=\emptyset$. Moreover, Equ.\
  \eqref{eq:bmg14} and $\rho_{\beta,r}\prec_T u$ imply that
  $\alpha\cap N^+_r(\beta)=\emptyset$ in $(T,\sigma)$, i.e.,
  $xy\notin E(G)$ for any $x\in\alpha$ and $y\in\beta$ since neither
  $(x,y)$ nor $(y,x)$ is an arc in $\G$. After contraction of $e$, we have
  $\rho_{\beta, r}\prec\rho_{\alpha,s}$, i.e.,
  $\beta \subseteq N^+_s(\alpha)$, but $\alpha\cap N^+_r(\beta)=\emptyset$
  in $(T_e,\sigma)$ by Equ.\ \eqref{eq:bmg14}. Thus we have
  $(x,y)\in E(\G)$ and $(y,x)\notin E(\G)$, which implies
  $xy\notin E(G(T_e,\sigma))$. In summary, we can therefore conclude that
  $(T_e,\sigma)$ still explains $(G,\sigma)$.

  Conversely, suppose that $e$ is a redundant edge. If there is no
  $\rthin$-class $\alpha$ with $v=\rho_{\alpha,s}$, then Equ.\
  \eqref{eq:bmg14} again implies that contraction of $e$ does not affect
  the out-neighborhoods of any $\rthin$-classes, thus $(T_e,\sigma)$
  explains $(G,\sigma)$. Hence assume, for contradiction, that there is a
  color $s\in\sigma(L(T(v)))\cap\sigma(L(T(u))\setminus L(T(v)))$ and an
  $\rthin$-class \NEW{$\beta\subseteq L(T(u))\setminus L(T(v))$ of color
    $s$} with $\rho_{\beta,r}\succeq u$, where $r\in S\setminus \{s\}$ such
  that there exists an $\rthin$-class $\alpha$ of color $\sigma(\alpha)=r$
  with $v=\rho_{\alpha,s}$. Note that this in particular means that there
  is no leaf $z$ of color $r$ in $L(T(u))\setminus L(T(v))$ as otherwise
  $\lca(\beta,z)\prec_T u=\rho_{\beta,r}=\lca(\beta,\alpha)$; a
  contradiction since $\alpha \in N_r^+(\beta)$ by Equ.\
  \eqref{eq:bmg14}. Since by construction $\alpha\prec v$, we have
  $u=\lca(\alpha,\beta)$ and therefore $\rho_{\beta,r}= u$. In particular,
  it holds $\rho_{\beta,r}\succ \rho_{\alpha,s}$. As a consequence, we have
  $\beta \cap N^+_s(\alpha)=\emptyset$ and $\alpha\subseteq N^+_r(\beta)$
  in $(T,\sigma)$, again by Equ.\ \eqref{eq:bmg14}. Thus, for any
  $x\in \alpha$ and $y\in \beta$ we have $(x,y)\notin E(\G)$ and
  $(y,x)\in E(\G)$, and therefore $xy\notin E(G)$. Since
  $\rho_{\beta,r}= u$, contraction of $e$ implies
  $\rho_{\beta,r}= \rho_{\alpha,s}$ in $(T_e,\sigma)$. Therefore,
  $(x,y)\in E(\G)$ and $(y,x)\in E(\G)$, which implies
  $xy\in E(G(T_e,\sigma))$. Thus $(T_e,\sigma)$ does not explain
  $(G,\sigma)$; a contradiction.  
\end{proof} 
Note that the characterization of redundant edges requires information on
(directed) best matches. In particular, Property \AX{(LR)} requires
$\rthin$-classes.

The next result, Lemma \ref{lem:lem2-2}, provides alternative sufficient
conditions for least resolved trees.  In particular, it shows whether inner
edges $uv$ can be contracted based on the particular colors of leaves below
the children of $u$.  We will show in the last section that the conditions
in Lemma \ref{lem:lem2-2} are also necessary for RBMGs that are cographs
(cf.\ Lemma \ref{lem:lem2-2-necessary}).  These conditions are thus
designed to fit in well within the framework of RBMGs that are cographs,
which will be introduced in more detail later, although these conditions
may be relaxed for the general case.
	
\begin{lemma}
  Let $(G,\sigma)$ be an RBMG explained by $(T,\sigma)$ and let $e=uv$ be
  an inner edge of $T$.  Moreover, for two vertices $x,y$ in $T$, we define
  $S_{x,\neg y}\coloneqq \sigma(L(T(x)))\setminus \sigma(L(T(y)))$.  Then
  $(T_e,\sigma)$ explains $(G,\sigma)$, if one of the following conditions
  is satisfied:
 \begin{itemize}
 \item[(1)] $\sigma(L(T(v'))) \cap \sigma(L(T(v))) = \emptyset$ for all
   $v'\in\child_T(u)$, or
 \item[(2)]
   $\sigma(L(T(v'))) \cap \sigma(L(T(v))) \in
   \{\sigma(L(T(v))),\sigma(L(T(v')))\}$ for all $v'\in\child_T(u)$, and
   either
   \begin{itemize}
   \item[(i)] $\sigma(L(T(v)))\subseteq \sigma(L(T(v')))$ for all
     $v'\in\child_T(u)$, or
   \item[(ii)] if $\sigma(L(T(v')))\subsetneq \sigma(L(T(v)))$ for some
     $v'\in\child_T(u)$, then, for
     every $w\in\child_T(v)$ that satisfies $S_{w,\neg v'} \neq \emptyset$,
		 it holds that $\sigma(L(T(v')))$ and $\sigma(L(T(w)))$ do
     not overlap and thus, $\sigma(L(T(v')))\subseteq \sigma(L(T(w)))$ .
   \end{itemize}
 \end{itemize}
 \label{lem:lem2-2}
\end{lemma}
\begin{proof} 
  Suppose that $e=uv$ satisfies one of the Properties (1) or (2). If
  Property (1) is satisfied, we clearly have
  $\sigma(L(T(v)))\cap \sigma(L(T(u))\setminus L(T(v)))=\emptyset$, which
  implies that Condition \AX{(LR)} of Lemma \ref{lem:lr} is trivially
  satisfied. Therefore, $e$ is redundant in $(T,\sigma)$ and, by Def.\
  \ref{def:redundant}, $(T_e,\sigma)$ explains $(G,\sigma)$.

  Now let
  $\sigma(L(T(v'))) \cap \sigma(L(T(v))) \in
  \{\sigma(L(T(v))),\sigma(L(T(v')))\}$ for all $v'\in\child_T(u)$ and
  assume that either Property (2.i) or (2.ii) is satisfied. In order to see
  that $(T_e,\sigma)$ explains $(G,\sigma)$, we show that $e$ is redundant
  in $(T,\sigma)$ by application of Lemma \ref{lem:lr}. Thus suppose
  $v=\rho_{\alpha,s}$ for some $\rthin$-class
  $\alpha\in\mathcal{N}(\G(T,\sigma))$. If there exists no $\rthin$-class
  $\beta\subseteq L(T(u))\setminus L(T(v))$ of $\G(T,\sigma)$ with
  $\sigma(\beta)=s$, then Lemma \ref{lem:lr} is again trivially satisfied
  and $(T_e,\sigma)$ explains $(G,\sigma)$.  Hence, suppose that there is
  an $\rthin$-class $\beta\subseteq L(T(u))\setminus L(T(v))$ of
  $\G(T,\sigma)$ with $\sigma(\beta)=s$.  Clearly, if
  $\beta\preceq_T x\prec_T u$ for some $x\in \child_T(u)\setminus \{v\}$
  with $\sigma(L(T(v)))\subseteq \sigma(L(T(x)))$, then
  $\rho_{\beta,\sigma(\alpha)}\preceq_T x\prec_T u$.

  Hence, if Property (2.i) holds, i.e.,
  $\sigma(L(T(v))) \subseteq \sigma(L(T(v')))$ for all $v'\in \child_T(u)$,
  we easily see that for all $\rthin$-classes
  $\beta\subseteq L(T(u)\setminus T(v))$ with $\sigma(\beta)=s$ we have
  $\rho_{\beta,\sigma(\alpha)}\preceq_T x\prec_T u$ for some
  $x\in \child_T(u)\setminus \{v\}$. Therefore, $e$ is redundant in
  $(T,\sigma)$ and $(T_e,\sigma)$ explains $(G,\sigma)$.
  
  Now suppose that Property (2.ii) holds. If
  $\sigma(\alpha)\in \sigma(L(T(v')))$ for each $v'\in \child_T(u)$, we
  easily see that $\rho_{\beta,\sigma(\alpha)}\preceq_T x \prec_T u$ for
  some $x\in\child_T(u)\setminus \{x\}$. Otherwise, there exists some
  $\tilde v\in\child_T(u)\setminus \{v\}$ such that
  $\sigma(\alpha)\notin \sigma(L(T(\tilde v)))$. By Property (2),
  $\sigma(L(T(\tilde v)))$ and $\sigma(L(T(v)))$ do not overlap. Therefore,
  $\sigma(L(T(\tilde v))) \subsetneq \sigma(L(T(v)))$.  In order to show
  that \AX{(LR)} is satisfied, we thus need to show that
  $s\notin \sigma(L(T(\tilde v)))$, otherwise $\rho_{\beta',s}=u$ for some
  $\rthin$-class $\beta' \subseteq L(T(u))\setminus L(T(v))$ of
  $\G(T,\sigma)$. Let $w\in\child_T(v)$ such that $a\preceq_T w$ for some
  $a\in \alpha$. Since $\sigma(\alpha)\notin \sigma(L(T(\tilde v)))$, it
  follows $S_{w,\neg \tilde v}\neq \emptyset$. Hence, by Property (2.ii),
  it must hold $\sigma(L(T(\tilde v)))\subseteq \sigma(L(T(w)))$. Since
  $\rho_{\alpha,s}=v$ by assumption, we necessarily have
  $s\notin \sigma(L(T(w)))$ and thus, as
  $\sigma(L(T(\tilde v)))\subseteq \sigma(L(T(w)))$, we can conclude
  $s\notin \sigma(L(T(\tilde v)))$.  Thus, for all children
  $v'\in \child_T(u)$, we either have $\sigma(\alpha)\in \sigma(L(T(v')))$
  or $\sigma(\alpha),\sigma(\beta)\not\in \sigma(L(T(v')))$.  Now, one can
  easily see that $\rho_{\beta,\sigma(\alpha)}\preceq_T x \prec_T u$ for
  some $x\in\child_T(u)\setminus \{x\}$.  Hence, Condition \AX{(LR)} from
  Lemma \ref{lem:lr} is always satisfied. Therefore, the edge $e$ is
  redundant in $(T,\sigma)$, i.e., $(T_e,\sigma)$ explains $(G,\sigma)$.
  
\end{proof}

\begin{definition}\label{def:series-edge-contract}
  Let $(G,\sigma)$ be an RBMG explained by $(T,\sigma)$. Then $(T,\sigma)$
  is \emph{least resolved w.r.t.\ $(G,\sigma)$} if $(T_A,\sigma)$ does not
  explain $(G,\sigma)$ for any non-empty series of edges $A$ of
  $(T,\sigma)$.
\end{definition}

Fig.\ \ref{fig:lr_nonunique} gives an example of least resolved trees that
are not unique. We summarize the discussion as
\begin{theorem}\label{thm:lr}
  Let $(G,\sigma)$ be an RBMG explained by $(T,\sigma)$. Then there exists
  a (not necessarily unique) least resolved tree
  $(T_{e_1\dots e_k},\sigma)$ explaining $(G,\sigma)$ obtained from
  $(T,\sigma)$ by a series of edge contractions $e_1 e_2 \dots e_k$ such
  that the edge \NEW{$e_1$ is redundant in $(T,\sigma)$ and} $e_{i+1}$ is
  redundant in $(T_{e_1\dots e_i},\sigma)$ for $i\in \{1,\dots,k-1\}$. In
  particular, $(T,\sigma)$ displays $(T_{e_1\dots e_k},\sigma)$.
\end{theorem}
\begin{proof}
  The Theorem follows directly from the definition of least resolved trees
  and the observation that for any two redundant edges $e\neq f$ of
  $(T,\sigma)$, the tree $(T_{ef},\sigma)$ does not necessarily explain
  $(G,\sigma)$.  Clearly, by definition, $(T_{e_1\dots e_k},\sigma)$ is
  displayed by $(T,\sigma)$.
  
\end{proof}

\section{$\sthin$-Thinness}\label{sect:thin}

Fig.\ \ref{fig:sameNeighbors} shows that $N(a)=N(b)$ does not necessarily
imply $\sigma(a)=\sigma(b)$ for RBMGs. We therefore work here with a
color-preserving variant of thinness.
\begin{definition} 
  Let $(G,\sigma)$ be an undirected colored graph. Then two vertices $a$
  and $b$ are in relation $\sthin$, in symbols $a\sthin b$, if
  $N(a) = N(b)$ and $\sigma(a) = \sigma(b)$.\\
  An undirected colored graph $(G,\sigma)$ is $\sthin$-thin if no two
  distinct vertices are in relation $\sthin$. We denote the $\sthin$-class
  that contains the vertex $x$ by $[x]$.
  \label{def:sthin}
\end{definition}

As a consequence of Lemma \ref{lem:bmg-basics} and the fact that every
RBMG $(G,\sigma)$ is the symmetric part of some BMG $\G(T,\sigma)$, we
obtain
\begin{lemma}
  Let $(G,\sigma)$ be an RBMG, $(T,\sigma)$ a tree explaining $(G,\sigma)$,
  and $\G(T,\sigma)$ the corresponding BMG. Then $x\rthin y$ in
  $\G(T,\sigma)$ implies that $x\sthin y$ in $(G,\sigma)$.
\label{lem:SRthin}
\end{lemma}
The converse of Lemma \ref{lem:SRthin} is not true, however. A
counterexample can be found in Fig.\ \ref{fig:sameNeighbors}.

For an undirected colored graph $(G,\sigma)$, we denote by $G/\sthin$ the
graph whose vertex set are exactly the $\sthin$-classes of $G$, and two
distinct classes $[x]$ and $[y]$ are connected by an edge in $G/\sthin$ if
there is an $x'\in [x]$ and $y'\in [y]$ with $x'y'\in E(G)$. Moreover,
since the vertices within each $\sthin$-class have the same color, the map
$\sigmasthin \colon V(G/\sthin) \to S$ with $\sigmasthin([x]) = \sigma(x)$
is well-defined.

\begin{lemma}\label{lem:sthin}
  $(G/\sthin, \sigmasthin)$ is $\sthin$-thin for every undirected colored
  graph $(G,\sigma)$. Moreover, $xy\in E(G)$ if and only if
  $[x][y] \in E(G/\sthin)$. Thus, $G$ is connected if and only if
  $G/\sthin$ is connected.
\end{lemma}
\begin{proof}
  First, we show that $xy\in E(G)$ if and only if $[x][y] \in E(G/\sthin)$.
  Assume $xy\in E(G)$. Since $G$ does not contain loops, we have
  $x\notin N(x)$. However, $x\in N(y)$. Therefore, $N(x)\neq N(y)$ and
  thus, $[x]\neq [y]$. By definition, thus, $[x][y] \in E(G/\sthin)$.
  
  Now assume $[x][y] \in E(G/\sthin)$. By construction, there exists
  $x'\in [x]$ and $y'\in [y]$ such that $x'y'\in E(G)$ and thus
  $x'\in N(y') = N(y)$ and $y'\in N(x') = N(x)$. In particular,
  $x'y'\in E(G)$ implies $\sigma(x')\neq\sigma(y')$ and thus,
  $\sigma(x)\neq\sigma(y)$ since by definition all vertices within an
  $\sthin$-class are of the same color. Therefore, $xy\in E(G)$ by
  definition of $\sthin$-thinness.

  Now suppose, for contradiction, that $(G/\sthin,\sigmasthin)$ is not
  $\sthin$-thin.  Then, there are two distinct vertices $[x],[y]$ in
  $G/\sthin$ that have the same neighbors $[v_1],\dots,[v_k]$ in $G/\sthin$
  and $\sigmasthin([x]) = \sigmasthin([y])$ and, in particular,
  $\sigma(x)=\sigma(y)$. From ``$xy\in E(G)$ if and only if
  $[x][y] \in E(G/\sthin)$'' we infer
  $N_G(x) = \bigcup_{i=1}^k \bigcup_{v\in [v_i]} \{v\}= N_G(y)$ and thus
  $[x] = [y]$; a contradiction. Thus, $(G/\sthin, \sigmasthin)$ must be
  $\sthin$-thin.  
\end{proof}

The map $\gamma_{\sthin} \colon V(G) \to V(G/\sthin): x\mapsto [x]$
collapses all elements of an $\sthin$-thin class in $(G,\sigma)$ to a
single node in $(G_{/\sthin},\sigmasthin)$. Hence, the
$\gamma_{\sthin}$-image of a connected component of $(G,\sigma)$ is a
connected component in $(G_{/\sthin},\sigmasthin)$. Conversely, the
pre-image of a connected component of $(G_{/\sthin},\sigmasthin)$ that
contains an edge is a single connected component of
$(G,\sigma)$. Furthermore, $(G_{/\sthin},\sigmasthin)$ contains at most one
isolated vertex of each color $r\in S$. If it exists, then its pre-image is
the set of all isolated vertices of color $r$ in $(G,\sigma)$; otherwise
$(G,\sigma)$ has no isolated vertex of color $r$.

The next lemma shows how a tree $(T,\sigma)$ that explains an RBMG
$(G,\sigma)$ can be modified to a tree that still explains $(G,\sigma)$ by
replacing edges that are connected to vertices within the same
$\sthin$-class. Although this lemma is quite intuitive, one needs to be
careful in the proof since changing edges in $(T,\sigma)$ may also change
the neighborhoods $N_G(x)$ of vertices $x\in V(G)$ and may result in a tree
that does not explain $(G,\sigma)$ anymore.

\begin{lemma}
  Let $(G,\sigma)$ be an RBMG that is explained by $(T,\sigma)$ on $L$.
  Let $x,x'\in [x]$ be two distinct vertices in an $\sthin$-class $[x]$
  of $(G,\sigma)$. Suppose that $x$ and $x'$ have distinct parents $v_x$
  and $v_{x'}$ in $T$, respectively.  Denote by $T_{x',v_x}$ the tree on
  $L$ obtained from $T$ by (i) removing the edge $(v_{x'},x')$, (ii)
  suppressing the vertex $v_{x'}$ if it now has degree $2$, and (iii)
  inserting the edge $(v_{x},x')$.  Then, $(T_{x',v_x},\sigma)$ explains
  $(G,\sigma)$.
  \label{lem:same_parent}
\end{lemma}
\begin{proof}	
  Let $[x]$ be an $\sthin$-class with vertices $x,x'\in[x]$ that have
  distinct parents $v_x$ and $v_{x'}$ in $T$, respectively. Put
  $T' = T_{x',v_x}$ and let $(G',\sigma)$ be the RBMG explained by
  $(T',\sigma)$.  We proceed with showing that $(G',\sigma)=(G,\sigma)$.
  To see this, we observe that $x,x'\in[x]$ implies that $N_G(x) = N_G(x')$
  and $\sigma(x)=\sigma(x')$. By construction we also have
  $N_{G'}(x)=N_{G'}(x')$ and $x'\notin N_{G'}(x)$. Moving $x'$ in $T$ does
  not affect the last common ancestors of $x$ and any $y\ne x'$, hence
  $N_{G'}(x)=N_{G}(x)$, and thus also $N_{G'}(x')=N_{G}(x)$.  Now consider
  $N_{G'}(y)$ and $N_{G}(y)$ for some $y\neq x,x'$ and assume, for
  contradiction, that $N_{G'}(y) \neq N_{G}(y)$. Then there exists a vertex
  $z\in N_{G}(y)\setminus N_{G'}(y)$ or $z\in N_{G'}(y)\setminus N_{G}(y)$,
  which in particular implies $N_G(z)\neq N_{G'}(z)$. As shown above,
  $N_{G'}(x)=N_{G}(x) = N_{G}(x') = N_{G'}(x')$.  Hence,
  $N_G(z)\neq N_{G'}(z)$ implies $z\ne x,x'$.  Moreover, since $z$ is
  adjacent to $y$ in either $G$ or $G'$, we have $\sigma(z)\neq
  \sigma(y)$. However, replacing $x'$ in $T$ cannot influence the
  adjacencies between vertices $u$ and $v$ with $\sigma(u)\neq \sigma(x')$
  and $\sigma(v)\neq\sigma(x')$.  Taken the latter arguments together, we
  can conclude that $\sigma(z) = \sigma(x)\neq \sigma(y)$.
	
  First assume  $z\in N_{G}(y)\setminus  N_{G'}(y)$. Then  
  \begin{align}
    &\lca_{T}(z,y) \preceq_{T} \lca_{T}(z',y)  \text{ for all } z'
      \text{ with } \sigma(z') = \sigma(z) \text{ and} \label{eq:z0} \\
    &\lca_{T}(z,y) \preceq_{T} \lca_{T}(z,y') \text{ for all } y'
      \text{ with } \sigma(y') = \sigma(y).
  \end{align}
  Since $z\notin N_{G'}(y)$, we additionally have 
  \begin{align}
    &\lca_{T'}(z,y) \succ_{T'} \lca_{T'}(z',y) \text{ for some } z'
      \text{ with } \sigma(z') = \sigma(z) \text{ or} \label{eq:z1}\\
    &\lca_{T'}(z,y) \succ_{T'} \lca_{T'}(z,y') \text{ for some } y'
      \text{ with } \sigma(y') = \sigma(y). 
  \end{align}
  The fact that  $T$ and $T'$ are identical up to the location of $x'$
  together with  $\sigma(z')=\sigma(x')\neq\sigma(y)$ and $x'\neq z$
  implies that in $T'$ we still have $\lca_{T'}(z,y) \preceq_{T'}
  \lca_{T'}(z,y')$ for all $y'$ with $\sigma(y') = \sigma(y)$.
  Hence, Equ.\ \eqref{eq:z1} must be satisfied.
  Equ.\ \eqref{eq:z0} and \eqref{eq:z1} together imply that $x'=z'$ and
  that $x'$ is the only vertex that satisfies Equ.\ \eqref{eq:z1}.  
  In $T'$ the vertices $x$ and $x'$ have the same parent. Together with 
  $x'=z'$ and Equ.\ \eqref{eq:z1} this implies
  $\lca_{T'}(x,y) = \lca_{T'}(x',y) \prec_{T'} \lca_{T'}(z,y)$. 
  Since $T$ and $T'$ are identical up to the location of $x'$, we also have
  $\lca_{T'}(x,y) = \lca_{T}(x,y)$ and $\lca_{T'}(y,z) = \lca_{T}(y,z)$.
  Combining these arguments, we obtain $\lca_{T}(x,y) \prec_{T}
  \lca_{T}(y,z)$, which contradicts Equ.\ \eqref{eq:z0} because 
  $\sigma(z) = \sigma(x)$.
  
  Assuming $z\in N_{G'}(y)\setminus N_{G}(y)$, and interchanging the role
  of $T$ and $T'$ in the argument above, we obtain
  \begin{align}
    &\lca_T(z,y) \succ_T \lca_T(x',y) \text{ and } \label{eq:z2}\\
    &\lca_{T}(z,y) \preceq_{T} \lca_{T}(z,y') \text{ for all } y'
      \text{ with } \sigma(y') = \sigma(y)
  \end{align}
  and that there is no other vertex $z^*\neq x'$ with
  $\sigma(z^*)=\sigma(x')$ and $\lca_T(z,y) \succ_T \lca_T(z^*,y)$.  Since
  $x$ and $x'$ have the same parent in $T'$ we have
  $\lca_T(x,y) = \lca_{T'}(x,y) = \lca_{T'}(x',y) \succeq_{T'}
  \lca_{T'}(z,y)=\lca_{T}(z,y) \succ_T \lca_T(x',y)$.  The fact that $T$
  and $T'$ are identical up to the location of $x'$ now implies that for
  all inner vertices $v,w$ of $T'$ we have $v\prec_{T'}w$ if and only if
  $v\prec_T w$. Hence we have
  $$\lca_T(x,y) \succeq_{T} \lca_{T}(z,y)  \succ_T \lca_T(x',y)$$
  implying that $T$ displays the triple $x'y|x$. Therefore $xy$ is not an
  edge in $(G,\sigma)$, whence 	$y\notin N_G(x) = N_G(x')$.\\
  Since there is no other vertex $z^*\neq x'$ with $\sigma(z^*)=\sigma(x')$
  and $\lca_T(z,y) \succ_T \lca_T(z^*,y)$, we have
  $\lca_T(z^*,y) \succ_T \lca_T(x',y)$ for all $z^*\neq x'$ with
  $\sigma(z^*)=\sigma(z)=\sigma(x')$.  Since $y\notin N_G(x')$, there must
  be a vertex $y'$ with $\sigma(y')=\sigma(y)$ such that
  $\lca_T(x',y)\succ_T \lca_T(x',y')$. We can choose $y'$ such that there
  is no other vertex $y^*\neq y'$ satisfying $\sigma(y^*)=\sigma(y')$ and
  $\lca_T(x',y') \succ_T \lca_T(x',y^*)$. Thus we have
  $$\lca_T(x',y')\prec_T\lca_T(x',y)\prec_T \lca_T(x,y),$$
  which implies $y'\not\in N_G(x)$. However, since $x'$ is unique w.r.t.\
  Equ.\ \eqref{eq:z2}, we must have $y'\in N_G(x')$; a contradiction to
  $N_G(x)=N_G(x')$.\\
  Therefore, we have $N_G(v) = N_{G'}(v)$ for all $v\in V(G)$, and thus
  $(G,\sigma)=(G',\sigma)$ as claimed.

\end{proof}

\begin{lemma}
  $(G,\sigma)$ is an RBMG if and only if $(G/\sthin, \sigmasthin)$ is an
  RBMG.  Moreover, every RBMG $(G,\sigma)$ is explained by a tree
  $(\widehat{T},\sigma)$ in which any two vertices $x,x'\in [x]$ of each
  $\sthin$-class $[x]$ of $(G,\sigma)$ have the same parent.
\label{lem:Sthin-tree}
\end{lemma}
\begin{proof}
  Consider an RBMG $(G,\sigma)$ explained by the tree $(T,\sigma)$, and let
  $[x]$ be an $\sthin$-class of $(G,\sigma)$.  If all the vertices within
  $[x]$ have the same parent $v$ in $T$, then we can identify the edges
  $vx'$ for all $x'\in [x]$ to obtain the edge $v[x]$.  If all children of
  $v$ are leaves of the same color, we additionally suppress $v$ in order
  to obtain a phylogenetic tree $T/[x]$. Note that in this case,
  $\parent(v)$ cannot be incident to any leaf $y$ of color $\sigma(x)$ in
  $(T,\sigma)$ as this would imply $N(x)=N(y)$ and therefore $x\sthin
  y$. Hence, suppression of $v$ has no effect on any of the neighborhoods
  and thus does not affect any of the reciprocal best matches in
  $(T,\sigma)$.  If all $\sthin$-classes are of this form, then the tree
  $(T/\sthin,\sigmasthin)$ obtained by collapsing each class $[x]$ to a
  single leaf and potential suppression of 2-degree nodes still explains
  $(G/\sthin,\sigmasthin)$.

  The construction of $T_{x',v_x}$ as in Lemma \ref{lem:same_parent} can be
  repeated until all vertices $x'$ of each $\sthin$-class $[x]$ have been
  re-attached to have the same parent $v_x$. After each re-attachment step,
  the tree still explains $(G,\sigma)$. The procedure stops when all
  $x'\in [x]$ are siblings of $x$ in the tree, i.e., a tree
  $(\widehat{T},\sigma)$ of the desired form is reached. The tree obtained
  by retaining only one representative of each $\sthin$-class $[x]$
  (relabeled as $[x]$), explains $(G/\sthin, \sigmasthin)$.

  Conversely, assume that $(G/\sthin, \sigmasthin)$ is an RBMG explained by
  the tree $(\tilde T,\sigmasthin)$. Each leaf in $\tilde T$ is an
  $\sthin$-class $[x]$. Consider the tree $(T,\sigma)$ obtained by
  replacing, for all $\sthin$-classes $[x]$ the edge $\parent([x])[x]$ in
  $T$ by the edges $\parent([x])x'$ and setting
  $\sigma(x') = \sigmasthin([x])$ for all $x'\in[x]$. By construction,
  $(T,\sigma)$ explains $(G,\sigma)$, and thus $(G,\sigma)$ is an RBMG.

\end{proof}
Lemma~\ref{lem:Sthin-tree} is illustrated in Fig.\ \ref{fig:SthinTree}.

\begin{lemma}\label{lem:2col}
  Let $(G,\sigma)$ be an $\sthin$-thin $n$-RBMG explained by $(T,\sigma)$
  with $n\geq 2$. Then $|\sigma(L(T(v)))|\ge2$ holds for every inner vertex
  $v\in V^0(T)$. 
\end{lemma}
\begin{proof}
  \NEW{Let $S=\sigma(V(G))$.}  Assume, for contradiction, that there exists
  an inner vertex $v\in V^0(T)$ such that $\sigma(L(T(v)))=\{r\}$ with
  $r\in S$. Since $(T,\sigma)$ is phylogenetic, there must be two distinct
  leaves $a,b\in L(T(v))$ with $\sigma(a)=\sigma(b)=r$. Since $(G,\sigma)$
  is $\sthin$-thin, $a$ and $b$ do not belong to the same
  $\sthin$-class. Hence, $\sigma(a)=\sigma(b)$ implies $N(a)\neq
  N(b)$. Since $|S|\ge2$, there is a leaf $c\in V(G)$ with
  $\sigma(c)=s\in S\setminus \{r\}$. On the other hand,
  $\sigma(L(T(v)))=\{r\}$ implies $\lca(a,c)=\lca(b,c)\succ v$.
  
  Now consider the corresponding BMG $\G(T,\sigma)$. Since
  $\sigma(L(T(v)))=\{r\}$, we have $c\in N^-(a)$ if and only if
  $c\in N^-(b)$, and $c\in N^+(a)$ if and only if $c\in N^+(b)$. Together,
  this implies $N(a)=N(b)$ in $G(T,\sigma)$ ; a contradiction.
  
\end{proof}

Any two leaves $x,y$ in $(T,\sigma)$ with $\sigma(x)=\sigma(y)$ and
$\parent(x)=\parent(y)$ obviously belong to the same $\sthin$-equivalence
class of $G(T,\sigma)$. The absence of such pairs of vertices in
$(T,\sigma)$ is thus a necessary condition for $G(T,\sigma)$ to be
$\sthin$-thin, it is not sufficient, however. We leave it as an open
question for future research to characterize the leaf-colored trees that
explain $\sthin$-thin RBMGs.

\section{Connected Components, Forks, and Color-Complete Subtrees}
\label{sect:connect}

Although RBMGs, like BMGs, are not hereditary, they satisfy a related,
weaker property that will allow us to restrict our attention to connected
RBMGs.
\begin{lemma}
  Let $(G,\sigma)$ be an RBMG with vertex set $L$ explained by $(T,\sigma)$
  and let $(T_{|L'},\sigma_{|L'})$ be the restriction of $(T,\sigma)$ to
  $L'\subseteq L$. Then the induced subgraph
  $(G,\sigma)[L']\coloneqq (G[L'],\sigma_{|L'})$ of $(G,\sigma)$ is a (not
  necessarily induced) subgraph of $G(T_{|L'},\sigma_{|L'})$.
  \label{lem:fact:1}
\end{lemma}
\begin{proof}
  Lemma~1 in \cite{Geiss:18x} states the analogous result for BMGs.  It
  obviously remains true for the symmetric part.
  
\end{proof}

The next result is a direct consequence of Lemma \ref{lem:fact:1} that will
be quite useful for proving some of the following results.
\begin{corollary}
  Let $(G,\sigma)$ be an RBMG that is explained by $(T,\sigma)$.  Moreover,
  let $v\in V(T)$ be an arbitrary vertex and $(G^*_v,\sigma^*_v)$ be a
  connected component of $G(T(v), \sigma_{|L(T(v))})$. Then,
  $(G^*_v,\sigma^*_v)$ is contained in a connected component
  $(G^*,\sigma^*)$ of $(G,\sigma)$.
  \label{cor:subconnComp}
\end{corollary}

We next ensure the existence of certain types of edges in any RBMG.
\begin{lemma}\label{lem:rbm-pairs}
  Let $(T,\sigma)$ be a leaf-colored tree on $L$ and let $v\in V(T)$. Then,
  for any two distinct colors $r,s\in \sigma(L(T(v)))$, there is an edge
  $xy\in E(G(T,\sigma))$ with $x\in L[r]\cap L(T(v))$ and
  $y\in L[s]\cap L(T(v))$.  In particular, all edges in
  $G(T(v), \sigma_{|L(T(v))})$ are contained in $G(T,\sigma)$.
\end{lemma}
\begin{proof}
  Let $v$ be a vertex of $(T,\sigma)$ such that $r,s\in \sigma(L(T(v)))$,
  $r\ne s$. Then there is always an inner vertex $w\preceq_T v$ such that
  (i) $\{r,s\}\subseteq\sigma(L(T(w)))$ and (ii) none of its children
  $w_i\in\child(w)$ satisfies $\{r,s\}\subseteq\sigma(L(T(w_i)))$. Any such
  $w$ has children $w_r,w_s\in \child(w)$ such that
  $r\in\sigma(L(T(w_r)))$, $s\notin\sigma(L(T(w_r)))$ and
  $s\in\sigma(L(T(w_s)))$, $r\notin\sigma(L(T(w_s)))$. Thus
  $\lca_T(x,y)\succeq_T w$ for every $x\in L(T(w_r))\cap L[r]\ne\emptyset$
  and $y\in L[s]$, with equality whenever $y\in L(T(w_s))$. Analogously,
  $\lca_T(y,x)\succeq_T w$ for every $y\in L(T(w_s))\cap L[s]\ne\emptyset$
  and $x\in L[r]$, with equality whenever $x\in L(T(w_r))$. Hence $xy$ is a
  reciprocal best match mediated by $\lca_T(x,y)=w$ whenever
  $x\in L(T(w_r))\cap L[r]$ and $y\in L(T(w_s))\cap L[s]$. Therefore
  $xy\in E(G(T,\sigma))$.

  In particular, the latter construction shows that the chosen leaves
  $x\in L(T(w_r))\cap L[r]$ and $y\in L(T(w_s))\cap L[s]$ are reciprocal
  best matches in $(T(v), \sigma_{|L(T(v))})$. Hence, every edge in
  $G(T(v), \sigma_{|L(T(v))})$ is also contained in $G(T,\sigma)$.  
\end{proof}
    
As a direct consequence of Lemma \ref{lem:rbm-pairs}, we obtain
\begin{corollary}
  If $(G,\sigma)$ is an RBMG with $|S|\ge 2$ colors, then there is at least
  one edge $xy\in E(G[L[r]\cup L[s]])$ for any two distinct colors
  $r, s\in S$.
  \label{cor:funfact}
\end{corollary}

\begin{theorem}\label{thm:conn_comp}
  Let $(G^*,\sigma^*)$ with vertex set $L^*$ be a connected component of
  some RBMG $(G,\sigma)$ and let $(T,\sigma)$ be a leaf-colored tree
  explaining $(G,\sigma)$. Then, $(G^*,\sigma^*)$ is again an RBMG and is
  explained by the restriction $(T_{|L^*},\sigma_{|L^*})$ of $(T,\sigma)$
  to $L^*$.
\end{theorem}
\begin{proof}
  Throughout this proof, all $N^{+}$-neighborhoods are taken w.r.t.\ the
  underlying BMG $\G(T,\sigma)$. It suffices to show that
  $G(T_{|L^*},\sigma_{|L^*})=(G^*,\sigma^*)$. Lemma \ref{lem:fact:1}
  implies that $(G^*,\sigma^*)$ is a (not necessarily induced) subgraph of
  $G(T_{|L^*},\sigma_{|L^*})$, i.e.,
  $E(G^*)\subseteq E(G(T_{|L^*},\sigma_{|L^*}))$. By assumption,
  $(G^*,\sigma^*)$ is an induced subgraph of $(G,\sigma)$. Thus, we only need
  to prove that $E(G(T_{|L^*},\sigma_{|L^*}))\subseteq E(G^*)$.

  Assume, for contradiction, that there exists an edge $xy$ in
  $G(T_{|L^*},\sigma_{|L^*})$ that is not contained in $(G^*,\sigma^*)$.
  By definition, $r\coloneqq \sigma(x) \neq s\coloneqq \sigma(y)$ and, in
  particular, $x,y\in L^*$. Let $u\coloneqq \lca_T(x,y)$. 
  By construction, any two
  vertices within $L^*$ have the same last common ancestor in $(T,\sigma)$
  and $(T_{|L^*},\sigma_{|L^*})$.  Since the edge $xy$ is not contained in
  $(G^*,\sigma^*)$, the edge $xy$ is not contained in $(G,\sigma)$
  either. Hence, $x$ and $y$ do not form reciprocal best matches in
  $(T,\sigma)$.  Thus, there must exist some $x'\in L[r]$ with
  $\lca_T(x',y)\prec_T \lca_T(x,y)$, or a leaf $y'\in L[s]$ with
  $\lca_T(x,y')\prec_T \lca_T(x,y)$.

  W.l.o.g.\ we assume that the first case is satisfied.  Since
  $\lca_T(x',y)\prec_T \lca_T(x,y)$, we must have $x'\in L\setminus L^*$,
  as otherwise,
  $\lca_{T_{|L^*}}(x',y)\prec_{T_{|L^*}} \lca_{T_{|L^*}}(x,y)$ and hence,
  $x$ cannot be a best match of $y$, which in turn would imply that $xy$ is
  not an edge in $G(T_{|L^*},\sigma_{|L^*})$.  We will re-use the latter
  argument and refer to it as \emph{Argument-1}.

  In the following, w.l.o.g.\ we choose $x'\in L[r]$ such that
  $\lca_T(x',y)\prec_T \lca_T(x,y)$ and $\lca(x',y)$ is $\preceq_T$-minimal
  among all least common ancestors that satisfy the latter condition.  We
  write $v\coloneqq \lca_T(x',y)$. By construction, we have $v\prec_T u$.  By
  contraposition of \emph{Argument-1}, we have for all $x''\in L^*$ with
  $\sigma(x'') = r$ it must hold that $\lca_T(x'',y)\succeq_T \lca_T(x,y)$
  and thus, $x''\notin L(T(v))$. In other words, we have
  \begin{equation} \label{eq:l}
    x''\notin L^* \text{ for all\ }  x''\in L(T(v))\cap L[r]. 
  \end{equation} 

  Let $v_{x'},v_y\in \child(v)$ with $x'\preceq_T v_{x'}$ and
  $y\preceq_T v_y$. The choice of $x'$ and the resulting
  $\preceq_T$-minimality of $\lca(x',y)$ implies that
  $\sigma(x)=r\notin \sigma(L(T(v_y)))$.  Therefore, $x'\in N^+_r(y)$.  We
  observe that $x'y\notin E(G)$ since, otherwise, $x'\in L^*$; a
  contradiction.  From $x'y\notin E(G)$ we conclude $y\notin N^+_s(x')$ and
  thus there exists an $y'\in L[s]$ such that
  $\lca_T(x',y')\prec_T \lca_T(x',y)=v$ and hence, $y'\prec_T v_{x'}$.
         
  The latter, in particular implies that $r,s\in \sigma(L(T(v_{x'}))$.
  Hence, we can apply Lemma \ref{lem:rbm-pairs} to conclude that there are
  two vertices $\tilde x\in L[r]\cap L(T(v_{x'}))$ and
  $\tilde y\in L[s]\cap L(T(v_{x'}))$ such that $\tilde x\tilde y\in
  E(G)$. Equ. \eqref{eq:l} now implies $\tilde x \notin L^*$.  Therefore,
  $\tilde x\tilde y\in E(G)$ now allows us to conclude that
  $\tilde y\in L\setminus L^*$.

  Now, let $\mathscr{P}_{xy}=(x=a_0a_1a_2\dots a_{k-1} a_{k}=y)$ be a
  shortest path in $(G^*, \sigma^*)$ connecting $x$ and $y$. Since $x$ and
  $y$ reside within the same connected component $(G^*, \sigma^*)$ of
  $(G,\sigma)$ and $xy\notin E(G^*)$, such a path exists and, in
  particular, it must contain at least one $a_i\neq x,y$, i.e., $k>1$.  By
  definition of a shortest path, $a_ia_j\notin E(G)$ for all
  $i,j\in\{0,1,\dots, k\}$ that satisfy $|i-j|>1$.  Since $a_i\in L^*$ for
  any $0\leq i \leq k$ but $\tilde x,\tilde y\in L\setminus L^*$, we
  have \begin{equation}
    \label{eq:edge}
    \tilde xa_i, \tilde ya_i \notin E(G)
  \end{equation}
  for any $0\leq i \leq k$, since otherwise, $\tilde x$ and $\tilde y$
  would be contained in the connected component $(G^*,\sigma^*)$ and thus,
  also in $L^*$; a contradiction.

  We proceed to show by induction that
  \begin{enumerate}
  \item[(I1)] $a_i\in L(T(v))$, $1\le i \le k$, and
  \item[(I2)] there exists a vertex
    $\tilde a_i \in L(T(v))\cap L[\sigma(a_i)]$ such that
    $\tilde a_i\notin L^*$, $1\le i \le k$.
  \end{enumerate}
  We start with $i=k$.  By construction, $y=a_{k}\in L(T(v))$ satisfies
  Property (I1).  Moreover, $\tilde a_{k} \coloneqq \tilde y$ satisfies
  Property (I2). For the induction step assume that, for a fixed $m\leq k$,
  Property (I1) and (I2) is satisfied for all $i$ with $m < i\leq k$.

  Now, consider the case $i=m$. For better readability we put
  $b\coloneqq a_{m+1}$ and $\tilde b\coloneqq \tilde a_{m+1}$. By induction
  hypothesis, $b$ and $\tilde b$ satisfy Property (I1) and (I2),
  respectively. Since $a_m b\in E(G)$, we know that
  $\sigma(a_m)\neq \sigma(b)$.  In what follows, we consider the two
  exclusive cases: either $\sigma(a_m)=\sigma(x) = r$ or
  $\sigma(a_m)\neq r$. If $\sigma(a_m) = r$, then we put
  $\tilde a_m = \tilde x$. Hence, Property (I2) is trivially satisfied for
  $\tilde a_m$.  Moreover, $a_m$ must then be contained in $L(T(v))$,
  otherwise $v\succeq_T\lca_T(b,\tilde x)$ implies that
  $\lca_T(b,\tilde x)\prec_T \lca(b,a_m)$, which contradicts
  $a_m b\in E(G)$, i.e., Property (I1) is satisfied as well.

  In case $\sigma(a_m)\neq r$ assume first, for contradiction, that
  $a_m\notin L(T(v))$. Since $b,\tilde b \in L(T(v))$ we observe that
  $\lca_T(b,a_m)=\lca_T(\tilde b,a_m) \succ_T v$.  Note that we have
  $b\in N^+_{\sigma(b)}(a_m)$ since $ba_m\in E(G)$ by definition of
  $\mathscr{P}_{xy}$.  Thus, $\lca_T(b,a_m)=\lca_T(\tilde b,a_m)$ implies
  $\tilde b\in N^+_{\sigma(b)}(a_m)$.  Since $a_m\in L^*$ (by definition)
  and $\tilde b\notin L^*$ (by Property (I2)), we can conclude that
  $a_m\tilde b\notin E(G)$.  The latter two arguments imply that
  $a_m \notin N^+_{\sigma(a_m)}(\tilde b)$.  Hence, there exists a leaf
  $a_m'$ with $\sigma(a_m)=\sigma(a_m')$ such that
  $\lca_T(\tilde b, a_m')\prec_T \lca_T(\tilde b, a_m)$.  There are two
  cases, either $a_m'\in L(T(v_{\tilde b}))$ or
  $a_m'\notin L(T(v_{\tilde b}))$, where $v_{\tilde b}\in \child(v)$ with
  $\tilde b \preceq_T v_{\tilde b}$.  If $a_m'\in L(T(v_{\tilde b}))$, then
  $\lca_T(b, a_m')\preceq_T v$ and we can re-use the fact
  $\lca_T(b,a_m) \succ_T v$ from above to conclude that
  $\lca_T(b,a_m')\prec_T \lca_T(b,a_m)$.  If
  $a_m'\notin L(T(v_{\tilde b}))$, then
  $\lca_T(b,a_m')\preceq_T\lca_T(\tilde b, a_m')$. Thus, we have
  $\lca_T(b,a_m')\preceq_T\lca_T(\tilde b, a_m')\prec_T \lca_T(\tilde b,
  a_m) = \lca_T(b, a_m)$.  Hence, in either case we obtain
  $\lca_T(b,a_m')\prec_T \lca_T(b,a_m)$, thus $a_mb\notin E(G)$; a
  contradiction.  Therefore, $a_m\in L(T(v))$, i.e., Property (I1) is
  satisfied by $a_m$.

  To summarize the argument so far, Property (I1) is always satisfied
  for $a_m$, independent of the particular color $\sigma(a_m)$.  Moreover,
  Property (I2) is satisfied, in case $\sigma(a_m) = r$.  Thus, it remains
  to show that Property (I2) is also satisfied in case
  $\sigma(a_m)\neq r$. To this end, let $v_m\in \child(v)$ such that
  $a_m\preceq_T v_m$. If $r\in \sigma(L(T(v_m)))$, then Lemma
  \ref{lem:rbm-pairs} implies that there must exist leaves
  $\tilde x_m, \tilde a_m \in L(T(v_m))$ with $\sigma(\tilde x_m)=r$ and
  $\sigma(\tilde a_m)=\sigma(a_m)$ such that
  $\tilde x_m \tilde a_m \in E(G)$.  By Equ.\ \eqref{eq:l}, no vertex
  in $L(T(v))\cap L[r]$ is contained in $L^*$, and thus, we have
  $\tilde x_m \notin L^*$ and, since $\tilde x_m \tilde a_m \in E(G)$, it
  must also hold $\tilde a_m \notin L^*$.

  Otherwise, if $r\notin \sigma(L(T(v_m)))$, then $\sigma(\tilde x) = r$
  and $\tilde x\preceq_T v_{x'}$ implies that $v_m\neq v_{x'}$. Hence,
  $\lca(a_m, \tilde x)=v$. In particular, there is no vertex $x''\in L[r]$
  such that $\lca_T(a_m,x'')\prec_T\lca(a_m, \tilde x)=v$, thus
  $\tilde x\in N^+_r(a_m)$.  Since $a_m\in L^*$ and $\tilde x\notin L^*$,
  it must hold that $a_m \tilde x\notin E(G)$.  Thus, there must exist a
  leaf $\tilde a_m \in L[\sigma(a_m)]$ such that
  $\lca_T(\tilde x, \tilde a_m)\prec_T \lca_T(\tilde x, a_m)=v$, i.e.,
  $\sigma(a_m)\in \sigma(L(T(v_{x'})))$. We can therefore apply Lemma
  \ref{lem:rbm-pairs} to conclude that there must exist
  $\tilde x_m \in L(T(v_{x'})) \cap L[r]$ and
  $\tilde a_m \in L(T(v_{x'})) \cap L[\sigma(a_m)]$ such that
  $\tilde x_m \tilde a_m \in E(G)$. Analogous argumentation as in the case
  $r\in \sigma(L(T(v_m)))$ shows $\tilde x_m, \tilde a_m \notin
  L^*$. Hence, Property (I2) is satisfied, which completes the induction
  proof.

  Property (I1) finally implies that $a_1 \in L(T(v))$. Moreover, by
  construction of $\mathscr{P}_{xy}$ we have $xa_1\in E(G^*)$.  Property
  \eqref{eq:edge}, on the other hand, implies $\tilde xa_1\notin
  E(G^*)$. Consequently, we have
  $\lca_T(a_1,x)\prec_T \lca_T(a_1,\tilde x)$. This, however, contradicts
  $\lca_T(a_1,x) = u \succ_T v = \lca_T(a_1,\tilde x)$.  The shortest path
  $\mathscr{P}_{xy}$ can, therefore, consist only of the single edge $xy$,
  and hence $E(G(T_{|L^*},\sigma_{|L^*}))\subseteq E(G^*)$. Therefore
  $G(T_{|L^*},\sigma_{|L^*})=(G^*,\sigma^*)$ and $(T_{|L^*},\sigma_{|L^*})$
  explains $(G^*,\sigma^*)$. In particular, the connected component
  $(G^*,\sigma^*)$ is again an RBMG.
  
\end{proof}

So far, we have shown that every connected component of an $n$-RBMG is
therefore a $k$-RBMG possibly with a strictly smaller number $k$ of
colors. We next ask when the disjoint union of RBMGs is again an RBMG.  To
this end, we identify certain vertices in the leaf-colored tree
$(T,\sigma)$ that, as we shall see below, are related to the decomposition
of $G(T,\sigma)$ into connected components.

\begin{definition}
  Let $(T,\sigma)$ be a leaf-colored tree with leaf set $L$. An inner
  vertex $u$ of $T$ is \emph{color-complete} if
  $\sigma(L(T(u)))=\sigma(L)$. Otherwise it is \emph{color-deficient}.
\label{def:color-compl}
\end{definition}
We will also refer to a subtree $(T(u),\sigma_{|L(T(u))})$ of $(T,\sigma)$
as \emph{color-complete} if its root is color-complete.
    
We write $A(u)$ for the set of color-deficient children of $u$, i.e.,
\begin{equation}
  \label{eq:A(u)}
  A(u)\coloneqq \{v\mid v\in \child(u), \sigma(L(T(v))) \subsetneq \sigma(L)\}
\end{equation}
and set
\begin{equation}
  \mathcal{L}(u) \coloneqq  \bigcup_{v\in A(u)} L(T(v)).
\end{equation}

\begin{figure}[t]
  \begin{tabular}{lcr}
    \begin{minipage}{0.5\textwidth}
      \begin{center}
        \includegraphics[width=\textwidth]{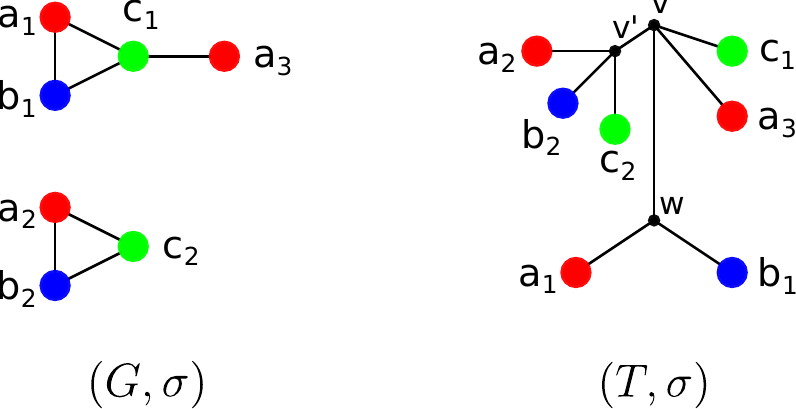}
      \end{center}
    \end{minipage}
    &    & 
    \begin{minipage}{0.4\textwidth}
      \caption{The 3-RBMG $(G,\sigma)$ on the left hand side can be
        explained by the tree $(T,\sigma)$ shown on the right. In
        $(T,\sigma)$, the inner vertex $v$ is a fork. The color-deficient
        children of $v$ are $c_1,a_3$ and $w$, thus
        $\mathcal{L}(v) = \{a_1,b_1,c_1,a_3\}$.  Also, $v'$ is a fork and
        $\mathcal{L}(v') =\{a_2,b_2,c_2\}$. The set of forks is
        $\zeta(T,\sigma)=\{v,v'\}$.}
      \label{fig:fork}
    \end{minipage}
  \end{tabular}
\end{figure}

\begin{definition}
  Let $(T,\sigma)$ a leaf-colored tree. An inner vertex $u\in V^0(T)$ is a
  \emph{fork} if $\sigma(\mathcal{L}(u))=\sigma(L)$. We write
  $\zeta(T,\sigma)$ for the set of forks in $T$.
\label{def:fork}
\end{definition}
For an illustration  see Fig.\ 
\ref{fig:fork}. As an immediate consequence
of the definition we have
\begin{lemma}
  \label{lem:fact:sum-cl} 
  Every fork in a leaf-colored tree $(T,\sigma)$ is color-complete, but not
  every color-complete vertex is a fork.
\end{lemma}
\begin{proof}
  For a fork $u$, we have
  $\sigma(L)=\sigma(\mathcal{L(T)}(u))\subseteq
  \bigcup_{v\in\child(u)}\sigma(L(T(v)))=\sigma(L(T(u))$.  Thus every fork
  must be color-complete. In order to see that not every color-complete
  vertex is a fork, consider a leaf-colored tree $(T,\sigma)$, where
  $\rho_T$ has exactly two children both of which are color-complete.
  Then, $\rho_T$ is color-complete but $A(\rho_T)=\emptyset$. Hence,
  $\rho_T$ is not a fork.
  
\end{proof}

\NEW{Clearly, there are no forks in a leaf-labeled tree $(T,\sigma)$ with
  $|\sigma(L(T))|=1$. In the following, we will therefore restrict our
  attention to trees and graphs with at least two colors, omitting the
  trivial case of 1-RBMGs which correspond to the edge-less graph $(G,\sigma)$ 
	that is explained by any leaf-colored tree $(T,\sigma)$ with leaf set $L(T) = V(G)$.  Next, we
  derive some useful technical results about forks and color-complete
  trees, which will be needed to prove the main result of this section.}

\begin{lemma}\label{lem:zeta_nonempty}
  Let $(T,\sigma)$ be a leaf-colored tree.  Then
  $\zeta(T,\sigma)\neq \emptyset$.
\end{lemma}
\begin{proof}
  \NEW{Let $L=L(T)$.}  Assume, for contradiction, that
  $\zeta(T,\sigma)=\emptyset$. Thus, in particular,
  $\rho_T\notin\zeta(T,\sigma)$. Since the root $\rho_T$ is always
  color-complete, we have
  $\sigma(\mathcal{L}(\rho_T))\neq \sigma(L(T(\rho_T)))=\sigma(L)$, which
  implies $A(\rho_T) \subsetneq \child(\rho_T)$. Hence,
  Equ.~\eqref{eq:A(u)} implies that there is a child $u_1$ of the root with
  $\sigma(L(T(u_1)))=\sigma(L)$. Since $\zeta(T,\sigma)=\emptyset$, the
  vertex $u_1$ is not a fork. Repeating the argument, $u_1$ must have a
  child $u_2$ with $\sigma(L(T(u_2)))=\sigma(L)$, and so on. Hence, there
  is a sequence of inner vertices
  $\rho_T\coloneqq u_0 \succ_T u_1\succ_T u_2\succ_T \dots\succ_T u_k$ such
  that $u_j$ has only color-complete children for $0\le j<k$.  Since $T$ is
  finite, all maximal paths of this form a finite, i.e., the final vertex
  $u_k$ in every maximal path has only color-deficient children, i.e.,
  $A(u)=\child(u)$. Since $u_k$ itself is color-complete by construction,
  $\sigma(\mathcal{L}(u))=\sigma(L(T(u)))=\sigma(L)$, i.e., $u_k$ is fork,
  a contradiction.  
\end{proof}

\begin{lemma}\label{lem:cl}
  Let $(G,\sigma)$ be an $n$-RBMG, $n\ge 2$, $(T,\sigma)$ a tree \NEW{with
    leaf set $L$} that explains $(G,\sigma)$ and $(T(u),\sigma_{|L(T(u))})$
  a color-complete subtree of $(T,\sigma)$ for some $u\in V^0(T)$. Then,
  $xy\notin E(G)$ for any two vertices $x,y\in L$ with $x\in L(T(u))$ and
  $y\in L\setminus L(T(u))$.
\end{lemma}
\begin{proof}
  If $u=\rho_T$, then $L\setminus L(T(u)) = \emptyset$ and the lemma is
  trivially true. Thus suppose $u\neq \rho_T$.  Let $x\in L(T(u))$ and
  assume for contradiction $xy\in E(G)$ for some $y\in L\setminus L(T(u))$,
  i.e., $x$ and $y$ are reciprocal best matches. By choice of $x$ and $y$,
  $\lca(x,y)\succ u$ and $\sigma(x)\neq \sigma(y)$.  Since
  $(T(u),\sigma_{|L(T(u))})$ is color-complete, there exists a leaf
  $y'\in L(T(u))$ with $\sigma(y')=\sigma(y)$. Hence in particular,
  $\sigma(y')\neq \sigma(x)$ and thus, $y'\neq x$.  Since $y'\in L(T(u))$,
  we have $\lca(x,y')\preceq u \prec \lca(x,y)$; a contradiction to the
  assumption that $x$ and $y$ are reciprocal best matches.
  
\end{proof}

\begin{lemma}\label{lem:sum-com}
  Let $(T,\sigma)$ be a leaf-colored tree with leaf set $L$  
	that explains \NEW{the $n$-RBMG} $(G,\sigma)$ with $n>1$, and
  let $u\in \zeta(T,\sigma)$ be a fork in $(T,\sigma)$.  Then, the
  following statements are true:
  \begin{description}
  \item[(i)] If $L^*$ is the vertex set of a connected component
    $(G^*,\sigma^*)$ of $(G,\sigma)$, then either
    $L^*\subseteq \mathcal{L}(u)$ or $L^*\cap \mathcal{L}(u)=\emptyset$.
  \item[(ii)] There is a connected component
    $(G^*,\sigma^*)$ of $(G,\sigma)$ with leaf set
    $L^*\subseteq\mathcal{L}(u)$ and $\sigma(L^*)=\sigma(L)$.
  \item[(iii)] Let $(G^*,\sigma^*)$ be a connected component of
    $(G,\sigma)$ with vertex set $L^*$ and $\sigma(L^*)=\sigma(L)$. 
		Then, $u'\coloneqq \lca(L^*)$ is a fork and
    $L^*\subseteq \mathcal{L}(u')$.
  \end{description}
\end{lemma}
\begin{proof}
  All $N^+$-neighborhoods in this proof are taken w.r.t.\ the underlying
  BMG $\G(T,\sigma)$.  By Lemma \ref{lem:zeta_nonempty}, we have
  $\zeta(T,\sigma)\neq \emptyset$ and thus, there exists a fork in
  $(T,\sigma)$. In what follows, let $u\in\zeta(T,\sigma)$ be chosen
  arbitrarily.

  \textit{(i)}\quad Let $(G^*,\sigma^*)$ be a connected component of
  $(G,\sigma)$ and $L^*$ its vertex set. The statement is trivially true if
  $|L^*|=1$. Hence, assume that $|L^*|\ge 2$.  By Lemma
  \ref{lem:fact:sum-cl}, $(T(u),\sigma_{|L(T(u))})$ is
  color-complete. Lemma \ref{lem:cl} implies $xy\notin E(G)$ for any pair
  of leaves $x\in L(T(u))$ and $y\in L\setminus L(T(u))$. Therefore either
  $L^*\subseteq L(T(u))$ or $L^*\cap L(T(u))=\emptyset$. In the latter
  case, we have
  $L^* \cap \mathcal{L}(u)\subseteq L^* \cap L(T(u))=\emptyset$.\\
  Now, suppose $L^*\subseteq L(T(u))$ and consider a vertex $x\in L^*$ and
  let $z\in L^*\setminus\{x\}$ be a neighbor of $x$, i.e., $xz\in
  E(G)$. Such a $z$ exists since $|L^*|\ge 2$ and $G^*$ is connected. If
  $x\in L(T(u))\setminus \mathcal{L}(u)$, then there exists a
  color-complete inner vertex $v\in\child(u)$ that satisfies $x\prec
  v$. Since $v$ is color-complete, Lemma \ref{lem:cl} implies that there is
  no edge between $L(T(v))$ and $L(T(u))\setminus L(T(v))$ and thus we have
  $z\in L(T(v))$. Therefore $z\notin \mathcal{L}(u)$.  Now suppose that
  $x\in \mathcal{L}(u)$. If $z\notin \mathcal{L}(u)$, then
  $z\in L(T(u))\setminus \mathcal{L}(u)$.  Thus we can apply analogous
  arguments and Lemma \ref{lem:cl} to conclude that there cannot be an edge
  between $x$ and $z$; a contradiction.  Hence, $z\in \mathcal{L}(u)$.  In
  summary, we have either $L^*\subseteq \mathcal{L}(u)$ or
  $L^*\cap \mathcal{L}(u)=\emptyset$.

  \textit{(ii)}\quad \NEW{Let $S\coloneqq \sigma(L)$ with $|S|=n>1$.  We
    proceed by induction.}
  For $n=2$, the statement is a direct
  consequence of Lemma \ref{lem:rbm-pairs}.\\
  For the induction step, suppose the statement is correct for RBMGs with a
  color set of less than $n$ colors. Recall that for any $v_i\in A(u)$ the
  color set of any subtree $(T(v_i),\sigma_{|L(T(v_i))})$ contains less
  than $n$ colors, i.e., $S_{v_i}\coloneqq \sigma(L(T(v_i))) \neq S$.  By
  Lemma \ref{lem:zeta_nonempty}, there must exist a fork
  $w\in \zeta(T(v_i), \sigma_{|L(T(v_i))})$ within the tree
  $(T(v_i), \sigma_{|L(T(v_i))})$. Since $w$ is a fork in
  $(T(v_i), \sigma_{|L(T(v_i))})$, it is therefore also color-complete in
  $(T(v_i), \sigma_{|L(T(v_i))})$.  However, by definition, we have
  $w \preceq v_i\in A(u)$ and thus, $w$ is not color-complete in
  $(T,\sigma)$. Nevertheless, we can apply the induction hypothesis to the
  RBMG $(G_{v_i},\sigma_{v_i})\coloneqq G(T(v_i),\sigma_{|L(T(v_i))})$ to
  ensure that there exists a connected component
  $(G^*_{v_i},\sigma^*_{v_i})$ with leaf set
  $L_{v_i}^*\subseteq \mathcal{L}(w)$ and $\sigma(L_{v_i}^*)=S_{v_i}$. Now,
  fix this index $i$.  By Cor.\ \ref{cor:subconnComp}, there is a connected
  component $(G^*,\sigma^*)$ with leaf set $L^*$ of $(G,\sigma)$ that
  contains $(G^*_{v_i},\sigma^*_{v_i})$.
 
  Assume, for contradiction, that $|\sigma(L^*)|<n$. Suppose first that
  $|S_{v_i}|=n-1$. Thus, $S\setminus S_{v_i} = \{r\}$ and for each color
  $s\in S\setminus \{r\}$ there is a vertex $z\in V(G^*_{v_i})$ with color
  $\sigma(z)=s$.  By construction, $u\in \zeta(T,\sigma)$ implies that
  there exists a vertex $v_j\in A(u)$ ($i\neq j$) such that $r\in
  S_{v_j}$. In particular, it follows from Equ.\ \eqref{eq:bmg14} that
  $L(T(v_j))\cap L[r] \subseteq N_r^+(x)$ for all $x\in L(T(v_i))$. Since
  $S_{v_i} \subseteq \sigma (L^*)$ but $|\sigma (L^*)|<n$, we have
  $|\sigma (L^*)|=n-1$, and we conclude that $xy\notin E(G)$ for every
  $y\in L(T(v_j))\cap L[r]$ and $x\in V(G^*_{v_i})$.  The latter two
  arguments imply that $x\notin N_{\sigma(x)}^+(y)$ for all
  $y\in L(T(v_j))\cap L[r]$ and $x\in V(G^*_{v_i})$.
  This, however, is only possible if 
  $L(T(v_j))$ contains leaves of all colors $s\neq r$, i.e.,
  $S_{v_i}\subsetneq S_{v_j}$ and thus $|S_{v_j}|=n$; a contradiction to
  $v_j\in A(u)$. 

  Now, suppose that $|S_{v_i}|<n-1$, i.e.,
  $S\setminus S_{v_i} = \{r_1,...,r_m\}$. Again, for any $r_j$
  ($1\le j \le m$), there is a vertex $v_j\in A(u)$ ($i\neq j$) such that
  $r_j\in S_{v_j}$. Note that $v_j=v_k$ may be possible for two different
  colors $r_j$ and $r_k$. If there exists a color $s_j\in S_{v_i}$ that is
  not contained in $S_{v_j}$, then, for any $x\in L(T(v_j))\cap L[r_j]$ and
  $y\in L(T(v_i))\cap L[s_j]$, we have $\lca_T(x,y)=u\prec_T \lca_T(x,y')$
  and $\lca_T(x,y)=u\prec_T \lca_T(x',y)$ for all
  $x'\in L[r_j], y'\in L[s_j]$, and hence $xy\in E(G)$. Thus, there is a
  connected component in $G(T(u),\sigma_{|L(T(u))})$ that contains at least
  all colors $S_{v_i}\cup \{r_j\}$. Consequently, if for any
  $j\in \{1,\dots,m\}$ there exists such a color
  $s_j\in S_{v_i}\setminus S_{v_j}$, then there must be a connected
  component in $G(T(u),\sigma_{|L(T(u))})$ that contains all colors in
  $S$. By Cor.\ \ref{cor:subconnComp}, every connected component of
  $G(T(u),\sigma_{|L(T(u))})$ is contained in a connected component of
  $(G,\sigma)$ and statement (ii) is true for this case.
      
  On the other hand, if there is at least one $j$ for which this property
  is not true, similar argumentation as in the case $|S_{v_i}|=n-1$ shows
  that $S_{v_i}\subset S_{v_j}$, hence in particular
  $|S_{v_j}|>|S_{v_i}|$. We can then apply the same argumentation for RBMG
  $(G_{v_j},\sigma_{v_j})\coloneqq G(T(v_j), \sigma_{|L(T(v_j))})$ and either
  obtain a connected component on $n$ colors in $G(T(u),\sigma_{|L(T(u))})$
  or some inner vertex $v_k\in A(u)$ with $|S_{v_i}|<|S_{v_j}|<|S_{v_k}|$.
  Repeating this argumentation, in each step we either obtain either an
  $n$-colored connected component or further increase the sequence
  $|S_{v_i}|<|S_{v_j}|<|S_{v_k}|<\dots$. Since $L$ is finite, this sequence
  must eventually terminate with $|S_{v_l}|=n$, contradicting
  $v_l\in A(u)$. 
  In summary, we have shown that  $|\sigma(L^*)|\neq n$ is not possible
  and hence $\sigma(L^*)=\sigma(L)$. Finally, $\emptyset \neq L^*
  \cap L(T(v_i))$ and $v_i\in A(u)$ implies that
  $L^*\cap \mathcal{L}(u)\neq\emptyset$. Thus, we can apply 
  Statement (i) to conclude  that $L^*\subseteq \mathcal{L}(u)$.

  \textit{(iii)}\quad By Statement (ii), there is a connected component
  $(G^*,\sigma^*)$ with vertex set $L^*$ and $\sigma(L^*)=\sigma(L)$. Put
  $u'\coloneqq \lca(L^*)$.  We start by showing
  $L^*\subseteq \mathcal{L}(u')$.  Assume, for contradiction, that there
  exists a leaf $a\in L^*$ such that $a\notin \mathcal{L}(u')$. Let
  $v'\in \child(u')$ be the (unique) child of $u'$ with $a\prec_T
  v'$. Since $a\notin \mathcal{L}(u')$, we can conclude that
  $v'\notin A(u')$. Thus $v'$ is color-complete and therefore,
  $(T(v'),\sigma_{|L(T(v'))})$ is color-complete. By Lemma \ref{lem:cl}, we
  thus have $b\prec_T v'$ for any $b\in L$ with $ab\in
  E(G)$. \NEW{Repeating this argument for any $b\in N(a)$ and $c\in N(b)$
    and so on, this finally implies $L^*\subseteq_T L(T(v'))$.  Therefore
    $\lca(L^*)\preceq_T v'\prec_T u'$;}
  a contradiction to
  $u' = \lca(L^*)$.  Thus we have $L^*\subseteq \mathcal{L}(u')$. As a
  consequence, $\sigma(\mathcal{L}(u'))=\sigma(L)$, i.e., $u'$ is a fork.
  
\end{proof}

\begin{corollary}\label{cor:n-color}
  Let $(G,\sigma)$ be an $n$-RBMG, $n\ge 2$, that is explained by a tree
  $(T,\sigma)$ with root $\rho_T$.
  \begin{itemize}
  \item[(i)] There exists an $n$-colored connected component $(G^*,\sigma^*)$ of
    $(G,\sigma)$.
  \item[(ii)] If $(G,\sigma)$ is connected, then $\zeta(T,\sigma)=\{\rho_T\}$.
  \end{itemize}
\end{corollary}
\begin{proof}
  (i) Since $\zeta(T,\sigma)\neq \emptyset$ (see Lemma
  \ref{lem:zeta_nonempty}), the existence of an $n$-colored connected
  component of $(G,\sigma)$ is a
  direct consequence of Lemma \ref{lem:sum-com}(ii).\\
  (ii) Lemma \ref{lem:sum-com}(iii) implies $\rho_T\in \zeta(T,\sigma)$. By
  (i) and Lemma \ref{lem:sum-com}(ii), we have
  $L(T)\subseteq \mathcal{L}(u)\subseteq L(T)$ for all
  $u\in \zeta(T,\sigma)$, hence $\mathcal{L}(u)= L(T)$. Since this is true
  only if $u=\rho_T$, assertion (ii) follows.  
\end{proof}

The following result helps to gain some understanding of the ambiguities
among the leaf-colored trees that explain the same RBMG.
\begin{lemma}\label{lem:tree-transfer}
  Let $(G,\sigma)$ be an $n$-RBMG, $n\ge 2$, explained by $(T,\sigma)$ and
  $u\in\zeta(T,\sigma)$ with $u\neq \rho_T$. Moreover, let $v\in\child(u)$,
  where $v$ is color-complete, and $(T',\sigma)$ the tree obtained from
  $(T,\sigma)$ by replacing the edge $uv$ by $\parent(u)v$. Then,
  $(T',\sigma)$ explains $(G,\sigma)$.
\end{lemma}
\begin{proof}
  First note that, since $u$ is a fork in $(T,\sigma)$, there must exist at
  least two color-deficient nodes $w_1, w_2 \in A(u)$. Since $v$ is
  color-complete, we have $v\neq w_1,w_2$, thus $\deg_{T'}(u)>2$, i.e.,
  $(T',\sigma)$ is phylogenetic.  In what follows, we show that
  $(G,\sigma) = G(T',\sigma)$. Put $L\coloneqq V(G)$.

  First, let $x,y\in L\setminus L(T(u))$. Then, by construction of
  $(T',\sigma)$, we have $\lca_T(x,y)=\lca_{T'}(x,y)$, and
  $\lca_T(x,y)\prec_T z$ implies $\lca_{T'}(x,y)\prec_{T'} z$ for all
  $z\in V(T)$. In other words, reciprocal best matches $xy$ with
  $x,y\notin L(T(u))$ remain reciprocal best matches in
  $(T',\sigma)$. Moreover, if $x$ and $y$ are not reciprocal best matches
  in $(T,\sigma)$, then we have w.l.o.g.\ $\lca_T(x,y)\succ_T\lca_T(x',y)$
  for some (fixed) $x'\in L[\sigma(x)]$. Clearly, if
  $x'\in L\setminus L(T(v))$, then we still have, by construction,
  $\lca_{T'}(x,y)=\lca_T(x,y)\succ_{T'}\lca_{T'}(x',y)=\lca_T(x',y)$. Thus,
  if $x'\in L\setminus L(T(v))$, then $x$ and $y$ do not form reciprocal
  best matches in $(T',\sigma)$.  If $x'\in L(T(v))$, then
  $\lca_T(x',y)\succeq_T \parent(u)$. Now, $\lca_T(x,y)\succ_T\lca_T(x',y)$
  implies that $\lca_T(x,y)\succ_T \parent(u)$.  In other words,
  $\lca_T(x',y)$ and $\lca_T(x,y)$ lie on the path from the root to
  $\parent(u)$. This and the construction of $(T',\sigma)$ implies that
  $\lca_T(x,y) = \lca_{T'}(x,y) \succ_{T'} \lca_T(x',y) = \lca_{T'}(x',y)$.
  Thus, $x$ and $y$ do not form reciprocal best matches in $(T',\sigma)$.
  In summary, $xy\in E(G)$ if and only if $xy\in E(G(T,\sigma))$ for all
  $x,y\in L\setminus L(T(u))$.

  Moreover, since $v$ is color-complete in both trees, we can apply Lemma
  \ref{lem:cl} to conclude that neither $(G,\sigma)$ nor $G(T',\sigma)$
  contains edges between $L(T(v))$ and $L\setminus L(T(v))$. Since
  $T'(v)=T(v)$ by construction, we additionally have
  $G(T'(v),\sigma_{|L(T(v))})=G(T(v),\sigma_{|L(T(v))})=
  (G[L(T(v))],\sigma_{|L(T(v))})$.

  It remains to show the case $x\in L'\coloneqq L(T(u))\setminus L(T(v))$,
  and either $y\in L'$ or $y \in L\setminus L(T(u))$.  Suppose first that
  $y \in L\setminus L(T(u))$.  Since $u$ is a fork, Lemma
  \ref{lem:sum-com}(ii) implies that there exists a connected component
  $(G^*,\sigma^*)$ of $(G,\sigma)$ with leaf set $L^*$ such that
  $L^*\subseteq \mathcal{L}(u)$. In particular, as $v$ is color-complete,
  it is not contained in $A(u)$. We therefore conclude that
  $L^*\subseteq L'$, i.e., the subtree $(T_{|L'},\sigma_{|L'})$ is
  color-complete as well. Since by construction $T'(u)=T(u)\setminus T(v)$,
  Lemma \ref{lem:cl} implies that there are no edges between $L(T(u))$ and
  $L\setminus L(T(u))$ in both $(G,\sigma)$ and $G(T',\sigma)$. In other
  words, $x$ and $y$ do not form reciprocal best matches, neither in
  $(T,\sigma)$ nor in $(T',\sigma)$ whenever
  $x\in L'\coloneqq L(T(u))\setminus L(T(v))$ and
  $y \in L\setminus L(T(u))$.

  Now, suppose that $y\in L'$. If $x$ and $y$ do not form reciprocal best
  matches in $(T,\sigma)$, then we have w.l.o.g.\
  $\lca_T(x,y)\succ_T\lca_T(x',y)$ for some (fixed) $x'\in L[\sigma(x)]$.
  This immediately implies that $x'\in L'$. Again, since $T'(u)=T_{|L'}$,
  we have
  $\lca_T(x,y)=\lca_{T'}(x,y)\succ_{T'} \lca_T(x',y) = \lca_{T'}(x',y)$.
  Hence, $x$ and $y$ do not form reciprocal best matches in $x$ and $y$ in
  $(T',\sigma)$.  Finally, if $x$ and $y$ are reciprocal best matches in
  $(T,\sigma)$, then $\lca_T(x,y) \preceq_T \lca_T(x',y)$ and
  $\lca_T(x,y) \preceq_T \lca_T(x,y')$ for all $x'\in L[\sigma(x)]$ and
  $y'\in L[\sigma(y)]$.  We first fix a leaf $x'\in L[\sigma(x)]$ for which
  the latter inequality is satisfied.  By construction,
  $ \lca_{T}(x,y) = \lca_{T'}(x,y) \preceq_{T'} u$.  Clearly, if
  $x'\in L'$, then the fact $T'(u)=T_{|L'}$
  implies that $\lca_{T'}(x,y) \preceq_{T'} \lca_{T'}(x',y)$.  On the other
  hand, if $x'\notin L'$, then $\lca_{T'}(x',y)\succ_{T'} u$ by
  construction of $(T',\sigma)$. We thus have
  $\lca_{T'}(x,y) \preceq_{T'} u \prec_{T'} \lca_{T'}(x',y)$.  Hence,
  $\lca_T(x,y) \preceq_T \lca_T(x',y)$ implies
  $\lca_{T'}(x,y) \preceq_{T'} \lca_{T'}(x',y)$ for all
  $x'\in L[\sigma(x)]$.  Analogous arguments hold for $y'\in L[\sigma(y)]$.
  Hence, $x$ and $y$ remain reciprocal best matches in $(T',\sigma)$.\\
  In summary, $xy\in E(G)$ if and only $xy\in E(G(T',\sigma))$.  
\end{proof}

Let $(G,\sigma)$ be an undirected, vertex colored graph with vertex set $L$
and $|\sigma(L)|=n$. We denote the connected components of $(G,\sigma)$ by
$(G_i^n,\sigma_i^n)$, $1\le i \le k$, with vertex sets $L_i^n$ if
$\sigma(L_i^n)=\sigma(L)$ and $(G_j^{<n},\sigma_j^{<n})$, $1\le j \le l$,
with vertex sets $L_j^{<n}$ if $\sigma(L_j^{<n})\subsetneq\sigma(L)$. That
is, the upper index distinguishes components with all colors present from
those that contain only a proper subset. Suppose that each
$(G_i^n,\sigma_i^n)$ and $(G_j^{<n},\sigma_j^{<n})$ is an RBMG. Then there
are trees $(T_i^n,\sigma_i^n)$ and $(T_j^{<n},\sigma_j^{<n})$ explaining
$(G_i^n,\sigma_i^n)$ and $(G_j^{<n},\sigma_j^{<n})$, respectively. The
roots of these trees are $u_i$ and $v_j$, respectively. We construct a tree
$(T_G^*,\sigma)$ with leaf set $L$ in two steps:
\begin{itemize}
\item[(1)] Let $(T',\sigma^n)$ be the tree obtained by attaching the trees
  $(T_i^n,\sigma_i^n)$ with their roots $u_i$ to a common root $\rho'$.
\item[(2)] First, construct a path $P=v_1v_2\dots v_{l-1}v_l\rho'$, where
  $\rho'$ is omitted whenever $T'$ is empty.  Now attach the trees
  $(T_j^{<n},\sigma_j^{<n})$, $1\le j\le l$, to $P$ by identifying the root
  of each $T_j^{<n}$ with the vertex $v_j$ in $P$.  Finally, if
  $(T',\sigma^n)$ exists, attach it to $P$ by identifying the root of $T'$
  with the vertex $\rho'$ in $P$. The coloring of $L$ is the one given for
  $(G,\sigma)$.
\end{itemize}
This construction is illustrated in Fig.~\ref{fig:tree_constr} \NEW{for
  $n\ge 2$. For $n=1$, the resulting tree is simply the star tree on $L$.}

Our goal for the remainder of this section is to show that every RBMG is
explained by a tree of the form $(T_G^*,\sigma)$. We start by collecting
some useful properties of $(T_G^*,\sigma)$.
\begin{fact}
  \label{fact:tree1}
  Let $(G,\sigma)$ be an undirected vertex colored graph \NEW{with
    $|\sigma(V(G))|\ge 2$} whose connected components are RBMGs and let
  $(T_G^*,\sigma)$ be the tree described above. Then
  \begin{itemize}
  \item[(i)] $\zeta(T^*_G,\sigma)=\{u_1,\dots,u_k\}$,
  \item[(ii)] Every subtree $(T_i^n,\sigma_i^n)$, $1\le i\le k$ and
    $(T^*(v_j),\sigma_{|L(T_G^*(v_j))})$ and $1\le j\le l$, resp., is
    color-complete.
  \end{itemize}
\end{fact}
\begin{proof}
  Statement (i) is an immediate consequence of Cor.\ \ref{cor:n-color}(ii).
  For Statement (ii) observe that, by construction,
  $\sigma(L^n_i) = \sigma(L)$ and thus, $(T_i^n,\sigma_i^n)$ is a
  color-complete subtree of $(T^*_G,\sigma)$, $1\leq i\leq k$.  By step (2)
  of the construction of $(T^*_G,\sigma)$, we have
  $u_1\prec_{T_G^*}\rho'\prec_{T_G^*} v_l\prec_{T_G^*}\dots\prec_{T_G^*}
  v_1$.  Since $u_1$ is color-complete by assumption, so is each of its
  ancestors.
  
\end{proof}

\begin{lemma}\label{lem:tree}
  Let $(G,\sigma)$ be an undirected vertex colored graph on $n$ colors
  whose connected components are RBMGs and there is at least one
  $n$-colored connected component, and let $(T^*_G,\sigma)$ be the tree
  described above.  Then $(T^*_G,\sigma)$ explains $(G,\sigma)$.
\end{lemma}
\begin{proof}
  \NEW{For $n=1$, $(T^*_G,\sigma)$ is simply the star tree on
    $V(G)$. Clearly, $(G,\sigma)$ must be the edge-less graph, which is
    explained by $(T^*_G,\sigma)$. Now suppose $n>1$.}  Let
  $(G^n_i,\sigma^n_i)$ be an $n$-colored connected component of
  $(G,\sigma)$, $i\in \{1,\dots,k\}$ and $k\geq 1$. It has vertex set
  $L_i^{n}=L(T^*_G(u_i))$. By construction,
  $(T^*_G(u_i), \sigma_{|L_i^{n}}) = (T_i^n,\sigma_i^n)$ explains
  $(G^n_i,\sigma^n_i)$ and
  $\NEW{(G[L_i^n],\sigma_{|L_i^n})} =
  (G^n_i,\sigma^n_i)$.  Moreover, $(T^*_G(u_i), \sigma^n_i)$ is a
  color-complete subtree of $(T^*_G,\sigma)$ that is rooted at
  $u_i$. Hence, Lemma \ref{lem:cl} implies that there are no edges in
  $G(T^*_G,\sigma)$ between $L_i^{n}$ and any other vertex in
  $L\setminus L_i^{n}$.  In other words, $(G^n_i,\sigma^n_i)$ remains a
  connected component in $G(T^*_G,\sigma)$, $i\in \{1,\dots,k\}$.
	 
  Now, suppose that there is a connected component
  $(G^{<n}_j,\sigma^{<n}_j)$, $j\in \{1,\dots,l\}$ and $l\geq 1$, which
  contains less than $n$ colors.  Again, by construction,
  $(T^*_G(v_j)_{|L_j^{<n}}, \sigma_{|L_j^{<n}}) = (T_j^{<n},\sigma_j^{<n})$
  explains $(G^{<n}_j,\sigma^{<n}_j)$ and
  $(G[L_i^{<n}],\sigma_{|L_i^{<n}})= (G^{<n}_j,\sigma^{<n}_j)$.
  Furthermore, we have $L^{<n}_{j'}\cap L(T^*_G(v_j))=\emptyset$ if and
  only if $j'<j$ by construction of the path $v_1v_2\dots v_l$ in $T^*_G$.
  By Observation \ref{fact:tree1}(ii), $v_j$ is color-complete and Lemma
  \ref{lem:cl} implies that there is no edge between $L^{<n}_{j}$ and any
  $L^{<n}_{j'}$ whenever $j'<j$.  In other words,
  $(G^{<n}_j,\sigma^{<n}_j)$, $j\in \{1,\dots,l\}$ remains a connected
  component in $G(T^*_G,\sigma)$.
	
  To summarize, \emph{all} connected components of $(G,\sigma)$ remain
  connected components in $G(T^*_G,\sigma)$ and are explained by
  restricting $(T^*_G,\sigma)$ to the corresponding leaf set, which
  completes the proof.

\end{proof}

\begin{theorem}\label{thm:connected}
  An undirected leaf-colored graph $(G,\sigma)$ is an RBMG if and only if
  each of its connected components is an RBMG and at least one connected
  component contains all colors.
\end{theorem}
\begin{proof}
  \NEW{For $n=1$, the statement trivially follows from the fact that an
    RBMG must be properly colored and thus, a 1-RBMG must be the edge-less graph. 
		Now suppose $n>1$.}  By Theorem \ref{thm:conn_comp} every
  connected component of an RBMG is again an RBMG. Cor.\
  \ref{cor:n-color}(i) ensures the existence of a connected component
  containing all colors. Conversely, if $(G,\sigma)$ is an undirected graph
  whose connected components are RBMGs and at least one of them contains
  all colors, then Lemma \ref{lem:tree} guarantees that it is explained by
  a tree of the form $(T^*_G,\sigma)$, and hence it is an RBMG.  
\end{proof}

\begin{corollary}
  \label{cor:tree2}
  Every RBMG can be explained by a tree of the form $(T^*_G,\sigma)$.
\end{corollary}

\section{Three Classes of Connected 3-RBMGs}
\label{sect:classes}

\subsection{Three Special Classes of Trees} 

We start with a rather technical result that allows us to simplify the
structure of trees explaining a given 3-RBMG.
\begin{lemma}\label{lem:2-color-subtree}
  Let $(G,\sigma)$ be an $\sthin$-thin 3-RBMG that is explained by
  $(T,\sigma)$. Moreover, let $u\in V^0(T)$ be a vertex that has two
  distinct children $v_1,v_2\in\child(u)$ such that
  $\sigma(L(T(v_1)))=\sigma(L(T(v_2)))\subsetneq \sigma(L(T))$ and
  $v_1\in V^0(T)$, and denote by $(T',\sigma)$ the tree obtained by
  replacing the edge $uv_2$ in $(T,\sigma)$ by $v_1v_2$ and possible
  suppression of $u$, in case $u$ has exactly two children in $(T,\sigma)$
  or removal of $u$ and its incident edge, in case $u=\rho_{T'}$.  Then
  $(T',\sigma)$ explains $(G,\sigma)$.
\end{lemma}
\begin{proof}
  It is easy to see that the resulting tree $(T',\sigma)$ is phylogenetic.
  We emphasize that this proof does not depend on whether $u$ has been
  suppressed or removed. Put $L\coloneqq L(T)$. Moreover, Lemma
  \ref{lem:2col} implies that $L(T(v_1))$ contains leaves of more than one
  color, hence $|\sigma(L(T(v_1)))|=2$.

  Let $S=\{r,s,t\}$ be the color set of $(G,\sigma)$ and
  $\sigma(L(T(v_1)))=\{r,s\}$. Since $L(T(v_1))$ and $L(T(v_2))$ do not
  contain leaves of color $t$, we have
  $\lca_T(y,z)=\lca_{T'}(y,z)\NEW{\succeq u}$ for every
  $y\in L[r]\cup L[s]$ and $z\in L[t]$. Hence, $yz \in E(G)$ if and only if
  $yz\in E(G(T',\sigma))$ for every $y\in L[r]\cup L[s]$ and $z\in
  L[t]$. It therefore suffices to consider
  $(T_{rs},\sigma_{rs})\coloneqq (T_{|L[r]\cup L[s]},\sigma_{|L[r]\cup
    L[s]})$ and
  $(T'_{rs},\sigma_{rs})\coloneqq (T'_{|L[r]\cup L[s]},\sigma_{|L[r]\cup
    L[s]})$.

  First note that, since $T'(v_2)=T(v_2)$, vertex $v_2$ is color-complete
  in both $(T_{rs},\sigma_{rs})$ and $(T'_{rs},\sigma_{rs})$. Hence, Lemma
  \ref{lem:cl} implies that neither $G(T_{rs},\sigma_{rs})$ nor
  $G(T'_{rs},\sigma_{rs})$ contains edges of the form $xy$, where
  $x\in L(T(v_2))$ and $y\notin L(T(v_2))$. Moreover, since
  $T'(v_2)=T(v_2)$, we have
  $G(T_{rs}(v_2),\sigma_{|L(T(v_2))})=G(T'_{rs}(v_2),\sigma_{|L(T(v_2))})$.
  Since $v_1$ is also color-complete in $(T_{rs},\sigma_{rs})$ and
  $(T'_{rs},\sigma_{rs})$, we can similarly conclude that both graphs
  $G(T_{rs},\sigma_{rs})$ and $G(T'_{rs},\sigma_{rs})$ contain no edges
  $xy$, where $x\in L(T(v_1))$ and $y\notin L(T(v_1))$. Hence, it suffices
  to consider edges between leaves in $L(T(v_1))$ If $v_1$ is a fork in
  $(T_{rs},\sigma_{rs})$, one can easily see that $(T,\sigma)$ is obtained
  from $(T',\sigma)$ by the same operation used in Lemma
  \ref{lem:tree-transfer}.  Hence, Lemma \ref{lem:tree-transfer} implies
  that $G(T_{rs},\sigma_{rs})=(G[L[r]\cup L[s]],\sigma_{rs})$.
  Suppose that $v_1$ is not a fork.  Note that any $w\in\child_T(v_1)$ with
  $|\sigma(L(T(w)))|=1$ must be a leaf as, otherwise, all leaves in
  $L(T(w))$ would be in a common $\sthin$-class and $(G,\sigma)$ would not
  be $\sthin$-thin. Therefore, any $w\in \child_T(v_1)$ is either
  color-complete or a leaf in $(T_{rs},\sigma_{rs})$.  Therefore, by Lemma
  \ref{lem:cl}, there are no edges between
  $G(T_{rs}(w_1),\sigma_{|L(T(w_1))})$ and
  $G(T_{rs}(w_2),\sigma_{|L(T(w_2))})$ as soon as one of the children $w_1$
  and $w_2$ is a non-leaf vertex. In other words, if there are edges
  between $G(T_{rs}(w_1),\sigma_{|L(T(w_1))})$ and
  $G(T_{rs}(w_2),\sigma_{| L(T(w_2))})$, then both vertices
  $w_1,w_2\in \child_T(v_1)$ are also contained in $L$. Since, by
  construction, $\child_T(v_1)\cap L=\child_{T'}(v_1)\cap L$, we have
  $w_1w_2\in E(G(T_{rs},\sigma_{rs}))$ if and only if
  $w_1w_2\in E(G(T'_{rs},\sigma_{rs}))$ for any $w_1,w_2\in
  \child_T(v_1)$. Moreover, by construction, we have
  $(T(w),\sigma_{|L(T(w))})=(T'(w),\sigma_{|L(T(w))})$ for any inner vertex
  $w\in \child_T(v_1)$, hence
  $G(T(w),\sigma_{|L(T(w))})=G(T'(w),\sigma_{|L(T(w))})$, which concludes
  the proof.
  
\end{proof}

\begin{definition} \cite{Harary:73}
  \label{def:caterpillar}
  A rooted tree $T$ is a \emph{caterpillar} if every inner vertex has a most
  one child that is an inner vertex.
\end{definition}

\begin{figure}[t]
  \begin{center}
    \includegraphics[width=1\textwidth]{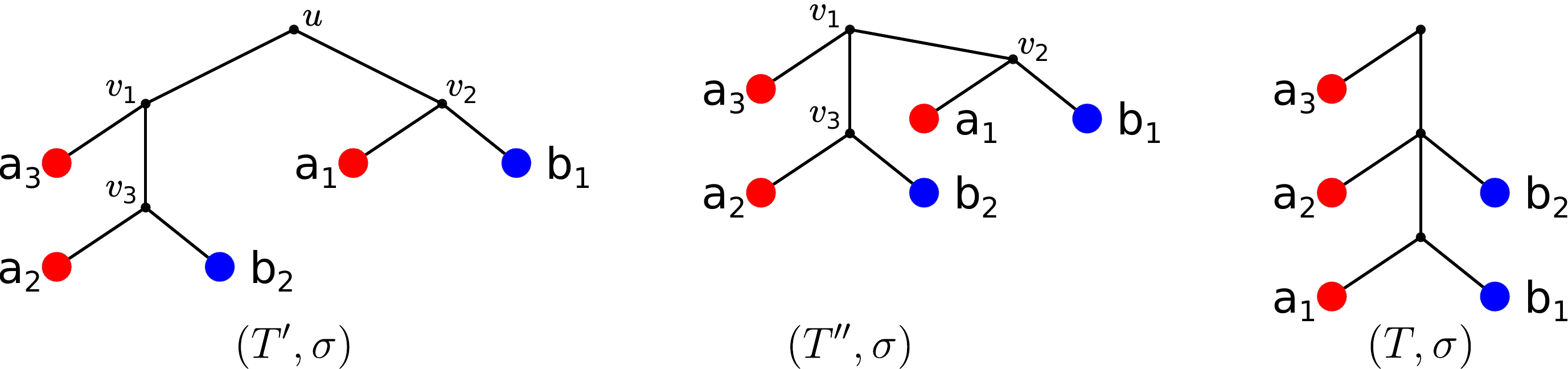}      
  \end{center}
  \caption{Assume that the tree $(T',\sigma)$ is the 2-colored restricted
    version of some tree that explains a 3-RBMG. It is easy to verify that
    all three trees $(T',\sigma)$, $(T'',\sigma)$, and $(T,\sigma)$ explain
    the same 2-RBMG $(G,\sigma)$.  According to the transformation of
    Lemma \ref{lem:2-color-subtree}, $(T'',\sigma)$ is obtained from
    $(T',\sigma)$ by deletion of the edge $uv_2$, inserting $v_1v_2$ and
    removal of $u$ and its single incident edge. Similarly, $(T,\sigma)$ is
    obtained from $(T'',\sigma)$ by deleting $v_1v_2$ and inserting
    $v_3v_2$. The final tree $(T,\sigma)$ is a caterpillar.}
  \label{fig:2-color-subtree}
\end{figure}

Repeated application of the transformation as in Lemma
\ref{lem:2-color-subtree} implies the following result, which is illustrated
in Fig.\ \ref{fig:2-color-subtree}.
\begin{corollary}\label{cor:2col-sub}
  Let $(G,\sigma)$ be an $\sthin$-thin 3-RBMG. Then there exists a tree
  $(T,\sigma)$ explaining $(G,\sigma)$ for which every 2-colored subtree
  $(T(u),\sigma_{|L(T(u))})$ with $|\sigma(L(T(u)))|=2$ is a caterpillar
  and $\sigma(L(T(v_1)))\neq \sigma(L(T(v_2)))$ for any distinct
  $v_1,v_2\in \child(u)$.
\end{corollary}
\begin{proof}
  Let $(T',\sigma)$ explain $(G,\sigma)$ and let $u\in V^0(T')$ be such
  that $(T'(u),\sigma_{|L(T'(u))})$ is a 2-colored subtree of
  $(T',\sigma)$. Suppose there exists an inner vertex $v\in V^0(T'(u))$
  with two distinct children that are again inner vertices, i.e.,
  $w_1,w_2\in \child(v)\cap V^0(T'(u))$. Since $(G,\sigma)$ is
  $\sthin$-thin, we can apply Lemma \ref{lem:2col} to conclude that
  $(T'(w_1),\sigma_{|L(T'(w_1))})$ and $(T'(w_2),\sigma_{|L(T'(w_2))})$ are
  both 2-colored subtrees, thus
  $\sigma(L(T'(w_1)))=\sigma(L(T'(w_2)))\subsetneq \sigma(L)$. By Lemma
  \ref{lem:2-color-subtree}, the tree $(T'',\sigma)$ that is obtained from
  $(T',\sigma)$ by deleting $vw_2$ and inserting $w_1w_2$, still explains
  $(G,\sigma)$ and satisfies
  $|\child(v)\cap V^0(T''(u))|=|\child(v)\cap V^0(T'(u))|-1$. Repeating
  this transformation until each inner vertex $v\in V^0(T)$ satisfies
  $\sigma(L(T(v_1)))\neq\sigma(L(T(v_2)))$ for any
  $v_1,v_2\in \child_T(v)$, finally yields a tree $(T,\sigma)$ for which
  $|\child(v)\cap V^0(T(u))|\le 1$, i.e., a caterpillar, that explains
  $(G,\sigma)$. In particular, we have $|\sigma(L(T(v_1)))|=1$ if and only
  if $v_1\in L$ (cf.\ Lemma \ref{lem:2col}), and $v$ cannot have two leaves
  of the same color as children because $(G,\sigma)$ is $\sthin$-thin.
  
\end{proof}

The restriction to connected $\sthin$-thin graphs with $3$ colors together
with the fact that all 2-colored subtrees can be chosen to be caterpillars
according to Cor.\ \ref{cor:2col-sub} identifies three distinct classes of
trees, see Fig.\ \ref{fig:categories}.

\begin{definition}\label{def:3-col-tree}
  Let $(T,\sigma)$ be a 3-colored tree with color set $S=\{r,s,t\}$.  The
  tree $(T,\sigma)$ is of
  \begin{description}
  \item[\textbf{Type \AX{\bf (I)}},] if there exists $v\in \child(\rho_T)$ such
    that $|\sigma(L(T(v)))|=2$ and
    $\child(\rho_T)\setminus \{v\}\subsetneq L$.
  \item[\textbf{Type \AX{\bf (II)}},] if there exists
    $v_1, v_2\in \child(\rho_T)$ such that
    $|\sigma(L(T(v_1)))|=|\sigma(L(T(v_2)))| = 2$,
    $\sigma(L(T(v_1)))\neq \sigma(L(T(v_2)))$ and
    $\child(\rho_T)\setminus \{v_1,v_2\}\subsetneq L$,
  \item[\textbf{Type \AX{\bf (III)}},] if there exists
    $v_1, v_2, v_3\in \child(\rho_T)$ such that
    $\sigma(L(T(v_1)))=\{r,s\}$, $\sigma(L(T(v_2)))=\{r,t\}$,
    $\sigma(L(T(v_3)))=\{s,t\}$, and
    $\child(\rho_T)\setminus \{v_1,v_2, v_3\} \subsetneq L$.
  \end{description}
\end{definition}

\begin{lemma}\label{lem:2col-subtrees}
  Let $(G,\sigma)$ be an $\sthin$-thin connected 3-RBMG with vertex set
  $L$ and color set $\sigma(L)=\{r,s,t\}$. Then, there is a tree
  $(T,\sigma)$ with root $\rho_T$ explaining $(G,\sigma)$ that satisfies
  the properties in Cor.\ \ref{cor:2col-sub} and is of Type \AX{(I)},
  \AX{(II)} or \AX{(III)}.  In particular, all leaves that are incident to
  the root of $(T,\sigma)$ must have pairwise distinct colors.
\end{lemma}
\begin{proof}
  Since $(G,\sigma)$ is an RBMG, there is a tree $(T,\sigma)$ that explains
  $(G,\sigma)$.  Denote its root by $\rho_{T}$.  Note, $|\sigma(L)|=3$
  implies that $|L|\geq 3$.

  If $|L|=3$, then it is easy to see that $G$ must be a complete graph on
  three vertices.  In this case, any tree $(T,\sigma)$ where $T$ is a
  triple explains $(G,\sigma)$ and satisfies Type \AX{(I)} and Cor.\
  \ref{cor:2col-sub}.

  Now suppose that $|L|>3$.  Since $(G,\sigma)$ is connected, Cor.\
  \ref{cor:n-color}(ii) implies that $\zeta(T,\sigma)=\{\rho_{T}\}$.  Lemma
  \ref{lem:sum-com}(i) then implies
  $L\subseteq \mathcal{L}(\rho_T) \subseteq L$, i.e.,
  $A(\rho_{T})=\child(\rho_{T})$, and thus $|\sigma(L(T(v)))|<3$ for every
  $v\in\child(\rho_{T})$. This is, every proper subtree of $(T,\sigma)$
  contains at most two colors.  As a consequence of Cor.\
  \ref{cor:2col-sub}, the tree $(T,\sigma)$ can be chosen such that there
  is no pair of distinct vertices $v_1,v_2\in\child(\rho_{T})$ for which
  $\sigma(L(T(v_1)))=\sigma(L(T(v_2)))$.  Moreover, as $|L|>3$ and
  $|\sigma(L)|=3$, it follows directly from Cor.\ \ref{cor:2col-sub} that
  $|\sigma(L(T(v)))| =1$ for every child $v\in \child(\rho_T)$ is not
  possible.  Thus, there is at least one child $v\in \child(\rho_T)$ with
  $|\sigma(L(T(v)))|\neq 1$ and thus, $|\sigma(L(T(v)))| =2$.
	
  In summary, there are only six possible subtrees
  $(T(v),\sigma_{L(T(v))})$ with $v\in\child(\rho_{T})$, three containing
  two colors and three containing only a single color, and each of these
  six types of subtrees can appear at most once, while there is, in
  addition, at least one child $v\in\child(\rho_{T})$ where
  $(T(v),\sigma_{L(T(v))})$ contains two colors.

  Therefore, we end up with the three cases \AX{(I)}, \AX{(II)}, and
  \AX{(III)}: If there is exactly one vertex $v\in \child(\rho_T)$ such
  that the subtree $(T(v),\sigma_{L(T(v))})$ contains two colors, any other
  leaf in $L\setminus L(T(v))$ must be directly attached to $\rho_T$, thus
  Condition \AX{(I)} is satisfied. Similarly, Condition \AX{(II)} and
  \AX{(III)}, respectively, correspond to the case where there exist two
  and three 2-colored subtrees below the root.  Since the three types of
  trees \AX{(I)}, \AX{(II)}, and \AX{(III)} differ by the number of
  two-colored subtrees of the root, no tree can belong to more than one
  type.  By the choice of $(T,\sigma)$, it satisfies Cor.\
  \ref{cor:2col-sub}.  

  Finally, if the root of a tree is incident to two leaves of the same
  color, then the graph explained by this tree cannot be
  $\sthin$-thin. Thus, the last statement must be satisfied.  
\end{proof}

The fact that every connected 3-RBMG can be explained by a tree with a very
peculiar structure can now be used to infer stringent structural
constraints on the 3-RBMGs themselves.

\begin{lemma}\label{lem:TreeTypes}
  Let $(G,\sigma)$ \NEW{with vertex set $L$} be an $\sthin$-thin connected
  3-RBMG with $\sigma(L)=\{r,s,t\}$ and $(T,\sigma)$ be a tree of Type
  \AX{(I)}, \AX{(II)}, or \AX{(III)} explaining $(G,\sigma)$. Consider
  $v\in\child(\rho_T)$ such that $\sigma(L(T(v))) = \{r,s\}$.  Then:
  \begin{itemize}
  \item[(i)] If $x\in L(T(v)) \cap L[r]$, then $xy\in E(G)$ for
    $\sigma(y)=s$ if and only if $\parent(x)=\parent(y)$ and thus,
    $y\in L(T(v))$.
  \end{itemize}
  If, in addition, there is a vertex $w\in\child(\rho_T)\setminus \{v\}$ with
  $\sigma(L(T(w)))=\{r,t\}$, i.e., $(T,\sigma)$ is of either Type \AX{(II)}
  or \AX{(III)}, then the following statements hold: 
  \begin{itemize} 
  \item[(ii)] For any $y\in L(T(v))$, $z\in L(T(w))$,
    we have $yz\in E(G)$ if and only if $y\in L[s]$ and $z\in L[t]$.
  \item[(iii)] If $(T,\sigma)$ is of Type \AX{(II)}, then $yz\in E(G)$ for
    every  $y\in L[s]$ and $z\in L[t]$.
  \item[(iv)] For any $a\in L(T(v))$, $b\in \child(\rho_T)\cap L$ with
    $\sigma(b)\neq \sigma(a)$, we have $ab\in E(G)$ if and only if
    $\sigma(b)\notin \sigma(L(T(v)))$.
  \end{itemize} 
\end{lemma}
\begin{proof}
  \smallskip\noindent(i) We assume $y\in L[s]$ and $xy\in E(G)$.  For
  contradiction, suppose $y\notin L(T(v))$. Since $L(T(v))$ contains at
  least one leaf $y'\neq y$ of color $s$, we have
  $\lca(x,y')\preceq v \prec \lca(x,y)$, which implies $xy\notin E(G)$; the
  desired contradiction. Hence, $y\in L(T(v))$. Now assume, again for
  contradiction, $\parent(x)\neq \parent(y)$. There are three cases: (a)
  $\parent(x)$ and $\parent(y)$ are incomparable in $(T,\sigma)$, (b)
  $\parent(x)\prec_T\parent(y)$, or (c) $\parent(y)\prec_T\parent(x)$.  In
  Case (a), Lemma \ref{lem:2col} implies that there is a leaf
  $y'\in L(T(\parent(x))\cap L[s]$ and therefore,
  $\lca(x,y')\prec \lca(x,y)$; again a contradiction to $xy\in E(G)$.
  Similar argumentation can be applied to the Cases (b) and (c). Hence, we
  conclude that
  $\parent(x)=\parent(y)$. \\
  Conversely, assume $\parent(x)=\parent(y)$ and $y\in L(T(v))$.  By
  construction, we have $\parent(x)=\lca(x,y)\preceq \lca(x,y')$ for all
  $y'\in L[s]$, thus $xy\in E(G)$.
 
  \smallskip\noindent(ii) Let $y\in L(T(v))$, $z\in L(T(w))$ and
  $yz\in E(G)$.  Assume, for contradiction, $\sigma(y)=r$. Since
  $(G,\sigma)$ does not contain edges between vertices of the same color,
  we have $z\in L[t]$. By construction of $(T,\sigma)$, there must be some
  $x\in L(T(w))$ of color $r$. Hence,
  $\lca(z,x)\preceq w \prec \lca(z,y)=\rho_T$; a contradiction to
  $yz\in E(G)$. Thus, $\sigma(y)=s$. An analogous argument yields
  $\sigma(z)=t$.  Conversely, let $y\in L(T(v))$ and $z\in L(T(w))$ such
  that $\sigma(y)=s$ and $\sigma(z)=t$.  Since neither
  $t\in \sigma(L(T(v)))$ nor $s\in \sigma(L(T(w)))$ and $(T,\sigma)$ is of
  Type \AX{(II) or \AX{(III)}}, we can immediately conclude that
  $\lca(y,z)=\rho_T = \lca(y,z')=\lca(y',z)$ for all $y'\in L[s]$ and all
  $z'\in L[t]$. Thus, $yz\in E(G)$.

  \smallskip\noindent(iii) Since, $s\notin\sigma(L(T(w)))$,
  $t\notin \sigma(L(T(v)))$, and $(T,\sigma)$ is of Type \AX{(II)}, we have
  $\lca_T(y,z)=\rho_T$ for any pair $y\in L[s]$, $z\in L[t]$. Thus,
  $yz\in E(G)$.

  \smallskip\noindent(iv) Let $\sigma(a)=r$ and suppose first
  $\sigma(b)=s$. Then, there is some $y\prec v$ with $\sigma(y)=s$, thus
  $\lca(a,y)\prec \lca(a,b)$. Therefore $a$ and $b$ cannot be reciprocal
  best matches, i.e., $ab\notin E(G)$. Now assume $\sigma(b)=t$. Since
  $t\notin \sigma(L(T(v)))$, we have $\lca(a,z)=\rho_T$ for every
  $z\in L[t]$.  In particular, we have $\lca(b, L[r])=\rho_T$, and
  therefore $ab\in E(G)$.
  
\end{proof}
We note in passing that Lemma \ref{lem:TreeTypes}(iv) is also satisfied by
Type \AX{(I)} trees.

In the following we also need a special form of Type \AX{(II)} trees:
\begin{definition}\label{def:typeiistar} 
  A tree $(T,\sigma)$ of Type \AX{(II)} with color set $S=\{r,s_1,s_2\}$
  and root $\rho_T$, where $v_1,v_2\in \child(\rho_T)$ with
  $\sigma(L(T(v_1)))=\{r,s_1\}$ and $\sigma(L(T(v_2)))=\{r,s_2\}$, is of
  Type \AX{(II$^*$)} if, for $i\in \{1,2\}$, it satisfies:
  \begin{description}
  \item[\AX{($\star$)}] If there is a vertex $w\in V^0(T(v_i))$ such that
    $\child(w) \cap L=\{x\}$ for some $x\in L[r]$, then there is a vertex
    $v\in \child(\rho_T)$ such that $\sigma(v)=s_i$.
  \end{description}
\end{definition}
\NEW{For leaf-colored trees explaining an $\sthin$-thin graph, Definition
  \ref{def:typeiistar} implies, in particular, that if there is some vertex
  $w\in V^0(T(v_i))$ with $\sigma(\child(w) \cap L)=\{r\}$ in a tree
  $(T,\sigma)$ of Type \AX{(II$^*$)}, then
  $ L[s_i]\setminus L(T(v_i))\neq \emptyset$. Moreover, note that in a
  leaf-colored tree explaining an $\sthin$-thin graph, the property
  $\sigma(\child(w) \cap L)=\{r\}$ always implies $|\child(w) \cap L|=1$. }


Given an arbitrary tree $(T,\sigma)$ of Type \AX{(II)} with colors and
subtrees as in Def.\ \ref{def:typeiistar}, one can easily construct a
corresponding tree $(T',\sigma)$ of Type \AX{(II$^*$)} using the following
rule for $i\in \{1,2\}$:
\begin{description}
\item[\AX{(r)}] If there is no vertex $v\in \child(\rho_T)$ such that
  $\sigma(v)=s_i$, then re-attach all vertices $x\in L[r]$ with
  $\child(\parent(x))\cap L =\{x\}$ to $\rho_T$ and suppress $\parent(x)$
  in case $\parent(x)$ has degree 2 after removal of the edge
  $\parent(x)x$.
\end{description}
By construction, the tree $(T',\sigma)$ has no vertices $w\in V^0(T(v_i))$
with $\sigma(\child(w) \cap L)=\{r\}$, and thus, $(T',\sigma)$ trivially
satisfies \AX{($\star$)}. Hence, $(T',\sigma)$ is of Type \AX{(II$^*$)}.

We proceed by showing that rule \AX{(r)} must be applied to at most one
leaf in order to obtain a tree $(T',\sigma)$ of Type \AX{(II$^*$)}.
\begin{lemma}
\label{lem:one-Ruler}
Let $(T,\sigma)$ be Type \AX{(II)} tree that is not of Type \AX{(II$^*$)}
and that explains a connected $\sthin$-thin 3-RBMG.  Let $\rho_T$ be the
root of $(T,\sigma)$ and $S=\{r,s_1,s_2\}$ its color set.  Moreover, let
$v_1,v_2\in\child(\rho_T)$ such that $\sigma(L(T(v_i)))=\{r,s_i\}$,
$i\in\{1,2\}$.  Then,
\begin{enumerate}
\item[(i)] no leaf of color $r$ is incident to $\rho_T$ and
\item[(ii)] if Rule \AX{(r)} is applied to some vertex $x\in L(T(v_i))$,
  then $x$ is the only leaf in $L[r]\cap L(T(v_i))$ with
  $\child(\parent(x))\cap L =\{x\}$ and all inner vertices in $L(T(v_j))$,
  $j\neq i$ satisfy Property \AX{($\star$)} in Def.\ \ref{def:typeiistar}.
\end{enumerate}
\end{lemma}
\begin{proof}
  First note that, since $(T,\sigma)$ is of Type \AX{(II)}, it satisfies
  $\child(\rho_T)\setminus \{v_1,v_2\}\subset L$. Since $(T,\sigma)$ is not
  of Type \AX{(II$^*$)}, there must be a leaf $x\in L[r]$ with
  $w\coloneqq \parent(x)\preceq_T v_i$ and $\sigma(\child(w)\cap L)=\{r\}$ such
  that there is no leaf of color $s_i$ incident to $\rho_T$, i.e.,
  $L[s_i]\subseteq L(T(v_i))$ for some $i\in \{1,2\}$. W.l.o.g.\ we can
  assume $i=1$. Now, $L[s_1]\subseteq L(T(v_1))$ and Lemma
  \ref{lem:TreeTypes}(i) implies that $N_{s_1}(x)=\emptyset$ in
  $(G,\sigma)$. However, since $(G,\sigma)$ is connected, there must exist
  some $z\in L[s_2]$ such that $xz\in E(G)$. Lemma
  \ref{lem:TreeTypes}(i)+(iv) then implies that every $z\in L[s_2]$ with
  $xz\in E(G)$ must be incident to $\rho_T$. However, Lemma
  \ref{lem:2col-subtrees} implies that $z$ is the only leaf of color $s_2$
  that is incident to the root.  Since $N_{s_1}(x)=\emptyset$, it holds
  that $z$ is the only vertex in $L$ that is adjacent to $x$ in $G$ and
  thus, $N(x)=\{z\}$ in $(G,\sigma)$.

  In order to show Statement (i), we now assume for contradiction that
  there exists another leaf $x'\neq x$ of color $r$ such that
  $x'\in \child(\rho_T)$. Then, as a consequence of Lemma
  \ref{lem:2col-subtrees} and since $L[s_1]\subseteq L(T(v_1))$, we have
  $\child(\rho_T)\cap L=\{x',z\}$. Thus Lemma \ref{lem:TreeTypes}(i)+(iv)
  implies that $x'$ is not adjacent to any vertex in
  $L(T(v_1))\cup L(T(v_2))$.  Moreover, we have
  $\lca_T(x',z) = \rho_T = \lca_T(x'',z) =\lca_T(x',z')$ for all
  $x''\in L[r]$ and $z'\in L[s_2]$, hence $x'z\in E(G)$. Taking the latter
  two arguments together with the observation that there is no leaf with
  color $s_1$ incident to the root $\rho_T$, we obtain $N(x')=N(x)=\{z\}$
  in $(G,\sigma)$; a contradiction to the $\sthin$-thinness of
  $(G,\sigma)$.

  We proceed with showing Statement (ii). Repeating the latter arguments,
  one easily checks that $N(x_1)=\{z\}$ for any vertex $x_1\in L(T(v_1))$
  with $\sigma(\child(\parent(x_1))\cap L)=\{r\}$. However, since
  $(G,\sigma)$ is $\sthin$-thin, we cannot have a further vertex
  $x_1\in L(T(v_1))$ with $N(x_1)=\{z\} = N(x)$.  Hence, there is exactly
  one $x\in L[r]\cap L(T(v_1))$ with $\child(\parent(x))\cap L =\{x\}$ to
  which Rule \AX{(r)} can be applied.  Moreover, the existence of
  $z\in \child(\rho_T)\cap L[s_2]$ immediately implies that Property
  \AX{($\star$)} in Def.\ \ref{def:typeiistar} is satisfied for every
  $w\in V^0(T(v_2))$ with $\sigma(\child(w)\cap L) = \{r\}$.  Thus,
  Statement (ii) is satisfied.  
\end{proof}

\begin{lemma}\label{lem:tree-star}
  If a connected $\sthin$-thin 3-RBMG can be explained by tree of Type
  \AX{(II)}, then it can be explained by a tree of Type \AX{(II$^*$)}.
\end{lemma}
\begin{proof}
  Assume that $(T,\sigma)$ is of Type \AX{(II)} and that $G(T,\sigma)$ is a
  connected $\sthin$-thin 3-RBMG.  Let $S=\{r,s,t\}$ be the color set of
  $L\coloneqq L(T)$ and $v_1,v_2\in \child(\rho_T)$ with
  $\sigma(L(T(v_1)))=\{r,s\}$ and $\sigma(L(T(v_2)))=\{r,t\}$. If
  $(T,\sigma)$ is already of Type \AX{(II$^*$)}, then the statement is
  trivially true.
	
  Now suppose that $(T,\sigma)$ is not of Type \AX{(II$^*$)}. Lemma
  \ref{lem:one-Ruler} implies that there is exactly one leaf $x\in L[r]$ to
  which Rule \AX{(r)} can be applied. Hence, by using Rule \AX{(r)} and
  thus re-attaching $x$ to the root, one obtains a tree $(T',\sigma)$ of
  Type \AX{(II$^*$)}.  In particular, Lemma \ref{lem:one-Ruler} implies
  that $x$ is the only vertex with color $r$ in $(T',\sigma)$ incident to
  the root.  W.l.o.g.\ assume that $x\in L(T(v_1))$.  Note, in particular,
  that the necessity of relocating $x$ implies
  $L[s]\setminus L(T(v_1))=\emptyset$ (cf.\ Rule \AX{(r)}), i.e.,
  $L[s]\subseteq L(T(v_1))$. Thus, $\child(\parent(x))\cap L=\{x\}$ and
  Lemma \ref{lem:TreeTypes}(i)+(iv) implies that $N_s(x)=\emptyset$ in
  $G(T,\sigma)$.

  Since $L=L(T')$, it suffices to show that
  $E(G(T',\sigma))=E(G(T,\sigma))$ to prove that $(T',\sigma)$ explains
  $G(T,\sigma)$.  One easily checks that the only edges that may be
  different between both sets are those containing the leaf $x$.

  We start by showing $N_s(x)=\emptyset$ in $G(T',\sigma)$.  Observe first
  that, as we have only changed the position of vertex $x\in L[r]$ to
  obtain $(T',\sigma)$, $L[s]\subseteq L(T(v_1))$ implies
  $L[s]\subseteq L(T'(v_1))$.  By Lemma \ref{lem:2col}, we have
  $\sigma(L(T(w))) = \{r,s\}$ for all inner vertices $w\preceq_{T}v_1$ in
  $T$.  Thus, there must be a vertex $w\in(L(T(v_1)))$ that is incident to
  two leaves $x'$ and $y$ with $\sigma(x') = r$ and $\sigma(y)=s$. Since
  $\{x',y\} \subseteq \child(\parent(y))$, it follows $x'\neq x$ and thus,
  $x'$ has not been re-attached.  The latter implies that
  $\sigma(L(T'(v_1))) = \{r,s\}$ and, by construction,
  $\lca_{T'}(x,y') =\rho_{T'} \succ_{T'} v_1 \succ_{T'} \lca_{T'}(x',y')$
  for all $x'\in L(T'(v_1))\cap L[r]$ and
  $y'\in L[s]\cap L(T'(v_1)) = L[s]$.  Therefore, there is no edge between
  $x$ and any $y'\in L[s]$ in $G(T',\sigma)$.  Hence, $N_s(x) = \emptyset$
  in $G(T',\sigma)$.
			
  It remains to show that $xz\in E(G(T,\sigma))$ if and only if
  $xz \in E(G(T',\sigma))$ for this particular re-located vertex $x$ and
  all $z\in L[t]$.  Since $G(T,\sigma)$ is connected and
  $N_s(x)=\emptyset$, there must exist some vertex $z$ with color $t$ such
  that $xz\in E(G(T,\sigma))$.  Note, Lemma \ref{lem:TreeTypes}(ii) implies
  that there are no edges $xz$ for all $z\in L[t]\cap L(T(v_2))$.  That is,
  $x$ and $z\in L[t]$ form an edge $xz$ in $G(T,\sigma)$ if and only if $z$
  is incident to the root $\rho_T$ (cf.\ Lemma
  \ref{lem:TreeTypes}(ii)+(iv)).  In particular, we have by construction of
  $(T',\sigma)$ that
  $\lca_{T'}(x,z) = \rho_{T'} = \lca_{T'}(x',z) =\lca_{T'}(x,z')$ for all
  $x'\in L[r]$ and all $z'\in L[t]$ and hence, $xz\in E(G(T',\sigma))$.

  Now assume that $xz\in E(G(T',\sigma))$ for this particular re-located
  vertex $x$ and some $z\in L[t]$.  Since $x$ has been re-attached, we have
  $\lca_{T'}(x,z) =\rho_{T'}$.  By Lemma \ref{lem:one-Ruler}, none of the
  vertices in $L(T(v_2))$ has been re-attached.  Hence,
  $L(T'(v_2)) = L(T(v_2))$ and thus, $\sigma(L(T(v_2))) = \{r,t\}$.
  Moreover, we have $xz'\notin E(G(T',\sigma))$ for all
  $z'\in L[t]\cap L(T(v_2))$ since
  $\lca_{T'}(x,z') =\rho_{T'} \succ_{T'} \lca_{T'}(x',z')$ for all
  $x'\in L[r]\cap L(T(v_2))$.  Thus, $z$ must be adjacent to
  $\rho_{T'}$. By construction, $z$ must be adjacent to $\rho_{T}$.  As
  argued above, $xz$ in $G(T,\sigma)$ if and only if $z$ is incident to the
  root $\rho_T$.  Therefore, $xz\in E(G(T,\sigma))$.

  In summary, $E(G(T',\sigma))=E(G(T,\sigma))$ and hence, $(T',\sigma)$
  explains $G(T,\sigma)$.  
\end{proof}
	
\subsection{Three classes of $\sthin$-thin $3$-RBMGs}

We are now in the position to use these results to show that connected 
components of 3-RBMGs can be grouped into three disjoint graph classes
that correspond to the three tree Types \AX{(I)}, \AX{(II)}, and
\AX{(III)}.  These three classes are shown in Fig.\ \ref{fig:categories}.

\begin{definition}
  An undirected, connected graph $(G,\sigma)$ on three colors is of
  \begin{description}
  \item[\textbf{Type \AX{\bf (A)}}] if $(G,\sigma)$ contains a $K_3$ on
    three colors but no induced $P_4$, and thus also no induced $C_n$,
    $n\ge 5$.
  \item[\textbf{Type \AX{\bf (B)}}] if $(G,\sigma)$ contains an induced
    $P_4$ on three colors whose endpoints have the same color, but no
     induced $C_n$ for $n\ge 5$.
  \item[\textbf{Type \AX{\bf (C)}}] if $(G,\sigma)$ contains an induced
    $C_6$ along which the three colors appear twice in the same
    permutation, i.e., $(r,s,t,r,s,t)$.
  \end{description}
  \label{def:typesABC}
\end{definition}

\begin{theorem}\label{thm:3c-types}
  Let $(G,\sigma)$ be an $\sthin$-thin connected 3-RBMG.  Then $(G,\sigma)$
  is either of Type \AX{(A)}, \AX{(B)}, or \AX{(C)}.  An RBMG of Type
  \AX{(A)}, \AX{(B)}, and \AX{(C)}, resp., can be explained by a tree of
  Type \AX{(I)}, \AX{(II)}, and \AX{(III)}, respectively.
\end{theorem}
\begin{proof}
  Let $(G,\sigma)$ be an $\sthin$-thin connected 3-RBMG. If $|L|=3$, then
  $|\sigma(L)|=3$ implies that $(G,\sigma)$ is the complete graph $K_3$ on
  three colors, i.e., a graph of Type \AX{(A)}. Every phylogenetic tree on
  three leaves explains $(G,\sigma)$, and all of them except the star are
  of Type \AX{(I)}. From here on we assume $|L|>3$.  By Lemma
  \ref{lem:2col-subtrees} every connected $\sthin$-thin 3-RBMG
  $(G,\sigma)$ is explained by a tree $(T,\sigma)$ of either Type \AX{(I)},
  \AX{(II)}, or \AX{(III)}.

  \smallskip\par\noindent\textit{Claim 1.} If $(T,\sigma)$ is of Type
  \AX{(I)}, then $(G,\sigma)$is of Type \AX{(A)}.
  \par\noindent\textit{Proof of Claim 1.}
  Let $(T,\sigma)$ be a Type \AX{(I)} tree, i.e., the root $\rho_T$ has one
  child $v$ such that $\sigma(L(T(v)))=\{r,s\}$ and all other children of
  $\rho_T$ are leaves. This and $|\sigma(L)|=3$ implies that there must be
  a leaf $z\in \child(\rho_T)\cap L[t]$. Since $(G,\sigma)$ is
  $\sthin$-thin, $z$ is the only leaf of color $t$ in $(T,\sigma)$, thus
  $xz\in E(G)$ for every $x\in L[r]$ and $yz\in E(G)$ for every
  $y\in L[s]$.  In particular, Lemma \ref{lem:2col} implies that there
  exists an inner vertex $u\preceq_T v$ such that $\child(u)=\{x^*,y^*\}$
  with $x^*\in L[r]$ and $y^*\in L[s]$, thus we have $x^*y^*\in E(G)$.
  Hence, the induced subgraph $G[x^*y^*z]$ forms a $K_3$.

  It remains to show that $(G,\sigma)$ contains no induced $C_n$, $n\ge 5$,
  or $P_4$. Since $L[t]=\{z\}$ and $xz,yz\in E(G)$ for any $x\in L[r]$ and
  $y\in L[s]$, we can conclude that there cannot be any induced $P_4$, and
  thus no induced $C_n$, $n\ge 5$, either, that contains color $t$. Now
  assume, for contradiction, that there is an induced $P_4$ that contains
  the two colors $r,s$. By construction, this $P_4$ must have subsequent
  coloring $(r,s,r,s)$, thus it contains three distinct vertices
  $x,y_1,y_2$ such that $x\in L[r]$ and $y_1,y_2\in L[s]$ and $x$ is
  adjacent to $y_1$ and $y_2$. Lemma \ref{lem:TreeTypes}(i) implies that
  $\parent(y_1)=\parent(y_2)$. Hence, $N(y_1)=N(y_2)$, which contradicts
  the $\sthin$-thinness of $(G,\sigma)$. Thus, there exists no induced
  $P_4$ and thus no induced $C_n$ with $n\ge 5$ containing only two colors.

  Hence, $(G,\sigma)$ is of Type \AX{(A)}.  \hfill$\triangleleft$
		
  \smallskip\noindent\textit{Claim 2.} If $(T,\sigma)$ is of Type
  \AX{(II)}, then $(G,\sigma)$ is of Type \AX{(B)}.
  \par\noindent\textit{Proof of Claim 2.}
  Let $(T,\sigma)$ be a Type \AX{(II)} tree, i.e., the root has two
  distinct children $v_1,v_2\in \child(\rho_T)$ such that
  $\sigma(L(T(v_1)))=\{r,s\}$ and $\sigma(L(T(v_2)))=\{r,t\}$, and all
  other children of the root are leaves.
	
  We start by showing that $(G,\sigma)$ contains the particular colored
  induced $P_4$.  Lemma \ref{lem:2col} implies that there must be a leaf
  $y_1\in L(T(v_1))\cap L[s]$ such that $\parent(y_1)=\parent(x_1)$ for
  some $x_1\in L(T(v_1))\cap L[r]$ and therefore $x_1y_1\in E(G)$.
  Similarly, there exist two leaves $z_1\in L(T(v_2))\cap L[t]$ and
  $x_2\in L(T(v_2))\cap L[r]$ such that $x_2z_1\in E(G)$.  Lemma
  \ref{lem:TreeTypes}(ii) implies that $x_1z_1\notin E(G)$ and
  $x_2y_1\notin E(G)$.  Clearly, $x_1x_2\notin E(G)$ since the two vertices
  have the same color.  Moreover, Lemma \ref{lem:TreeTypes}(iii) implies
  that $y_1z_1\in E(G)$. Hence, $\langle x_1y_1z_1x_2\rangle$ forms an
  induced $P_4$ in $(G,\sigma)$ on three colors whose endpoints have the
  same color.
	
  We proceed by showing that $(G,\sigma)$ does not contain 
 	an induced $C_n$ with $n\ge 5$.  First note that Lemma
  \ref{lem:TreeTypes}(iii) implies $yz\in E(G)$ for any two leaves
  $y\in L[s]$, $z\in L[t]$. Thus, $(G,\sigma)$ cannot contain an induced
  $C_n$ for some $n\ge 5$ on colors $s$ and $t$
  only. Therefore, assume, for contradiction, that there exists an induced
  $C_n$ for some fixed $n\ge 5$ in $G(T,\sigma)$ that contains a leaf $x$
  of color $r$. Note that this necessarily implies $|N(x)|>1$ in
  $G$. Suppose first that $x\in \child(\rho_T)$. Since
  $(T(v_1),\sigma_{|L(T(v_1))})$ and $(T(v_2),\sigma_{|L(T(v_2))})$ both
  contain leaves of color $r$, any vertex that is adjacent to $x$ in $G$
  must be incident to $\rho_T$ in $T$.  Hence, as $(G,\sigma)$ is
  $\sthin$-thin, $|N(x)|>1$ in $G$ implies that
  $\child(\rho_T)\cap L =\{x,y,z\}$, where $y\in L[s]$ and $z\in L[t]$,
  i.e., we have $N(x)=\{y,z\}$. Thus, any induced $C_n$, $n\geq 5$
  containing $x$ must also contain both $y$ and $z$. However, as $x$, $y$,
  and $z$ have the same parent in $T$, they clearly form a $K_3$ in $G$; a
  contradiction to $x$, $y$, and $z$ being part of an induced $C_n$.

  Now suppose $x\in L(T(v_1))\cap L[r]$.  Since
  $(T(v_2),\sigma_{|L(T(v_2))})$ contains the colors $r$ and $t$,
  $(G,\sigma)$ cannot contain an edge $xz$ with $z\in L(T(v_2))$ (cf.\
  Lemma \ref{lem:TreeTypes}(ii)).  Hence, $N_t(x)\neq \emptyset$ if and
  only if there exists a leaf $z$ of color $t$ that is directly attached to
  the root $\rho_T$.  Since $G(T,\sigma)$ is $\sthin$-thin, there can be at
  most one leaf of color $t$ that is attached to $\rho_T$, thus
  $|N_t(x)|\le 1$. This and $|N(x)|>1$ in $G$ implies that there must be a
  leaf $y\in L[s]$ such that $y\in N_s(x)$. By Lemma
  \ref{lem:TreeTypes}(i), this is the case if and only if
  $y\in L(T(v_1))\cap L[s]$ and $\parent(x)=\parent(y)$. Since
  $G(T,\sigma)$ is $\sthin$-thin, there exists at most one leaf of color
  $s$ with this property, hence in particular $N(x)=\{y,z\}$.  Using Lemma
  \ref{lem:TreeTypes}(iv), we can conclude that $yz\in E(G)$. Thus, $x$,
  $y$ and $z$ form a $K_3$. Therefore, these three leaves cannot be
  contained together in an induced $C_n$, $n\geq 5$.

  Since an analogous argumentation holds if $x\in L(T(v_2))$, we conclude
  that there cannot be an induced $C_n$, $n\geq 5$ containing
  a leaf of color $r$.

  In summary, $G$ does not contain an induced  $C_n$ for $n\ge 5$, and thus, $(G,\sigma)$ is of Type
  \AX{(B)}.  \hfill$\triangleleft$
    
  \smallskip\noindent\textit{Claim 3.} If $(T,\sigma)$ is of Type
  \AX{(III)}, then $(G,\sigma)$ is of Type \AX{(C)}.
  \par\noindent\textit{Proof of Claim 3.}
  Let $(T,\sigma)$ be a Type \AX{(III)} tree, i.e., the root $\rho_T$ has
  three children $v_1,v_2,v_3\in \child(\rho_T)$ such that
  $\sigma(L(T(v_1)))=\{r,s\}$, $\sigma(L(T(v_2)))=\{r,t\}$, and
  $\sigma(L(T(v_3)))=\{s,t\}$, and all remaining children are
  leaves. Again, Lemma \ref{lem:2col} and Lemma \ref{lem:TreeTypes}(i)
  imply that there exist $x_1,y_1\in L(T(v_1))$ with $x_1y_1\in E(G)$,
  $x_2,z_1\in L(T(v_2))$ with $x_2z_1\in E(G)$ and $y_2,z_2\in L(T(v_3))$
  with $y_2z_2\in E(G)$, where $x_i\in L[r]$, $y_i\in L[s]$ and
  $z_i\in L[t]$. Applying Lemma \ref{lem:TreeTypes}(ii), we can in addition
  conclude that $y_1z_1,x_2y_2,x_1z_2\in E(G)$, and $(G,\sigma)$ contains
  none of the edges $x_1z_1$, $x_1y_2$, $y_1x_2$, $y_1z_2$, $z_1y_2$, or
  $x_2z_2$. Moreover, $(G,\sigma)$ does not contain edges between vertices
  of the same color.  Hence, $(G[C],\sigma_{|C})$ with
  $C=\{x_1,y_1,z_1,x_2,y_2,z_2\}$ forms the desired induced
  $C_6$. Therefore, $(G,\sigma)$ is of Type \AX{(C)}.
  \hfill$\triangleleft$

  By definition, the three classes of 3-RBMGs \AX{(A)}, \AX{(B)}, and
  \AX{(C)} are disjoint. Lemma \ref{lem:2col-subtrees} states that the
  three classes of trees \AX{(I)}, \AX{(II)}, and \AX{(III)} are disjoint,
  hence there is a 1-1 correspondence between the tree Types \AX{(I)},
  \AX{(II)}, and \AX{(III)} and the graph classes \AX{(A)}, \AX{(B)}, and
  \AX{(C)}.
  
\end{proof}

An undirected, colored graph $(G,\sigma)$ contains an induced $K_3$, $P_4$,
or $C_6$, respectively, if and only if $(G/\sthin,\sigma/\sthin)$ contains
an induced $K_3$, $P_4$, or $C_6$, resp., on the same colors (cf.\ Lemma
\ref{lem:sthin}). An immediate consequence of this fact is
\begin{theorem}\label{thm:thinABC}
  A connected (not necessarily $\sthin$-thin) 3-RBMG $(G,\sigma)$ is
  either of Type \AX{(A)}, \AX{(B)}, or \AX{(C)}.
\end{theorem}

\subsection{Characterization of Type \AX{(A)} 3-RBMGs}

As an immediate consequence of Theorem \ref{thm:3c-types} and the
well-known properties of cographs \cite{Corneil:81} we obtain
\begin{fact}\label{fact:cograph}
  Let $(G,\sigma)$ be a connected, $\sthin$-thin 3-RBMG.  Then it is of
  Type \AX{(A)} if and only if it is a cograph.
\end{fact}

\begin{definition}
  Let $G=(V,E)$ be an undirected graph. A vertex $x\in V(G)$ such that
  $N(x)=V\setminus \{x\}$ is a \emph{hub-vertex}.
\label{def:hub}
\end{definition}

\begin{lemma}\label{lem:charA} 
  A properly vertex colored, connected, $\sthin$-thin graph $(G,\sigma)$ on
  three colors with vertex set $L$ is a 3-RBMG of Type \AX{(A)} if and
  only if $G \notin \mathscr{P}_3$ and it satisfies the following
  conditions:
  \begin{enumerate}
  \item[\AX{(A1)}] $G$ contains a hub-vertex $x$, i.e.,
    $N(x)=V(G)\setminus \{x\}$
  \item[\AX{(A2)}] $|N(y)|<3$ for every $y\in V(G)\setminus \{x\}$.
  \end{enumerate}
\end{lemma}
\begin{proof}
  By definition, a 3-RBMG is properly colored and has $|L|\ge 3$
  vertices. If $|L|=3$ there are only two connected graphs: $K_3$ and
  $P_3$. Both satisfy \AX{(A1)} and \AX{(A2)} since the three vertices have
  distinct colors. However, only $K_3$ is a 3-RBMG: it is explained by any
  tree on three leaves with pairwise distinct colors.  From here on we
  assume $|L|\ge 4$.

  We start with the ``only-if-direction'' and show that every 3-RBMG of
  Type \AX{(A)} satisfies \AX{(A1)} and \AX{(A2)}. We set $S=\{r,s,t\}$,
  and assume that $(G,\sigma)$ is a 3-RBMG of Type \AX{(A)}. Theorem
  \ref{thm:3c-types} implies that there exists a tree $(T,\sigma)$ with
  root $\rho_T$ explaining $(G,\sigma)$ that is of Type \AX{(I)}, i.e.,
  there is a vertex $v\in\child(\rho_T)$ such that
  $\sigma(L(T(v)))=\{s,t\}$ and $\child(\rho_T)\setminus \{v\}\subset
  L$. Thus every leaf $x$ with color $\sigma(x)=r$ is a child of
  $\rho_T$. Since $(G,\sigma)$ is $\sthin$-thin, this implies $|L[r]|=1$
  and therefore, $xy\in E(G)$ for every $y\neq x$, hence \AX{(A1)} is
  satisfied.  In order to show \AX{(A2)}, consider
  $y\in V(G)\setminus \{x\}$, where $x$ is again the unique vertex with
  color $r$.  Since $(G,\sigma)$ is properly colored, $\sigma(y)\neq
  r$. W.l.o.g., let $\sigma(y)=s$. Assume, for contradiction, that
  $|N(y)|\geq 3$. Then there are at least two distinct vertices
  $z,z'\in N_t(y)$. Assume first that $y\in L(T(v))$. Hence, there is
  $z^*\in L[t]$ with $\lca(y,z^*)\preceq v \prec \rho_T$. Hence, we must
  have $z,z'\in L(T(v))$. However, Lemma \ref{lem:TreeTypes}(i) implies
  that $z$ and $z'$ must be siblings and therefore $N(z)=N(z')$; a
  contradiction, since $(G,\sigma)$ was assumed to be $\sthin$-thin.  Now
  assume that $y\in \child(\rho_T)$.  Then, Lemma \ref{lem:TreeTypes}(iv)
  and $z,z'\in N_t(y)$ imply that $z$ and $z'$ both have to be adjacent to
  $\rho_T$; again this contradicts the assumption that $(G,\sigma)$ is
  $\sthin$-thin.  Thus, \AX{(A2)} is satisfied.
 
  We proceed with showing the ``if-direction''. Suppose $(G,\sigma)$ is a
  properly vertex colored, connected, $\sthin$-thin graph satisfying
  \AX{(A1)} and \AX{(A2)}.  In order to show that $(G,\sigma)$ is a Type
  \AX{(A)} RBMG it suffices, by Theorem \ref{thm:3c-types}, to construct a
  Type \AX{(I)} tree that explains $(G,\sigma)$.  Let $x$ be a vertex that
  is adjacent to all others, which exists by $\AX{(A1)}$. Assume w.l.o.g.\
  that $\sigma(x)=r$. Since $(G,\sigma)$ is $\sthin$-thin and it does not
  contain edges between vertices of the same color, $x$ must be the only
  vertex of color $r$. We define $L_2\coloneqq \{y \mid y\neq x, |N(y)|=2\}$.
  Since $|L|>3$ and thus $|N(x)|\ge 3$, we have $x\in L\setminus L_2$.
  Note, each vertex is adjacent to $x$ and thus, $N(y) = \{x,z\}$ for all
  $y\in L_2$ and some vertex $z\in L[t]$, $t\neq \sigma(y)$.  Property
  \AX{(A2)} implies that there are $|L\setminus L_2|-1$ vertices with
  degree $1$, all incident to $x$. Since $(G,\sigma)$ is $\sthin$-thin and
  $|S|=3$, there are at most two vertices with degree $1$, at most one of
  each color different from $\sigma(x)$, and thus $|L\setminus L_2|\le 3$.

  We first construct a caterpillar $(T_2,\sigma_{|L_2})$ with leaf set $L_2$
  and root $\rho_{T_2}$ such that $\parent(y)=\parent(z)$ for any
  $y,z\in L_2$ with $\sigma(y)\neq \sigma(z)$ if only if $yz\in E(G)$.  As
  $N(y) = \{x,z\}$ for all $y\in L_2$ and some vertex $z\in L[t]$,
  $t\neq \sigma(y)$, we can conclude that each connected component of
  $(G[L_2], \sigma_{|L_2})$ is a single edge $yz$.  Thus, it is easy to see
  that $(T_2,\sigma_{|L_2})$ explains $(G[L_2], \sigma_{|L_2})$.
	
  For the construction of $(T,\sigma)$, we then distinguish two cases: (i)
  If $L\setminus L_2=\{x,w_1\}$, then $(T,\sigma)$ is obtained by attaching
  the vertices $x$ and $w_1$ as well as $\rho_{T_2}$ as children of the
  root $\rho_T$. (ii) If $L\setminus L_2 = \{x,w_1,w_2\}$ for distinct
  vertices $x$, $w_1$, and $w_2$ in $(G,\sigma)$, we first build an
  auxiliary tree $(T',\sigma_{|L_2\cup \{w_2\}})$ with root $\rho_{T'}$ by
  attaching $\rho_{T_2}$ and vertex $w_2$ to $\rho_{T'}$. The tree
  $(T,\sigma)$ is then constructed from $(T',\sigma_{|L_2\cup \{w_2\}})$ by
  attaching $x$, the other vertex $w_1$ and $\rho_{T'}$ as children of
  $\rho_T$. It remains to show that $(T,\sigma)$ explains $(G,\sigma)$.  By
  construction, $(T,\sigma)$ is a tree of Type \AX{(I)} where the vertices
  $\rho_{T_2}$ and $\rho_{T'}$ play the role of $v$ in Def.\
  \ref{def:typeiistar} in Case (i) and (ii), respectively.  In the
  following let $v=\rho_{T_2}$ or $v= \rho_{T'}$ depending on whether we
  have Case (i) and (ii).
  
  Theorem \ref{thm:3c-types} implies that $G(T,\sigma)$ is 3-RBMG of Type
  \AX{(A)}.  It is easy to see that
  $G(T,\sigma)[L_2] = (G[L_2], \sigma_{|L_2})$.  Any remaining edges in
  $G(T,\sigma)$ are thus adjacent to vertices in $L\setminus L_2$.
  We first consider edges that may be incident to vertex $x$ in
  $G(T,\sigma)$.  Since $\lca_T(z, x) = \rho_T$ and
  $r\notin \sigma(L(T(v)))$, we have $xz \in E(G(T,\sigma))$ for all
  $z\in L\setminus \{x\}$. Hence, \AX{(A1)} is satisfied by $x$ in
  $G(T,\sigma)$.
	 
  Now, consider edges that may be incident to vertex $w_1$ in
  $G(T,\sigma)$.  First note, that in both Cases (i) and (ii), the vertex
  $w_1$ is adjacent to the root $\rho_T$ in $(T,\sigma)$.  Since
  $(T,\sigma)$ is a tree of Type \AX{(I)}, we can apply Lemma
  \ref{lem:TreeTypes}(i) to conclude that there are no edges in
  $G(T,\sigma)$ between $w_1$ and any vertex in
  $L(T(v))$. 
  This and the arguments above show that
  $N(w_1)=\{x\}$.  In other words, $w_1z \in
  E(G(T,\sigma))$ if and only if $w_1z \in E(G)$ for all $z\in
  L$. This in particular shows
  $(G,\sigma)=G(T,\sigma)$ conforms to Case (i).

  Finally, assume Case (ii) and consider edges that are incident to vertex
  $w_2$ in $(G,\sigma)$ By construction, $s,t\in
  \sigma(L_2)$. Since $\lca(w_2,z) =v \succ_T \rho_{T_2} \succeq
  \lca(w',z)$ for all $w',z\in L_2$ with $\sigma(w_2)=\sigma(w) \neq
  \sigma(z)$, we can conclude that
  $w_2$ is not adjacent to any other vertex in
  $L_2$.  This and the arguments above show that
  $N(w_2)=\{x\}$.  Therefore, $w_2z \in
  E(G(T,\sigma))$ if and only if $w_1z \in E(G)$ for all $z\in L$.
	
  In summary, $G(T,\sigma)=(G,\sigma)$. Therefore, $(G,\sigma)$ is of Type
  \AX{(A)}.
  
\end{proof}

For later reference, we record a simple property of hub-vertices.
\begin{corollary}\label{cor:hub-vertex}
  Let $x$ be a hub-vertex of some connected $\sthin$-thin 3-RBMG
  $(G,\sigma)$ of Type \AX{(A)} with vertex set $L$ and $|L|>3$. Then, $x$
  is the only vertex of its color in $(G,\sigma)$, i.e.,
  $L[\sigma(x)]= \{x\}$. Moreover, for any $(T,\sigma)$ explaining
  $(G,\sigma)$, $x$ must be incident to the root of $T$.
\end{corollary}
\begin{proof}
  Since $(G,\sigma)$ does not contain edges between vertices of the same
  color, the first statement immediately follows from Property \AX{(A1)} and
  $\sthin$-thinness of $(G,\sigma)$.  
  
  For the second statement, let $(T,\sigma)$ be an arbitrary tree with root
  $\rho_T$ that explains $(G,\sigma)$. Let $v\in \child(\rho_T)$ with
  $x\preceq_T v$. Assume, for contradiction, $v\neq x$. Thus Lemma
  \ref{lem:2col} implies that there exists a leaf $y\in L$ with
  $\sigma(y)\neq \sigma(x)$ such that $y\preceq_T v$. Then, since $x$ is
  connected to any vertex in $L\setminus \{x\}$, all vertices of color
  $\sigma(y)$ must be contained in the subtree $T(v)$; otherwise
  $\lca_T(x,y)\prec_T \lca_T(x,y')=\rho_T$ for some vertex
  $y'\in L[\sigma(y)]$, $y'\neq y$, which yields a contradiction to
  $xy'\in E(G)$. As $(T,\sigma)$ is phylogenetic, the root $\rho_T$ has at
  least two different children, i.e., there is some $w\in \child(\rho_T)$,
  $w\neq v$. Let $r\neq \sigma(x),\sigma(y)$ be the third color in
  $(G,\sigma)$. We already argued
  $\sigma(x),\sigma(y)\notin\sigma(L(T(w)))$, thus
  $\sigma(L(T(w)))=\{r\}$. In particular, since $(G,\sigma)$ is
  $\sthin$-thin, Lemma \ref{lem:2col} implies that $w$ must be a
  leaf. Since $(G,\sigma)$ is connected, we can apply the same arguments as
  for $L[\sigma(y)]$ to conclude that $r\notin \sigma(L(T(v)))$, thus
  $|L[r]|=1$. Since $x$ is the only leaf of its color in $(T,\sigma)$ and
  $\sigma(L(T(v)))=\{\sigma(x),\sigma(y)\}$, we can again apply Lemma
  \ref{lem:2col} to conclude that $|L[\sigma(y)]|=1$. In summary, we have
  therefore shown $|L|=3$; a contradiction. Hence, $x$ must be incident to
  $\rho_T$.  
\end{proof}

\begin{lemma}
  Let $(G,\sigma)$ be an $\sthin$-thin graph satisfying \AX{(A1)} and
  \AX{(A2)}. Then $G$ is a cograph.
\end{lemma}
\begin{proof}
  Since $G$ contains a hub-vertex by \AX{(A1)}, it can be written as join
  $G'\join K_1$, where the $K_1$ corresponds to the hub-vertex. As a
  consequence of \AX{(A2)}, $(G',\sigma)$ is a 2-colored graph with vertex
  degree at most $1$. The number of isolated vertices in $G'$ cannot exceed
  $2$, one of each color, since otherwise two vertices that are isolated in
  $G'$ would have the same color and thus share the hub as their only
  neighbor in $G$, contradicting $\sthin$-thinness of $(G,\sigma)$. Hence,
  $G'$ is the disjoint union of an arbitrary number of $K_2$ and at most
  two copies of $K_1$:
  $G = \left((\bigcupdot^{n_1} K_1)\cup (\bigcupdot^{n_2} K_2)\right)\join
  K_1$ with $0\le n_1\le 2$ and $n_2\ge 0$. Thus $G$ is a cograph
  \cite{Corneil:81}.  
\end{proof}

\subsection{Characterization of Type \AX{(B)} 3-RBMGs}

\begin{definition}
  Let $(G,\sigma)$ be an undirected, connected, properly colored,
  $\sthin$-thin graph with vertex set $L$ and color set
  $\sigma(L)=\{r,s,t\}$, and assume that $(G,\sigma)$ contains the induced
  path $P\coloneqq \langle \hat x_1 \hat y \hat z \hat x_2 \rangle$ with
  $\sigma(\hat x_1)=\sigma(\hat x_2)=r$, $\sigma(\hat y)=s$, and
  $\sigma(\hat z)=t$. Then $(G,\sigma)$ is \emph{B-like w.r.t.\ $P$} if (i)
  $N_{r}(\hat y)\cap N_{r}(\hat z)=\emptyset$, and (ii) $G$ does not
  contain an induced cycle $C_n$, $n\ge 5$.
\label{def:Ltsr}
\end{definition}

For a $3$-colored, $\sthin$-thin graph $(G,\sigma)$ that is B-like w.r.t.\
the induced path
$P\coloneqq \langle \hat x_1 \hat y \hat z \hat x_2 \rangle$ we define the
following subsets of vertices:
\begin{align*} 	
  L_{t,s}^{P} \coloneqq & \{y \mid  \langle xy\hat z\rangle \in
                          \mathscr{P}_3 \text{ for any } x\in N_{r}(y)\}\\
  L_{t,r}^{P} \coloneqq & \{x\mid N_{r}(y)=\{x\} \text{ and }
                          \langle xy\hat z\rangle \in \mathscr{P}_3\}
                          \cup \\
                        & \{x\mid x\in L[r],\, N_{s}(x)=\emptyset,\,
                          L[s]\setminus L_{t,s}^P\neq\emptyset\}\\
  L_{s,t}^{P} \coloneqq & \{z \mid \langle xz\hat y\rangle \in \mathscr{P}_3
                          \text{ for any } x\in N_{r}(z)\} \\ 
  L_{s,r}^{P} \coloneqq & \{x\mid N_{r}(z)=\{x\} \text{ and }
                            xz\hat y \in \mathscr{P}_3\} \cup \\
                        & \{x\mid x\in L[r], N_{t}(x)=\emptyset,
                          L[t]\setminus L_{s,t}^P\neq\emptyset\}\\
\end{align*}
The first subscripts $t$ and $s$ refer to the color of the vertices
$\hat z$ and $\hat y$, respectively, that ``anchor'' the $P_3$s within the
defining path $P$. The second index identifies the color of the vertices in
the respective set, since by definition we have
$L_{t,s}^{P}\subseteq L[s]$, $L_{t,r}^{P}\subseteq L[r]$,
$L_{s,t}^{P}\subseteq L[t]$ and $L_{s,r}^{P}\subseteq L[r]$. Furthermore,
we set
\begin{align*}
  L_{t}^{P} \coloneqq & L_{t,s}^{P} \cup L_{t,r}^{P} \\
  L_{s}^{P} \coloneqq & L_{s,t}^{P} \cup L_{s,r}^{P} \\
  L_*^P     \coloneqq & L \setminus (L_t^P \cup L_s^P). 
\end{align*}
By definition, $L_{s,r}^P=L_s^P\cap L[r]$, $L_{t,r}^P=L_t^P\cap L[r]$,
$L_{s,t}^P=L_s^P\cap L[t]$, and $L_{t,s}^P=L_t^P\cap L[s]$. For simplicity
we will often write $L_*^P[i] \coloneqq L_*^P\cap L[i]$ for $i\in \{s,t\}$.

These vertex sets arise naturally from trees of Type \AX{(II$^*$)}:
\begin{lemma}\label{lem:L^B-tree}
  Let $(G,\sigma)$ be a connected $\sthin$-thin 3-RBMG of Type \AX{(B)}
  with vertex set $L$ and color set $S=\{r,s,t\}$. Then, the colors can be
  permuted such that there are $\hat x_1, \hat x_2\in L[r]$,
  $\hat y\in L[s]$, $\hat z \in L[t]$ such that $(G,\sigma)$ is B-like
  w.r.t.\ $P= \langle\hat x_1\hat y\hat z \hat x_2 \rangle$. Moreover,
  there exists a tree $(T,\sigma)$ of Type \AX{(II$^*$)} explaining
  $(G,\sigma)$ such that
  \begin{itemize}
  \item[(i)] $L_t^P=L(T(v_1))$ and $L_s^P=L(T(v_2))$ for
    $v_1,v_2\in\child(\rho_T)\setminus L$, and
  \item[(ii)] $L_*^P=\child(\rho_T)\cap L$.
  \end{itemize}
\end{lemma}
\begin{proof}
  Let $(G,\sigma)$ be a connected, $\sthin$-thin 3-RBMG of Type
  \AX{(B)}. Then, by Lemmas \ref{lem:2col-subtrees} and
  \ref{lem:tree-star}, there is a tree $(T,\sigma)$ with root $\rho_T$
  explaining $(G,\sigma)$ that is of Type \AX{(II$^*$)}. In particular, the
  colors can be chosen such that there are $v_1,v_2\in\child(\rho_T)$ with
  $\sigma(L(T(v_1)))=\{r,s\}$, $\sigma(L(T(v_2)))=\{r,t\}$, and
  $\child(\rho_T)\setminus \{v_1,v_2\}\subset L$. Applying the same
  argumentation as in the proof of Thm.\ \ref{thm:3c-types} (Claim 2), we
  conclude that there are leaves $\hat x_1,\hat y \prec_T v_1$ and
  $\hat x_2,\hat z\prec_T v_2$, where $\hat x_1,\hat x_2\in L[r]$,
  $\hat y\in L[s]$, $\hat z\in L[t]$, such that
  $\langle \hat x_1\hat y\hat z\hat x_2\rangle$ is an induced $P_4$ in $G$.
  By $\sthin$-thinness of $(G,\sigma)$ and Lemma \ref{lem:TreeTypes}(i), we
  have $N_r(\hat y)=\{\hat x_1\}$ and $N_r(\hat z)=\{\hat x_2\}$, and thus
  $N_r(\hat y)\cap N_r(\hat z)=\emptyset$. Since 3-RBMGs of Type \AX{(B)}
  do not contain induced cycles on more than six vertices, $(G,\sigma)$ is
  B-like w.r.t.\ $\langle\hat x_1\hat y\hat z\hat x_2\rangle$. It remains
  to show Properties (i) and (ii).

  In what follows, we put for simplicity $L_1\coloneqq L_t^P$.  To
  establish Property (i) we treat vertices of colors $s$ and $r$
  separately.  Consider $y\in L[s]$. We first show that $y\in L(T(v_1))$
  implies $y\in L_1$. Clearly, since $L(T(v_1))$ contains leaves of color
  $r$, any $x\in N_r(y)$ must satisfy $x\prec_T v_1$. Lemma
  \ref{lem:TreeTypes}(i) implies that $x\in N_r(y)$ if and only if
  $\parent(x)=\parent(y)$. Therefore $|N_r(y)|\le 1$ because $(G,\sigma)$
  is $\sthin$-thin. If $N_r(y)=\emptyset$, then $y\in L_1$ by
  definition. Otherwise, $N_r(y)=\{x\}$ with $x\prec_T v_1$. By
  construction of $(T,\sigma)$, there is a leaf $x'\prec_T v_2$ of color
  $r$; this implies $\lca(\hat z,x)\succ_T\lca(\hat z,x')$ and hence
  $x\hat z\notin E(G)$. Since $y\hat z\in E(G)$ by Lemma
  \ref{lem:TreeTypes}(iii), we have
  $\langle xy\hat z \rangle \in \mathscr{P}_3$ and thus, $y\in L_1$. Hence,
  $L(T(v_1))\cap L[s] \subseteq L_1\cap L[s]$ as claimed.  We now show that
  $y\in L_1$ implies $y\in L(T(v_1))$. To this end, consider
  $y\in L_1\cap L[s]$, i.e., either $N_r(y)=\emptyset$ or
  $\langle xy\hat z \rangle \in \mathscr{P}_3$ for every $x\in
  N_r(y)$. Assume, for contradiction, that $y\notin L(T(v_1))$. Then $y$
  must be incident to the root $\rho_T$. Since $L(T(v_2))$ contains no leaf
  of color $s$, Lemma \ref{lem:TreeTypes}(iv) implies
  $y\hat x_2, y\hat z \in E(G)$. Since $\hat x_2 \hat z \in E(G)$, the
  vertices $\hat x_2, y,\hat z$ induce a $K_3$, thus
  $\langle \hat x_2 y \hat z \rangle\notin \mathscr{P}_3$, and therefore
  $y\notin L_1$; a contradiction. Hence, we can conclude
  $L_1\cap L[s] \subseteq L(T(v_1))\cap L[s]$. In summary we therefore have
  $L(T(v_1))\cap L[s] = L_1\cap L[s]$.

  Consider $x\in L[r]$. We show that $x\in L(T(v_1))$ implies $x\in L_1$.
  If there exists a leaf $y\in L[s]$ incident to $\parent(x)$, then
  $N_r(y)=\{x\}$ by Lemma \ref{lem:TreeTypes}(i) and
  $\langle xy\hat z \rangle \in \mathscr{P}_3$ by Lemma
  \ref{lem:TreeTypes}(ii)+(iii), implying $x\in L_1$. Otherwise,
  $\sthin$-thinness of $G$ implies that $\child(\parent(x))\cap L =
  \{x\}$. In this case, $N_s(x)=\emptyset$ by Lemma
  \ref{lem:TreeTypes}(i). Moreover, since $(T,\sigma)$ is of Type
  \AX{(II$^*$)} and $\child(\parent(x))\cap L=\{x\}$, we can apply
  Condition \AX{($\star$)} in Def.\ \ref{def:typeiistar} to conclude
  $L[s]\setminus L_{t,s}^P =L[s]\setminus L(T(v_1)) \neq \emptyset$, where
  equality holds because $L(T(v_1))\cap L[s] = L_1\cap L[s]=L_{t,s}^P$. In
  summary, we have thus shown that $x\in L_1$.  Hence,
  $L(T(v_1))\cap L[r]\subseteq L_1\cap L[r]$ as claimed.  Conversely, we
  show that $x\in L_1$ implies $x\in L(T(v_1))$. Assume that
  $x\in L_1\cap L[r]$. Then, by definition of $L_1$, we have
  $x\in L_{t,r}^P$. Thus, either (a) there is a leaf $y\in L[s]$ such that
  $N_r(y)=\{x\}$ and $\langle xy \hat z \rangle \in \mathscr{P}_3$, or, (b)
  $N_s(x)=\emptyset$ and
  $L[s]\setminus L_{t,s}^P=L[s]\setminus L(T(v_1))\neq \emptyset$.  In Case
  (a), assume first for contradiction that $y$ is adjacent to the root
  $\rho_T$.  Lemma \ref{lem:TreeTypes}(iv) implies that $x'y \in E(G)$ for
  any $x'\in L(T(v_2))\cap L[r]$. Since $L(T(v_2))\cap L[r]\neq \emptyset$
  and $|N_r(y)|=1$, we have $L(T(v_2))\cap L[r]=\{x\}$ and thus, as
  $\hat x_2\in L(T(v_2))$, it follows $x=\hat x_2$. However, since
  $x = \hat x_2$ and $\hat z$ are adjacent in $(G,\sigma)$, $xy\hat z$
  cannot form an induced $P_3$. We therefore conclude that $y$ cannot be
  adjacent to the root $\rho_T$.  Since $s\notin \sigma(L(T(v_2)))$, it
  must thus hold that $y\in L(T(v_1))$. Lemma \ref{lem:TreeTypes}(ii)+(iv)
  then implies that we have $x\in L(T(v_1))$.  In Case (b),
  $L[s]\setminus L_{t,s}^P=L[s]\setminus L(T(v_1))\neq \emptyset$ implies
  that there exists a leaf $y^*\in L[s]\setminus L(T(v_1))$. Since
  $s\notin \sigma(L(T(v_2)))$, this vertex $y^*$ must be incident to the
  root $\rho_T$. On the other hand, we have $xy^*\notin E(G)$ because
  $N_s(x) = \emptyset$, hence $x$ cannot be incident to $\rho_T$. Applying
  Lemma \ref{lem:TreeTypes}(iv) thus implies $x\in L(T(v_1))$. Therefore,
  $L_1\cap L[r]=L(T(v_1))\cap L[r]$.
  
  In summary we have shown $L_1=L(T(v_1))$. By symmetry of the definitions,
  analogous arguments imply $L_s^P = L(T(v_2))$, completing the proof of
  statement (i). Property (ii) now immediately follows from
  $\child(\rho_T)\cap L= L\setminus (L(T(v_1))\cup L(T(v_2)))= L\setminus
  (L_t^P\cup L_s^P)$. 
  
\end{proof}

The following remark will be useful for the design of algorithms to
recognize Type \AX{(B)} RBMGs. It implies, in particular, that testing
whether $(G,\sigma)$ is B-like w.r.t.\ some induced $P_4$ strongly depends
on the reference $P_4$, i.e., it is necessary to identify all $P_4$s in
$(G,\sigma)$.
\begin{fact}\label{fact:NOindep-P-choice}
  A connected $\sthin$-thin 3-RBMG $(G,\sigma)$ of Type \AX{(B)} may
  contain distinct induced $P_4$s $P$ and $P'$, both of the form
  $(r,s,t,r)$ for distinct colors $r,s,t$, such that $(G,\sigma)$ is B-like
  w.r.t.\ $P$ but not B-like w.r.t.\ $P'$. An example is given in
  Fig.\ \ref{fig:NOindep-P-choice}.
\end{fact}

Using the previous result, we obtain the following characterization for
3-colored RBMGs of Type \AX{(B)}.

\begin{lemma}\label{lem:charB}
  Let $(G,\sigma)$ be an undirected, connected, $\sthin$-thin and properly
  3-colored graph with color set $S=\{r,s,t\}$ and let $x\in L[r]$,
  $y\in L[s]$ and $z\in L[t]$. Then, $(G,\sigma)$ is a 3-RBMG of Type
  \AX{(B)} if and only if the following conditions are satisfied, after
  possible permutation of the colors:
  \begin{description}
  \item[\AX{(B1)}] $(G,\sigma)$ is B-like w.r.t.\
    $P = \langle\hat x_1\hat y\hat z \hat x_2 \rangle$ for some
    $\hat x_1, \hat x_2\in L[r]$, $\hat y\in L[s]$, $\hat z \in L[t]$,
  \item[\AX{(B2.a)}] If $x\in L_*^P$, then $N(x)=L_*^P\setminus\{x\}$,
  \item[\AX{(B2.b)}] If $x\in L_t^P$, then $N_s(x)\subset L_t^P$ and
    $|N_s(x)|\le 1$, and $N_t(x)= L_*^P[t]$,
  \item[\AX{(B2.c)}] If $x\in L_s^P$, then
    $N_t(x)\subset L_s^P$ and $|N_t(x)|\le 1$, and $N_s(x)= L_*^P[s]$
  \item[\AX{(B3.a)}] If $y\in L_*^P$, then
    $N(y)=L_s^P\cup (L_*^P\setminus\{y\})$, 
  \item[\AX{(B3.b)}] If $y\in L_t^P$, then
    $N_r(y)\subset L_t^P$ and $|N_r(y)|\le 1$, and $N_t(y)=L[t]$,
  \item[\AX{(B4.a)}] If $z\in L_*^P$, then
    $N(z)=L_t^P\cup (L_*^P\setminus\{z\})$, 
  \item[\AX{(B4.b)}] If $z\in L_s^P$, then
    $N_r(z)\subset L_s^P$ and $|N_r(z)|\le 1$, and $N_s(z)=L[s]$.
  \end{description} 
  In particular, $L_t^P$, $L_s^P$, and $L_*^P$ are pairwise disjoint and
  $\hat x_1, \hat y \in L_t^P$, $\hat x_2,\hat z\in L_s^P$.
\end{lemma}

\begin{proof}
  Suppose first that $(G,\sigma)$ satisfies Conditions \AX{(B1)} -
  \AX{(B4.b)}. By Condition \AX{(B1)}, $(G,\sigma)$ is B-like, thus in
  particular it contains no induced $C_n$, $n\ge 5$. Therefore, if
  $(G,\sigma)$ is an RBMG, then it must be of Type \AX{(B)}.

  In order to prove that $(G,\sigma)$ is indeed an RBMG, we construct a
  tree $(T,\sigma)$ based on the sets $L_t^P$, $L_s^P$, and $L_*^P$ and
  show that it explains $(G,\sigma)$. To this end, we show first that the
  sets $L_t^P$, $L_s^P$, and $L_*^P$ are pairwise disjoint. By definition,
  $L_*^P$ is disjoint from $L_t^P$ and $L_s^P$. Moreover, by definition,
  $\sigma(L_t^P)\cap \sigma(L_s^P)=\{r\}$. Thus, it suffices to show that
  any vertex $x\in L[r]\setminus L_*^P$ is is contained in exactly one of
  the sets $L_t^P$ or $L_s^P$. Assume, for contradiction, that there exists
  a leaf $x$ that is contained in both $L_t^P$ and $L_s^P$. Hence,
  $x\in L_{t,r}^P\cap L_{s,r}^P$. Then, by definition of $L_{t,r}^P$,
  either (a) $N_s(x)=\emptyset$ and
  $L[s]\setminus L_{t,s}^P\neq \emptyset$, or (b) $N_r(y)=\{x\}$ and
  $\langle xy\hat z \rangle\in \mathscr{P}_3$ for some $y\in L[s]$.
  Suppose first Case (a). Since $N_s(x)=\emptyset$ and $(G,\sigma)$ is
  connected, there must be a vertex $z\in L[t]$ such that $xz\in
  E(G)$. Since $x\in L_s^P$, Condition \AX{(B2.c)} implies $N_t(x)=\{z\}$.
  Furthermore, by Condition \AX{(B2.c)}, $N_t(x)\subset L_s^P$ and thus,
  $z\in L_s^P$. On the other hand, $x\in L_t^P$ and \AX{(B2.b)} imply
  $z\in L_*^P$; a contradiction since $L_s^P$ and $L_*^P$ are
  disjoint. Analogously, in Case (b), Condition \AX{(B2.b)} implies
  $y\in L_t^P$, whereas $y\in L_*^P$ by Condition \AX{(B2.c)}, which again
  yields a contradiction. We therefore conclude that $L_t^P$, $L_s^P$, and
  $L_*^P$ are disjoint.

  Moreover, for the construction of $(T,\sigma)$, we show that $G[L_i^P]$
  is the disjoint union of an arbitrary number of $K_1$'s and $K_2$'s with
  $i\in \{s,t\}$. By definition of $L_t^P$, we have
  $\sigma(L_t^P) \subseteq \{r,s\}$. Since $N_s(x) \subset L_t^P$ and
  $|N_s(x)|\le 1$ for any $x\in L_{t,r}^P$ by \AX{(B2.b)} as well as
  $N_r(y) \subset L_t^P$ and $|N_r(y)|\le 1$ for any $y\in L_{t,s}^P$ by
  \AX{(B3.b)}.  Therefore, any vertex of $L_t^P$ has at most one neighbor
  in $L_t^P$.  Similar arguments and application of Properties \AX{(B2.c)},
  resp., \AX{(B4.b)} show that any vertex of $L_s^P$ has at most one
  neighbor in $L_s^P$.  Thus, $G[L_i^P]$ is the disjoint union of an
  arbitrary number of $K_1$'s and $K_2$'s with $i\in \{s,t\}$.

  We are now in the position to construct a tree $(T,\sigma)$ based on the
  sets $L_t^P$, $L_s^P$, and $L_*^P$ and to show that it explains
  $(G,\sigma)$.  First, for $i\in \{s,t\}$, we construct a caterpillar
  $(T_i,\sigma_i)$, with root $\rho_{T_i}$, on the leaf set $L_i^P$ such
  that $\parent(a)=\parent(b)$ for any $a,b\in L_i^P$ with
  $\sigma(a)\neq \sigma(b)$ if and only if $ab\in E(G)$.  Since $G[L_i^P]$
  is the disjoint union of an arbitrary number of $K_1$'s and $K_2$'s, the
  tree $T_i$ is well-defined. It is, however, not unique as the order of
  inner vertices in $T_i$ is arbitrary.  Then $(T,\sigma)$ is given by
  attaching $\rho_{T_t}$, $\rho_{T_s}$ and $L_*^P$ to the root
  $\rho_T$. Since $L_t^P$, $L_s^P$, and $L_*^P$ are pairwise disjoint, the
  tree $(T,\sigma)$ is well-defined.
 
  We now show that $(T,\sigma)$ is of Type \AX{(II)} by verifying that
  $\sigma(L_t^P)=\{r,s\}$ and $\sigma(L_s^P) = \{r,t\}$.  It is easy to see
  that $\hat z\in L_{s,t}^P$ and $\hat y\in L_{t,s}^P$ and thus,
  $s\in \sigma(L_t^P)$ and $t\in \sigma(L_s^P)$.  Since
  $\hat z\in L_{s}^P$, we can apply Property \AX{(B4.b)} to conclude that
  $\hat x_2\in N_r(\hat z) \subset L_s^P$. Hence, $r\in \sigma(L_s^P)$.
  Applying \AX{(B3.b)}, one similarly shows $\hat x_1\in L_{t}^P$ and thus,
  $r\in \sigma(L_t^P)$.  By construction, $t\notin \sigma(L_t^P)$ and
  $s\notin \sigma(L_s^P)$.  Thus, $\sigma(L_t^P)=\{r,s\}$ and
  $\sigma(L_s^P) = \{r,t\}$ and hence, $(T,\sigma)$ is of Type \AX{(II)}.

  It remains to show that $G(T,\sigma)=(G,\sigma)$.  To this end, we put
  $L_1 \coloneqq L_t^P$ and $L_2 \coloneqq L_s^P$ as well as
  $v_1 \coloneqq \rho_{T_t}$ and $v_2 \coloneqq \rho_{T_s}$.  Therefore,
  $L_1 = L(T(v_1))$ and $\sigma(L_1) = \{r,s\}$ as well as
  $L_2 = L(T(v_2))$ and $\sigma(L_1) = \{r,t\}$.

  In order to show $G(T,\sigma)=(G,\sigma)$, we first consider the
  adjacencies between vertices $L[s]$ and $L[t]$.  By Conditions
  \AX{(B3.a)} and \AX{(B3.b)}, we have $yz\in E(G)$ for any $y\in L[s]$,
  $z\in L[t]$. The same is true for $G(T,\sigma)$ by Lemma
  \ref{lem:TreeTypes}(iii). Thus, the edges between vertices of color $s$
  and $t$ in $G(T,\sigma)$ and $ (G,\sigma)$ coincide.

  Next, we show that the neighborhood w.r.t.\ $r$ of any vertex of color
  $s$ and $t$, respectively, coincide in $(G,\sigma)$ and $G(T,\sigma)$.
  For each $y\in L_1$ with $\sigma(y)=s$, we have
  $\lca(y,x)\prec_T \lca(y,x')$ for any $x\in L(T(v_1))$,
  $x'\notin L(T(v_1))$ with $\sigma(x)=\sigma(x')=r$. Therefore,
  $N_r(y)\subset L_1$ in $G(T,\sigma)$ for all $y\in L_1\cap L[s]$.  By
  Condition \AX{(B3.b)}, we also have $N_r(y)\subset L_1$ in $(G,\sigma)$
  for all $y\in L_1\cap L[s]$.  Clearly, for any $x\in L(T(v_1))$, we have
  $xy\in E(G(T,\sigma))$ if and only if $\parent(x)=\parent(y)$.  Moreover
  we constructed $(T,\sigma)$ such that $\parent(x)=\parent(y)$ if and only
  if $xy\in E(G)$.  Hence, the neighborhoods $N_r(y)$ in $G(T,\sigma)$ and
  $(G,\sigma)$ coincide for all $y\in L_1\cap L[s]$.  By similar arguments
  and application of \AX{(B4.b)} one can show that the neighborhoods
  $N_r(z)$ in $G(T,\sigma)$ and $(G,\sigma)$ coincide for all
  $z\in L_2\cap L[t]$.

  Now suppose that $y\in L[s]$ is not contained in $L_1$, thus
  $y\in L_*^P$, and let $x\in L[r]$.  By construction of $(T,\sigma)$, we
  have $\lca(x,y) = \rho_T$ and thus, $\lca(x,y)\preceq_T\lca(x,y')$ for
  any $y'$ of color $s$ if and only if $x\in L_2\cup L_*^P$, hence
  $N_r(y)=(L_2\cup L_*^P) \cap L[r]$ in $G(T,\sigma)$. By Condition
  \AX{(B3.a)}, we have $N_r(y)=(L_2\cup L_*^P) \cap L[r]$ in $(G,\sigma)$
  as well.  Hence, the neighborhoods $N_r(y)$ coincide in $G(T,\sigma)$ and
  $(G,\sigma)$ for all $y\in L_*^P$.  Similar arguments and application of
  \AX{(B4.a)} shows that the neighborhoods $N_r(z)$ in $G(T,\sigma)$ and
  $(G,\sigma)$ coincide for all $z\in L_*^P$.
	
  So far, we have shown that the neighborhoods $N(y)$ and $N(z)$ of all
  $y\in L[s]$, resp., $z\in L[t]$ are the same in both, $G(T,\sigma)$ and
  $(G,\sigma)$. It remains to show that this is also true for vertices
  $x\in L[r]$. Since $y\in N(x)\cap L[s]$ if and only if
  $x\in N(y)\cap L[r]$ and the $N(y)$ neighborhoods for all $y\in L[s]$
  coincide in both graphs, we can conclude that $N_s(x)$ w.r.t.\
  $G(T,\sigma)$ coincides with $N_s(x)$ w.r.t.\ $(G,\sigma)$. The same is
  true for the $N_t(x)$ neighborhoods. Hence, the $N(x)$ neighborhoods in
  $G(T,\sigma)$ and $(G,\sigma)$, resp., are identical.  In summary, we
  have shown that $G(T,\sigma)=(G,\sigma)$, i.e., $(T,\sigma)$ explains
  $(G,\sigma)$. Hence, $(G,\sigma)$ is a 3-RBMG.
	
  Now let $(G,\sigma)$ be a 3-RBMG of Type \AX{(B)}. By Lemma
  \ref{lem:L^B-tree}, $(G,\sigma)$ is B-like w.r.t.\
  $\langle\hat x_1\hat y\hat z \hat x_2\rangle$ for some
  $\hat x_1, \hat x_2\in L[r]$, $\hat y\in L[s]$, $\hat z \in L[t]$, which
  proves \AX{(B1)}. Moreover, again by Lemma \ref{lem:L^B-tree}, the tree
  $(T,\sigma)$ that explains $(G,\sigma)$ can be chosen in a way that it is
  of Type \AX{(II$^*$)} and satisfies $L_1=L(T(v_1))$, $L_2=L(T(v_2))$ for
  $v_1,v_2\in\child(\rho_T)\setminus L$, and $L_*^P=\child(\rho_T)\cap
  L$. Now, careful application of Lemma \ref{lem:TreeTypes}(i)-(iv), which
  we leave to the reader, shows that Conditions \AX{(B2.a)} to \AX{(B4.b)}
  are satisfied.
  
\end{proof}

We note that some conditions in Lemma \ref{lem:charB} are redundant. For
instance, \AX{(B4.a)} and \AX{(B4.b)} are a consequence of
\AX{(B2.a)}-\AX{(B3.b)}. We find them convenient, however, to describe the
structure of Type \AX{(B)} 3-RBMGs since they emphasize the symmetric
structure of the conditions and somewhat simplify the arguments. We will
give a non-redundant set of conditions in Theorem \ref{thm:char3cBMG} at
the end of this section.

We give here an alternative receipt to reconstruct a 3-RBMG $(G,\sigma)$ of
Type \AX{(B)} that is B-like w.r.t.\ some $P$ as in Def.\ \ref{def:Ltsr},
based on its induced subgraphs
$(G_*,\sigma_*)\coloneqq (G[L_*^P],\sigma_{|L_*^P})$,
$(G_1,\sigma_1)\coloneqq (G[L_t^P],\sigma_{|L_t^P})$, and
$(G_2,\sigma_2)\coloneqq (G[L_s^P],\sigma_{|L_s^P})$.  This particular
reconstruction and the knowledge about the structure of Type \AX{(B)} RBMGs
may be potentially useful for orthology detection, more precisely for the
identification of false positive and false negative orthology assignments
\cite{GGL:19}. By \AX{(B2.a)} and \AX{(B2.b)}, $(G_1,\sigma_1)$ and
$(G_2,\sigma_2)$ are both disjoint unions of an arbitrary number of $K_1$'s
and $K_2$'s.  By Lemma \ref{lem:L^B-tree} there exists a tree of Type
\AX{(II$^*$)} that explains $(G,\sigma)$ and satisfies $L(T(v_1)) = L_t^P$,
$L(T(v_2)) = L_s^P$ for $v_1,v_2\in\child(\rho_T)\setminus L$, and
$L_*^P=\child(\rho_T)\cap L$.  Hence, by Lemma \ref{lem:TreeTypes},
$(G,\sigma)$ can be obtained by inserting edges all edges $ab$ with
\begin{itemize}
\item[(i)] $a\in L_t^P$, $b\in L_s^P$ and
  $\sigma(a),\sigma(b)\in\sigma(L_t^P)\triangle \sigma(L_s^P)$, where
  $\triangle$ denotes the symmetric difference, and
\item[(ii)] $a\in L_i^P$, $b\in L_*^P$ and $\sigma(b)\notin \sigma(L_i^P)$
  for $i\in \{s,t\}$
\end{itemize}
into the disjoint union of $(G_1,\sigma_1)$, $(G_2,\sigma_2)$, and
$(G_*,\sigma_*)$.  

However, the assignment of leaves to one of the sets $L_t^P$, $L_s^P$, or
$L_*^P$ strongly depends on the choice of the corresponding induced $P_4$.
We refer to Fig.\ \ref{fig:P_4-classes} for an example. The 3-RBMG
$(G,\sigma)$ contains the induced $P_4$s $\langle a_1b_1c_1a_2 \rangle$ and
$\langle a_1c_2b_2a_2 \rangle$, where $a_1,a_2\in L[r]$, $b_1,b_2\in L[s]$
and $c_1,c_2\in L[t]$.  If $L_t^P$, $L_s^P$, or $L_*^P$ are defined w.r.t.\
$P = \langle a_1b_1c_1a_2 \rangle$, then one obtains $L_t^P=\{a_1,b_1\}$,
$L_s^P=\{a_2,c_1\}$, and $L_*^P=\{b_2,c_2\}$, from which one constructs the
tree $(T_1,\sigma)$. On the other hand, if
$P = \langle a_1c_2b_2a_2 \rangle$ is chosen as the corresponding $P_4$, it
yields $L_t^P=\{a_1,c_2\}$, $L_s^P=\{a_2,b_2\}$, and $L_*^P=\{b_1,c_1\}$
and the tree $(T_2,\sigma)$.

We will return to the induced $P_4$s with endpoints of the same color in
Section~\ref{sect:P4} below. We shall see that they fall into two distinct
classes, which we call \emph{good} and \emph{bad}.  All good $P_4$s in
$(G,\sigma)$ imply the same vertex sets $L_t^P$, $L_s^P$, and $L_*^P$. In
contrast, different bad $P_4$s results in different vertex sets.

\subsection{Characterization of Type \AX{(C)} 3-RBMGs}

The construction of Type \AX{(B)} 3-RBMGs can be extended to a similar
characterization of Type \AX{(C)} 3-RBMGs.

\begin{definition} \label{def:C-like}
  Let $(G,\sigma)$ be an undirected, connected, properly colored,
  $\sthin$-thin graph.  Moreover, assume that $(G,\sigma)$ contains the
  hexagon
  $H\coloneqq \langle \hat x_1 \hat y_1 \hat z_1 \hat x_2 \hat y_2 \hat z_2
  \rangle$ such that $\sigma(\hat x_1)=\sigma(\hat x_2) =r$,
  $\sigma(\hat y_1)=\sigma(\hat y_2)=s$, and
  $\sigma(\hat z_1)=\sigma(\hat z_2) =t$.  Then, $(G,\sigma)$ is
  \emph{C-like w.r.t.\ $H$} if there is a vertex
  $v\in\{\hat x_1 ,\hat y_1, \hat z_1, \hat x_2 ,\hat y_2, \hat z_2\}$ such
  that $|N_c(v)|>1$ for some color $c\neq \sigma(v)$.  Suppose that
  $(G,\sigma)$ is C-like w.r.t.\
  $H =\langle \hat x_1 \hat y_1 \hat z_1 \hat x_2 \hat y_2 \hat z_2
  \rangle$ and assume w.l.o.g. that $v=\hat x_1$ and $c=t$, i.e.,
  $|N_t(\hat x_1)|>1$. Then we define the following sets:
  \begin{align*}
    L_t^H &\coloneqq \{x \mid \langle x \hat z_2 \hat y_2 \rangle
            \in \mathscr{P}_3\}\cup
            \{y \mid \langle y \hat z_1 \hat x_2 \rangle\in \mathscr{P}_3\} \\
    L_s^H &\coloneqq \{x \mid \langle x \hat y_2 \hat z_2 \rangle
            \in \mathscr{P}_3\} \cup
            \{z \mid \langle z \hat y_1 \hat x_1 \in \rangle
            \mathscr{P}_3\} \\
    L_r^H &\coloneqq \{ y \mid \langle y \hat x_2 \hat z_1 \rangle
            \in \mathscr{P}_3\} \cup
            \{z \mid \langle z \hat x_1 \hat y_1 \rangle\in \mathscr{P}_3\}\\
    L_*^H &\coloneqq V(G)\setminus (L_r^H\cup L_s^H\cup L_t^H).
  \end{align*}		
\end{definition}

Again, there is a close connection between these vertex sets and trees of
Type \AX{(III)}.
\begin{lemma}\label{lem:L^C-tree}
  Let $(G,\sigma)$ be an $\sthin$-thin 3-RBMG of Type \AX{(C)} with
  $|L|>6$ and color set $S=\{r,s,t\}$. Then, up to permutation of colors,
  $(G,\sigma)$ is C-like w.r.t.\ the hexagon
  $H = \langle \hat x_1\hat y_1\hat z_1 \hat x_2 \hat y_2 \hat z_2 \rangle$
  for some $\hat x_i\in L[r]$, $\hat y_i\in L[s]$, $\hat z_i \in L[t]$ and
  there exists a tree $(T,\sigma)$ of Type \AX{(III)} explaining
  $(G,\sigma)$ such that
  \begin{itemize}
  \item[(i)] $L_t^H=L(T(v_1))$, $L_s^H=L(T(v_2))$, and
    $L_r^H=L(T(v_3))$ where $v_1,v_2, v_3\in\child(\rho_T)$, and
  \item[(ii)] $L_*^H=\child(\rho_T)\cap L$.
  \end{itemize}
\end{lemma} 
\begin{proof}
  We argue along the lines of the proof of Lemma \ref{lem:L^B-tree}.  Let
  $(G,\sigma)$ be a 3-RBMG of Type \AX{(C)}. Then, Lemma
  \ref{lem:2col-subtrees} implies that there exists a tree $(T,\sigma)$
  with root $\rho_T$ explaining $(G,\sigma)$ that is of Type \AX{(III)},
  thus in particular there are vertices $v_1,v_2,v_3\in \child(\rho_T)$
  with $\sigma(L(T(v_1)))=\{r,s\}$, $\sigma(L(T(v_2)))=\{r,t\}$, and
  $\sigma(L(T(v_3)))=\{s,t\}$, and
  $\child(\rho_T)\setminus \{v_1,v_2, v_3\}\subset L$.  Similar
  argumentation as in the proof of Thm.\ \ref{thm:3c-types} (Claim 3) shows
  that there are leaves $\hat x_1,\hat y_1 \prec_T v_1$,
  $\hat x_2,\hat z_1\prec_T v_2$, and $\hat y_2, \hat z_2\prec_T v_3$,
  where
  $\hat x_1,\hat x_2\in L[r], \hat y_1, \hat y_2\in L[s], \hat z_1, \hat
  z_2\in L[t]$, such that
  $\langle \hat x_1\hat y_1\hat z_1\hat x_2\hat y_2 \hat z_2 \rangle$ is a
  hexagon.  Since $|L|>6$, there exists an additional leaf $z'$.
  W.l.o.g. we can assume that this additional vertex has color
  $\sigma(z')=t$.  Since $(T,\sigma)$ is of Type \AX{(III)}, there are
  three mutually exclusive cases: $z'\prec_T v_2$ or $z'\prec_T v_3$ or
  $z'$ is incident to the root $\rho_T$.

  Suppose first that $z'\prec_T v_2$. Lemma \ref{lem:TreeTypes}(ii) implies
  $z'\hat y_1\in E(G)$.  Since in addition $\hat z_1\hat y_1\in E(G)$, we
  can conclude $|N_t(\hat y_1)|>1$.  Similarly, if $z'\prec_T v_3$, then
  Lemma \ref{lem:TreeTypes}(ii) implies $z'\hat x_1\in E(G)$ and thus, as
  $\hat z_2\hat x_1\in E(G)$, we have $|N_t(\hat x_1)|>1$.  Finally, if
  $z'$ is incident to the root $\rho_T$, then Lemma \ref{lem:TreeTypes}(iv)
  implies $z'\hat x_1\in E(G)$ and we again obtain $|N_t(\hat x_1)|>1$.  In
  summary, if $|L|>6$ and $(G,\sigma)$ is of Type \AX{(C)}, then there is
  always a hexagon $H$ and a vertex $v$ in $H$ such that $|N_c(v)|>1$ for
  some color $c\neq \sigma(v)$.  Therefore, $(G,\sigma)$ is C-like
  w.r.t. some hexagon in $G$.

  It remains to show Properties (i) and (ii).  Since $(G,\sigma)$ is C-like
  w.r.t. some hexagon $H$ in $G$ and one can always shift the vertex labels
  along
  $H=\langle \hat x_1\hat y_1\hat z_1 \hat x_2 \hat y_2 \hat z_2 \rangle$
  as well as permuting the colors in $(G,\sigma)$, we can w.l.o.g.\
  assume $v=\hat x_1$ and $c=t$, i.e., $|N_t(\hat x_1)|>1$.  Let
  $x\in L[r]$ and assume first $x\in L(T(v_1))$. We show that this implies
  $x\in L_t^H$. Lemma \ref{lem:TreeTypes}(ii) implies $x\hat z_2\in E(G)$
  and $x\hat y_2\notin E(G)$.  Since $\hat y_2\hat z_2 \in E(G)$ by
  definition of $H$, we can conclude
  $\langle x\hat z_2 \hat y_2 \rangle \in \mathscr{P}_3$. Thus,
  $x\in L_t^H$.  Hence, we have
  $L(T(v_1))\cap L[r] \subseteq L_t^H\cap L[r]$. Now let $x\in L_t^H$,
  i.e., $\langle x \hat z_2\hat y_2\rangle$ forms an induced $P_3$.  Since
  $x \hat y_2 \notin E(G)$, Lemma \ref{lem:TreeTypes}(ii) implies that
  $x \notin L(T(v_2))$. In addition, $x \hat y_2 \notin E(G)$ and Lemma
  \ref{lem:TreeTypes}(iv) imply that $x$ cannot be incident to the root
  $\rho_T$.  Moreover, $x\notin L(T(v_3))$ because
  $r\notin \sigma(L(T(v_3)))$ by construction of $(T,\sigma)$.  Hence, $x$
  must be contained in $L(T(v_1))$.  Therefore, we have
  $L(T(v_1))\cap L[r] \supseteq L_t^H\cap L[r]$, which implies
  $L(T(v_1))\cap L[r] = L_t^H\cap L[r]$.
		 
  Analogously, one shows $L(T(v_1))\cap L[s] = L_t^H\cap L[s]$ and
  consequently, $L(T(v_1))= L_t^H$. By symmetry, one then obtains
  $L(T(v_2)) = L_s^H$ and $L(T(v_3)) = L_r^H$, showing property
  (i). Property (ii) is a direct consequence of Property (i) because
  $L_*^H=L\setminus (L_t^H\cup L_s^H \cup L_r^H)=L\setminus (L(T(v_1))\cup
  L(T(v_2))\cup L(T(v_3)))=\child(\rho_T)\cap L$.\\
  
\end{proof}

\begin{lemma}\label{lem:charC}
  Let $(G,\sigma)$ be an undirected, connected, $\sthin$-thin, and properly
  3-colored graph with color set $S=\{r,s,t\}$ and let $x\in L[r]$,
  $y\in L[s]$ and $z\in L[t]$. Then, $(G,\sigma)$ is a 3-RBMG of Type
  \AX{(C)} if and only if $(G,\sigma)$ is either a hexagon or $|L|>6$ and,
  up to permutation of colors, the following conditions are satisfied:
  \begin{description}
  \item[\AX{(C1)}] $(G,\sigma)$ is C-like w.r.t.\ the hexagon
    $H = \langle \hat x_1\hat y_1\hat z_1 \hat x_2 \hat y_2 \hat z_2
    \rangle$ for some $\hat x_i\in L[r]$, $\hat y_i\in L[s]$,
    $\hat z_i \in L[t]$ with $|N_t(\hat x_1)|>1$,
  \item[\AX{(C2.a)}] If $x\in L_*^H$, then
    $N(x)=L_r^H\cup (L_*^H\setminus\{x\})$,
  \item[\AX{(C2.b)}] If $x\in L_t^H$, then
    $N_s(x)\subset L_t^H$ and $|N_s(x)|\le 1$, and
    $N_t(x)= L_*^H[t]\cup L_r^H[t]$,
  \item[\AX{(C2.c)}] If $x\in L_s^H$, then
    $N_t(x)\subset L_s^H$ and $|N_t(x)|\le 1$, and
    $N_s(x)= L_*^H[s]\cup L_r^H[s]$
  \item[\AX{(C3.a)}] If $y\in L_*^H$, then
    $N(y)=L_s^H\cup (L_*^H\setminus\{y\})$,
  \item[\AX{(C3.b)}] If $y\in L_t^H$, then
    $N_r(y)\subset L_t^H$ and $|N_r(y)|\le 1$, and
    $N_t(y)=L_*^H[t]\cup L_s^H[t]$,
  \item[\AX{(C3.c)}] If $y\in L_r^H$, then
    $N_t(y)\subset L_r^H$ and $|N_t(y)|\le 1$, and
    $N_r(y)=L_*^H[r]\cup L_s^H[r]$,
  \item[\AX{(C4.a)}] If $z\in L_*^H$, then
    $N(z)=L_t^H\cup (L_*^H\setminus\{z\})$,
  \item[\AX{(C4.b)}] If $z\in L_s^H$, then
    $N_r(z)\subset L_s^H$ and $|N_r(z)|\le 1$, and
    $N_s(z)=L_*^H[s]\cup L_t^H[s]$,
  \item[\AX{(C4.c)}] If $z\in L_r^H$, then
    $N_s(z)\subset L_r^H$ and $|N_s(z)|\le 1$, and
    $N_r(z)=L_*^H[r]\cup L_t^H[r]$.
  \end{description} 
  In particular, $L_t^H$, $L_s^H$, $L_r^H$, and $L_*^H$ are pairwise
  disjoint and $\hat x_1, \hat y_1\in L_t^H$,
  $\hat x_2, \hat z_1\in L_s^H$, $\hat y_2, \hat z_2\in L_r^H$.
 \end{lemma}
\begin{proof}
  If $|L|\le6$, it follows from Theorem \ref{thm:3c-types} that
  $(G,\sigma)$ is a 3-RBMG of Type \AX{(C)} if and only if $|L|=6$ and
  $(G,\sigma)$ is a hexagon. 
  Hence, we assume that $|L|>6$ and $(G,\sigma)$ satisfies conditions
  \AX{(C1)} - \AX{(C4.c)}. As a consequence of \AX{(C1)}, if $(G,\sigma)$
  is an RBMG, then it must be of Type \AX{(C)}. In order to prove that
  $(G,\sigma)$ is indeed an RBMG, we construct a tree $(T,\sigma)$ based on
  the sets $L_t^H$, $L_s^H$, $L_r^H$, and $L_*^H$ and show that
  $(T,\sigma)$ explains $(G,\sigma)$.

  To this end, we first show that the sets $L_t^H$, $L_s^H$, $L_r^H$, and
  $L_*^H$ are pairwise disjoint.  By definition, $L_*^H$ is disjoint from
  $L_t^H$, $L_s^H$, and $L_r^H$. Now, let $x\in L[r]$ and assume that
  $x\in L_t^H$. Hence,
  $\langle x \hat z_2 \hat y_2 \rangle\in \mathscr{P}_3$ which in
  particular implies $x\hat z_2 \in E(G)$. Therefore,
  $\langle x \hat y_2 \hat z_2 \rangle\notin \mathscr{P}_3$ and thus,
  $x\notin L_s^H$. Repeated analogous argumentation shows that $L_t^H$,
  $L_s^H$, and $L_r^H$ are pairwise disjoint.
		
  Moreover, for the construction of $(T,\sigma)$, we show that $G[L_i^H]$
  is the disjoint union of an arbitrary number of $K_1$'s and $K_2$'s with
  $i\in \{r,s,t\}$.  By definition of $L_t^H$, we have
  $\sigma(L_t^H) \subseteq \{r,s\}$.  Since $N_s(x) \subset L_t^H$ and
  $|N_s(x)|\le 1$ for any $x\in L_{t}^H$ by \AX{(C2.b)} as well as
  $N_r(y) \subset L_t^H$ and $|N_r(y)|\le 1$ for any $y\in L_{t}^H$ by
  \AX{(C3.b)}, any vertex of $L_t^H$ has at most one neighbor in $L_t^H$.
  Thus, $G[L_t^H]$ is the disjoint union of an arbitrary number of $K_1$'s
  and $K_2$'s. Similar arguments and application of Properties \AX{(C2.c)}
  and \AX{(C4.b)}, resp., \AX{(C3.c)} and \AX{(C4.c)}, show that
  $G[L_s^H]$, resp., $G[L_r^H]$, is the disjoint union of an arbitrary
  number of $K_1$'s and $K_2$'s.

  We are now in the position to construct a tree $(T,\sigma)$ based on the
  sets $L_t^H$, $L_s^H$, $L_r^H$, and $L_*^H$.  First, for
  $i\in \{r,s,t\}$, we construct a caterpillar $(T_i,\sigma_i)$, with root
  $\rho_{T_i}$, on the leaf set $L_i^H$ such that $\parent(a)=\parent(b)$
  for any $a,b\in L_i^H$ with $\sigma(a)\neq \sigma(b)$ if and only if
  $ab\in E(G)$.  Since $G[L_i^H]$ is the disjoint union of an arbitrary
  number of $K_1$'s and $K_2$'s, the tree $T_i$ is well-defined. It is,
  however, not unique as the order of inner vertices in $T_i$ is arbitrary.
  Then $(T,\sigma)$ is given by attaching $\rho_{T_t}$, $\rho_{T_s}$,
  $\rho_{T_r}$ and $L_*^H$ to the root $\rho_T$.  Since $L_t^H$, $L_s^H$,
  $L_r^H$, and $L_*^H$ are pairwise disjoint, the tree $(T,\sigma)$ is
  well-defined.  It is easy to verify that $\hat x_1, \hat y_1 \in L_t^H$,
  $\hat x_2, \hat z_1 \in L_s^H$, and $\hat y_2, \hat z_2 \in L_r^H$.
  Therefore, $\sigma(L_t^H)=\{r,s\}$, $\sigma(L_s^H) = \{r,t\}$ and
  $\sigma(L_r^H) = \{s,t\}$.  This implies that $(T,\sigma)$ is of Type
  \AX{(III)}.  To this end, let $v_1\coloneqq \rho_{T_t}$,
  $v_2\coloneqq \rho_{T_s}$ and $v_3\coloneqq \rho_{T_r}$ and thus
  $L(T(v_1)) = L_t^H$, $L(T(v_2)) = L_s^H$ and $L(T(v_3)) = L_r^H$,
  respectively.
	
  It remains to show that $G(T,\sigma)=(G,\sigma)$.  We first consider the
  adjacencies of the vertices with color $r$.  Let $x\in L[r]$. Suppose
  first $x\in \child(\rho_T)$, i.e., $x\in L_*^H$ in $(G,\sigma)$. Clearly,
  any leaf (with color different from $\sigma(x)$) that is incident to the
  root $\rho_T$ is a neighbor of $x$ in $G(T,\sigma)$, i.e.,
  $L_*^H\setminus \{x\}\subseteq N(x)$. Moreover, since there is no leaf of
  color $r$ in $L(T(v_3))$, we have $L_r^H\subseteq N(x)$ in $G(T,\sigma)$
  by Lemma \ref{lem:TreeTypes}(iv). Hence,
  $L_r^H \cup (L_*^H\setminus \{x\})\subseteq N(x)$ in $G(T,\sigma)$.
  Furthermore, since $r$ is contained in $\sigma(L(T(v_1)))$ as well as in
  $\sigma(L(T(v_2)))$, we can apply Lemma \ref{lem:TreeTypes}(iv) to
  conclude that $x$ is not adjacent to any vertex in $L(T(v_1)) = L_t^H$
  and $L(T(v_2)) = L_s^H$ in $G(T,\sigma)$.  Thus,
  $N(x) \subseteq L_r^H \cup (L_*^H\setminus \{x\})$ and therefore,
  $N(x) = L_r^H \cup (L_*^H\setminus \{x\})$ in $G(T,\sigma)$ for all
  $x\in L[r]\cap L_*^H$.  By Property (C2.a), the latter is also satisfied
  in $(G,\sigma)$ for all $x\in L[r]\cap L_*^H$.  Hence, the respective
  neighborhoods of all $x\in L[r]\cap L_*^H$ in $G(T,\sigma)$ and
  $(G,\sigma)$ coincide.

  Now let $x\in L(T(v_1)) = L_t^H$. By construction and Lemma
  \ref{lem:TreeTypes}(i), we have $xy\in E(G(T,\sigma))$, resp.,
  $xy\in E(G)$ for $y\in L[s]$ if and only if $\parent(x)=\parent(y)$.
  Hence, the respective neighborhoods $N_s(x)$ of all
  $x\in L(T(v_1))\cap L[r]=L_t^H\cap L[r]$ in $G(T,\sigma)$ and
  $(G,\sigma)$ coincide.  Now consider the neighborhood $N_t(x)$ in
  $G(T,\sigma)$.  Since $r\in \sigma(L(T(v_2)))=L_s^H$, Lemma
  \ref{lem:TreeTypes}(ii) implies that $x$ is not adjacent to any vertex in
  $L_s^H$. Hence, as $t\notin \sigma(L_t^H)$, if follows
  $N_t(x)\subseteq L_*^H[t]\cup L_r^H[t]$.  Since there is no leaf of color
  $t$ in $L(T(v_1))$, we have $L_*^H[t]\subseteq N_t(x)$ by Lemma
  \ref{lem:TreeTypes}(iv).  Moreover, since $r\notin \sigma(L(T(v_3)))$,
  Lemma \ref{lem:TreeTypes}(ii) implies
  $L_r^H[t]=L(T(v_3))\cap L[t]\subseteq N_t(x)$. Hence,
  $L_*^H[t]\cup L_r^H[t] \subseteq N_t(x)$ and we therefore conclude
  $N_t(x) = L_*^H[t]\cup L_r^H[t]$.  By Property (C2.b), the latter is also
  satisfied in $(G,\sigma)$ for all $x\in L(T(v_1)) = L_t^H$.  Hence, the
  respective neighborhoods $N_t(x)$ of all $x\in L_t^H\cap L[r]$ are
  identical in $G(T,\sigma)$ and $(G,\sigma)$. Since
  $N_t(x)\cup N_s(x) = N(x)$, the neighborhoods $N(x)$ coincide in
  $G(T,\sigma)$ and $(G,\sigma)$ for every $x\in L_t^H\cap L[r]$.  By
  similar arguments, one can show that the same is true for any
  $x\in L_s^H\cap L[r]$.
	
  By symmetry, analogous arguments show that the neighborhoods of leaves
  with color $s$ or $t$ are the same in $(G,\sigma)$ and $G(T,\sigma)$. We
  therefore conclude $(G,\sigma)=G(T,\sigma)$, i.e., $(T,\sigma)$ explains
  $(G,\sigma)$.
 
  Conversely, let $(G,\sigma)$ be a connected, $\sthin$-thin 3-RBMG of Type
  \AX{(C)}.  By Theorem \ref{thm:3c-types}, $(G,\sigma)$ is either a
  hexagon, or $|L|>6$ and it contains a hexagon $H$ of the form
  $(r,s,t,r,s,t)$.  In the latter case, Lemma \ref{lem:L^C-tree} implies
  that $(G,\sigma)$ is always C-like w.r.t.\ some hexagon
  $H = \langle \hat x_1\hat y_1\hat z_1 \hat x_2 \hat y_2 \hat z_2 \rangle$
  with $\hat x_i\in L[r]$, $\hat y_i\in L[s]$, $\hat z_i \in L[t]$. Similar
  arguments as in the proof of Lemma \ref{lem:L^C-tree} show that w.l.o.g.\
  we can assume $|N_t(\hat x_1)|>1$.  Hence, $(G,\sigma)$ satisfies
  Property (C1).  Moreover, Lemma \ref{lem:L^C-tree} implies that there
  exists a tree $(T,\sigma)$ of Type \AX{(III)} that explains $(G,\sigma)$
  and such that $L(T(v_1)) = L_t^H$, $L(T(v_2)) = L_s^H$,
  $L(T(v_3)) = L_r^H$ and $L_*^H=\child(\rho_T)\cap L$.  Now, careful
  application of Lemma \ref{lem:TreeTypes}(i)-(iv), which we leave to the
  reader, shows that Conditions (C2.a) to (C4.c) are satisfied.  
\end{proof}

If $(G,\sigma)$ is a 3-RBMG of Type \AX{(C)}, an analogous construction as
in the case of Type \AX{(B)} 3-RBMGs can be used to obtain $(G,\sigma)$
from the sets $L_t^H$, $L_s^H$, $L_r^H$, and $L_*^H$. Again, this
information may be useful for correcting the orthology graph.  If $|L|=6$
then $(G,\sigma)$ is already a hexagon
$H = \langle x_1 y_1 z_1 x_2 y_2 z_2 \rangle $ such that, up to permutation
of the colors, $\sigma(x_i)=r$, $\sigma(y_i)=s$, and $\sigma(z_i)=t$,
$i\in\{1,2\}$.  This 3-colored graph is explained \NEW{by} the two distinct trees
$T_1\coloneqq ((x_1,y_1),(z_1,x_2),(y_2,z_2))$ and
$T_2\coloneqq ((y_1,z_1),(x_2,y_2),(z_2,x_1))$, given in standard Newick
tree format.  These two trees induce different leaf sets $L(T(v_i))$, where
$v_i\in \child(\rho_T)\cap V^0(T)$ in the corresponding tree.  One can
show, however, that for $|L|>6$, every hexagon defines the same sets
$L_i^H$, $i\in\{t,s,r\}$ and $L_*^H$.  To this end we will need the
following technical result:

\begin{lemma}\label{lem:C6}
  Let $(T,\sigma)$ be a tree of Type \AX{(III)} with root $\rho_T$
  explaining a connected $\sthin$-thin 3-RBMG $(G,\sigma$) and let
  $H\coloneqq \langle \hat x_1 \hat y_1\hat z_1\hat x_2\hat y_2\hat z_2
  \rangle$ be a hexagon in $(G,\sigma)$ such that $\hat x_i\in L[r]$,
  $\hat y_i\in L[s]$ and $\hat z_i\in L[t]$ for distinct colors $r,s,t$ and
  $1\leq i\leq 2$.  Then,
  $\hat x_i, \hat y_i, \hat z_i \notin \child(\rho_T)$, $1\leq i\leq 2$.
\end{lemma}
\begin{proof}
  By definition of $(T,\sigma)$, there exist distinct
  $v_1,v_2,v_3\in \child(\rho_T)$ such that $\sigma(L(T(v_1)))=\{r,s\}$,
  $\sigma(L(T(v_2)))=\{r,t\}$, and $\sigma(L(T(v_3)))=\{s,t\}$, and
  $\child(\rho_T)\setminus \{v_1,v_2,v_3\} \subset L$. Assume, for
  contradiction, that $\hat x_1\in \child(\rho_T)$.  Then, either
  $\hat y_1\in \child(\rho_T)$ or, by Lemma \ref{lem:TreeTypes}(iv),
  $\hat y_1\preceq_T v_3$. In the latter case, Lemma \ref{lem:TreeTypes}(i)
  implies $\hat z_1\preceq_T v_3$ and thus $\hat x_1\hat z_1\in E(G)$ by
  Lemma \ref{lem:TreeTypes}(iv); contradicting that $H$ is a hexagon.
  Hence, $\hat x_1\notin \child(\rho_T)$. Due to symmetry, we can apply
  similar arguments to the remaining vertices
  $\hat x_2, \hat y_i, \hat z_i$, $1\leq i\leq 2$, to show that none of
  them is contained in $\child(\rho_T)$.  
\end{proof}

We are now in the position to prove the uniqueness of $L_i^H$,
$i\in\{r,s,t\}$, and $L_*^H$.
\begin{lemma}\label{lem:uniqueLH}
  Let $(G,\sigma)$ be a connected $\sthin$-thin 3-RBMG of Type \AX{(C)}
  with leaf set $|L|>6$. Moreover, let the sets $L_t^H$, $L_s^H$, $L_r^H$,
  and $L_*^H$ be defined w.r.t.\ an induced hexagon
  $ H \coloneqq \langle \hat x_1\hat y_1\hat z_1\hat x_2\hat y_2\hat z_2
  \rangle$ with $|N_t(\hat x_1)|>1$, where $\hat x_i\in L[r]$,
  $\hat y_i\in L[s]$ and $\hat z_i\in L[t]$ for distinct colors
  $r,s,t$. Then, for any hexagon $ H'$ of the form $(r,s,t,r,s,t)$ we have
  $L_t^{H'}=L_t^H$, $L_s^{H'}=L_s^H$, $L_r^{H'}=L_r^H$, and
  $L_*^{H'}=L_*^H$.
\end{lemma}
\begin{proof}
  Let $(T,\sigma)$ be a leaf-colored tree explaining $(G,\sigma)$, which,
  by Theorem \ref{thm:3c-types} can be chosen to be of Type \AX{(III)},
  i.e., there are distinct $v_1,v_2,v_3\in \child(\rho_T)$ such that
  $\sigma(L(T(v_1)))=\{r,s\}$, $\sigma(L(T(v_2)))=\{r,t\}$, and
  $\sigma(L(T(v_3)))=\{s,t\}$, and
  $\child(\rho_T)\setminus \{v_1, v_2, v_3\} \subset L$. In particular,
  Lemma \ref{lem:L^C-tree} implies that $(T,\sigma)$ can be chosen such
  that $L_t^H=L(T(v_1))$, $L_s^H=L(T(v_2))$, $L_r^H=L(T(v_3))$, and
  $L_*^H=\child(\rho_T)\cap L$. Lemma \ref{lem:C6} implies
  $V(H)\cap \child(\rho_T)=\emptyset$. Since $\hat x_1$ has more than one
  neighbor of color $t$, Lemma \ref{lem:TreeTypes}(i) and $\sthin$-thinness
  of $(G,\sigma)$ imply that $\hat x_1$ cannot be contained in $L(T(v_2))$
  as otherwise $|N_t(\hat x_1)|\le 1$. Hence, $\hat x_1\preceq
  v_1$. Applying Lemma \ref{lem:TreeTypes}(i)+(ii), we then conclude that
  $\hat y_1\preceq_T v_1$, $\hat x_2, \hat z_1\preceq_T v_2$, and
  $\hat y_2, \hat z_2\preceq_T v_3$. In other words,
  $\hat x_1,\hat y_1\in L_t^H$, $\hat x_2,\hat z_1\in L_s^H$, and
  $\hat y_2,\hat z_2\in L_r^H$.
  
  We proceed to show that the leaf sets $L_t^H$, $L_s^H$, $L_r^H$, and
  $L_*^H$ remain unchanged if they are taken w.r.t.\ some other vertex
  $v\in V(H)\setminus \{\hat x_1\}$ with $|N_c(v)|>1$ for some color
  $c\neq \sigma(v)$. Note that, as $(G,\sigma)$ is explained by
  $(T,\sigma)$ and $\hat x_1\in L_t^H$, we can apply Property \AX{(C2.b)}
  to conclude $|N_s(\hat x_1)|\le 1$.  Suppose first $v=\hat y_1$ and
  $|N_c(\hat y_1)|>1$. Since $\hat y_1\in L_t^H$, we can apply Property
  \AX{(C3.b)} and obtain $|N_r(\hat y_1)|\le 1$. Hence we have $c=t$. The
  definition of $L_t$ with $v=\hat y_1$ and $c=t$ implies that
  $L_t^{H'}=\{y\mid \langle y\hat z_1 \hat x_2\rangle\in
  \mathscr{P}_3\}\cup \{x\mid\langle x\hat z_2 \hat y_2\rangle\in
  \mathscr{P}_3\}=L_t^H$. Similarly, one obtains $L_s^{H'}=L_s^H$ and
  $L_r^{H'}=L_r^H$.  Now, let $v=\hat z_1$ and $|N_c(\hat z_1)|>1$. Then,
  as $(T,\sigma)$ explains $(G,\sigma)$ and $\hat z_1\in L_s^H$, Property
  \AX{(C4.b)} implies $|N_r(\hat z_1)|\le 1$. Hence $c=s$.

  Again, the definition of $L_t^H$ with $v=\hat z_1$ and $c=s$ implies that
  $L_t^{H'}=L_t^H$.  Similarly, the definition of $L_s^H$ and $L_r^H$ with
  $v=\hat z_1$ and $c=s$ shows that $L_s^{H'}=L_s^H$, and
  $L_r^{H'}=L_r^H$. Applying similar arguments to $v=\hat y_2$,
  $v=\hat z_2$ and $v=\hat x_2$ under the assumption that $|N_c(v)|>1$ for
  some color $c\neq \sigma(v)$, shows that $v$ and $c$ always induce the
  same leaf sets $L_t^H$, $L_s^H$, $L_r^H$.  The latter implies that also
  the set $L_*^H$ is independent from the particular choice of the vertices
  $v$ in $H$.

  Now let $H'\coloneqq \langle x_1 y_1 z_1 x_2 y_2 z_2\rangle \neq H$ with
  $x_i\in L[r]$, $y_i\in L[s]$, $z_i\in L[t]$. Lemma \ref{lem:C6} implies
  that $x_1$ and $x_2$ are not incident to the root of $(T,\sigma)$, hence
  $x_1,x_2\in L(T(v_1))\cup L(T(v_2))$. Assume, for contradiction, that
  they are contained in the same subtree, say $x_1,x_2\in L(T(v_1))$. Then,
  as $x_2z_1\in E(G)$ and $\sigma(L(T(v_1)))=\{r,s\}$, Lemma
  \ref{lem:TreeTypes}(ii) implies that $z_1$ cannot reside within a subtree
  that contains leaves of color $r$, \NEW{thus $z_1\in
    L(T(v_3))$}. Therefore, we can again apply Lemma
  \ref{lem:TreeTypes}(ii) to conclude that $x_1z_1\in E(G)$; a
  contradiction since $H'$ is a hexagon. Analogously one shows that $x_1$
  and $x_2$ cannot be both located in the subtree $T(v_2)$. Hence, we can
  w.l.o.g.\ assume $x_1\in L(T(v_1))$. Then, by construction of
  $(T,\sigma)$, we have $\lca_T(x_1,z)=\lca_T(\hat x_1,z)$ for any
  $z\in L[t]$, thus $N_t(x_1)=N_t(\hat x_1)$ and in particular
  $|N_t(x_1)|>1$. Applying Lemma \ref{lem:C6} and analogous argumentation
  as for $H$ yields $x_1, y_1\preceq_T v_1$, $x_2, z_1\preceq_T v_2$, and
  $y_2, z_2\preceq_T v_3$.  Thus, In other words, $ x_1, y_1\in L_t^H$,
  $ x_2, z_1\in L_s^H$, and $ y_2, z_2\in L_r^H$.

  Consider first $L_t^H$ and let $x\in L[r]$. By definition, $x\in L_t^H$
  if and only if $\langle x\hat z_2\hat y_2\rangle$ is an induced
  $P_3$. Since $\sigma(L(T(v_1)))=\{r,s\}$ and $\sigma(L(T(v_3)))=\{s,t\}$,
  Lemma \ref{lem:TreeTypes}(ii) implies
  $\langle x z_2 y_2\rangle\in \mathscr{P}_3$, i.e., $x\in
  L_t^{H'}$. Conversely, suppose $x\in L_t^{H'}$, thus
  $\langle x z_2 y_2\rangle\in \mathscr{P}_3$. Since $(T,\sigma)$ explains
  $(G,\sigma)$ and $y_2, z_2\preceq_T v_3$, Lemma
  \ref{lem:TreeTypes}(ii)+(iv) implies $x\in L(T(v_1))=L_t^H$.  Hence,
  $L_t^H\cap L[r]=L_t^{H'}\cap L[r]$. Similar arguments show
  $L_t^H\cap L[s]=L_t^{H'}\cap L[s]$ and thus $L_t^H=L_t^{H'}$.  By
  symmetry, analogous arguments yield $L_s^H=L_s^{H'}$ and
  $L_r^H=L_r^{H'}$. Taken together, this implies
  $L_*^H=L_*^{H'}$. Analogous argumentation as used for $H$ shows that any
  vertex $v\in V(H')$ with $|N_c(v)|>1$, $c\neq\sigma(v)$, induces the
  same leaf sets $L_t^{H'}$, $L_s^{H'}$, $L_r^{H'}$, and $L_*^{H'}$, which
  finally completes the proof.
  
\end{proof}

In contrast to Observation \ref{fact:NOindep-P-choice} for Type \AX{(B)}
3-RBMGs, we obtain the following result.
\begin{corollary}\label{cor:indep-H-choice}
  Let $(G,\sigma)$ be a connected $\sthin$-thin 3-RBMG of Type \AX{(C)}.
  If $(G,\sigma)$ is C-like w.r.t.\ a hexagon $H$, then $(G,\sigma)$ is
  C-like w.r.t.\ every hexagon of the form $(r,s,t,r,s,t)$.  
\end{corollary}

\subsection{Characterization of 3-RBMGs and Algorithmic Results}

For later reference, finally, we summarize the main results of this
section, i.e., Theorem \ref{thm:3c-types} and the characterizations of the
three types in Lemmas \ref{lem:charA}, \ref{lem:charB}, and
\ref{lem:charC}:

\begin{theorem}
  \label{thm:char3cBMG} 
  An undirected, connected, properly 3-colored, $\sthin$-thin graph
  $(G,\sigma)$ is a 3-RBMG if and only if it satisfies either conditions
  \AX{(A1)} and \AX{(A2)}, \AX{(B1)}-\AX{(B3.b)}, or \AX{(C1)}-\AX{(C3.c)}
  and thus, is of Type \AX{(A)}, \AX{(B)}, or \AX{(C)}.
\end{theorem}

\begin{proof}
  By Theorem \ref{thm:3c-types}, any $\sthin$-thin connected 3-RBMG
  $(G,\sigma)$ must be either of Type \AX{(A)}, \AX{(B)}, or
  \AX{(C)}. Lemma \ref{lem:charA} implies that $(G,\sigma)$ is a 3-RBMG of
  Type \AX{(A)} if and only if it satisfies \AX{(A1)} and \AX{(A2)}. By
  Lemma \ref{lem:charB}, $(G,\sigma)$ is a 3-RBMG of Type \AX{(B)} if and
  only if Properties \AX{(B1)}-\AX{(B4.b)} are satisfied. However, as the
  neighborhoods of all vertices of one color can clearly be recovered from
  the neighborhoods of all vertices of different color, Properties
  \AX{(B4.a)} and \AX{(B4.b)} are redundant, i.e., $(G,\sigma)$ is a Type
  \AX{(B)} 3-RBMG if and only if \AX{(B1)} to \AX{(B3.b)} are
  satisfied. One analogously argues that Properties \AX{(C4.a)}-\AX{(C4.c)}
  are redundant.
  
\end{proof}

\begin{algorithm}[tbp]
\caption{3-RBMG Recognition and Construction of Tree} 
\label{alg:3RBMG}
\begin{algorithmic}[1]
\REQUIRE Properly 3-colored connected graph $(G',\sigma')$. 
\STATE $(G,\sigma) \gets (G'/\sthin, \sigmasthin')$ \label{l:thin}
\IF { \texttt{Test\_Type\_A}$(G,\sigma)$  = \texttt{true}}  \label{l:test-a}
   \STATE  $(T,\sigma)\gets$ \texttt{Build-Tree}($G,\sigma$) \label{l:tree-a}
   \STATE \textbf{goto} Line \ref{l:final-tree}
\ELSE
   \STATE Find one hexagon $H$ of the form $(r,s,t,r,s,t)$  \label{l:findH}
   \IF{ $(G, \sigma)$ is C-like w.r.t.\ $H$} \label{l:test-c}
      \STATE compute $L_s^H$, $L_t^H$, $L_r^H$, $L_*^H$ \label{l:LH}
      \STATE $(T,\sigma)\gets$
      \texttt{Build-Tree}($(G,\sigma)$,$L_t^H$,$L_s^H$,$L_r^H$,$L_*^H$)
      \label{l:TH}
      \COMMENT{cf.\ Lemma \ref{lem:L^C-tree}}
      \IF {$(T,\sigma)$ explains $(G,\sigma)$ } \label{l:checkC}
          \STATE \textbf{goto} Line \ref{l:final-tree}
      \ENDIF
   \ENDIF	  
   \ELSIF{ $(G, \sigma)$ satisfies Def.\ \ref{def:Ltsr}(i) for
     some $P=\langle xyzx'\rangle\in \mathscr{P}_4$
     with $\sigma(x)=\sigma(x')$ \label{l:B-start}}
      \STATE compute $L_s^P$, $L_t^P$, $L_*^P$ \label{l:LP}
      \STATE $(T,\sigma)\gets$
      \texttt{Build-Tree}($(G,\sigma)$,$L_t^P$,$L_s^P$,$L_*^P$)
      \label{l:TP}
      \COMMENT{cf.\ Lemma \ref{lem:L^B-tree}}
      \IF {$(T,\sigma)$ explains $(G,\sigma)$  \label{l:B-end}}
         \STATE \textbf{goto} Line \ref{l:final-tree}
      \ENDIF
   \ENDIF
   \STATE \textbf{return} ``$(G,\sigma)$ is not a 3-RBMG'' \label{l:no}
   \STATE construct final tree $(T',\sigma')$ for $(G',\sigma')$ based on
   $(T,\sigma)$\label{l:final-tree} \label{l:goto}
   \STATE \textbf{return} $(T,\sigma)$ and $(T',\sigma')$  \label{l:return-true}
\end{algorithmic}
\end{algorithm}

Let us now consider the question how difficult it is to decide whether a
given graph is a 3-RBMG or not. It easy to see that all conditions in
Theorem \ref{thm:char3cBMG} can be tested in polynomial time.  In case
$(G,\sigma)$ is a 3-RBMG, we are also interested in a tree that can
explain $(G,\sigma)$. Unless $(G,\sigma)$ is of Type \AX{(A)}, we have we
have to construct the leaf sets $L_s^P$, $L_t^P$, $L_*^P$, or $L_s^H$,
$L_t^H$, $L_r^H$, $L_*^H$, respectively.  Instead of checking each of the
conditions for Type \AX{(B)} or Type \AX{(C)} graphs in Theorem
\ref{thm:char3cBMG}, we can directly construct the tree $(T,\sigma)$
directly from the sets $L_i^X$, $i\in \{r,s,t\}$, $X\in \{P,H\}$ (cf.\
Lemma \ref{lem:L^B-tree}, resp., \ref{lem:L^C-tree}) and test whether or
not $(T,\sigma)$ explains $(G,\sigma)$.  The overall structure of this
algorithm is summarized in Algorithm \ref{alg:3RBMG}.  We first show in
Lemma \ref{lem:correct:alg-3rbmg} that Algorithm \ref{alg:3RBMG} indeed
recognizes 3-RBMGs and, in the positive case, returns a tree.  The proof
of Lemma \ref{lem:correct:alg-3rbmg} provides at the same time a
description of the single steps of Algorithm \ref{alg:3RBMG}.  We then
continue to show in Lemma \ref{lem:runtime:alg-3rbmg} that Algorithm
\ref{alg:3RBMG} runs in $ O(|V(G/\sthin)|^2 |E(G/\sthin)| + |E(G)|)$ time
for a given input graph $(G,\sigma)$.

\begin{lemma}
  Algorithm \ref{alg:3RBMG} determines if a given properly 3-colored
  connected graph $(G',\sigma')$ is a 3-RBMG and, in the positive case,
  returns a tree $(T',\sigma')$ that explains $(G',\sigma')$
  \label{lem:correct:alg-3rbmg}
\end{lemma}
\begin{proof} 
  Given a properly 3-colored connected graph $(G',\sigma')$, we first
  compute $(G,\sigma) = (G'/\sthin, \sigmasthin')$. By construction,
  $(G,\sigma)$ remains properly 3-colored and, by Lemma \ref{lem:sthin},
  $(G,\sigma)$ is $\sthin$-thin and connected.

  In Line \ref{l:test-a}, if $(G,\sigma)$ is of Type \AX{(A)}, then we can
  compute the tree $(T,\sigma)$ that explains $(G,\sigma)$ as constructed
  for \NEW{the} ``if-direction'' in the proof of Lemma \ref{lem:charA}, and jump to
  Line \ref{l:goto}.

  If $(G,\sigma)$ is not of Type \AX{(A)}, then we proceed by testing if
  $(G,\sigma)$ is of Type \AX{(C)}. To this end, we search first for one
  hexagon $H$ of the form $(r,s,t,r,s,t)$ in Line \ref{l:findH}.  If such a
  hexagon $H$ exists, we check if $(G,\sigma)$ is C-like w.r.t.\ $H$. By
  Cor.\ \ref{cor:indep-H-choice}, it is indeed sufficient to test
  C-likeness for one hexagon only.  If $(G,\sigma)$ is C-like w.r.t.\ $H$,
  then we compute the sets $L_t^H$, $L_s^H$, $L_r^H$, $L_*^H$ (Line
  \ref{l:LH}).  We proceed in Line \ref{l:TH} to construct a tree
  $(T,\sigma)$ based on the set $L_t^H$, $L_s^H$, $L_r^H$, $L_*^H$
  according to Lemma \ref{lem:L^C-tree}.  Now, to test if $(G,\sigma)$ is
  of Type \AX{(C)}, we can again apply Lemma \ref{lem:L^C-tree} which
  implies that it suffices show that $(T,\sigma)$ explains $(G,\sigma)$.
  If this is the case, we again jump to Line \ref{l:goto} and, if not, we
  proceed to check if $(G,\sigma)$ is of Type \AX{(B)}.

  If $(G,\sigma)$ is neither of Type \AX{(A)} nor \AX{(C)}, then either
  $(G,\sigma)$ is not a 3-RBMG or it must be of Type \AX{(B)}.  Thus, we
  continue in Line \ref{l:B-start}-\ref{l:B-end} to test if $(G,\sigma)$
  can be explained by some tree $(T,\sigma)$.  To this end, Observation
  \ref{fact:NOindep-P-choice} implies that we must check for every
  $P\in \mathscr{P}_4$ (for which the two endpoints have the same color),
  whether $(G,\sigma)$ satisfies Def.\ \ref{def:Ltsr}(i).  If this is not
  the case for any such induced $P_4$, then Lemma \ref{lem:L^B-tree}
  implies that $(G,\sigma)$ is not of Type \AX{(B)}. Together with the
  preceding tests, we can conclude that $(G,\sigma)$ is not a
  3-RBMG. Hence, the algorithm stops in Line \ref{l:no} and returns
  ``$(G,\sigma)$ is not a 3-RBMG''.  Otherwise, if $(G,\sigma)$ satisfies
  Def.\ \ref{def:Ltsr}(i) w.r.t.\ $P$, we construct a tree $(T,\sigma)$
  based on the set $L_s^P$, $L_t^P$, $L_*^P$ according to Lemma
  \ref{lem:L^B-tree}.  Again by Lemma \ref{lem:L^B-tree}, it is now
  sufficient to show that $(T,\sigma)$ explains $(G,\sigma)$ in order to
  test if $(G,\sigma)$ is a 3-RBMG.  Since the preceding tests already have
  established that $(G,\sigma)$ is neither of Type \AX{(A)} nor \AX{(C)},
  we can conclude that $(G,\sigma)$ is of Type \AX{(B)}.  If $(G,\sigma)$
  is a 3-RBMG, then we jump to Line \ref{l:goto}, otherwise we stop again
  in Line \ref{l:no} and the algorithm returns ``$(G,\sigma)$ is not a
  3-RBMG''.

  Finally, after having verified that $(G,\sigma)$ is indeed a 3-RBMG and
  constructed $(T,\sigma)$, the algorithm reaches Line \ref{l:goto}.  Lemma
  \ref{lem:Sthin-tree} implies that $(G',\sigma')$ is a 3-RBMG.  Moreover,
  the construction in the last part of the proof of Lemma
  \ref{lem:Sthin-tree} shows how to obtain a tree $(T',\sigma')$ that
  explains $(G',\sigma')$ from $(T,\sigma)$.  In Line \ref{l:return-true},
  the respective trees $(T',\sigma')$ and $(T,\sigma)$ are returned.  
\end{proof}

\begin{lemma}
  Let $(G',\sigma')$ is an undirected, properly 3-colored, connected graph
  and let $n=|V(G'/\sthin)|$, $m=|E(G'/\sthin)|$ and $m'= |E(G')|$.
  Algorithm \ref{alg:3RBMG} processes $(G',\sigma')$ in $O(mn^2+m')$ time.
  \label{lem:runtime:alg-3rbmg}
\end{lemma}
\begin{proof} 
  In a worst case, Observation \ref{fact:NOindep-P-choice} implies that we
  need to list all induced $P_4$s
  $\langle \hat x_1 \hat y \hat z \hat x_2 \rangle $ with
  $\sigma(\hat x_1) = \sigma(\hat x_2)$.  Since for any edge $yz$ in $G$,
  there exist at most $(n-2)(n-3)$ possible combinations of vertices $x$
  and $x'$ such that $\langle xyzx'\rangle$ forms an induced $P_4$, there
  are at most $O(mn^2)$ such paths. Hence, the global runtime of the
  algorithm cannot be better than $O(mn^2)$.  Thus, we only provide rough
  upper bounds for all other subtask to show that they stay within
  $O(mn^2)$ time.
	
  The computation of the relation $\sthin$, its equivalence classes and
  $(G=(V,E),\sigma) = (G'/\sthin, \sigmasthin')$ in Line \ref{l:thin} can
  be done in a similar fashion as outlined by \citet[Section
  24.4]{Hammack:2011a} in $O(|E(G')|)$ time, cf.\ \cite[Lemma
  24.10]{Hammack:2011a}.

  To test whether $(G,\sigma)$ is of Type \AX{(A)}, we first apply Cor.\
  \ref{cor:hub-vertex} and check for which colors $i\in \{r,s,t\}$ we have
  $|L[i]|=1$, which can be done in $O(n)$ time.  We then apply Lemma
  \ref{lem:charA} and verify if $G\notin \mathscr{P}_3$. Note that the
  latter task can be done in constant time, since we can check if $n=3$
  and, in the positive case, if the three vertices of $G$ are pairwisely
  connected by an edge in constant time $O(1)$. We apply Lemma
  \ref{lem:charA} again, and check for all colors $i\in \{r,s,t\}$ with
  $L[i]=\{x\}$, if $x$ is a hub-vertex and if $|N(y)|<3$ for every
  $y\in V\setminus \{x\}$.  Both of the latter tasks can be done in $O(n)$
  time.  If such a color and vertex exists, then $(G,\sigma)$ is of Type
  \AX{(A)} and we can build the tree $(T,\sigma)$ that explains
  $(G,\sigma)$.  To this end, we apply the construction as in the
  ``if-direction'' of the proof of Lemma \ref{lem:charA}.  We first
  construct the caterpillar $(T_2 , \sigma_{|L_2})$ with leaf set
  $L_2 = \{y\mid y\neq x, |N(y)|=2\}$ and root $\rho_{T_2}$.  It is easy to
  see that $L_2$ can be constructed in $O(n)$ time.  For the tree $T_2$ we
  add vertices such that $\parent(y) = \parent(z)$ for any $y,z \in L_2$
  with $\sigma(y) \neq \sigma(z)$ if and only if $yz \in E(G)$. Clearly,
  this task can be done in $O(m)$ time.  To construct the final tree
  $(T,\sigma)$ we need to check if $|V\setminus L_2|=2$ or
  $|V\setminus L_2|=3$, which can be done trivially in $O(n^2)$ time. All
  remaining steps to construct $(T,\sigma)$ can be done in constant time.
  Hence, to construct $(T,\sigma)$ we need $O(m+n^2)=O(n^2)$ time.  In
  summary, Line \ref{l:test-a} and \ref{l:tree-a} have overall time
  complexity $O(n^2)$.

  We continue by testing if $(G,\sigma)$ is of Type \AX{(C)} in Line
  \ref{l:findH}-\ref{l:checkC}.  In Line \ref{l:findH}, we first check if
  $(G,\sigma)$ is C-like w.r.t.\ some hexagon $H$.  Note, all candidate
  hexagons must be of the form $(r,s,t,r,s,t)$. In order to find such
  hexagons, we first compute the pairwise distances between all vertices in
  $O(n^3)\subseteq O(n^2m)$ time (Floyd-Warshall).  Then, we fix one of the
  colors, say $r$.  Clearly, two vertices in $L[r]$ that have distance
  larger or smaller than $3$, cannot be both located on such a
  hexagon. Thus, for all vertices $x,x'\in L[r]$ with distance $d(x,x')=3$
  we proceed as follows: We check for all edges $yz$ with $y\in L[s]$,
  $z\in L[t]$, if $x\in N_r(y)$, $x'\in N_r(z)$, $x'\notin N_r(y)$,
  $x \notin N_r(z)$.  If this is the case,
  $\langle xyzx' \rangle\in \mathscr{P}_4$ and we store
  $\langle xyzx' \rangle$ in the list $P_{(x,x')}[s,t]$.  Similarly, if for
  the edge $yz$ with $y\in L[s]$, $z\in L[t]$ we have $x\in N_r(z)$,
  $x'\in N_r(y)$, $x'\notin N_r(z)$, $x \notin N_r(y)$, then we put
  $\langle xzyx' \rangle$ in the list $P_{(x,x')}[t,s]$.  For each edge,
  the latter tests can be done in constant time, e.g.\ by using the
  adjacency matrix representation of $(G, \sigma)$.  As soon as we have
  found two vertices $x,x'\in L[r]$ such that each list $P_{(x,x')}[s,t]$
  and $P_{(x,x')}[t,s]$ contains at least one element
  $\langle xyzx' \rangle$ and $\langle xz'y'x' \rangle$ such that $yz'$ and
  $zy'$ do not form an edge, we have found a hexagon
  $H=\langle xyzx'y'z' \rangle$ of the form $(r,s,t,r,s,t)$.  Thus, for a
  given pair $x,x'\in L[r]$, finding a hexagon that contains $x$ and $x'$
  can be done in $O(m)$ time. As the latter may be repeated for all
  $x,x'\in L[r]$, we can conclude that finding a hexagon of the form
  $(r,s,t,r,s,t)$ in Line \ref{l:findH}, can be done
  $O(|L[r]|^2 m) = O(n^2 m)$ time.  Clearly, the test if $(G,\sigma)$ is
  C-like w.r.t.\ $H$ in Line \ref{l:test-c} can be done in constant time.
  Now, the sets $L_s^H$, $L_t^H$, $L_r^H$, $L_*^H$ are computed in Line
  \ref{l:LH}.  To determine these sets, we compute for each edge $uv$ in
  $H$ all vertices $w\in V\setminus (L[\sigma(u)]\cup L[\sigma(v)])$ such
  that $\langle wuv \rangle \in \mathscr{P}_3$. The latter can be done in
  $O(n)$ for each edge in $H$.  Since $H$ has only a constant number of
  edges, all sets $L_s^H$, $L_t^H$, $L_r^H$ can be constructed in $O(n)$
  time. The set $L_*^H$ can then be trivially constructed in $O(n^2)$ time.
  Now, we continue in Line \ref{l:TH} to construct a tree $(T,\sigma)$ as
  in Lemma \ref{lem:L^C-tree}.  Similar arguments as in the Type \AX{(A)}
  case show that $(T,\sigma)$ can be constructed in $O(m)$ time.  Finally,
  we check in Line \ref{l:checkC} if $(T,\sigma)$ explains $(G,\sigma)$.
  To this end, we note that $T$ has $O(n)$ vertices. Moreover it was shown
  in \cite{Baruch:88}, that the last common ancestor of $x$ and $y$ can be
  accessed in constant time, after an $O(|V(T)|) = O(n)$ time preprocessing
  step.  Hence, for each edge $xy\in E(G)$, we check if
  $\lca(x,y) \preceq_T \lca(x,y')$ and $\lca(x,y) \preceq_T \lca(x',y)$ for
  all $x'\in L[\sigma(x)]$ and $y'\in L[\sigma(y)]$ in $O(n^2)$. As this
  has to repeated for all edges of $G$, Line \ref{l:checkC} takes $O(mn^2)$
  time.  In summary, testing if $(G,\sigma)$ is of Type \AX{(C)} in Line
  \ref{l:findH}-\ref{l:checkC} can be done in $O(mn^2)$ time.

  In Line \ref{l:B-start}, we verify if $(G,\sigma)$ satisfies
  Def. \ref{def:Ltsr}(i) w.r.t.\ some $P\in \mathscr{P}_4$. Note, there are
  at most $mn^2$ $P_4$s $\langle abcd\rangle$ with $\sigma(a) = \sigma(d)$
  in $(G,\sigma)$. Listing all such induced $P_4$s can therefore trivially
  be done $O(mn^2)$ time.  For each induced $P_4$ $\langle abcd\rangle$
  with $\sigma(a)=\sigma(d)$ we can verify the condition in
  Def. \ref{def:Ltsr}(i) in \NEW{at most $O(n)$} time.  Thus, Line
  \ref{l:B-start} requires $O(mn^2)$ time.

  In Line \ref{l:LP}, we need to construct the sets $L_s^P$, $L_t^P$, and
  $L_*^P$.  Assume that
  $P = \langle \hat x_1 \hat y \hat z \hat x_2 \rangle $ is of the form
  $(r,s,t,r)$. To construct the set $L_{t,s}^P$, we have $y\in L_{t,s}^P$
  if the edge $y\hat z$ and every $x\in N_r(y)$ form an induced $P_3$.
  For each edge $y\hat z$ the latter can be tested in $O(n)$ time. To
  obtain $L_{t,s}^P$ the latter test must be repeated for all edges
  $y\hat z$ with $y\in L[s]$.  Thus, $L_{t,s}^P$ can be constructed in
  $O(mn)$ time.  The set $L_{t,r}^P$ is the disjoint union of two sets $L'$
  and $L''$, where the first set $L'$ contains all $x\in L[r]$ for which
  $x,y,\hat z$ induce a $P_3$ with $N_r(y)=\{x\}$ and the second set $L''$
  contains all $x\in L[r]$ with $N_s(x)=\emptyset$ whenever
  $L[s]\setminus L_{t,s}^P\neq \emptyset$.  By similar arguments as for
  $L_{t,s}^P$, the set $L'$ can be constructed in $O(mn)$ time.  For the
  set $L''$, observe that $L[s]\setminus L_{t,s}^P\neq \emptyset$ can be
  trivially verified in at most $O(n^2)$ time and $N_s(x)=\emptyset$ can be
  verified in $O(n)$ time for a given $x\in L[r]$.  To obtain $L''$ we must
  repeat the latter for all $x\in L[r]$ and hence, end up with a time
  complexity $O(n^3) \subseteq O(n^2m)$.  In summary, the set $L_{t,s}^P$
  and $L_{t,r}^P$ can be constructed in $O(mn)$ and $O(n^2m)$ time,
  respectively.  Therefore, $L_t^P$ can be constructed in $O(n^2m)$ time. By
  symmetry, the construction of $L_s^P$ can be done in $O(n^2m)$ time as
  well. The set $L_*^P = V\setminus (L_t^P\cup L_s^P)$ can then trivially
  be constructed in $O(n^2)$ time.  Now, we continue in Line \ref{l:TP} to
  construct a tree $(T,\sigma)$ as in Lemma \ref{lem:L^B-tree}.  Similar
  arguments as in the Type \AX{(A)} case show that $(T,\sigma)$ can be
  constructed in $O(m)$ time. Finally, we check in Line \ref{l:B-end} if
  $(T,\sigma)$ explains $(G,\sigma)$. By similar arguments as in the Type
  \AX{(C)} case, the latter task can be done in $O(mn^2)$ time.  In
  summary, Lines \ref{l:B-start}-\ref{l:B-end} require $O(mn^2)$ time.

  Finally we construct, in Line \ref{l:goto} the tree $(T',\sigma')$ for
  $(G',\sigma')$ based on the tree $(T,\sigma)$. Given the equivalence
  classes as computed in the first step (Line \ref{l:thin}), one can
  construct $(T',\sigma')$ as in the last part of the proof of Lemma
  \ref{lem:Sthin-tree}.  Thus, for each of the $n$ leaves $x$ we can check
  in $O(n)$ time in which class it is contained and then expand the leaf
  $x$ by $|[x]|$ vertices.  As there are at most $O(n')$ vertices that we
  may additionally add to $(T,\sigma)$, we can construct $(T',\sigma')$ in
  $O(n+n')=O(n')\subseteq O(m')$ time.  Since the task of computing the
  quotient graph $(G,\sigma)$ already takes $O(m')$ time, we end up with an
  overall runtime of $O(mn^2 + m')$.

\end{proof}

\section{The Good, the Bad, and the Ugly: induced $P_4$s}
\label{sect:P4}

In order to gain a better understanding of Type (B) 3-RBMGs, we consider
here in more detail the influence of the choice of the ``reference'' $P_4$
on the definition of the vertex sets $L_t^P$, $L_s^P$, and $L_*^P$ that
determine the structure of $(G,\sigma)$. The $P_4$s can be classified as
so-called \emph{good}, \emph{bad}, and \emph{ugly} quartets.  Quartets will
play an essential role for the characterization of 3-RBMGs as we shall see
later. In particular, the sets $L_t^P$, $L_s^P$, and $L_*^P$ can be
determined by good quartets and are independent of the choice of the
respective good quartet. As shown by \citet{GGL:19}, good quartets also
play an important role for the detection of false positive and false
negative orthology assignments.

\begin{fact}
  An $n$-RBMG does not contain an induced $P_4$ with two colors. Moreover,
  any induced $P_4$ with three distinct colors is either of the Type
  $\langle xyzx'\rangle$ or $\langle xyx'z\rangle$ with
  $\sigma(x)=\sigma(x')$.
  \label{obs:P4}
\end{fact}
\begin{proof}
  As shown by \citet[Cor.\ 6]{Geiss:18x}, there is no induced $P_4$ with
  only two colors since all 2-RBMGs are complete bipartite graphs. Hence,
  if we have three distinct colors, then exactly two vertices have
  the same color. Since RBMGs are properly colored, these vertices cannot
  be adjacent, leaving only the two alternatives $\langle xyzx'\rangle$ and
  $\langle xyx'z\rangle$.
  
\end{proof}
  
We emphasize that an RBMG on more than three colors may also contain
induced $P_4$s with four distinct colors. Consider, for instance, the tree
$((a_1,b_1,c),(a_2,b_2,d))$, given in Newick format, where $\sigma(a_i)=A$,
$\sigma(b_i)=B$, $\sigma(c)=C$, and $\sigma(d)=D$, $i\in\{1,2\}$ where $A$,
$B$, $C$, and $D$ are pairwise distinct colors. Then the RBMG $G(T,\sigma)$
contains the 4-colored induced $P_4$ $\langle a_1 c d b_2\rangle$. A
characterization for $n$-RBMGs that are cographs will be given later in
Theorem \ref{thm:cographA}.  For now we will restrict our attention to
$P_4$s with three colors only.

\begin{definition}[good, bad and ugly quartets\footnote{Best
    enjoyed with proper soundtrack at
    \url{https://www.youtube.com/watch?v=XjehlT1VjiU}}.]
  Let $(\G,\sigma)$ be a BMG with symmetric part $(G,\sigma)$ and let
  $Q\coloneqq \{x,x',y,z\} \subseteq L$ with $x,x'\in L[r]$, $y\in L[s]$,
  and $z\in L[t]$. The set $Q$, resp., the induced subgraph
  $\G[Q],\sigma_{|Q})$ is
  \begin{itemize}
  \item a \emph{good quartet} if (i) $\langle xyzx'\rangle$ is an induced
    $P_4$ in $(G,\sigma)$ and (ii) $(x,z),(x',y)\in E(\G)$ and
    $(z,x),(y,x')\notin E(\G)$,
  \item a \emph{bad quartet} if (i) $\langle xyzx'\rangle$ is an induced
    $P_4$ in $(G,\sigma)$ and (ii) $(z,x),(y,x')\in E(\G)$ and
    $(x,z),(x',y)\notin E(\G)$,
  \item an \emph{ugly quartet} if $\langle xyx'z\rangle$ is an induced
    $P_4$ in $(G,\sigma)$.
  \end{itemize}
  \label{def:GoodBadUgly}
\end{definition}
Fig.\ \ref{fig:P_4-classes} shows an example of an RBMG containing a good
quartet. Note that good, bad, and ugly quartets cannot appear in RBMGs
whose induced 3-colored subgraphs are all Type \AX{(A)} 3-RBMGs: by
definition, these do not contain induced $P_4$s.

\begin{figure}[t]
  \begin{center}
    \includegraphics[width=1\textwidth]{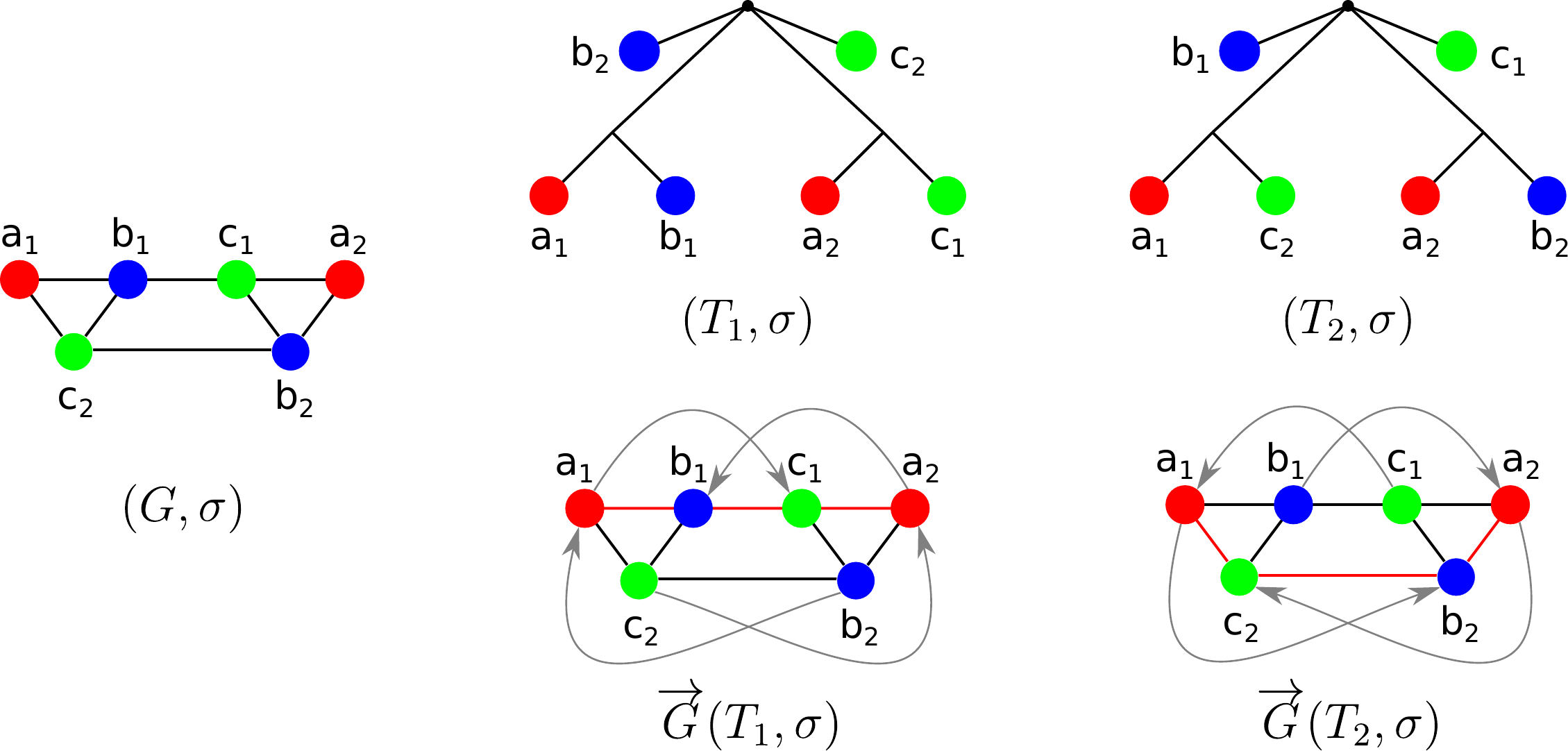}      
  \end{center}
  \caption{The $\sthin$-thin 3-RBMG $(G,\sigma)$ is explained by two trees
    $(T_1,\sigma)$ and $(T_2,\sigma)$ that induce distinct BMGs
    $\G(T_1,\sigma)$ and $\G(T_2,\sigma)$.  In $\G(T_1,\sigma)$,
    $P^1 =\langle a_1b_1c_1a_2\rangle$ defines a good quartet, while
    $P^2 =\langle a_1c_2b_2a_2\rangle$ induces a bad quartet. In
    $\G(T_2,\sigma)$ the situation is reversed. Moreover, the quartets for
    $P^1$ and $P^2$ induce different leaf sets. Denoting the leaf colors
    ``red'', ``blue'', and ``green'' by $r$, $s$, and $t$, respectively, we
    obtain $L_t^{P^1}=\{a_1,b_1\}$, $L_s^{P^1}=\{a_2,c_1\}$, and
    $L_*^{P^1}=\{b_2,c_2\}$, while $L_s^{P^2}=\{a_1,c_2\}$,
    $L_t^{P^2}=\{a_2,b_2\}$, and $L_*^{P^2}=\{b_1.c_1\}$. The good quartets
    in $\G(T_1,\sigma)$ and $\G(T_2,\sigma)$ are indicated by red
    edges. The induced paths $\langle a_1 b_1 c_1 b_2\rangle$ and
    $\langle a_2 c_1 b_1 c_2\rangle$ are examples of ugly quartets.}
  \label{fig:P_4-classes}
\end{figure}

\begin{lemma}\label{lem:P4-classes}
  Let $(G,\sigma)$ be an RBMG, $Q$ a set of four vertices with three
  colors, $G[Q]\in \mathscr{P}_4$, and $(\G,\sigma)$ a BMG containing
  $(G,\sigma)$. Then $Q$ is either a good, a bad, or an ugly quartet.
  \label{lem:type-Q}
\end{lemma}
\begin{proof}
  By Obs.\ \ref{obs:P4}, any induced $P_4$ is either of the form
  $\langle xyx'z\rangle$ or $\langle xyzx'\rangle$. In the first case $Q$
  is an ugly quartet. For the remainder of the proof we thus assume
  $\langle xyzx'\rangle$, and w.l.o.g, we suppose \NEW{that} the vertex colors are
  $\sigma(x)=\sigma(x')=r$, $\sigma(y)=s$, and $\sigma(z)=t$.
  
  Let $(T,\sigma)$ be a leaf-colored tree that explains $(\G,\sigma)$, and
  thus, by assumption, also $(G,\sigma)$. Since $\langle xyzx'\rangle$ is
  an induced $P_4$ in $(G,\sigma)$, the edge $xz$ cannot be contained in
  $E(G)$. Hence, we are left with three cases: (i) $(x,z)\in E(\G)$ and
  $(z,x)\notin E(\G)$, (ii) $(z,x)\in E(\G)$ and $(x,z)\notin E(\G)$, and
  (iii) $(x,z),(z,x)\notin E(\G)$.

  \par\noindent Case (i). We have $x'\in N^+_{r}(z)$ and $x\notin
  N^+_{r}(z)$, i.e., $\lca_T(x',z)\prec_T\lca_T(x,z)=:u$. This implies
  $\lca_T(x,x')=u$. Moreover, $(x,y)\in E(G)$ implies
  $\lca_T(x,y)\preceq_T \lca_T(x',y)$. In case of equality, we have
  $\lca_T(x,y)= \lca_T(x',y)\succeq _T u$. Thus, $x\in N^+_r(y)$ implies
  $x'\in N^+_r(y)$.  Hence, since $x'y\notin E(G)$, there must exist a
  leaf $y'\in L[s]$ such that $\lca_T(x',y')\prec_T \lca_T(x',y)$.
  Then, $\lca_T(x',z) \prec_T u$ implies either
  $\lca_T(z,y')\prec_T u \preceq_T \lca_T(x',y)$ or
  $\lca_T(z,y')=\lca_T(x',y')\prec_T \lca_T(x',y)$. Either alternative
  contradicts $yz\in E(G)$. Therefore, $\lca_T(x,y)\prec_T \lca_T(x',y)$.
  Together with $\lca_T(x,x')=u$ this implies $\lca_T(x,y)\prec_T u$.
  Now let $u_y \succeq_T \lca_T(x,y)$ and $u_z\succeq _T \lca_T(x',z)$,
  where $u_y, u_z \in \child(u)$. Then, $yz\in E(G)$ implies that
  $t\notin \sigma(L(T(u_y)))$ and $s\notin \sigma(L(T(u_z)))$.
  Hence, $(x,z),(x',y)\in E(\G)$ and $(z,x),(y,x')\notin E(\G)$, i.e.,
  $\{x,y,z,x'\}$ forms a good quartet. 

  \par\noindent Case (ii). We have $x,x'\in N^+_r(z)$, thus
  $u\coloneqq \lca_T(x,z)=\lca_T(x',z)$. On the other hand,
  $(x,z)\notin E(\G)$ implies that there exists a leaf $z'\in L[t]$ such
  that $\lca_T(x,z')\prec_T\lca_T(x,z)$. Hence, \NEW{as $z\in N^+_t(x')$,}
  we have distinct $u_x, u_{x'}, u_z \in \child(u)$ such that
  $x,z' \prec_T u_x$, $x' \preceq_T u_{x'}$, and $z\preceq _T u_z$.
  Moreover, $yz\in E(G)$ implies $\lca_T(y,z) \preceq _T \lca_T(y,z')$,
  thus we either have (a) $\lca_T(y,z) \prec _T \lca_T(y,z')$ or (b)
  $\lca_T(y,z) = \lca_T(y,z')$. In both cases, we have
  $u_x\prec_T \lca(x,y)$, thus, since $xy\in E(G)$, it follows
  $s\notin \sigma(L(T(u_x)))$. Similarly, since $zx'\in E(G)$, we have
  $r\notin \sigma(L(T(u_z)))$ and $t\notin \sigma(L(T(u_{x'})))$. \NEW{This
    in particular implies $\lca(x',y)\preceq \lca(x,y)$.}  Moreover,
  \NEW{as} $x'y\notin E(G)$, there must exist a leaf $y'\in L[s]$ with
  $\lca_T(x',y')\prec_T \lca_T(x',y)$. In Case (a), we have
  $y\in L(T(v_z))$ and thus $\lca_T(x',y)=u$, \NEW{which implies}
  $y'\preceq_T u_{x'}$.  In summary, this implies for Case (a)
  $x,z'\prec_T u_x$, $x',y'\prec_T u_{x'}$, and $y,z\prec_T u_z$ as well as
  $\sigma(L(T(u_x)))=\{r,t\}$, $\sigma(L(T(u_{x'})))=\{r,s\}$, and
  $\sigma(L(T(u_z)))=\{s,t\}$. Hence, $(z,x),(y,x')\in E(\G)$ and
  $(x,z),(x',y)\notin E(\G)$, i.e., $\{x,y,z,x'\}$ is a bad quartet in
  $(\G,\sigma)$. In Case (b), if $\lca_T(x',y')\succeq u$, then
  $\lca_T(x,y')=\lca_T(x',y')\prec_T \lca_T(x',y)$, contradicting
  $xy\in E(G)$. Hence, $y'\preceq_T u_{x'}$. This implies $\lca_T(x,y)=u$
  since otherwise $\lca_T(x,y')=u\prec_T \lca_T(x,y)$; again a
  contradiction to $xy\in E(G)$. Let $u_y\in \child(u)$ be such that
  $y\preceq_T u_y$. Since $xy, yz \in E(G)$, we conclude
  $\sigma(L(T(u_y)))=\{s\}$. Moreover, $yz\in E(G)$ then implies
  $s\notin \sigma(L(T(u_z)))$. Summarizing Case (b), we thus have
  $x,z'\prec_T u_x$, $x',y'\prec_T u_{x'}$, $y\preceq_T u_y$, and
  $z\prec_T u_z$ as well as $\sigma(L(T(u_x)))=\{r,t\}$,
  $\sigma(L(T(u_{x'})))=\{r,s\}$, $\sigma(L(T(u_y)))=\{s\}$, and
  $\sigma(L(T(u_z)))=\{t\}$. One now easily checks that $\{x,y,z,x'\}$
  again forms a bad quartet in $(\G,\sigma)$.

  \par\noindent Case (iii). Let $u\coloneqq \lca_T(x,x')$. Then
  $(x,z),(z,x)\notin E(\G)$ implies $\lca_T(x',z)\prec_T \lca_T(x,z)$ and
  there must exist some leaf $z'\in L[t]$ such that
  $\lca_T(x,z')\prec_T \lca_T(x,z)$. Hence, there are
  $u_x, u_{x'}\in \child(u)$ with $\lca_T(x,z')\preceq_T u_x$ and
  $\lca_T(x',z)\preceq_T u_{x'}$. By construction, we therefore have either
  $\lca_T(x,y)\prec_T\lca_(x',y)$ or $\lca_T(x,y)=\lca_(x',y)\succeq_T u$.
  The first case implies $\lca_T(y,z')\prec_T u=\lca_T(y,z)$, which
  contradicts $yz\in E(G)$. Hence, it must hold
  $\lca_T(x,y)=\lca_(x',y)\succeq_T u$ and thus, $x'\in N_r^+(y)$ because
  $x\in N_r^+(y)$. Consequently, since $x'y\notin E(G)$, there must be some
  $y'\in L[s]$ such that $\lca_T(x',y')\prec_T \lca_T(x',y)$. The same
  argumentation as in Case (i) shows that $\lca_T(x',z) \prec_T u$ implies
  either $\lca_T(z,y')\prec_T u \preceq_T \lca_T(x',y)$ or
  $\lca_T(z,y')=\lca_T(x',y')\prec_T \lca_T(x',y)$, which in either case
  contradicts $yz\in E(G)$. We therefore conclude that Case (iii) is
  impossible.

\end{proof}

We immediately find the following result that links good quartets to
3-RBMGs of Type \AX{(B)}:
\begin{lemma}
  Let $(G,\sigma)$ be an undirected, connected, $\sthin$-thin and properly
  3-colored graph that contains an induced path $P$ on four vertices.  If
  \NEW{$(G,\sigma)$} satisfies \AX{(B1)} to \AX{(B4.b)} \NEW{w.r.t.\ $P$},
  then there exists a tree $(T,\sigma)$ explaining $(G,\sigma)$ such that
  $P$ is a good quartet in $\G(T,\sigma)$.
\end{lemma}
\begin{proof}
	Suppose that $(G,\sigma)$ satisfies \AX{(B1)} to \AX{(B4.b)} w.r.t.\ $P$.
  Then, by Lemma
  \ref{lem:charB}, $(G,\sigma)$ is a 3-RBMG of Type \AX{(B)}. Thus,
  according to Lemma \ref{lem:L^B-tree}, there exists a Type \AX{(II$^*$)}
  tree $(T,\sigma)$ with root $\rho_T$ such that $L_t^P=L(T(v_1))$,
  $L_s^P=L(T(v_2))$ for distinct $v_1,v_2\in\child_T(\rho_T)\setminus L$
  and $L_*^P=\child_T(\rho_T)\cap L$, that explains $(G,\sigma)$. In
  particular, by Property \AX{(B1)}, we have
  $P\coloneqq \langle x y z x'\rangle$ with $\sigma(x)=\sigma(x')=r$,
  $\sigma(y)=s$ and $\sigma(z)=t$ for distinct colors $r,s,t$. Now, as
  $x,y\in L_t^P$ and $x',z\in L_s^P$ by Lemma \ref{lem:charB} and, by
  definition, $\sigma(L_t^P)=\{r,s\}$, $\sigma(L_s^P)=\{r,t\}$, one easily
  checks that $P$ is indeed a good quartet in $\G(T,\sigma)$.  
\end{proof}
We continue with some basic result about ugly quartets before analyzing
good and bad quartets in more details.
\begin{lemma}
  Let $\langle xyx'z\rangle$ be an ugly quartet in some connected
  $\sthin$-thin 3-RBMG $(G,\sigma)$ and $(T,\sigma)$ with root $\rho_T$ a
  tree of Type \AX{(II)} or \AX{(III)} that explains $(G,\sigma)$. Then the
  children $v_a\in \child(\rho_T)$ with $a\in \{x,x',y,z\}$ and
  $a\preceq_T v_a$ satisfy exactly one of the following conditions:
  \begin{enumerate}
  \item[(i)] $v_x=x$, $v_{x'}=v_z$, $v_y\neq y$, and $v_x$, $v_y$, $v_z$
    are pairwise distinct,
  \item[(ii)] $v_x=v_{x'}=v_z\neq v_y$ and $v_y\neq y$,
  \item[(iii)] $v_x\neq x$, $v_{x'}=x'$, $v_y=v_z$, and $v_x$, $v_{x'}$,
    $v_y$ are pairwise distinct,
  \item[(iv)] $v_x\neq x$, $v_{x'}=x'$, $v_y\neq y$, $v_z=z$, and $v_x$,
    $v_{x'}$, $v_y$, $v_z$ are pairwise distinct,
  \item[(v)] $v_x\neq x$, $v_{x'}=x'$, $v_y=y$, $v_z\neq z$, and $v_x$,
    $v_{x'}$, $v_y$, $v_z$ are pairwise distinct.
  \end{enumerate}
  In particular, $x$, $x'$, and $y$ can never reside within the same
  subtree $T(v_x)$. \NEW{Indeed, all cases may appear.}
\end{lemma}
\begin{proof}
  Let $L$ be the vertex set of $G$ and $r,s,t$ be three distinct colors in
  $(G,\sigma)$, where $\sigma(x)=\sigma(x')=r$, $\sigma(y)=s$, and
  $\sigma(z)=t$.

  We start by considering the two Cases (a) $v_x=x$, i.e.,
  $x\in \child(\rho_T)$ and (b) $v_x\neq x$.

  We first assume Case (a), i.e., $x\in \child(\rho_T)$. Note first that
  this immediately implies $x'\notin \child(\rho_T)$, i.e., $v_{x'}\neq x'$
  because $(G,\sigma)$ is $\sthin$-thin (cf.\ Lemma
  \ref{lem:2col-subtrees}). Moreover, we have
  $s\notin \sigma(L(T(v_{x'})))$ because $x'y\in E(G)$ and thus, since
  $v_{x'}\neq x'$, Lemma \ref{lem:2col} implies that
  $t\in \sigma(L(T(v_{x'})))$. Consequently, $z\in L(T(v_{x'}))$ as
  otherwise, there is some $z'\in L(T(v_{x'}))$ such that
  $\lca(x',z')\prec_T \lca(x',z)$, which contradicts $x'z\in E(G)$.  Since
  $xy\in E(G)$, $x\in \child(\rho_T)$ implies either $y\in \child(\rho_T)$
  or $\sigma(L(T(v_y)))=\{s,t\}$ as otherwise
  $\lca(x'',y)\prec_T \lca(x,y)$ for some $x''\in L(T(v_y))\cap L[r]$,
  contradicting $xy\in E(G)$. If the first case is true, we have, by
  construction, $\lca(y,z)\preceq_T \lca(y,z')$ and
  $\lca(y,z)\preceq_T \lca(y',z)$ for any $y'\in L[s]$ and $z'\in L[t]$,
  i.e., $yz\in E(G)$; a contradiction since $\langle xyx'z\rangle$ is an induced
  $P_4$. Hence, we have $\sigma(L(T(v_y)))=\{s,t\}$, which in particular
  implies $v_y\neq y$ and $v_y\neq v_{x'}, v_x$. Using Lemma
  \ref{lem:TreeTypes}, one easily checks that, if $\parent(x')=\parent(z)$,
  $\langle xyx'z\rangle$ is indeed an induced $P_4$ in $(G,\sigma)$, which
  finally shows Statement (i).

  Now assume that Case (b) is true, i.e., $v_x\neq x$. Then, by Lemma
  \ref{lem:2col}, the subtree $(T(v_x),\sigma_{|L(T(v_x))})$ must contain
  at least two colors. Assume first, for contradiction, that
  $s\in \sigma(L(T(v_x)))$. Then, as $xy,x'y\in E(G)$, Lemma
  \ref{lem:TreeTypes} implies that $x'\in L(T(v_x))$ and
  $\parent(x)=\parent(x')=\parent(y)$, which contradicts the
  $\sthin$-thinness of $(G,\sigma)$. Hence, $\sigma(L(T(v_x)))=\{r,t\}$,
  and in particular $v_x\neq v_y$.
  \\
  If $v_x=v_{x'}$, it follows from $x'z\in E(G)$ that $z\preceq_T v_x$
  (cf.\ Lemma \ref{lem:TreeTypes}), i.e., $v_x=v_z$. Moreover, since
  $s\notin \sigma(L(T(v_x)))$ and $yz\notin E(G)$, Lemma
  \ref{lem:TreeTypes}(iv) implies $v_y\neq y$. Hence, by Lemma
  \ref{lem:2col}, it must hold $\sigma(L(T(v_y)))=\{s,t\}$. Choosing
  $\parent(x')=\parent(z)\neq \parent(x)$, one can again use Lemma
  \ref{lem:TreeTypes} in order to show that $\langle xyx'z\rangle$ forms an
  induced $P_4$, which proves (ii).
  \\
  On the other hand, if $v_x\neq v_{x'}$, then $xy,x'y\in E(G)$ requires
  $\lca(x,y)=\lca(x',y)$, hence $v_y\neq v_x,v_{x'}$. Again, $x'y\in E(G)$
  then implies $s\notin \sigma(L(T(v_{x'})))$. Since
  $\sigma(L(T(v_{x'})))= \sigma(L(T(v_x)))=\{r,t\}$ is not possible by
  construction of \NEW{a tree of Type \AX{(II)} or \AX{(III)}}
  contains only color
  $r$, hence $v_{x'}=x'$ by Lemma \ref{lem:2col}. It finally remains to
  distinguish the two cases $v_y=v_z$ and $v_y\neq v_z$. In the first case,
  if $\parent(y)\neq \parent(z)$ in $(T,\sigma)$, we can apply Lemma
  \ref{lem:TreeTypes} to show $yz\notin E(G)$ and furthermore, that
  $\langle xyx'z\rangle$ is again an induced $P_4$. This yields Statement
  (iii).
  \\
  If the latter case is true, i.e., $v_y\neq v_z$, then, as
  $xy,x'y,x'z\in E(G)$, we have in particular
  $r\notin \sigma(L(T(v_y))),\sigma(L(T(v_z)))$. Furthermore, since
  $yz\notin E(G)$, there is either some $z'\in L(T(v_y))\cap L[t] $ such
  that $\lca(y,z')\prec_T\lca(y,z)$, or some $y'\in L(T(v_z))\cap L[s] $
  such that $\lca(y',z)\prec_T\lca(y,z)$. Note that, since $v_y\neq v_z$,
  $\sigma(L(T(v_y)))=\sigma(L(T(v_z)))=\{s,t\}$ is not possible by
  construction of $(T,\sigma)$. Hence, by applying Lemma \ref{lem:2col} to
  these two cases, we either obtain (iv) or (v). Again, Lemma
  \ref{lem:TreeTypes} easily shows that $\langle xyx'z\rangle$ is an
  induced $P_4$ in $(G,\sigma)$, which concludes the proof.  
\end{proof}

We will from now on focus on good and bad quartets only as they are of
particular interest for the characterization of Type \AX{(B)} 3-RBMGs. We
start with some basic result.

\begin{lemma}\label{lem:P_4-trees}
  Let $(G,\sigma)$ be an RBMG and $\NEW{P\coloneqq}\langle xyzx'\rangle$ an induced $P_4$
  in $(G,\sigma)$ with $\sigma(x)=\sigma(x')$, let $(T,\sigma)$ be a tree
  explaining $(G,\sigma)$ and let $v\coloneqq \lca_T(x,x',y,z)$. Then the
  distinct children $v_i \in \child(v)$ satisfy exactly one of the three
  alternatives
  \begin{enumerate}
  \item[(i)]   $x,y\preceq_T v_1$ and $x',z\preceq_T v_2$,
  \item[(ii)]  $x\preceq_T v_1$, $y,z\preceq_T v_2$, and $x' \preceq_T v_3$,
  \item[(iii)] $x\preceq_T v_1$, $y\preceq_T v_2$, $x'\preceq_T v_3$, and
    $z\preceq_T v_4$.
  \end{enumerate}
Indeed, all three cases may appear.
\end{lemma}
\begin{proof}
  Let $\sigma(x)=\sigma(x')=r$, $\sigma(y)=s$, and $\sigma(z)=t$.\\
  Suppose first $x,y \in L(T(v_1))$ for some $v_1\in \child(v)$. If
  $z\preceq_T v_1$, then $x'\in L(T(v_1))$ since otherwise
  $\lca_T(x,z)\prec_T\lca_T(x',z)$, contradicting $zx'\in E(G)$. Thus,
  $\lca_T(x,x',y,z)\prec_T v$; a contradiction to the definition of
  $v$. Hence, $z\notin L(T(v_1))$. Then, $yz\in E(G)$ implies that
  $t\notin \sigma(L(T(v_1)))$, hence in particular
  $v=\lca_T(x,z)\preceq _T \lca_T(x,z')$ for any $z'\in L[t]$, i.e.,
  $z\in N^+_t(x)$. Since $xz\notin E(G)$, it must therefore hold
  $\lca_T(x',z)\prec_T \lca_T(x,z)=v$, thus $x',z\in L(T(v_2))$ for some
  $v_2\in \child(v)\setminus \{v_1\}$. This implies Case (i).

  Now suppose $x$ and $y$ are located in different subtrees below $v$,
  i.e., $x\preceq_T v_1$ and $y\preceq_T v_2$ for distinct children
  $v_1,v_2\in\child(v)$. Since $xy\in E(G)$ and $\lca_T(x,y)=v$, we
  conclude that $r\notin \sigma(L(T(v_2)))$ and
  $s \notin \sigma(L(T(v_1)))$. This immediately implies
  $x'\notin L(T(v_1))$ as otherwise (a) $\lca_T(x,y)=\lca_T(x',y)$ results
  in $x'\in N^+_r(y)$, and (b) $\lca_T(x',y)\preceq_T \lca_T(x',y')$ for
  any $y'\in L[s]$, hence $x'y\in E(G)$; a contradiction. Therefore, there
  must be a child $v_3\neq v_1,v_2$ of $v$ such that $x'\preceq_T v_3$. As
  a consequence, $z$ cannot be contained in $L(T(v_1))$ since this would
  imply $\lca_T(x,z)\prec_T \lca_T(x',z)$, contradicting $x'z\in
  E(G)$. Suppose $z\in L(T(v_3))$. Then, since $yz\in E(G)$, we have
  $t\notin \sigma(L(T(v_2)))$ and $s\notin \sigma(L(T(v_3)))$. As we
  already know $r\notin \sigma(L(T(v_2)))$, we conclude
  $\sigma(L(T(v_2)))=\{s\}$ and $\sigma(L(T(v_3)))=\{r,t\}$. Clearly, this
  implies $yx'\in E(G)$; a contradiction. Therefore, $z\notin L(T(v_3))$,
  thus we either have $z\preceq_T v_2$ or there exists another child $v_4$
  of $v$ ($v_4\neq v_1,v_2,v_3$) such that $z\preceq_T v_4$.  The latter
  shows that one of the Cases (ii) and (iii) may occur.  However, we need
  to ensure that both can happen given the existence of the induced $P_4$
  $\langle xyzx'\rangle$.  Let us first assume $z\in L(T(v_2))$, thus
  $\sigma(L(T(v_2)))=\{s,t\}$. If $\sigma(L(T(v_1)))=\{r,t\}$ and
  $\sigma(L(T(v_3)))=\{r,s\}$, then one easily checks that
  $\langle xyzx'\rangle$ is an induced $P_4$ in $(G,\sigma)$, which implies
  statement (ii). On the other hand, if $z\preceq_T v_4$ and
  $\sigma(L(T(v_1)))=\{r,t\}$, $\sigma(L(T(v_2)))=\{s\}$,
  $\sigma(L(T(v_3)))=\{r,s\}$ $\sigma(L(T(v_4)))=\{t\}$, then
  $\langle xyzx'\rangle$ also forms an induced $P_4$ in $(G,\sigma)$, i.e.,
  Case (iii) is true.  
\end{proof}

It turns out that the location of good quartets in any tree is strictly
constrained:
\begin{lemma}\label{lem:iP_4-tree}
  Let $(\G,\sigma)$ be a BMG containing a good quartet
  $\langle xyzx' \rangle$, $(T,\sigma)$ a tree explaining $(\G,\sigma)$ and
  $v\coloneqq \lca(x,x',y,z)$.  Then, $x,y \preceq_T v_1$ and
  $x',z\preceq _T v_2$ for some distinct $v_1, v_2\in \child(v)$.
\end{lemma}
\begin{proof}
  Let $v\coloneqq \lca_T(x,x',y,z)$ and $v_1\in \child(v)$ such that
  $x\preceq_T v_1$. Suppose first $y\notin L(T(v_1))$, hence in particular
  $\lca_T(y,x')\preceq_T \NEW{v} = \lca_T(y,x)$. Since
  $xy\in E(G(T,\sigma))$, this implies $x'\in N^+_{\sigma(x)}(y)$ in
  $(\G,\sigma)$; a contradiction to $\langle xyzx' \rangle$ forming a good
  quartet. Hence, $y\preceq_T v_1$.  As $(T,\sigma)$ must satisfy one of
  the three cases of Lemma \ref{lem:P_4-trees} and the only possible case
  is (i), we can now conclude $x,y\preceq_T v_1$ and $x',z\preceq _T v_2$.
  
\end{proof}
\NEW{Note that the latter result in addition shows that the Cases (ii) and
  (iii) in Lemma~\ref{lem:P_4-trees} must correspond to bad quartets.
  Lemma~\ref{lem:iP_4-tree} now can be used to show that the any good
  quartet in an BMG is endowed with the same coloring.}

\begin{corollary}\label{cor:iP_4-color}
  Let $(G,\sigma)$ be a connected $\sthin$-thin 3-RBMG of Type \AX{(B)},
  $(\G,\sigma)$ a BMG containing $(G,\sigma)$ as symmetric part, and let
  $Q=\langle xyzx'\rangle$ with $\sigma(x)=\sigma(x')$ be a good quartet in
  $(\G,\sigma)$. Then every good quartet with
  $\langle x_1y_1z_1x_1'\rangle\in \mathscr{P}_4$ has colors
  $\sigma(x_1)=\sigma(x_1')=\sigma(x)$, $\sigma(y_1)=\sigma(y)$, and
  $\sigma(z_1)=\sigma(z)$.
\end{corollary}
\begin{proof}
  Since $(G,\sigma)$ a of Type \AX{(B)}, any leaf-colored tree $(T,\sigma)$
  with root $\rho_T$ explaining $(G,\sigma)$ is of Type \AX{(II)}
  (cf.\ Theorem \ref{thm:3c-types}). Hence, there are distinct
  $v_1, v_2\in \child(\rho_T)$ with
  $|\sigma(L(T(v_1)))|=|\sigma(L(T(v_2)))|=2$ and
  $\child(\rho_T)\setminus \{v_1,v_2\}\subset L$.
  By Lemma \ref{lem:iP_4-tree}, we have $x,y\preceq _T v_1$ and
  $x',z\preceq_T v_2$, hence $\sigma(L(T(v_1)))=\{\sigma(x),\sigma(y)\}$
  and $\sigma(L(T(v_2)))=\{\sigma(x),\sigma(z)\}$. Therefore,
  the statement follows directly from Lemma \ref{lem:iP_4-tree}.
  
\end{proof}

We are now in the position to formulate one of the main results of this
section:
\begin{lemma}
  Let $(G,\sigma)$ be a connected $\sthin$-thin 3-RBMG of Type \AX{(B)} and
  $(\G,\sigma)$ a BMG containing $(G,\sigma)$ as its symmetric
  part. Moreover, let $L_s^Q$, $L_t^Q$, and $L_*^Q$ be defined w.r.t.\ a
  good quartet $Q\coloneqq\langle x_1y_1z_1x_1'\rangle$ where
  $x_1,x_1'\in L[r]$, $y_1\in L[s]$, and $z_1\in L[t]$ for distinct colors
  $r,s,t$. Then, for any good quartet $Q'$, it holds $L_s^{Q'}=L_s^Q$,
  $L_t^{Q'}=L_t^Q$, and $L_*^{Q'}=L_*^Q$.
\end{lemma}
\begin{proof}
  Let $(T,\sigma)$ with root $\rho_T$ be a leaf-colored tree that explains
  $(\G,\sigma)$ and thus, also $(G,\sigma)$.  Since $(G,\sigma)$ is of Type
  \AX{(B)}, we can choose $(T,\sigma)$ to be of Type \AX{(II)} by Theorem
  \ref{thm:3c-types}, i.e., there are distinct children
  $v_1,v_2\in \child(\rho_T)$ with
  $|\sigma(L(T(v_1)))|=|\sigma(L(T(v_2)))|=2$ such that these two subtrees
  have exactly one color in common, and
  $\child(\rho_T)\setminus \{v_1,v_2\}\subset L$. Applying Lemma
  \ref{lem:L^B-tree}, we can choose $(T,\sigma)$ such that it is of Type
  \AX{(II$^*$)} and satisfies $L_t^Q=L(T(v_1))$, $L_s^Q=L(T(v_2))$ and
  $L_*^Q=\child(\rho_T)\cap L$.  On the other hand, Lemma
  \ref{lem:iP_4-tree} implies $x_1,y_1 \preceq_T v_1$ and
  $x_1',z_1\preceq_T v_2$.

  Now let $Q'\coloneqq \langle x_2y_2z_2x_2'\rangle\ne Q$ be another good
  quartet in $(\G,\sigma)$. By Cor.\ \ref{cor:iP_4-color}, we have
  $x_2,x_2'\in L[r]$, $y_2\in L[s]$ and $z_2\in L[t]$. Lemma
  \ref{lem:iP_4-tree} \NEW{and the structure of Type \AX{(II)} trees} then
  imply $x_2,y_2\preceq_T v_1$ and $x_2',z_2\preceq_T v_2$.  Consider first
  $L_t^{Q'}$ and let $x\in L[r],y\in L[s]$. Then, by definition,
  $y\in L_t^{Q'}$ if and only if $\langle x'yz_2 \rangle\in \mathscr{P}_3$
  for each $x'\in N_r(y)$. Since any two leaves of color $s$ and $t$ are
  reciprocal best matches in $(G,\sigma)$ by Lemma \ref{lem:TreeTypes}(iii)
  and $x'z_2 \notin E(G)$ for any $x'\in N_r(y)$ only if
  $x'\notin L(T(v_2))$ by cf.\ Lemma \ref{lem:TreeTypes}(i)+(iv), we have
  $\langle x'yz_2 \rangle \in \mathscr{P}_3$ for any $x'\in N_r(y)$ if and
  only if $\langle x'yz_1 \rangle\in \mathscr{P}_3$ for any $x'\in N_r(y)$.
  Hence, $y\in L_t^{Q'}$ if and only if $y\in L_t^Q$, i.e.,
  $L_t^Q\cap L[s]=L_t^{Q'}\cap L[s]$. If $N_s(x)=\emptyset$, the latter by
  definition of $L_t^Q$ immediately implies that $x\in L_t^{Q'}$ if and
  only if $x\in L_t^Q$. On the other hand, if $N_s(x)\neq\emptyset$, then
  $x\in L_t^{Q'}$ if and only if $N_r(y')=\{x\}$ for some induced $P_3$
  $\langle xy'z_2 \rangle$. This can only be true if $y'\in L(T(v_1))$
  since otherwise, $y'z_2 x_2'$ forms a circle, thus
  $|N_r(y')|>1$. Consequently, $x\in L(T(v_1))$ by Lemma
  \ref{lem:TreeTypes}(i). Since $y'z'\in E(G)$ for any
  $y'\in L[s], z'\in L[t]$ by Lemma \ref{lem:TreeTypes}(ii) and
  $x'z\notin E(G)$ for any $z\in L(T(v_2))\cap L[t]$ by Lemma
  \ref{lem:TreeTypes}(ii), we conclude that $\langle xy'z_2\rangle$ is an
  induced $P_3$ with $N_r(y')=\{x\}$ if and only if $\langle xy'z_1\rangle$
  is an induced $P_3$ with $N_r(y')=\{x\}$, hence
  $L_t^Q\cap L[r]=L_t^{Q'}\cap L[r]$. We therefore conclude
  $L_t^Q=L_t^{Q'}$.
  
  By symmetry, an analogous argument shows $L_s^Q=L_s^{Q'}$. Together, this
  finally implies $L_*^Q=L_*^{Q'}$, completing the proof.  
\end{proof}
Fig.\ \ref{fig:P_4-classes} shows that good and bad quartets do not
necessarily imply the same leaf sets $L_s^P$, $L_t^P$.

We simplify the notation using the following abbreviations:
\begin{definition}
  $(G_{rst},\sigma_{rst})\coloneqq (G[L[r]\cup L[s]\cup L[t]],
  \sigma_{|L[r]\cup L[s]\cup L[t]})$ and
  $(T_{rst},\sigma_{rst})\coloneqq (T_{|L[r]\cup L[s]\cup L[t]},
  \sigma_{|L[r]\cup L[s]\cup L[t]})$ for any three colors $r,s,t\in S$.
\label{def:restrict-3}
\end{definition}
The restriction of a BMG $\G(T,\sigma)$ to a subset $S'\subset S$ of colors
is an induced subgraph of $\G(T,\sigma)$ explained by the restriction of
$(T,\sigma)$ to the leaves with colors in $S'$, and thus again a BMG
\cite[Observation 1]{Geiss:18x}.  Since $G(T,\sigma)$ is the symmetric part
of $\G(T,\sigma)$, it inherits this property. In particular, we have
\begin{fact}\label{obs:induced_sub}
  If $(G,\sigma)$ is an $n$-RBMG, $n\ge 3$ explained by $(T,\sigma)$, then
  for any three colors $r,s,t \in S$, the restricted tree
  $(T_{rst},\sigma_{rst})$ explains $(G_{rst},\sigma_{rst})$, and
  $(G_{rst},\sigma_{rst})$ is an induced subgraph of $(G,\sigma)$.
\end{fact} 

This observation will play an important role in the proof of the
following lemma as well as in Section \ref{sect:n}.
\begin{lemma}\label{lem:3colgood}
  Let $(\G,\sigma)$ be a BMG.  Then, the symmetric part $(G,\sigma)$
  contains an induced 3-colored $P_4$ whose endpoints have the same color
  if and only if $(\G,\sigma)$ contains a good quartet.
\end{lemma}
\begin{proof}

  
  \NEW{Note that, if $(\G,\sigma)$ contains a good quartet, then its
    symmetric part $(G,\sigma)$ by definition contains a 3-colored $P_4$
    $\langle abcd \rangle$ with $\sigma(a)=\sigma(d)$.

    Conversely, suppose that $(G,\sigma)$ contains a 3-colored induced
    $P_4$ $\langle abcd \rangle$ whose endpoints have the same color
    $\sigma(a)=\sigma(d)$. Moreover, let $S$ be the color set of
    $(G,\sigma)$ and $(T,\sigma)$ be a tree explaining $(\G,\sigma)$ and
    thus also $(G,\sigma)$.} W.l.o.g.\ we can assume that the three colors
  of $\langle abcd \rangle$ are $r,s,t\in S$ without explicitly stating the
  particular coloring of the vertices $a,b,c,$ and $d$.  By definition,
  $(G_{rst},\sigma_{rst})$ contains the induced path
  $\langle abcd \rangle$.  W.l.o.g. we may assume that
  $(G_{rst},\sigma_{rst})$ is connected; otherwise the proof works
  analogously for the connected component of $(G_{rst},\sigma_{rst})$ that
  contains $\langle abcd \rangle$.

  Let us first assume that $(G_{rst},\sigma_{rst})$ is $\sthin$-thin. In
  this case, we can apply Theorem \ref{thm:3c-types} and conclude that
  $(G_{rst},\sigma_{rst})$ is either of Type \AX{(B)} or \AX{(C)}, i.e.,
  the restricted tree $(T_{rst},\sigma_{rst})$ with root
  $\rho\coloneqq \lca_T(L[r]\cup L[s] \cup L[t])$ that explains
  $(G_{rst},\sigma_{rst})$ must be of Type \AX{(II)} or \AX{(III)}. Hence,
  there are distinct children $v_1,v_2\in \child(\rho)$ with
  $\sigma(L(T_{rst}(v_1)))=\{r,s\}$ and $\sigma(L(T_{rst}(v_2)))=\{r,t\}$
  for distinct colors $r,s,t$ (up to permutation of the colors). Then there
  exist leaves $x,y\in L(T_{rst}(v_1))$ and $x',z\in L(T_{rst}(v_2))$ with
  $x,x'\in L[r]$, $y\in L[s]$, and $z\in L[t]$ such that
  $xy,x'z\in E(G_{rst})$ (cf. Lemma \ref{lem:rbm-pairs}). Moreover, Lemma
  \ref{lem:TreeTypes}(ii) implies that $yz\in E(G_{rst})$ as well as
  $xz,x'y\notin E(G_{rst})$.  Hence, $\langle xyzx'\rangle$ is an induced
  $P_4$ in $(G_{rst},\sigma_{rst})$ and thus, by Obs.\
  \ref{obs:induced_sub}, also in $(G,\sigma)$. Since
  $t\notin\sigma(L(T(v_1)))$, we have
  $\rho=\lca_T(x,z)\preceq_T \lca_T(x,z')$ for any $z'\in L[t]$, i.e.,
  $z\in N_t^+(x)$ in $(T,\sigma)$. In particular, $xz\notin E(G)$ then
  immediately implies $x\notin N_r^+(z)$. One similarly argues that
  $y\in N_s^+(x')$ and $x'\notin N_r^+(y)$ in $(T,\sigma)$, hence
  $(x,z), (x',y)\in E(\G)$ and $(z,x),(y,x')\notin E(\G)$. Therefore,
  $\langle xyzx'\rangle$ is a good quartet \NEW{in $(\G,\sigma)$}.

  Now, suppose that $(G_{rst},\sigma_{rst})$ is not $\sthin$-thin.  In this
  case, we apply the same arguments as above to the quotient graph
  $(G_{rst}/\sthin,\sigma_{rst}/\sthin)$ and conclude that there exists a
  good quartet $\langle [a][b][c][d] \rangle$ induced by the tree
  $(T_{rst},\sigma_{rst})$ that explains
  $(G_{rst}/\sthin,\sigma_{rst}/\sthin)$.  Let
  $x\in [a], y\in [b], z\in [c]$ and $x'\in [d]$.  Lemma \ref{lem:sthin}
  implies that $\langle xyzx'\rangle$ is an induced $P_4$ with
  $\sigma(x) = \sigma(x')$ in $(G_{rst},\sigma_{rst})$ and thus, by Obs.\
  \ref{obs:induced_sub}, also in $(G,\sigma)$.  To conclude that
  $\langle xyzx'\rangle$ is a good quartet it remains to show that
  $(x,z), (x',y)\in E(\G)$ and $(z,x),(y,x')\notin E(\G)$.  In order to see
  this, observe first that
  $([a],[c]), ([d],[b])\in E(\G(T_{rst},\sigma_{rst}))$ and
  $([c],[a]), ([b],[d])\notin E(\G(T_{rst},\sigma_{rst}))$ since
  $\langle [a][b][c][d] \rangle$ is a good quartet induced by
  $(T_{rst},\sigma_{rst})$. To obtain a tree $(\hat T, \hat \sigma)$ that
  explains $(G_{rst}, \sigma_{rst})$, we can proceed as in the proof of the
  ``if-direction'' in Lemma \ref{lem:Sthin-tree} and simply replace all
  edges $\parent([v])[v]$ in $T_{rst}$ by edges $\parent([v])v'$ for all
  $v'\in [v]$ and putting $\hat \sigma(v') = \sigma_{rst}([v])$.  Clearly,
  the latter construction and
  $([a],[c]), ([d],[b])\in E(\G(T_{rst},\sigma_{rst}))$ and
  $([c],[a]), ([b],[d])\notin E(\G(T_{rst},\sigma_{rst}))$ implies that
  $(x,z), (x',y)\in E(\G(\hat T, \hat \sigma))$ and
  $(z,x),(y,x')\notin E(\G(\hat T, \hat \sigma))$.  Therefore,
  $\langle xyzx'\rangle$ is a good quartet induced by
  $(\hat T, \hat \sigma)$ and thus, in $(G_{rst},\sigma_{rst})$.  Now, by
  Obs.\ \ref{obs:induced_sub} implies that $\langle xyzx'\rangle$ is a good
  quartet in $(\G,\sigma)$.  
\end{proof}

Since every bad quartet induces in particular an induced $P_4$ with
endpoints of the same color, Lemma \ref{lem:3colgood} immediately implies:
\begin{corollary}
  If a BMG $(\G,\sigma)$ contains a bad quartet, then it contains a good
  quartet. In particular, any BMG $(\G,\sigma)$ whose symmetric part
  contains a 3-RBMG of Type \AX{(B)} or \AX{(C)} as induced subgraph,
  contains a good quartet.
  \label{cor:bad2good}
\end{corollary}
\begin{proof}
  The first statement is an immediate consequence of Lemma
  \ref{lem:3colgood}. If $(G,\sigma)$ is a 3-RBMG of Type \AX{(B)}, then,
  by definition, it contains an induced $P_4$ of the form $(r,s,t,r)$ for
  distinct colors $r,s,t$. Hence, by Lemma \ref{lem:P4-classes}, this $P_4$
  is either a good or a bad quartet in $(\G,\sigma)$, where $(\G,\sigma)$
  is a BMG whose symmetric part contains $(G,\sigma)$ as induced subgraph.
  Lemma \ref{lem:3colgood} thus implies that $(\G,\sigma)$ contains a good
  quartet. If $(G,\sigma)$ is a 3-RBMG of Type \AX{(C)}, it must, by
  definition, contain an induced $C_6$ of the form $(r,s,t,r,s,t)$. In
  particular, it contains an induced $P_4$ whose endpoints are of the same
  color, thus we can again apply the same argumentation to complete the
  proof.  
\end{proof}

The converse of Cor.\ \ref{cor:bad2good} is, however, not true.  As an
example, consider the tree $T=((x,y),(x'z))$ in Newick format with
$\sigma(x)=\sigma(x')=r$, $\sigma(y)=s$, $\sigma(z)=t$. The graph
$G(T,\sigma)$ is the $P_4$ $\langle xyzx'\rangle$, which is of course
uniquely defined. The BMG $\G(T,\sigma)$ contains the directed edges
$(x,z), (x',y)$ but not $(z,x), (y,x')$, hence $Q=\{x,y,z,x'\}=V(T)$ is a
good quartet.

We close this section with a variation of Lemma \ref{lem:iP_4-tree} for
Type \AX{(C)} RBMGs:
\begin{lemma}
  Let $(G,\sigma)$ be a connected $\sthin$-thin 3-RBMG of Type \AX{(C)}
  that contains an induced hexagon
  $H\coloneqq \langle x_1 y_1 z_1 x_2 y_2 z_2 \rangle$ with $|N_t(x_1)|>1$,
  where $x_i\in L[r], y_i\in L[s], z_i\in L[t]$ for distinct colors
  $r,s,t$. Moreover, let $(T,\sigma)$ explain $(G,\sigma)$ and
  $v\coloneqq \lca_T(x_1,x_2,y_1,y_2,z_1,z_2)$. Then,
  \begin{itemize}
  \item[(i)] $\langle x_1y_1z_1x_2\rangle$, $\langle z_1x_2y_2 z_2\rangle$,
    and $\langle y_2z_2x_1y_1\rangle$ are good quartets in $\G(T,\sigma)$,
  \item[(ii)] $\langle y_1z_1x_2y_2\rangle$,
    $\langle x_2y_2 z_2x_1\rangle$, and $\langle z_2x_1y_1z_1\rangle$ are
    bad quartets in $\G(T,\sigma)$, and
  \item[(iii)] $x_1,y_1\preceq _T v_1$, $x_2,z_1 \preceq_T v_2$, and
    $y_2,z_2\preceq_T v_3$ for some distinct $v_1,v_2,v_3\in\child_T(v)$.
  \end{itemize}
\end{lemma}
\begin{proof}
  We start with proving Properties (i) and (ii).  By definition and Lemma
  \ref{lem:type-Q}, $\langle x_1y_1z_1x_2\rangle$ is either a good or a bad
  quartet in $\G(T,\sigma)$. Assume, for contradiction, that it is a bad
  quartet in $\G(T,\sigma)$, thus, in particular,
  $(z_1,x_1)\in E(\G(T,\sigma))$ and $(x_1,z_1)\notin
  E(\G(T,\sigma))$. Hence, as also $\langle z_2x_1y_1z_1\rangle$ must be
  either a good or a bad quartet in $\G(T,\sigma)$, this immediately
  implies that $\langle z_2x_1y_1z_1\rangle$ is a good quartet in
  $\G(T,\sigma)$. Let $w\coloneqq \lca_T(z_2,x_1,y_1,z_1)$. Lemma
  \ref{lem:iP_4-tree} then implies that there exist distinct
  $w_1,w_2\in\child_T(w)$ such that $z_2, x_1 \preceq_T w_1$ and
  $y_1,z_1\preceq_T w_2$. Clearly, as $x_1y_1\in E(G)$ and
  $\lca_T(x_1,y_1)=w$, we must have $s\notin \sigma(L(T(w_1)))$. Since
  $|N_t(x_1)|>1$, there is a leaf $z\in L[t]\setminus \{z_1,z_2\}$ such
  that $x_1z\in E(G)$. By Lemma \ref{lem:rbm-pairs}, there exists an edge
  $x'z'\in E(G(T,\sigma))$ with $x'\in L[r]\cap L(T(w'))$,
  $z'\in L[t]\cap L(T(w'))$ for any inner vertex $w'\preceq_T w$. One
  easily verifies that this and $x_1z_2\in E(G)$ necessarily implies that
  the leaves $x_1$, $z_2$, and $z$ must all be incident to the same parent
  in $T$. However, we then have $N(z)=N(z_2)$, i.e., $z$ and $z_2$ belong
  to the same $\sthin$-class; a contradiction since $(G,\sigma)$ is
  $\sthin$-thin. We therefore conclude that $\langle x_1y_1z_1x_2\rangle$
  must be a good quartet. Hence, $(x_2,y_1), (x_1,z_1)\in E(\G(T,\sigma))$
  and $(y_1,x_2), (z_1,x_1)\notin E(\G(T,\sigma))$, which, as a consequence
  of Lemma \ref{lem:type-Q}, immediately implies that
  $\langle y_1z_1x_2y_2\rangle$ and $\langle z_2x_1y_1z_1\rangle$ are bad
  quartets in $\G(T,\sigma)$. This similarly implies that
  $\langle z_1x_2y_2 z_2\rangle$ and $\langle y_2z_2x_1y_1\rangle$ are good
  quartets, from which we finally conclude that
  $\langle x_2y_2 z_2x_1\rangle$ is a bad quartet in $\G(T,\sigma)$.

  We continue with showing Property (iii).  Property (i) implies that
  $\langle x_1y_1z_1x_2\rangle$ and $\langle z_1x_2y_2 z_2\rangle$ are good
  quartets in $\G(T,\sigma)$. Hence, by Lemma \ref{lem:iP_4-tree}, we have
  $x_1,y_1\preceq _T w_1$, $x_2,z_1 \preceq_T w_2$ for distinct
  $w_1,w_2\in \child_T(u_1)$, where $u_1\coloneqq \lca_T(x_1,y_1,z_1,x_2)$,
  and $y_2,z_2\preceq _T w'_1$, $x_2,z_1 \preceq_T w'_2$ for distinct
  $w'_1,w'_2\in \child_T(u_2)$, where
  $u_2\coloneqq \lca_T(y_2,z_2,z_1,x_2)$, respectively. Since $u_1$ and
  $u_2$ are both located on the path from $z_1$ to the root of $T$, they
  must be comparable.  Next, we show that $u_1 = u_2$.  Assume first, for
  contradiction, $u_1\prec_T u_2$.  Then, by construction, we have
  $\lca_T(x_2,y_1) = u_1\prec_T u_2=\lca_T(x_2,y_2)$; a contradiction to
  $x_2y_2\in E(G)$. Similarly, $u_2\prec_T u_1$ yields a contradiction and
  thus, $u_1=u_2=v$ and, in particular, $w_2=w'_2$. It remains to show
  $w_1\neq w'_1$.  Assume, for contradiction, $w_1= w'_1$. Then
  $\lca_T(x_1,y_2) \preceq_T w_1\prec_T v=\lca_T(x_2,y_2)$; a contradiction
  to $x_2y_2\in E(G)$. Hence, $w_1$, $w_2$, and $w'_1$ are distinct
  children of $v$ in $T$, which completes the proof.  
\end{proof}
  
\section{Characterization of $n$-RBMGs}\label{sect:n}

\subsection{The General Case: Combination of 3-RBMGs} 

The key idea of characterizing $n$-RBMGs is to combine the information
contained in their 3-colored induced subgraphs $(G_{rst},\sigma_{rst})$,
cf.\ Def.\ \ref{def:restrict-3}.  Observation \ref{obs:induced_sub} shows
that $(G_{rst},\sigma_{rst})$ is a 3-colored induced subgraph of an
$n$-RBMG explained by $(T_{rst},\sigma_{rst})$, and thus a 3-RBMG.
Unfortunately, the converse of Observation \ref{obs:induced_sub} is in
general not true. Fig.~\ref{fig:supertree_counterex} shows a 4-colored
graph that is not a 4-RBMG while each of the four subgraphs induced by a
triplet of colors is a 3-RBMG.  We can, however, rephrase Observation
\ref{obs:induced_sub} in the following way:
\begin{fact}\label{fact:induced_sub}
  Let $(G,\sigma)$ be an $n$-RBMG for some $n\ge 3$. Then, $(T,\sigma)$
  explains $(G,\sigma)$ if and only if $(T_{rst},\sigma_{rst})$ explains
  $(G_{rst},\sigma_{rst})$ for all triplets of colors $r,s,t \in S$.
\end{fact}

\begin{lemma}\label{lem:edge_con} 
  Let $(T_e,\sigma)$ be the tree obtained by contracting an inner edge of
  $(T,\sigma)$. Then $G(T,\sigma)$ is a subgraph of $G(T_e,\sigma)$.
\end{lemma}
\begin{proof}
  Consider the edge $e=xy$ in $T$. By construction
  $(T_e,\sigma)\le (T,\sigma)$, thus Equ.\ \eqref{eq:bmg14} implies
  $N^+_{T}(u) \subseteq N^+_{T_e}(u)$ for all $u\in L\setminus L(T(y))$ and
  $N^+_{T}(v) = N^+_{T_e}(v)$ for all $v\in L(T(y))$.  It immediately
  follows $N^-_{T}(w) \subseteq N^-_{T_e}(w)$ for all $w\in L$. Hence,
  $E(G)\subseteq E(G(T_e))$. Since the leaf set remains unchanged,
  $L(T)=L(T_e)$, we conclude that the $G(T,\sigma)$ is a subgraph of
  $G(T_e,\sigma)$.  
\end{proof}

\begin{definition} 
  Let $(G,\sigma)$ be an $n$-RBMG. Then the \emph{tree set of $(G,\sigma)$}
  is the set
  $\mathcal{T}(G,\sigma)\coloneqq \{(T,\sigma) \mid \NEW{(T,\sigma)}  \text{ is least
    resolved and } G(T,\sigma)=(G,\sigma)\}$ of all leaf-colored trees
  explaining $(G,\sigma)$. Furthermore, we write
  $\mathcal{T}_{rst}(G,\sigma)$ for the set of all least resolved trees
  explaining the induced subgraphs $(G_{rst},\sigma_{rst})$.
\label{def:Grst-Trst}
\end{definition}
 
The counterexample in Fig.\ \ref{fig:supertree_counterex} shows that the
existence of a supertree for the tree set
$\mathcal{P}\coloneqq
\{T\in\mathcal{T}_{rst}(G_{rst},\sigma_{rst}),\,r,s,t\in S\}$ is not
sufficient for $(G,\sigma)$ to be an $n$-RBMG. We can only obtain a
substantially weaker result.
  
\begin{theorem}\label{thm:n-crbmg}
  A (not necessarily connected) undirected colored graph $(G,\sigma)$ is an
  $n$-RBMG if and only if (i) all induced subgraphs
  $(G_{rst},\sigma_{rst})$ on three colors are 3-RBMGs and (ii) there
  exists a supertree $(T,\sigma)$ of \NEW{the} tree set
  $\mathcal{P}\coloneqq \{T\in\mathcal{T}_{rst}(G,\sigma) \mid r,s,t\in
  S\}$, such that $G(T,\sigma)=(G,\sigma)$.
\end{theorem}
\begin{proof}
  Suppose $(G,\sigma)$ is an $n$-RBMG explained by some tree
  $(T,\sigma)$. Then Obs.\ \ref{fact:induced_sub} implies that
  $(G_{rst},\sigma_{rst})$ is a 3-RBMG that is explained by
  $(T_{rst},\sigma_{rst})$ for all triplets of colors $r,s,t \in S$.  By
  definition, each $(T_{rst},\sigma_{rst})$ is displayed by $(T,\sigma)$
  and thus, $(T,\sigma)$ is a supertree of these trees. Hence, the
  conditions are necessary.  Conversely, the existence of some supertree
  $(T,\sigma)$ with $G(T,\sigma)=(G,\sigma)$, i.e., $(T,\sigma)$ explains
  $(G,\sigma)$, clearly implies that $(G,\sigma)$ is an $n$-RBMG.  
\end{proof}

\subsection{Characterization of $n$-RBMGs that are cographs}

\begin{fact}
  Any undirected colored graph $(G,\sigma)$ is a cograph if and only if the
  corresponding $\sthin$-thin graph $(G/\sthin,\sigmasthin)$ is a cograph.
  \label{fact:ortho-cograph}
\end{fact} 
\begin{proof}
  It directly follows from Lemma \ref{lem:sthin} that $(G,\sigma)$ contains
  an induced $P_4$ if and only if $(G/\sthin,\sigmasthin)$ contains an
  induced $P_4$, which yields the statement.
  
\end{proof}
We note in passing that point-determining (thin) cographs recently have
attracted some attention in the literature \cite{Li:12}.

\begin{theorem}\label{thm:cographA}   
  Let $(G,\sigma)$ be an $n$-RBMG with $n\ge3$, and denote by
  $(G'_{rst},\sigma'_{rst})\coloneqq (G_{rst}/S,\sigma_{rst}/S)$ the
  $\sthin$-thin version of the 3-RBMG that is obtained by restricting
  $(G,\sigma)$ to the colors $r$, $s$, and $t$.  Then $(G,\sigma)$ is a
  cograph if and only if every \emph{3-colored} connected component of
  $(G'_{rst},\sigma'_{rst})$ is a 3-RBMG of Type \AX{(A)} for all triples
  of distinct colors $r,s,t$.
\end{theorem} 
\begin{proof}
  We first emphasize that distinct $\sthin$-classes of some $\sthin$-thin
  $n$-RBMG $(G,\sigma)$ may belong to the same $\sthin$-class in
  $(G'_{rst},\sigma'_{rst})\coloneqq (G_{rst}/S,\sigma_{rst}/S)$.
  Likewise, distinct $\sthin$-classes in $(G'_{rst},\sigma'_{rst})$ may
  belong to the same $\sthin$-classes in $(G'_{r's't'},\sigma'_{r's't'})$.
  In the following, the vertex set of a connected component
  $(G_{rst}^*,\sigma_{rst}^*)$ of $(G'_{rst},\sigma'_{rst})$ will be
  denoted by $L^*_{rst}$.
	
  Recall that $(G,\sigma)$ is a cograph if and only if all of its connected
  components are cographs.  Clearly, if $(G,\sigma)$ is an RBMG, then
  $(G_{rst},\sigma_{rst})$ and thus in particular
  $(G'_{rst},\sigma'_{rst})$ is a 3-RBMG (cf.\ Obs.\
  \ref{obs:induced_sub}) for any three distinct colors $r,s,t$. Moreover,
  since $(G,\sigma)$ is a cograph, it cannot contain an induced $P_4$, thus
  its induced subgraph on $L[r]\cup L[s] \cup L[t]$ and therefore also
  $(G'_{rst},\sigma'_{rst})$ do not contain an induced $P_4$ either, i.e.,
  each connected component of $(G'_{rst},\sigma'_{rst})$ is again a
  cograph. Hence, by Obs.\ \ref{fact:cograph}, each of the connected
  components with three colors must be of Type \AX{(A)}.
 
  Conversely, suppose that, for any distinct colors $r,s,t$, each connected
  component $(G_{rst}^*,\sigma_{rst}^*)$ of $(G'_{rst},\sigma'_{rst})$ is a
  3-RBMG of Type \AX{(A)}. Thus, $(G_{rst}^*,\sigma_{rst}^*)$ is again
  $\sthin$-thin. Obs.\ \ref{fact:cograph} implies that
  $(G_{rst}^*,\sigma_{rst}^*)$ must be a cograph. Hence, in particular,
  $(G,\sigma)$ cannot contain an induced $P_4$ on two or three colors (cf.\
  Obs.\ \ref{obs:induced_sub}). Assume, for contradiction, that
  $(G,\sigma)$ contains an induced $P_4$ $\langle abcd \rangle$ on four
  distinct colors, where $\sigma(a)=A$, $\sigma(b)=B$, $\sigma(c)=C$, and
  $\sigma(d)=D$. By abuse of notation, we will write $a$, $b$, $c$, and $d$
  for the $\sthin$-classes $[a]$, $[b]$, $[c]$, and $[d]$, respectively.
  Hence, Lemma \ref{lem:sthin} implies that $(G/\sthin,\sigmasthin)$
  contains the induced $P_4$ $\langle abcd \rangle$ on four distinct
  colors.  By Obs.\ \ref{obs:induced_sub}, $\langle abc\rangle$ must again
  be an induced $P_3$ in some connected component
  $(G_{ABC}^*,\sigma_{ABC}^*)$ of $(G'_{ABC},\sigma'_{ABC})$. Let
  $L^*_{ABC}$ be the leaf set of $(G_{ABC}^*,\sigma_{ABC}^*)$. Since
  $G_{ABC}^*\notin \mathscr{P}_3$ as a consequence of Lemma
  \ref{lem:charA}, $(G_{ABC}^*,\sigma_{ABC}^*)$ must contain at least four
  vertices, i.e., $|L^*_{ABC}|>3$. Hence, as $(G_{ABC}^*,\sigma_{ABC}^*)$
  is of Type \AX{(A)}, Lemma \ref{lem:charA} implies that
  $(G_{ABC}^*,\sigma_{ABC}^*)$ contains a hub-vertex.  By Property
  \AX{(A1)}, the hub-vertex must be connected to any other vertex in
  $(G_{ABC}^*,\sigma_{ABC}^*)$. Hence, neither $a$ nor $c$ can be the
  hub-vertex, since $ac\notin E(G)$.  By Cor.\ \ref{cor:hub-vertex}, the
  hub-vertex must be the only vertex of its color in
  $(G_{ABC}^*,\sigma_{ABC}^*)$.  Therefore, no vertex of color $A$ or $C$
  in $(G_{ABC}^*,\sigma_{ABC}^*)$ can be the hub-vertex. We therefore
  conclude that the hub-vertex must be $b$.
                        
  Applying the same argumentation to the connected component
  $(G_{BCD}^*,\sigma_{BCD}^*)$ of $(G'_{BCD},\sigma'_{BCD})$ that contains
  the induced $P_3$ $\langle bcd\rangle$, we can conclude that
  $|L^*_{BCD}|>3$ and $c$ must be the hub-vertex of
  $(G_{BCD}^*,\sigma_{BCD}^*)$.  Thus, in particular, it is the only vertex
  of color $C$ in $(G_{BCD}^*,\sigma_{BCD}^*)$.

  Moreover, since $|L^*_{ABC}|>3$ and $b$ is the only vertex of color $B$,
  there must be at least one other vertex of color $A$ or $C$ in the
  connected component $(G_{ABC}^*,\sigma_{ABC}^*)$. 

  Assume, for contradiction, that $L^*_{ABC}$ contains a leaf $c'$ of color
  $C$ such that $c'\neq c$ in $(G_{ABC}^*,\sigma_{ABC}^*)$. Since $b$ is
  the hub-vertex of $(G_{ABC}^*,\sigma_{ABC}^*)$, we have
  $bc'\in E(G_{ABC}^*)$ and thus $bc'\in E(G)$.  As $c$ is the only leaf of
  color $C$ in $(G_{BCD}^*,\sigma_{BCD}^*)$, we must ensure $c=c'$ in
  $G_{BCD}^*$. Thus, $cd\in E(G)$ implies $c'd\in E(G)$.  Since $c\neq c'$
  in $G_{ABC}^*$, $c$ must be adjacent to a vertex $\tilde a$ of color $A$
  that is not adjacent to $c'$, or \emph{vice versa}.
  Suppose first that $\tilde ac'\in E(G)$ and $\tilde a c\notin E(G)$. In
  this case $a=\tilde a$ in $G_{ABC}^*$ is possible.  Since $b$ is the
  hub-vertex of $(G_{ABC}^*,\sigma_{ABC}^*)$, we have $\tilde a b\in E(G)$.
  Consider $(G'_{ACD},\sigma'_{ACD})$. By construction, the vertices
  $\tilde a, c, c',d$ \NEW{are} 
  contained in the same connected component $(G_{ACD}^*,\sigma_{ACD}^*)$ of
	 Type \AX{(A)} in the  3-RBMG $(G'_{ACD},\sigma'_{ACD})$.
  Since $\tilde a c\notin E(G)$, the hub-vertex of
  $(G_{ACD}^*,\sigma_{ACD}^*)$ must be of color $D$.  Hence, as the
  hub-vertex is the only vertex of its color in
  $(G_{ACD}^*,\sigma_{ACD}^*)$, we can conclude that \NEW{$d$}
  is the hub-vertex. Therefore, $\tilde{a}d\in E(G)$, which in particular
  implies $\tilde a \neq a$ in $(G_{ACD}^*,\sigma_{ACD}^*)$ because
  $ad\notin E(G)$ by assumption. Hence,
  $\langle a b \tilde a d\rangle\in \mathscr{P}_4$ in
  $(G_{ABD}^*,\sigma_{ABD}^*)$; a contradiction.  Now assume
  $\tilde ac\in E(G)$ and $\tilde a c'\notin E(G)$. Again, analogous
  argumentation implies $\langle a b \tilde a d\rangle\in \mathscr{P}_4$ in
  $(G^*_{ABD},\sigma_{ABD}^*)$; again a contradiction.

  We therefore conclude that $L^*_{ABC}$ does not contain a leaf $c'\neq c$
  of color $C$ and hence, $L^*_{ABC}[C]=\{c\}$.  Together with
  $L^*_{ABC}[B]=\{b\}$ and $|L^*_{ABC}|>3$, this implies that there must
  exist a leaf $a'\neq a$ of color $A$ in $L^*_{ABC}$. We have
  $a'b\in E(G)$ because $b$ is the hub-vertex of
  $(G_{ABC}^*,\sigma_{ABC}^*)$. Hence, since $(G_{ABC}^*,\sigma_{ABC}^*)$
  is $\sthin$-thin, the neighborhoods of $a$ and $a'$ must differ.  The
  latter and $L^*_{ABC}[B] \cup L^*_{ABC}[C]=\{b,c\}$ implies
  $a'c\in E(G_{ABC}^*)$, i.e., $a'c\in E(G)$. One now easily checks that,
  by $\sthin$-thinness of $(G^*_{ABC},\sigma^*_{ABC})$, there cannot be a
  third vertex of color $A$ in $L^*_{ABC}$, thus
  $L^*_{ABC}=\{a,a',b,c\}$. Applying analogous arguments to
  $(G_{BCD}^*,\sigma_{BCD}^*)$ shows that $L^*_{BCD}=\{b,c,d,d'\}$, where
  $\sigma(d')=D$ and $bd', cd'\in E(G)$.  Moreover, we have
  $a'd\notin E(G)$, because otherwise $ad,bd\notin E(G)$ would imply that
  $\langle aba'd \rangle$ is an induced $P_4$ in
  $(G^*_{ABD},\sigma^*_{ABD})$; a contradiction since
  $(G^*_{ABD},\sigma^*_{ABD})$ is of Type \AX{(A)}.  Similarly,
  $ad'\notin E(G)$ as otherwise $(G^*_{ACD},\sigma^*_{ACD})$ would contain
  the induced $P_4$ $\langle ad'cd \rangle$.

  In summary, $\langle abcd \rangle$ is an induced $P_4$ in $(G,\sigma)$
  and there are vertices $a'$ and $d'$ such that $a'b,a'c,d'b,d'c\in E(G)$
  and $a'd,ad'\notin E(G)$, see Fig.\ \ref{fig:shortProof}(A). 

  \begin{figure}[t]
    \begin{tabular}{lcr}
      \begin{minipage}{0.54\textwidth}
        \begin{center}
          \includegraphics[width=\textwidth]{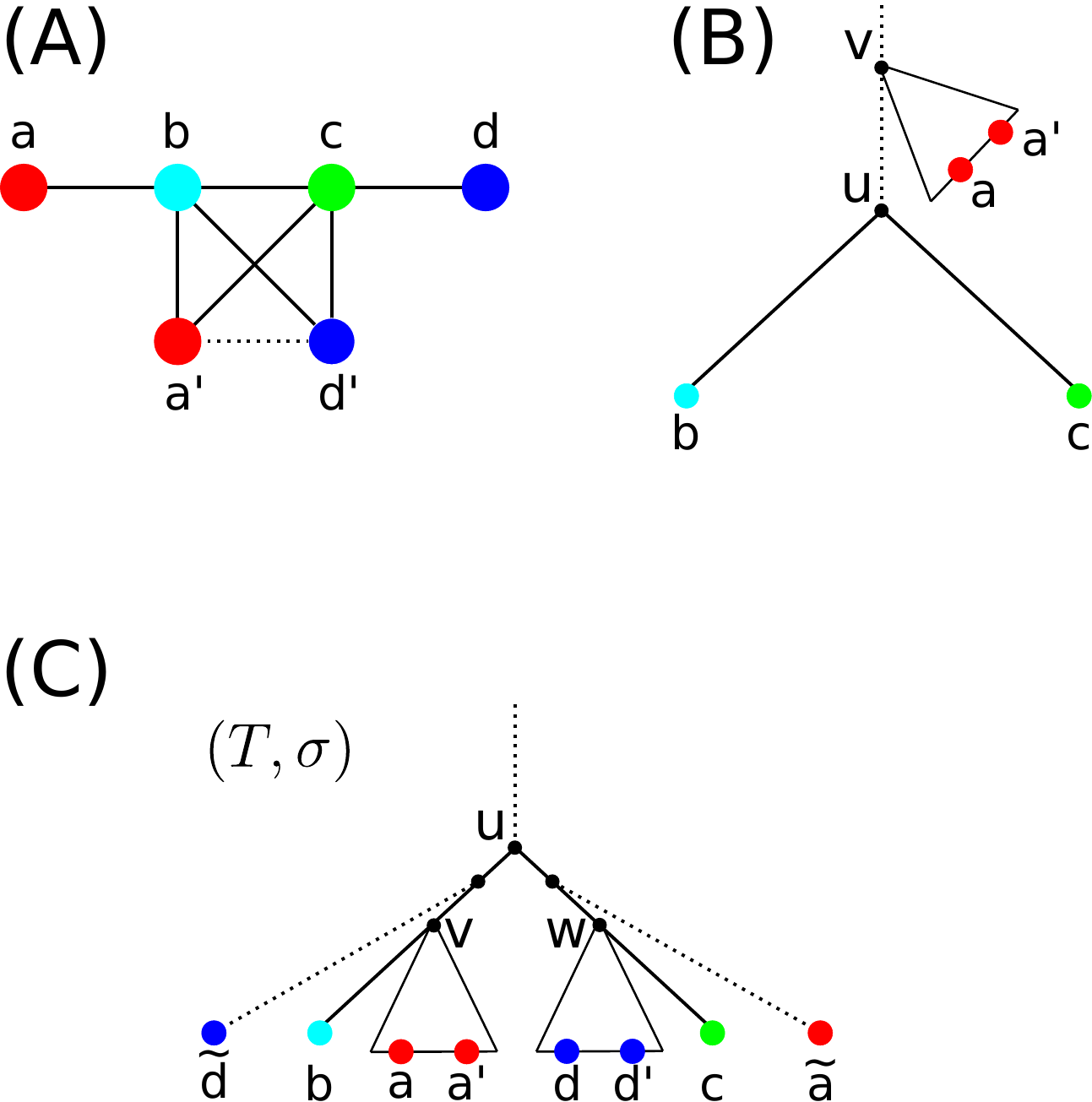}
        \end{center}
      \end{minipage}
      &    &
      \begin{minipage}{0.37\textwidth}
        \caption{Panel (A) shows the induced subgraph
          $(G[L'],\sigma_{|L'})$ with $L'=\{a,a',b,c,d,d'\}$ of
          $(G,\sigma)$ that is used in the proof of Theorem
          \ref{thm:cographA}.  In both trees, $u\coloneqq \lca(b,c)$ and
          $v\coloneqq \lca(a,b) = \lca(a',b)$.  Panel (B) shows a sketch of
          the subtree of $(T,\sigma)$ in case $v \succeq_T u$.  Panel (C)
          shows a sketch of a possible subtree of $(T,\sigma)$ in case
          $v \prec_T u$ and $w\coloneqq \lca(c,d) = \lca(c,d') \prec_T u$.
          In this representation of $(T,\sigma)$ we have
          $\lca(a,\widetilde d)\succ_T v$ and
          $\lca(\widetilde a,d)\succ_T w$.  However,
          $\lca(a,\widetilde d)\preceq_T v$ or
          $\lca(\widetilde a,d)\preceq_T w$ may be possible.  Dashed lines
          represent edges in $(G[L'],\sigma_{|L'})$ and paths in the trees
          that may or may not be present. Solid lines in the trees
          represent paths.}
        \label{fig:shortProof}
      \end{minipage}
    \end{tabular}
  \end{figure}
  
  Let $(T,\sigma)$ be a tree that explains $(G,\sigma)$ and
  $u\coloneqq \lca(b,c)$. The steps of the subsequent proof are illustrated
  in Fig.\ \ref{fig:shortProof}.  Since $ab \in E(G)$, we have
  $\lca(a,b) \preceq_T \lca(a',b)$. Similarly, $a'b\in E(G)$ implies
  $\lca(a',b) \preceq_T \lca(a,b)$.  Hence, we clearly have
  $v\coloneqq \lca(a,b) = \lca(a',b)$.  Note, $v$ and $u$ are both
  ancestors of $b$ and are thus located on the path from the root $\rho_T$
  to $b$. In other words, $v$ and $u$ are always comparable in $T$, i.e.,
  we have either $v\succeq_T u$ or $v\prec_T u$.  If $v\succeq_T u$, then
  $v=\lca(a,c)=\lca(a',c)$.  Since $a'c\in E(G)$, we have
  $v \preceq_T \lca(a',\widetilde c)$ and
  $v\preceq_T \lca(\widetilde a, c)$ for all $\widetilde{a} \in L[A]$ and
  all $\widetilde{c} \in L[C]$.  Together with $v=\lca(a,c)$, this implies
  $ac\in E(G)$; a contradiction.  Thus, only the case $v\prec_T u$ is
  possible and hence, $v$ must be located on the path from some child of
  $u$ to $b$ in $T$. Similarly, we have $w\coloneqq \lca(c,d) = \lca(c,d')$
  because $cd,cd'\in E(G)$. As $bd'\in E(G)$, $bd\notin E(G)$, we can apply
  an analogous argumentation as for $u$ and $v$ to conclude that only the
  case $w\prec_T u$ is possible. Thus $w$ must be located on the path from
  some child of $u$ to $c$ in $T$. In particular, we have $\lca(v,w)=u$ by
  definition of $u$ and therefore, $u=\lca(a,d)$. Since $ad\notin E(G)$, we
  have $u=\lca(a,d)\succ_T \lca(a,\widetilde d)$ for some $\widetilde a\in L[A]$ or
  $u=\lca(a,d)\succ_T \lca(\widetilde a,d)$ for some $\widetilde d\in L[D]$.  Assume that
  $u=\lca(a,d)\succ_T \lca(a,\widetilde d)$ for some
  $\widetilde d\in L[D]$. In this case,
  $\lca(a,\widetilde d)$ must be located on the path from some child $u'$
  of $u$ to $a$.  Thus, $u \succ_T u'\succ_T a,\widetilde d$ and by
  construction, $u'\succ_T b$.  Hence,
  $\lca(b,\widetilde d) \prec_T u = \lca(b,d')$ and thus
  $bd'\notin E(G)$; a contradiction.  Analogously,
  $u=\lca(a,d)\succ_T \lca(\widetilde a, d)$ would imply that
  $\lca(\widetilde a, d)$ is located on the path from some child of $u$ to
  $d$, which would contradict $a'c\in E(G)$.  Therefore, the tree
  $(T,\sigma)$ does not explain $(G,\sigma)$; a contradiction.

  Thus, if for any distinct colors $r,s,t$, each 3-colored connected
  component $(G_{rst}^*,\sigma_{rst}^*)$ of $(G'_{rst},\sigma'_{rst})$ is a
  3-RBMG of Type \AX{(A)}, then $(G,\sigma)$ must be a cograph.  
\end{proof}

\NEW{As a result of the previous proof, the induced subgraph shown in Panel
  (A) of Fig.\ \ref{fig:shortProof} is a forbidden subgraph of general
  RBMGs.
}

\subsection{Hierarchically Colored Cographs}

\begin{definition}\label{def:co-RBMG}
  A graph that is both a cograph and an RBMG is a \emph{co-RBMG}.
\end{definition}

\begin{lemma}
  Let $(G,\sigma)$ be a co-RBMG that is a explained by $(T,\sigma)$. Then,
  for any $v\in V^0(T)$ and each pair of distinct children
  $w_1,w_2\in \child_T(v)$, the sets $\sigma(L(T(w_1)))$ and
  $\sigma(L(T(w_2)))$ do not overlap.
\label{lem:child-overlap}
\end{lemma}
\begin{proof}
  Assume, for contradiction, that there exists some $v\in V^0(T)$ and
  distinct $w_1,w_2 \in \child(v)$ such that
  $r\in \sigma(L(T(w_1)))\cap \sigma(L(T(w_2)))$, $s$ is contained in
  $\sigma(L(T(w_1)))$ but not in $\sigma(L(T(w_2)))$, and $t$ is contained
  in $\sigma(L(T(w_2)))$ but not in $\sigma(L(T(w_1)))$ for three distinct
  colors $r,s,t$ in $(G,\sigma)$. Then, by Cor.\ \ref{cor:funfact}, there
  is a pair $x,y\in L(T(w_1))$ with $\sigma(x)=r$, $\sigma(y)=s$, and
  $xy\in E(G)$, as well as a pair $x',z\in L(T(w_2))$ with
  $\sigma(x')=r, \sigma(z)=t$ and $x'z\in E(G)$. Since
  $t\notin \sigma(L(T(w_1)))$ and $s\notin \sigma(L(T(w_2)))$, we have
  $\lca_T(y,z)=v\preceq_T \lca_T(y,z')$ for any $z'\in L[t]$ and
  $\lca_T(y,z)=v\preceq_T \lca_T(y',z)$ for any $y'\in L[s]$, hence
  $yz\in E(G)$. Moreover, as $\lca_T(z,x')\prec_T v=\lca_T(z,x)$ and
  $\lca_T(y,x)\prec_T v=\lca_T(y,x')$, the edges $xz$ and $x'y$ are not
  contained in $(G,\sigma)$. We therefore conclude that
  $\langle xyzx'\rangle$ is an induced $P_4$ in $(G,\sigma)$; a
  contradiction since $(G,\sigma)$ is a cograph.
  
\end{proof}

\begin{definition}
  Let $(H_1,\sigma_{H_1})$ and $(H_2,\sigma_{H_2})$ be two vertex-disjoint
  colored graphs.  Then
  $(H_1,\sigma_{H_1}) \join (H_2,\sigma_{H_2}) \coloneqq (H_1\join
  H_2,\sigma)$ and
  $(H_1,\sigma_{H_1}) \cupdot (H_2,\sigma_{H_2}) \coloneqq (H_1\cupdot
  H_2,\sigma)$ denotes their join and union, respectively, where
  $\sigma(x) = \sigma_{H_i}(x)$ for every $x\in V(H_i)$, $i\in\{1,2\}$.
	\label{def:coRBMG2}
\end{definition}

\begin{lemma}\label{lem:join} 
  Let $(G,\sigma)$ be a properly colored undirected graph such that either
  $(G,\sigma)=(H,\sigma_H) \join (H',\sigma_{H'})$ with
  $\sigma(V(H))\cap \sigma(V(H'))=\emptyset$ or
  $(G,\sigma)= (H,\sigma_H) \cupdot (H',\sigma_{H'})$ with
  $\sigma(V(H))\cap \sigma(V(H')) \in \{\sigma(V(H)),\sigma(V(H'))\}$,
  where $(H,\sigma_H)$ and $(H',\sigma_{H'})$ are disjoint RBMGs.  Then,
  $(G,\sigma)$ is an RBMG.

  Moreover, let $(H,\sigma_H)$ and $(H',\sigma_{H'})$ be explained by the
  trees $(T_H, \sigma_H)$ and $(T_{H'}, \sigma_{H'})$, respectively and let
  $(T,\sigma)$ be the tree obtained by joining $(T_H,\sigma_H)$ and
  $(T_{H'},\sigma_{H'})$ by a common root $\rho_T$. Then $(T,\sigma)$
  explains $(G,\sigma)$.
\end{lemma}
\begin{proof}
  Let $(T,\sigma)$ be the tree that is obtained by joining
  $(T_H, \sigma_H)$ and $(T_{H'}, \sigma_{H'})$ with roots $\rho_{H}$ and
  $\rho_{H'}$, respectively, under a common root $\rho_T$.  By construction,
  $\sigma(V(H)) = \sigma_H$ and $\sigma(V(H')) = \sigma_{H'}$.  We first show
  that $(T(\rho_H),\sigma_H)$ and $(T(\rho_{H'}),\sigma_{H'})$ explain
  $(H,\sigma_H)$ and $(H',\sigma_{H'})$, respectively. By construction, we
  have $\lca_T(x,y)=\lca_{T_H}(x,y)$ and $\lca_T(x,y)\prec_T \lca_T(x,z)$
  for any $x,y\in V(H)$ and $z\in V(H')$. It is therefore easy to see that
  $G(T(\rho_H),\sigma_H)=(H,\sigma_H)$. Analogous arguments show
  $G(T(\rho_{H'}),\sigma_{H'})=(H',\sigma_{H'})$.  Therefore, in order to
  show that $(T,\sigma)$ explains $(G,\sigma)$, it remains to show that all
  edges between vertices in $V(H)$ and $V(H')$ are identical in
  $(G,\sigma)$ and $G(T,\sigma)$.

  Suppose first $(G,\sigma)=(H,\sigma_H)\join (H',\sigma_{H'})$ with
  $\sigma(V(H))\cap \sigma(V(H'))=\emptyset$.  Thus we need to show
  $xy\in E(G(T,\sigma))$ for any $x\in L(T(\rho_H)), y\in L(T(\rho_{H'}))$.
  Since $\sigma(V(H))$ and $\sigma(V(H'))$ form a partition of
  $\sigma(V(G))$, we have $\lca_T(x,y)=\rho_T$ for any
  $x\in L[r], y\in L[s]$ with $r \in \sigma(V(H))$ and
  $s\in \sigma(V(H'))$.  Hence, $xy\in E(G(T,\sigma))$ for any
  $x\in L(T(\rho_H)), y\in L(T(\rho_{H'}))$ and therefore,
  $G(T,\sigma)=(G,\sigma)$.

  Now suppose that $(G,\sigma)= (H,\sigma_H) \cupdot (H',\sigma_{H'})$ with
  $\sigma(V(H))\cap \sigma(V(H')) \in \{\sigma(V(H)),\sigma(V(H'))\}$.
  Thus, we need to show $xy\notin E(G(T,\sigma))$ for any
  $x\in L(T(\rho_H)), y\in L(T(\rho_{H'}))$.  W.l.o.g.\ assume
  $\sigma(V(H'))\subseteq \sigma(V(H))$.  Hence, $(T(\rho_H),\sigma_H)$ is
  color-complete.  We can therefore apply Lemma \ref{lem:cl} to conclude
  that $xy\notin E(G(T,\sigma))$ for any $x\in V(H), y\in V(H')$. Hence,
  $G(T,\sigma)=(G,\sigma)$.
  
\end{proof}

\begin{definition}
  An undirected colored graph $(G,\sigma)$ is a \emph{hierarchically
    colored cograph (\hc-cograph)} if 
 \begin{itemize}
  \item[\AX{(K1)}] $(G,\sigma)=(K_1,\sigma)$, i.e., a colored vertex, or 
  \item[\AX{(K2)}] $(G,\sigma)= (H,\sigma_H) \join (H',\sigma_{H'})$ and
    $\sigma(V(H))\cap \sigma(V(H'))=\emptyset$, or 
  \item[\AX{(K3)}] $(G,\sigma)= (H,\sigma_H) \cupdot (H',\sigma_{H'})$ and
    $\sigma(V(H))\cap \sigma(V(H')) \in \{\sigma(V(H)),\sigma(V(H'))\}$,
\end{itemize}
where both $(H,\sigma_H)$ and $(H',\sigma_{H'})$ are \hc-cographs. 
For the color-constraints (cc) in \AX{(K2)} and \AX{(K3)}, we simply write
\AX{(K2cc)} and \AX{(K3cc)}, respectively.
\label{def:hc-cograph}
\end{definition}
Omitting the color-constraints reduces Def.\ \ref{def:hc-cograph} to Def.\
\ref{def:cograph}. Therefore we have
\begin{fact}\label{fact:hc-cographx}
  If $(G,\sigma)$ is an \hc-cograph, then $G$ is cograph.
\end{fact}

The recursive construction of an \hc-cograph $(G,\sigma)$ according to
Def.\ \ref{def:hc-cograph} immediately produces a binary \hc-cotree
$T^G_{\hc}$ corresponding to $(G,\sigma)$. The construction is essentially
the same as for the cotree of a cograph (cf.\ \citet[Section
3]{Corneil:81}): Each of its inner vertices is labeled by $1$ for a $\join$
operation and $0$ for a disjoint union $\cupdot$, depending on whether
\AX{(K2)} or \AX{(K3)} \NEW{is} used \NEW{in the} construction steps. We write
$t:V^0(T^G_{\hc})\to\{0,1\}$ for the labeling of the inner vertices. By
construction, $(T^G_{\hc},t)$ is a not necessarily discriminating cotree
for $G$, see \ref{fig:hc-cograph}. The relationships of \hc-cographs and
their corresponding \hc-cotrees are described at length in the main text.

\begin{lemma}\label{lem:hc-properties}
  Every \hc-cograph $(G,\sigma)$ is a properly colored cograph.
\end{lemma}
\begin{proof}
  Let $(G,\sigma)$ be an \hc-cograph.  In order to see that $(G,\sigma)$ is
  properly colored, observe that any edge $xy$ in $(G,\sigma)$ must be the
  result of some (possible preceding) join
  $(H,\sigma_H)\join (H',\sigma_{H'})$ during the recursive construction of
  $(G,\sigma)$ such that $x\in V(H)$ and $y\in V(H')$.  Condition (K2)
  implies that $\sigma(V(H))\cap \sigma(V(H'))=\emptyset$ and hence,
  $\sigma(x)\neq \sigma(y)$ for every edge $xy$ in $(G,\sigma)$.  
\end{proof}

\begin{theorem}
  A vertex labeled graph $(G,\sigma)$ is a co-RBMG if and only if it is an
  \hc-cograph.
\label{thm:hc-cograph}
\end{theorem} 
\begin{proof}
  Suppose that $(G,\sigma)$ with vertex set $V$ and edge set $E$ is a
  co-RBMG. We show by induction on $|V|$ that $(G,\sigma)$ is an
  \hc-cograph. This is trivially true in the base case $|V|=1$.

  For the induction step assume that any co-RBMG with less than $N$
  vertices is at the same time an \hc-cograph and consider a co-RBMG with
  $|V|=N$.  Since $(G,\sigma)$ is an $n$-RBMG, there exists a tree
  $(T,\sigma)$ with root $\rho_T$ that explains $(G,\sigma)$.  By Lemma
  \ref{lem:child-overlap}, none of the color sets $\sigma(L(T(v)))$ and
  $\sigma(L(T(w)))$ overlap for any two children $v,w\in \child(\rho_T)$.
  Moreover, Lemma \ref{lem:child-overlap} allows us to define a partition
  $\Pi$ of $\child(\rho_T)$ into classes $P_1,\dots,P_k$ such that each
  pair of vertices $v\in P_i$ and $w\in P_j$, $i\neq j$ satisfies
  $\sigma(L(T(v))) \cap \sigma(L(T(w))) = \emptyset$. Note that each $P_i$
  may contain distinct elements $v,w$ such that
  $\sigma(L(T(v))) \cap \sigma(L(T(w))) \in \{\emptyset, \sigma(L(T(v))),
  \sigma(L(T(w))) \}$. 

  First assume that the partition $\Pi$ of $\child(\rho_T)$ is trivial,
  i.e., it consists of a single class $P_1$. Since none of the sets
  $\sigma(L(T(v)))$ and $\sigma(L(T(w)))$ overlap for any $v,w\in P_1$,
  there is an element $w\in P_1$ such that $\sigma(L(T(w)))$ is
  inclusion-maximal, i.e., $\sigma(L(T(v))) \subseteq \sigma(L(T(w)))$ for
  all $v\in P_1=\child(\rho_T)$.  Let
  $L_{\neg w}= \bigcup_{v\in P_1\setminus w} L(T(v))$ and $L_w = L(T(w))$.
  Since $\rho_T$ is always color-complete, $w$ must be color-complete,
  i.e., $\sigma(L_w) = \sigma(V)$.  Hence, we have
  $\sigma(L_{\neg w}) \subseteq \sigma(L_w)$.

  We continue by showing that $(G[L_{\neg w}], \sigma_{|L_{\neg w}})$ and
  $(G[L_w],\sigma_{|L_{w}})$ are RBMGs that are explained by
  $(T_{|L_{\neg w}}, \sigma_{|L_{\neg w}})$ and
  $(T_{|L_{w}}, \sigma_{|L_{w}})$, respectively.  By construction, we have
  $\lca_T(x,y)=\lca_{T_{|L_{\neg w}}}(x,y)$ and
  $\lca_T(x,y)\preceq_T \lca_T(x,z)$ for any $x,y\in L_{\neg w}$ and
  $z\in L_w$.  It is therefore easy to see that
  $G(T_{|L_{\neg w}},\sigma_{|L_{\neg w}})=(G[L_{\neg w}],\sigma_{|L_{\neg
      w}})$. Analogous arguments show that
  $G(T_{|L_{w}},\sigma_{|L_{w}})=(G[L_{w}],\sigma_{|L_{w}})$.  Hence,
  $(G[L_{\neg w}], \sigma_{|L_{\neg w}})$ and $(G[L_w],\sigma_{|L_{w}})$
  are RBMGs.  Since any induced subgraph of a cograph is again a cograph,
  we can thus conclude that $(G[L_{\neg w}], \sigma_{|L_{\neg w}})$ and
  $(G[L_w],\sigma_{|L_{w}})$ are co-RBMGs.  Hence, by induction
  hypothesis, $(G[L_{\neg w}], \sigma_{|L_{\neg w}})$ and
  $(G[L_w],\sigma_{|L_{w}})$ are \hc-cographs.  Moreover, since $w$ is
  color-complete, we can apply Lemma \ref{lem:cl} to conclude that
  $(G,\sigma)= (G[L_{\neg w}], \sigma_{|L_{\neg w}}) \cupdot
  (G[L_w],\sigma_{|L_{w}})$ is the disjoint union of two \hc-cographs and
  in addition satisfies $\sigma(L_{\neg w}) \subseteq \sigma(L_w)$. Hence,
  $(G,\sigma)$ satisfies Property \AX{(K3)} and is therefore an
  \hc-cograph.

  Now assume that $\Pi$ is non-trivial, i.e., there are at least two
  classes $P_1,\dots,P_k$. Then, by construction, we have
  $\sigma(L(T(v))) \cap \sigma(L(T(w))) = \emptyset$ for all $v\in P_i$ and
  $w\in P_j$, $i\neq j$.  Let $L_i = \bigcup_{v\in P_i} L(T(v))$.  By
  construction, $\sigma(L_i)\cap \sigma(L_j) = \emptyset$ for all distinct
  $i,j$.  Hence, $\sigma(L_1), \dots, \sigma(L_k)$ form a partition of
  $\sigma(V)$.  Thus, we have $\lca_T(x,y)=\rho_T$ for any
  $x\in L[r], y\in L[s]$ with $r \in \sigma(L_i)$ and $s\in \sigma(L_j)$
  for distinct $i,j\in \{1,\dots,k\}$, which clearly implies $xy\in
  E(G)$. Hence, $G=G[L_1]\join G[L_2] \join \dots \join G[L_k]$.  Thus, by
  setting $H = \join_{i=1}^{k-1} G[L_i]$ and $H'=G[L_k]$, we obtain
  $(G,\sigma) = (H,\sigma_{|V(H)})\join (H',\sigma_{|V(H')})$.  We proceed
  to show that $(H,\sigma_{|V(H)})$ and $(H',\sigma_{|V(H')})$ are
  \hc-cographs.  Since any induced subgraph of a cograph is again a
  cograph, we can conclude that $H$ and $H'$ are cographs. Thus it remains
  to show that $(H,\sigma_{|V(H)})$ and $(H',\sigma_{|V(H')})$ are RBMGs.
  By similar arguments as in the case for one class $P_1$, one shows that
  $(T_{|V(H)}, \sigma_{|V(H)})$ and $(T_{|V(H')}, \sigma_{|V(H')})$ explain
  $(H,\sigma_{|V(H)})$ and $(H',\sigma_{|V(H')})$, respectively.  In
  summary, $(H,\sigma_{|V(H)})$ and $(H',\sigma_{|V(H')})$ are co-RBMGs and
  thus, by induction hypothesis, \hc-cographs. Since
  $\sigma(V(H)) = \bigcup_{i=1}^{k-1} \sigma(L_i)$,
  $\sigma(V(H')) = \sigma(L_k)$, and $\sigma(L_i), \sigma(L_j)$ are
  disjoint for distinct $i,j\in \{1,\dots,k\}$, we can conclude that
  $\sigma(V(H)) \cap \sigma(V(H')) = \emptyset$.  Hence,
  $(G,\sigma) = (H,\sigma_{|V(H)}) \join (H',\sigma_{|V(H')})$ satisfies
  Property (K2). Thus $(G,\sigma)$ is an \hc-cograph.

  Now suppose that $(G=(V,E),\sigma)$ is an \hc-cograph.  Lemma
  \ref{lem:hc-properties} implies that $G$ is a cograph.  In order to show
  that $(G,\sigma)$ is an RBMG, we proceed again by induction on $|V|$. The
  base case $|V|=1$ is trivially satisfied. For the induction hypothesis,
  assume that any \hc-cograph $(G,\sigma)$ with $|V|<N$ is an RBMG. Now let
  $(G,\sigma)$ with $|V|=N>1$ be an \hc-cograph. By definition of
  \hc-cographs and since $|V|>1$, there exist disjoint \hc-cographs
  $(H,\sigma_H)$ and $(H',\sigma_{H'})$ such that either (i)
  $(G,\sigma)=(H,\sigma_{H}) \join (H',\sigma_{H'})$ and
  $\sigma(V(H)) \cap \sigma(V(H'))=\emptyset$ or (ii)
  $(G,\sigma)=(H,\sigma_{H}) \cupdot (H',\sigma_{H'})$ and
  $\sigma(V(H)) \cap \sigma(V(H'))\in \{\sigma(V(H)), \sigma(V(H'))\}$.  By
  induction hypothesis, $(H,\sigma(V(H)))$ and $(H',\sigma(V(H')))$ are
  RBMGs. Hence, we can apply Lemma \ref{lem:join} to conclude that
  $(G,\sigma)$ an RBMG.  
\end{proof}

\begin{theorem}
  Every co-RBMG $(G,\sigma)$ is explained by its cotree
  $(T^G_{\hc},\sigma)$.
\label{thm:explainThc}
\end{theorem}
\begin{proof}
  We show by induction on $|V|$ that $(G,\sigma)$ is explained by
  $(T^G_{\hc},\sigma)$. This is trivially true for the base case $|V|=1$.
  Assume that any co-RBMG with less than $N$ vertices is explained by its
  cotree $(T^G_{\hc},\sigma)$.  Now let $(G=(V,E),\sigma)$ be a co-RBMG
  with $|V|= N$.  By Thm.\ \ref{thm:hc-cograph}, $(G,\sigma)$ is an
  \hc-cograph.  Thus, $(G,\sigma)= (H,\sigma_H) \star (H',\sigma_{H'})$,
  $\star\in \{\join, \cupdot\}$ such that \AX{(K2)} and \AX{(K3)}, resp.,
  are satisfied. Thm.\ \ref{thm:hc-cograph} implies that $(H,\sigma_H)$ and
  $(H',\sigma_{H'})$ are co-RBMGs.  By induction hypothesis, the co-RBMGs
  $(H,\sigma_H)$ and $(H',\sigma_{H'})$ are explained by their \hc-cotrees
  $(T^H_{\hc},\sigma_H)$ and $(T^{H'}_{\hc},\sigma_{H'})$, respectively.
  By construction, $(T^G_{\hc},\sigma)$ is the tree that is obtained by
  joining $(T^H_{\hc},\sigma_H)$ and $(T^{H'}_{\hc},\sigma_{H'})$ under a
  common root.  Lemma \ref{lem:join} now implies that the
  $(T^G_{\hc},\sigma)$ explains $(G,\sigma)$.
  
\end{proof}

\subsection{Recognition of \hc-cographs}
\label{sec:rec-hc-cogr}

Although (discriminating) cotrees can be constructed and cographs can be
recognized in linear time \cite{HABIB2005183,BCHP:08,CPS:85}, these results
cannot be applied directly to the construction of \hc-cotrees and the
recognition of \hc-cographs. The key problem is that whenever $(G,\sigma)$
comprises $k>2$ connected components, there are $2^{k-1}-1$ bipartitions
$G=G_1\cupdot G_2$. For each of them \AX{(K3cc)} needs to be checked, and
if it is satisfied, both $G_1$ and $G_2$ need to be tested for being
\hc-cotrees. In general, this incurs exponential effort. The following
results show, however, that it suffices to consider a single, carefully
chosen bipartition for each disconnected graph $(G,\sigma)$.

\begin{lemma}
  Every connected component of an \hc-cograph is an \hc-cograph.
\label{lem:connComp-hc}
\end{lemma}
\begin{proof}
  Let $(G,\sigma)$ be an \hc-cograph. By Thm.\ \ref{thm:hc-cograph},
  $(G,\sigma)$ is a co-RBMG.  Since each connected component of a cograph
  is again a cograph, each connected component of $(G,\sigma)$ must be a
  cograph.  In addition, Thm.\ \ref{thm:connected} implies that each
  connected component of $(G,\sigma)$ is \NEW{an} RBMG.  The latter two arguments
  imply that each connected component of $(G,\sigma)$ is a co-RBMG.
  Hence, Thm.\ \ref{thm:hc-cograph} implies that each connected component
  of $(G,\sigma)$ is an \hc-cograph.  
\end{proof}

\begin{lemma}
  Let $(G,\sigma)$ be a disconnected \hc-cograph with connected components
  $G_{1}=(V_1,E_1)$, \dots, $G_k=(V_k,E_k)$ and let $G_{\ell}$,
  $1\leq \ell \leq k$ be a connected component whose color set is minimal
  w.r.t.\ inclusion, i.e., there is no $i\in\{1,\dots,k\}$ with
  $i\neq \ell$ such that $\sigma(V_{i}) \subsetneq \sigma(V_{\ell})$.
  Denote by $G-G_{\ell}$ the graph obtained from $(G,\sigma)$ by
  deleting the connected component $G_{\ell}$. Then
  \begin{equation*}
    (G,\sigma) = (G_{\ell},\sigma_{|V_{\ell}})
    \cupdot (G-G_{\ell},\sigma_{|V\setminus V_{\ell}})
  \end{equation*}
  satisfies Property \AX{(K3)}.
  \label{lem:minimal}
\end{lemma}
\begin{proof}
  Let $(G=(V,E),\sigma)$ be a disconnected \hc-cograph with connected
  components
  $(G_{1} = (V_1,E_1),\sigma_1), \dots, (G_k=(V_k,E_k),\sigma_k)$ and put
  $\sigma_i\coloneqq \sigma_{|V_i}$, \NEW{$1\le i\le k$}.  Let
  $(G_{\ell},\sigma_{\ell})$, $1\leq \ell \leq k$ be a graph such that
  $\sigma(V_{\ell})$ is minimal w.r.t.\ inclusion. We write
  $G'=(V',E')\coloneqq G - G_{\ell}$ and
  $\sigma'\coloneqq \sigma_{|V\setminus V_{\ell}}$, thus
  $(G', \sigma') = (G- G_{\ell},\sigma_{|V\setminus V_{\ell}})$ and
  $V'=V\setminus V_{\ell}$.  In order to prove that
  $(G,\sigma) = (G_{\ell},\sigma_{\ell}) \cupdot (G', \sigma')$ satisfies
  \AX{(K3)}, we must show (i) that $(G_{\ell},\sigma_{\ell})$ and
  $(G', \sigma')$ are \hc-cographs and (ii) that
  $\sigma(V_{\ell}) \cap \sigma(V') \in \{\sigma(V_{\ell}), \sigma(V') \}$.

  (i) By Lemma \ref{lem:connComp-hc}, each connected component
  $(G_{1},\sigma_1), \dots, (G_k,\sigma_k)$ is an \hc-cograph. Thus, in
  particular, $(G_{\ell},\sigma_{\ell})$ is an \hc-cograph.  Furthermore,
  Thm.\ \ref{thm:hc-cograph} implies that each connected component
  $(G_{1},\sigma_1), \dots, (G_k,\sigma_k)$ is a co-RBMG.  \NEW{By Thm.\
    \ref{thm:connected}, the latter two arguments imply that
    $(G', \sigma')$ is a co-RBMG. Moreover, Thm.\ \ref{thm:connected}
    implies that there exists $1\le j\le k$ such that
    $\sigma(V_j)=\sigma(V)$ and $j\neq l$ since $\sigma(V_l)$ is minimal
    w.r.t.\ inclusion.}

  (ii) By Theorem \ref{thm:hc-cograph}, $(G,\sigma)$ is an RBMG. Applying
  Cor.\ \ref{cor:n-color}, we can conclude that $(G,\sigma)$ contains a
  connected component $(G^*,\sigma^*)$ with
  $\sigma(V(G^*))=\sigma(V)$. Since $\sigma(V_{\ell})$ is minimal w.r.t.\
  inclusion, we can w.l.o.g.\ assume that $(G^*,\sigma^*)$ is contained in
  $(G',\sigma')$. Hence,
  $\sigma(V_{\ell})\cap \sigma(V')=\sigma(V_{\ell})\cap
  \sigma(V)=\sigma(V_{\ell})$ and thus,
  $\sigma(V_{\ell})\subseteq \sigma(V')$, which implies that
  $(G,\sigma) = (G_{\ell},\sigma_{\ell}) \cupdot (G', \sigma')$ satisfies
  Property \AX{(K3cc)}.  
\end{proof}

The choice of the graph $G_{\ell}$ in Lemma \ref{lem:minimal} will in
general not be unique. As a consequence, there may be distinct cotrees
$(T_{\hc}^G,\sigma)$ that explain the same co-RBMG, see Fig.\
\ref{fig:hc-cograph} for an illustrative example.

While Thm.\ \ref{thm:cographA} allows the recognition of co-RBMGs in
polynomial time, it does not provide an explaining tree. The equivalence of
co-RBMGs and \hc-cographs together with Lemmas \ref{lem:join} and
\ref{lem:minimal} yields an alternative polynomial-time recognition
algorithm that is constructive in the sense that it explicitly provides a
tree explaining $(G,\sigma)$.

\begin{theorem}
  Let $(G,\sigma)$ be a properly colored undirected graph.  Then it can be
  decided in polynomial time whether $(G,\sigma)$ is a co-RBMG and, in the
  positive case, a tree $(T,\sigma)$ that explains $(G,\sigma)$ can be
  constructed in polynomial time.
  \label{thm:coRBMG-recog}
\end{theorem}

\begin{proof}
  Let $\overline G$ denote the complement of $G$.  Testing if $(G,\sigma)$
  is the join or disjoint union of graphs can clearly be done in polynomial
  time.

  Assume first that $(G,\sigma)$ is the join of graphs. In this case,
  $(\overline G,\sigma)$ decomposes into connected components
  $(\overline G_1,\sigma_1),\dots, (\overline G_k,\sigma_k)$, $k\geq2$,
  i.e.,
  $(\overline G,\sigma) = \bigcupdot_{i=1}^k (\overline G_i,\sigma_i)$.
  Therefore,
  $(G,\sigma) = (\overline{\overline G},\sigma) =
  \overline{\bigcupdot_{i=1}^k (\overline G_i,\sigma_i)}= \join_{i=1}^k
  (\overline{\overline G}_i,\sigma_i)=\join_{i=1}^k (G_i,\sigma_i)$, where
  none of the graphs $(G_i,\sigma_i)$ can be written as the join of two
  graphs and $k$ is maximal.  Note, we can ignore the parenthesis in the
  latter equation, since the $\join$ operation is associative. It follows
  from \AX{(K2cc)} that $(G,\sigma)$ is an \hc-cograph if and only if (1)
  all $(G_i,\sigma_i)$ are \hc-cographs and (2) all color sets
  $\sigma(V(G_i))$ are pairwise disjoint.  In this case, every binary
  tree with leaves $(G_1,\sigma_1),\dots, (G_k,\sigma_k)$ and all \NEW{inner} vertices
  labeled $1$ may appear in $(T^G_{hc},t)$.

  Now assume that $(G,\sigma)$ is the disjoint union of the connected
  graphs $(G_i,\sigma_i)$, $1\leq i\leq k$.  Let
  $G_{\ell} = (V_{\ell},E_{\ell})$, $1\leq \ell \leq k$ be a connected
  component such that $\sigma(V_{\ell})$ is minimal w.r.t.\ inclusion.
  Such a component can be clearly identified in polynomial time.  By Lemma
  \ref{lem:minimal},
  $(G,\sigma) = (G_{\ell},\sigma_{|V_{\ell}}) \cupdot
  (G-G_{\ell},\sigma_{|V\setminus V_{\ell}})$ satisfies \AX(K3), whenever
  $(G,\sigma)$ is a co-RBMG.  Again, Lemma \ref{lem:join} implies that the
  two trees that explains $(G_{\ell},\sigma_{|V_{\ell}})$ and
  $(G-G_{\ell},\sigma_{|V\setminus V_{\ell}})$, respectively, can then be
  joined under a common root in order to obtain a tree that explains
  $(G,\sigma)$. The effort is again polynomial.

  Finally, each of the latter steps must be repeated recursively on the
  connected components of either $G$ or $\overline G$. This either results
  in an \hc-cotree $(T^G_{\hc},t,\sigma)$ or we encounter a violation of
  \AX{(K2cc)} or \AX{(K3cc)} on the way. That is, we obtain a join
  decomposition such that the color sets $\sigma(V(G_i))$ are not
  pairwise disjoint, or a graph $G_{\ell}$ such that
  $(G_{\ell},\sigma_{|V_{\ell}}) \cupdot (G-G_{\ell},\sigma_{|V\setminus
    V_{\ell}})$ violates \AX{(K3cc)}. In either case, the recursion
  terminates and reports ``$(G,\sigma)$ is not an \hc-cograph''.  Since
  $T^G_{\hc}$ has $O(|V(G)|)$ vertices, the polynomial-time decomposition
  steps must be repeated at most $O(|V(G)|)$ times, resulting in an overall
  polynomial-time algorithm.
  
\end{proof}

Recall that $T_e$ denotes the tree that is obtained from $T$ by contraction
of the edge $e$ and $(T,t,\sigma)$ or $(T,t)$ is a cotree for $(G,\sigma)$
if $t(\lca(x, y))=1$ if and only if $xy \in E(G)$.

\begin{figure}[t]
  \begin{tabular}{lcr}
    \begin{minipage}{0.45\textwidth}
      \begin{center}  
        \includegraphics[width=\textwidth]{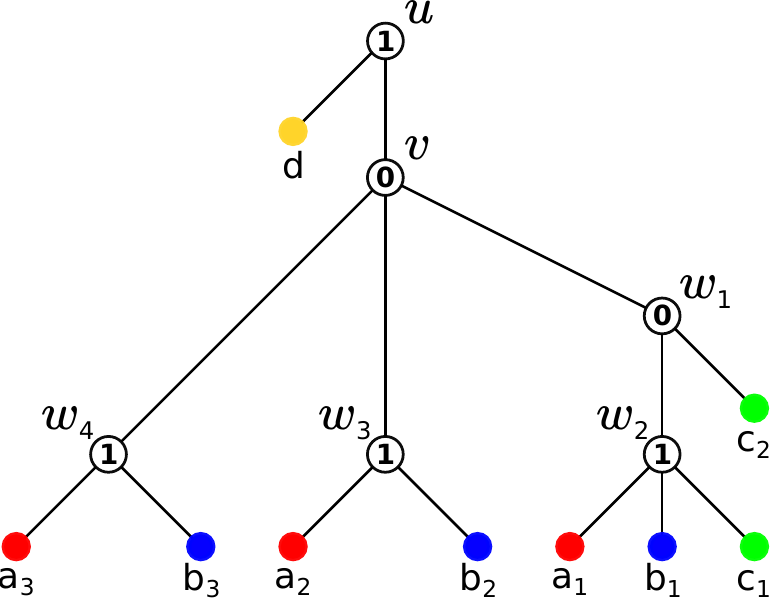}
      \end{center}
    \end{minipage}
    &    & 
    \begin{minipage}{0.45\textwidth}
      \caption{The cotree $(T,t,\sigma)$ explains a 4-colored co-RBMG
        $(G,\sigma)$. Cor.\ \ref{cor:1-0-edge} implies  that the edge
        $uv$ is redundant.  However, by Prop.\ \ref{prop:cograph2}, the
        tree $(T_e,t')$ is not a cotree for the co-RBMG $(G,\sigma)$ for
        all possible labelings $t':V^0\NEW{(T)}\to \{0,1\}$. Moreover, Lemma
        \ref{lem:lem2-2-necessary}(2) implies that the two inner edges $vw_3$
        and $vw_4 $ of $T$ are redundant. However, contracting both edges
        at the same time gives a tree $(T_1,\sigma)$ with
        $\parent(a_2)=\parent(a_3)$, thus $a_2$ and $a_3$ belong to the
        same $\sthin$-class in $G(T_1,\sigma)$. Hence, $(T_1,\sigma)$ does
        not explain $(G,\sigma)$ (cf.\ Lemma \ref{lem:01-edge}). Finally, one
        checks that the edge $vw_1$ is relevant because the edge $a_3c_2$
        is contained in $G(T_2,\sigma)$, where $(T_2,\sigma)$ is obtained
        from $(T,\sigma)$ by contraction of $vw_1$, but not in $(G,\sigma)$
        (cf.\ Lemma \ref{lem:lem2-2-necessary}).}
      \label{fig:cotree_00edge}
    \end{minipage}
  \end{tabular}
\end{figure}

Fig.\ \ref{fig:sameNeighbors} gives an example of a least resolved tree
$(T,\sigma)$ that explains a co-RBMG $(G,\sigma)$. However, one easily
verifies that $(T,\sigma)$ is not a refinement of the discriminating cotree
for $(G,\sigma)$.  Moreover, as shown in Fig.\ \ref{fig:hc-cograph}, there
might exist several cotrees that explain a given co-RBMG. By Prop.\
\ref{prop:cograph2}, the discriminating cotree for a co-RBMG $(G,\sigma)$
is unique. It does not necessarily explain $(G,\sigma)$, however. In order
to see this, consider the example in Fig.\ \ref{fig:cotree_00edge}, where
the edge $vw_1$ with $t(v)=t(w_1)=0$ cannot be contracted without violating
the property that the resulting tree still explains the underlying
co-RBMG.

To shed some light on the question how cotrees for a co-RBMG $(G,\sigma)$
and least resolved trees that explain $(G,\sigma)$ are related, we identify
the edges of a cotree for $(G,\sigma)$ that can be contracted.  To this
end, we show that the sufficient conditions in Lemma \ref{lem:lem2-2} are
also necessary for co-RBMGs.
\begin{lemma}
  Let $(T,t,\sigma)$ be a not necessarily binary cotree explaining the
  co-RBMG $(G,\sigma)$ that is also a cotree for $(G,\sigma)$ and let
  $e=uv$ be an inner edge of $T$. Then $(T_e,\sigma)$ explains $(G,\sigma)$
  if and only if either Property (1) or (2) from Lemma \ref{lem:lem2-2} is
  satisfied.
  \label{lem:lem2-2-necessary}
\end{lemma}
\begin{proof}
  By Lemma \ref{lem:lem2-2}, Properties (1) and (2) ensure that
  $(T_e,\sigma)$ explains $(G,\sigma)$.

  Conversely, suppose that $(T_e,\sigma)$ explains $(G,\sigma)$. Since
  $(G,\sigma)$ is a co-RBMG, it is an \hc-cograph by Theorem
  \ref{thm:hc-cograph} and thus, $(T,t,\sigma)$ is an \hc-cotree for
  $(G,\sigma)$. Let $\mathscr{G}_x\coloneqq (G,\sigma)[L(T(x))]$ for
  $x\in V(T)$. Clearly, $(T(u),t_{|L(T(u))},\sigma_{|L(T(u))})$ is an
  \hc-cotree for $\mathscr{G}_u$ and thus, $\mathscr{G}_u$ is an
  \hc-cograph. Hence, $\mathscr{G}_u$ can be written as
  $\mathscr{G}_v\star (H,\sigma_{H})$, where $\star\in \{\join,\cupdot\}$
  and $(H,\sigma_{H}) = \star_{x\in C}\mathscr{G}_x$ for
  $C\coloneqq \child_T(u)\setminus \{v\}$. Clearly, either $t(u)=1$ (in
  case $\star= \join$) or $t(u)=0$ (in case $\star= \cupdot$). As
  $(T(u),t_{|L(T(u))},\sigma_{|L(T(u))})$ is an \hc-cotree, we have either
  $\sigma(L(T(v'))) \cap \sigma(L(T(v))) = \emptyset$ for all
  $v'\in\child_T(u)$ (in case $\star= \join$) or
  $\sigma(L(T(v'))) \cap \sigma(L(T(v))) \in
  \{\sigma(L(T(v))),\sigma(L(T(v')))\}$ for all $v'\in\child_T(u)$ (in case
  $\star= \cupdot$).

  Thus, if $\star=\join$, then we immediately obtain Property (1).  In the
  second case, where $\mathscr{G}_u=\mathscr{G}_v\cupdot(H',\sigma_{H'})$,
  the color constraint \AX{(K3cc)} implies
  $\sigma(L(T(v)))\subseteq \sigma(L(T(v')))$ or
  $\sigma(L(T(v')))\subsetneq \sigma(L(T(v)))$ for any
  $v'\in\child_T(u)$. If the first case is true for all children of $u$ in
  $T$, we obtain Property (2.i) of Lemma \ref{lem:lem2-2}.  Thus, suppose
  there exists some vertex $v'\in\child_T(u)$ with
  $\sigma(L(T(v')))\subsetneq \sigma(L(T(v)))$. Assume, for contradiction,
  that there is a vertex $w\in\child_T(v)$ with
  $S_{w,\neg v'} \neq \emptyset$ and a color $s$ such that $s$ is contained
  in $\sigma(L(T(v')))$ but not in $\sigma(L(T(w)))$. Thus, in particular,
  there exists a vertex $b\preceq_T v'$ with $\sigma(b)=s$.  Moreover,
  there is vertex $a\preceq_T w$ with $\sigma(a)=r\in S_{w,\neg v'}$.
  Since $r\notin \sigma(L(T(v')))$, the leaf $a$ must be contained in the
  out-neighborhood of $b$ in $\G(T,\sigma)$. Since
  $\sigma(L(T(v')))\subsetneq \sigma(L(T(v)))$ and
  $s\notin\sigma(L(T(w)))$, there exists a vertex
  $b'\in L(T(v))\setminus L(T(w))$ with $\sigma(b')=s$.  Hence,
  $\lca(a,b')=v$.  Thus, $b$ is not contained in the out-neighborhood of
  $a$, i.e., $ab\notin E(G)$. However, if we contract $e=uv$, we obtain the
  new vertex $uv = \lca_{T_e}(a,b) $ in $T_e$.  Since
  $r\notin \sigma(L(T(v')))$ and $s\notin \sigma(L(T(w)))$, we immediately
  obtain $\lca_{T_e}(a,b) \preceq_{T_e} \lca_{T_e}(a,b')$ and
  $\lca_{T_e}(a,b) \preceq_{T_e} \lca_{T_e}(a',b)$ for all $a'$ of color
  $\sigma(a)$ and $b'$ of color $\sigma(b)$. Thus $ab \in E(G(T_e,\sigma))$
  and hence, $(T_e,\sigma)$ does not explain $(G,\sigma)$; a contradiction.
  
\end{proof}

As an immediate consequence of Lemma \ref{lem:lem2-2-necessary}(1) we
obtain
\begin{corollary}\label{cor:1-0-edge}
  Let $(T,t,\sigma)$ be a not necessarily binary cotree explaining the
  co-RBMG $(G,\sigma)$ that is also a cotree for $(G,\sigma)$.  If
  $t(u)=1$ for an inner edge $e=uv$ of $T$, then the tree $(T_e,\sigma)$
  explains $(G,\sigma)$.
\end{corollary}
Now, Cor.\ \ref{cor:1-0-edge} and Prop.\ \ref{prop:cograph2} imply
\begin{corollary}\label{cor:1-1-edge}
  Let $(T,t,\sigma)$ be a not necessarily binary cotree explaining the
  co-RBMG $(G,\sigma)$ and let $e=uv$ be an inner edge of $T$ with
  $t(u)=t(v)=1$.  Then the tree $(T_e,t_e,\sigma)$ explains $(G,\sigma)$
  and is a cotree for $(G,\sigma)$, where the vertex $w=uv$ obtained by
  contracting the edge $uv$ is labeled by $t_e(w)=1$ and $t_e(w')=t(w')$
  for all other vertices $w'\neq w$.
\end{corollary}

Thus, if $(T^G_{\hc},t,\sigma)$ is a least resolved tree that explains
$(G,\sigma)$, then it will not have any adjacent vertices labeled by
$1$. The situation is more complicated for $0$-labeled vertices. Fig.\
\ref{fig:cotree_00edge} shows that not all edges $xy$ in
$(T^G_{\hc},t,\sigma)$ with $t(x)=t(y)=0$ can be contracted.  However, we
obtain the following characterization, which is an immediate consequence of
Prop.\ \ref{prop:cograph2}, Lemma \ref{lem:lem2-2-necessary} and Cor.\
\ref{cor:1-1-edge}.
\begin{corollary}
  Let $(T,t,\sigma)$ be a not necessarily binary cotree explaining the
  co-RBMG $(G,\sigma)$ that is also a cotree for $(G,\sigma)$.  Let $e=uv$
  be an inner edge of $T$.  The following two statements are equivalent:
  \begin{enumerate}
  \item $(T_e,t_e,\sigma)$ explains $(G,\sigma)$ and is a cotree for
    $(G,\sigma)$, where the vertex $w=uv$ obtained by contracting the edge
    $uv$ is labeled by $t_e(w)=t(u)$ and $t_e(w')=t(w')$ for all other
    vertices $w'\neq w$,
  \item $t(u)=t(v)$ and, if $t(u)=0$, then $e$ satisfies Properties (1) and
    (2) in Lemma \ref{lem:lem2-2}.
  \end{enumerate}
  \label{cor:e-cograph}
\end{corollary}

If we apply Cor.\ \ref{cor:1-1-edge} and \ref{cor:e-cograph}, then Prop.\
\ref{prop:cograph2} implies that we always obtain a cotree
$(T_e,t_e,\sigma)$ for $(G,\sigma)$. Hence, we can repeatedly apply Cor.\
\ref{cor:1-1-edge} and \ref{cor:e-cograph} and conclude that the least
resolved tree $(T,t,\sigma)$ explaining $(G,\sigma)$ does neither contain
edges $xy$ with $t(x)=t(y)=1$ nor edges $xy$ with $t(x)=t(y)=0$ satisfying
Lemma \ref{lem:lem2-2}(1) and (2).  Moreover, Cor.\ \ref{cor:1-0-edge}
allows us to contract edges $xy$ with $t(x)=1\neq t(y)$.  In this case,
however, Prop.\ \ref{prop:cograph2} implies that $(T_e,t')$ is not a cotree
for the co-RBMG $(G,\sigma)$ for all possible labelings
$t':V^0\to \{0,1\}$.  Hence, the question arises, how often we can apply
Cor.\ \ref{cor:1-0-edge}. An answer is provided by the next result:

\begin{lemma}\label{lem:10-edge}
  Let $(T,t,\sigma)$ be a not necessarily binary cotree that explains the
  co-RBMG $(G,\sigma)$ and that is a cotree for $(G,\sigma)$.  Let $A$ be
  the set of all inner edges $e=uv$ of $T$ with $t(u)=1$.  Then,
  $(T_{B},\sigma)$ explains $(G,\sigma)$ for all $B\subseteq A$.
\end{lemma}
\begin{proof}
  Any edge $e=uv \in B$ is contracted to some vertex $u_e$ in
  $(T_{B},\sigma)$.

  Let $e=uv\in A$. By definition of $A$, we have $t(u)=1$.  Clearly,
  $(T(u),t_{|L(T(u))},\sigma_{|L(T(u))})$ is an \hc-cotree.  Hence,
  $\sigma(L(T(v_1)))\cap \sigma(L(T(v_2)))=\emptyset$ for any two distinct
  vertices $v_1,v_2\in\child_{T}(u)$.  Now assume that we have contracted
  $e$ to obtain $(T_e,\sigma)$.  By Cor.\ \ref{cor:1-0-edge},
  $(T_e,\sigma)$ explains $(G,\sigma)$.  Moreover suppose that there exists
  another edge $f = uv'\in A$ which corresponds to $u_ev'$ in
  $(T_e,\sigma)$.  In $(T,\sigma)$, we have
  $\sigma(L(T(v)))\cap \sigma(L(T(v')))=\emptyset$ which, in particular,
  implies $\sigma(L(T(v')))\cap \sigma(L(T(w)))=\emptyset$ for all
  $w\in \child_T(v)$.  In $(T_e,\sigma)$, the children of $u_e$ are now
  $\child_{T_e}(u_e) = (\child_T(u)\setminus \{v\})\cup \child_T(v)$.
  Thus, $\sigma(L(T(v')))\cap \sigma(L(T(v'')))=\emptyset$ for all
  $v''\in \child_{T_e}(u_e)$.  Lemma \ref{lem:lem2-2}(1) implies that $f$
  can be contracted to obtain the tree $(T_{ef},\sigma)$ that explains
  $(G,\sigma)$.  Repeated application of the latter arguments shows that
  all edges incident to vertex $u$ in $(T,\sigma)$ can be contracted.
	
  Finally, the contraction of the edges can be performed in a top-down
  fashion. In this case, the contraction of edges incident to $u$ does not
  influence the children of any vertex $u'$ that is incident to some edge
  $e'= u'v'$ having label $t(u')=1$.  That is, we can apply the latter
  arguments to all edges in $B$ independently, from which we conclude that
  $(T_{B},\sigma)$ explains $(G,\sigma)$ for all $B\subseteq A$.
  
\end{proof}

For the contraction of edges $xy$ with $t(x)=0\neq t(y)$, however, the
situation becomes more complicated.
\begin{lemma}\label{lem:01-edge}
  Let $(T,t,\sigma)$ be a not necessarily binary cotree that explains the
  co-RBMG $(G,\sigma)$ and is a cotree for $(G,\sigma)$. Moreover, let
  $u\in V^0(T)$ be an inner vertex with $t(u)=0$ and $A=\{e_1,\dots,e_k\}$
  be the set of all inner edges $e_i=uv_i\in E(T)$ with
  $v_i\in \child_T(u)$ such that $t(v_i)=1$ and $e_i$ is redundant in \NEW{$(T,\sigma)$}.
  Then, $(T_e,\sigma)$ explains $(G,\sigma)$ for all $e\in A$ but
  $(T_{B},\sigma)$ does not explain $(G,\sigma)$ for all $B\subseteq A$
  with $|B|\geq 2$.
\end{lemma}
\begin{proof}
  Let the edges $e=uv$ and $f=uv'$ be contained in $A$. Since $e$ and $f$
  are redundant, $(T_e,\sigma)$ and $(T_f,\sigma)$ both explain
  $(G,\sigma)$. It suffices to show that $(T_{ef},\sigma)$ does not explain
  $(G,\sigma)$.  Following the same argumentation as in the beginning of
  the proof of Lemma \ref{lem:lem2-2-necessary}, we conclude that
  $(T(v),t_{|L(T(v))},\sigma_{|L(T(v))})$ is an \hc-cotree. This and
  $t(v)=1$ implies $\sigma(L(T(w_i)))\cap \sigma(L(T(w_j)))=\emptyset$ for
  any two distinct vertices $w_i, w_j\in\child_T(v)$.  Hence,
  $\sigma(L(T(v)))$ is partitioned into the sets
  $\sigma(L(T(w_1))), \dots, \sigma(L(T(w_k)))$ with $w_i\in\child_T(v)$,
  $1\le i \le k$.  Analogously,
  $\sigma(L(T(w'_1))),\dots, \sigma(L(T(w'_m)))$ with
  $w'_i\in\child_T(v')$, $1\le i \le m$ forms a partition of
  $\sigma(L(T(v')))$.  Consider an arbitrary but fixed vertex
  $w\in \child_T(v)$.  Assume, for contradiction, that $(T_{ef},\sigma)$
  explains $(G,\sigma)$ and denote by $u_{ef}$ the inner vertex in $T_{ef}$
  that is obtained by contracting the edges $uv$ and $uv'$.  Since
  $(G,\sigma)$ is an \hc-cograph and $t(u)=0$, the color sets
  $\sigma(L(T(v)))$ and $\sigma(L(T(v')))$ are neither disjoint nor do they
  overlap.  As $(T_{ef},\sigma)=(T_{fe},\sigma)$, we can w.l.o.g.\ assume
  that $\sigma(L(T(v'))) \subseteq \sigma(L(T(v)))$.
		
  Now, let $w'\in \child_T(v')$ such that there is some $z'\in L(T(w'))$
  with $\sigma(z')=t \notin \sigma(L(T(w)))$. Let $x\in L(T(w))$ with
  $\sigma(x)=r$. Since $t(u)=0$ and $u=\lca_T(x,z')$, we have
  $xz'\notin E(G)$. However, as
  $\sigma(L(T(v')))\subseteq \sigma(L(T(v)))$, there is some child
  $\widetilde w\in \child_T(v)$ distinct from $w$ such that
  $t\in \sigma(L(T(\widetilde w)))$. Let
  $\widetilde z\in L(T(\widetilde w))$ with $\sigma(\widetilde z)=t$. Since
  $t(v)=1$, we have $x\widetilde z\in E(G)$.  As $(T_{ef},\sigma)$ explains
  $(G,\sigma)$, $x\widetilde z$ must be an edge in $G(T_{ef},
  \sigma)$. Hence,
  $\lca_{T_{ef}}(x,\widetilde z) = u_{ef}\preceq_{T_{ef}}
  \lca_{T_{ef}}(x,z'')$ and
  $\lca_{T_{ef}}(x,\widetilde z) = u_{ef}\preceq_{T_{ef}}
  \lca_{T_{ef}}(x'',\widetilde z)$ for all $x''\in L[\sigma(x)]$ and
  $z''\in L[t]$. However, by construction of $T_{ef}$, we have
  $\lca_{T_{ef}}(x,z') = \lca_{T_{ef}}(x,\widetilde z) = u_{ef}$.  Hence,
  if $r\notin \sigma(L(T(w')))$, then
  $x\widetilde z \in E(G(T_{ef},\sigma))$ implies that $x z'$ is an edge in
  $G(T_{ef},\sigma)$; a contradiction to $xz'\notin E(G)$.  Now assume that
  $r\in \sigma(L(T(w')))$.  Clearly, $v'$ must have at least one further
  child $w''$ with $s\in \sigma(L(T(w'')))$ and $s\notin \sigma(L(T(w')))$.
  In particular, $r,t\notin \sigma(L(T(w'')))$.  Since
  $\sigma(L(T(v'))) \subseteq \sigma(L(T(v)))$, there exists a leaf
  $y \in L(T(v))$ with $\sigma(y)=s$.  Now, we either have $y\preceq_T w$
  or $y\preceq_T \widetilde w$ or
  $y\preceq_T \hat w\in \child_T(v)\setminus \{w,\widetilde w\}$.  In any
  case, $t(v)=1$ implies that at least one of the edges $xy$ or
  $y\widetilde z$ must be contained in $G$.  Assume $xy\in E(G)$. Since
  $t(u)=0$, we have $xy''\notin E(G)$ for any $y''\in L(T(w'')) \cap L[s]$.
  On the other hand, as $r\notin \sigma(L(T(w'')))$, we can apply the
  preceding argumentation to infer from $xy\in E(G(T_{ef},\sigma))$ that
  $x y''$ is an edge in $G(T_{ef},\sigma)$; a contradiction.  Analogously,
  the existence of an edge $y\widetilde z$ in $G$ yields  a
  contradiction as well.  Hence, $(T_{ef},\sigma)$ does not explain
  $(G,\sigma)$.
  
\end{proof}

The previous results can finally be used to obtain a least resolved tree
$(T,\sigma)$ from a cotree $(T',t,\sigma)$ for a given co-RBMG $(G,\sigma)$
that also explains $(G,\sigma)$. Instead of checking all inner edges of
$T'$ for redundancy, Lemma \ref{lem:10-edge} can be applied to identify
promptly many redundant edges, which considerably reduces the number of
edges that need to be checked. This idea is implemented in Algorithm
\ref{alg:lr-tree}, which returns a least resolved tree that explains
$(G,\sigma)$. Lemma \ref{lem:01-edge}, however, suggests that this least
resolved tree is not necessarily unique for $(G,\sigma)$.

\begin{algorithm}[tbp]
  \caption{From Cotrees to Least Resolved Trees} 
  \label{alg:lr-tree}
  \begin{algorithmic}[1]
    \REQUIRE Leaf-labeled cotree $(T,t,\sigma)$ 
    \STATE $A\gets \emptyset$
    \FORALL {inner edges $e=uv$ with $t(u)=1$} \label{alg:start-for}
       \STATE $A\gets A\cup\{e\}$
    \ENDFOR \label{alg:end-for}
    \STATE $(T,\sigma) \gets (T_A,\sigma)$ \label{alg:contract}  
    \WHILE {$(T,\sigma)$ contains redundant inner edges $e=uv$ } 
    \label{alg:start-while2}
       \STATE Contract $e$ to obtain  $(T_e,\sigma)$  
       \STATE $(T,\sigma) \gets (T_e,\sigma)$
    \ENDWHILE \label{alg:end-while2}
    \STATE \textbf{return} $(T,\sigma)$
  \end{algorithmic}
\end{algorithm}

\begin{theorem}
Let $(G,\sigma)$ be a co-RBMG that is explained by $(T,t,\sigma)$ such that
$(T,t,\sigma)$ is also a cotree for $(G,\sigma)$.
Then,  Algorithm \ref{alg:lr-tree} returns a
least resolved tree that explains $(G,\sigma)$, in polynomial time. 
\label{thm:last}
\end{theorem}
\begin{proof}
  Lemma \ref{lem:10-edge} implies that all inner edges $e=uv$ with $t(u)=1$
  can be contracted, which is done in Line
  \ref{alg:start-for}-\ref{alg:contract}.  Afterwards, we check for all
  remaining inner edges $e=uv$ whether they are redundant or not, and, if
  so, contract them.  In summary, the algorithm is correct.  Clearly, all
  steps including the check for redundancy as in Lemma \ref{lem:lr} can be
  done in polynomial time.  
\end{proof}

\end{appendix} 
\end{document}